\documentclass[journal,final,twoside]{IEEEtran}

\usepackage{graphicx}
\usepackage{amsthm}
\usepackage{amsmath,amssymb}
\usepackage{array}
\usepackage{xcolor}
\newcolumntype{L}[1]{>{\raggedright\let\newline\\\arraybackslash\hspace{0pt}}m{#1}}
\usepackage{booktabs,multirow}
\usepackage[labelfont=bf,font=small]{caption}
\usepackage{subcaption}
\usepackage{placeins}
\usepackage{adjustbox}
\usepackage[inline]{enumitem}
\usepackage{tabularx,ragged2e}
\usepackage{float}
\usepackage{chngcntr} 
\usepackage[figuresleft]{rotating} 

\newcolumntype{Y}{>{\centering\arraybackslash}X}
\newcolumntype{S}{>{\hsize=.4\hsize}X}
\newcolumntype{s}{>{\raggedleft\arraybackslash\hsize=.3\hsize}X}
\newcolumntype{M}{>{\RaggedRight\arraybackslash\hspace{0pt}}X}
\usepackage{url}
\theoremstyle{definition}
\newtheorem{definition}{Definition}[section]
\newtheorem{proposition}{Proposition}[section]

\newtheorem{corollary}{Corollary}[section]

\usepackage{lineno,hyperref}
\modulolinenumbers[5]

\newcommand{\review}[1]{\textcolor{black}{#1}}
\newcommand{\reviewpar}{\color{black}}
\newcommand{\reviewminor}[1]{\textcolor{black}{#1}}
\newcommand{\reviewminorpar}{\color{black}}

\usepackage{todonotes}
\usepackage{colortbl}

\newlength{\foursubht}
\newsavebox{\foursubbox}

\usepackage[absolute,showboxes]{textpos}

\begin{document}
%
\title{\review{Optimization for Medical Image Segmentation: Theory and Practice when evaluating with Dice Score or Jaccard Index}}
%
%
%

\ifCLASSOPTIONpeerreview
\author{}
\else
\author{Tom~Eelbode,
        Jeroen~Bertels,
        Maxim~Berman,
        Dirk~Vandermeulen,~\IEEEmembership{Fellow,~IEEE,}
        Frederik~Maes,~\IEEEmembership{Senior Member,~IEEE,}
        Raf~Bisschops,
        and~Matthew~B.~Blaschko
\thanks{Copyright (c) 2019 IEEE. Personal use of this material is permitted. Permission from IEEE must be obtained for all other uses, in any current or future media, including reprinting/republishing this material for advertising or promotional purposes, creating new collective works, for resale or redistribution to servers or lists, or reuse of any copyrighted component of this work in other works.}
\thanks{Submitted for review on November 16, 2019.}
\thanks{Published on June 15, 2020.}
\thanks{T. Eelbode and J. Bertels contributed equally to this work.}
\thanks{T. Eelbode, J. Bertels, M. Berman, D. Vandermeulen,
F. Maes and M. Blaschko are with the Department of Electrical Engineering (ESAT),
Center for Processing Speech and Images, KU Leuven, Belgium.
E-mail: \url{tom.eelbode@kuleuven.be}}
\thanks{R. Bisschops is with the Department of Gastroenterology and Hepatology, UZ Leuven, Belgium.}
\thanks{This work is funded in part by Internal Funds KU Leuven (grant \# C24/18/047). The computational resources were partly provided by the Flemish Supercomputer Center (VSC). J.B. is part of NEXIS, a project that has received funding from the European Union's Horizon 2020 Research and Innovations Programme (grant \# 780026). R.B. is supported by FWO and Fujifilm. M.B.\ and M.B.B.\ acknowledge support from FWO (grant \# G0A2716N), an Amazon Research Award, an NVIDIA GPU grant, and the Facebook AI Research Partnership. This research received funding from the Flemish Government under the “Onderzoeksprogramma Artifici\"{e}le Intelligentie (AI) Vlaanderen” programme. The authors thank H. Willekens, C. Camps, C. Hassan, E. Coron, P. Bhandari, H. Neumann and O. Pech for their effort and collaboration.}}
\fi

\ifCLASSOPTIONpeerreview
\markboth{IEEE Transactions on Medical Imaging, VOL. X, NO. X, JUNE 2020}{}%
\else
\markboth{IEEE Transactions on Medical Imaging, VOL. X, NO. X, JUNE 2020}%
{T. Eelbode J. Bertels et al.}
\fi
%

\maketitle

\begin{abstract}
\reviewpar
In many medical imaging and classical computer vision tasks, the Dice score and Jaccard index are used to evaluate the segmentation performance. Despite the existence and great empirical success of metric-sensitive losses, i.e. relaxations of these metrics such as soft Dice, soft Jaccard and Lov\'{a}sz-Softmax, many researchers still use per-pixel losses, such as (weighted) cross-entropy to train CNNs for segmentation. Therefore, the target metric is in many cases not directly optimized. We investigate from a theoretical perspective, the relation within the group of metric-sensitive loss functions and question the existence of an optimal weighting scheme for weighted cross-entropy to optimize the Dice score and Jaccard index at test time. We find that the Dice score and Jaccard index approximate each other relatively and absolutely, but we find no such approximation for a weighted Hamming similarity. For the Tversky loss, the approximation gets monotonically worse when deviating from the trivial weight setting where soft Tversky equals soft Dice.
We verify these results empirically in an extensive validation on six medical segmentation tasks and can confirm that metric-sensitive losses are superior to cross-entropy based loss functions in case of evaluation with Dice Score or Jaccard Index. This further holds in a multi-class setting, and across different object sizes and foreground/background ratios. These results encourage a wider adoption of metric-sensitive loss functions for medical segmentation tasks where
the performance measure of interest is the Dice score or Jaccard index.

\end{abstract}


%
\IEEEpeerreviewmaketitle

\section{Introduction}
\IEEEPARstart{I}{n} medical image segmentation, the most commonly used performance metrics are the Dice score and Jaccard index~\cite{Kamnitsas2017,Ronneberger2015a,brats2018,isles2017,isles2018}. These metrics are used rather than pixel-wise accuracy because they are a more adequate indicator for the perceptual quality of a segmentation. This is also recognized by Zijdenbos et al. who state that the Dice score better reflects size and localization agreement for object segmentation \cite{Zijdenbos1994}. They also show that Dice score is a special case of the Kappa statistic in case the number of background voxels greatly outweighs that of the foreground voxels.

When we train a learning-based method for image segmentation, we are performing risk minimization. In this setting, it is important that we optimize a loss function during training which we will also use for evaluation during test time \cite{Vapnik:1995:NSL:211359}. However, in the MICCAI 2018 proceedings, 68 out of 96 learning-based segmentation papers used a loss function which does not directly optimize Dice score or Jaccard index even though evaluation was performed with one of these two metrics (more details in Appendix~\ref{sec:MICCAI2018methodology}). Pixel-wise (weighted) cross-entropy loss is still frequently used in these papers \cite{Kamnitsas2017,Chen2017}. Nevertheless, differentiable approximations for Dice score and Jaccard index have been proposed for training discriminative models with a gradient-based optimization algorithm, such as stochastic gradient descent (SGD). For example, the soft Dice \cite{Sudre2017} and soft Jaccard \cite{tarlow2012revisiting,nowozin2014optimal} are relaxations for their respective metrics and can be used to surrogate Dice score and Jaccard index during training. The Lov\'{a}sz-softmax is a more recent convex extension for the Jaccard index \cite{Berman2018a}.

This great variety of loss functions being used in the literature raises the question whether there is any rationale for choosing a specific group of loss functions over the others and if this choice has an influence on the quality of the predictions. In this work, we theoretically investigate the link between Dice score and Jaccard index and find that they approximate each other under risk minimization. Despite the existence of this approximation bound, we can prove that no such approximation exists for cross-entropy and that no weighting will allow it to surrogate Dice or Jaccard. Analogously, we look for an optimal weighting of false positives and false negatives in the Tversky loss and find that this is the case when the latter is equivalent to the Dice loss. We then empirically validate these theories on \review{six} different medical image segmentation tasks. We find that the performance when training with (weighted) cross-entropy is inferior to using any of the surrogates for Dice and Jaccard. However, there is no statistically significant difference between these surrogates.

This article is an extended version of \cite{Bertels2019a}.
\review{We continued the analysis on an extra dataset from the WMH 2017 challenge, in which the images contain multiple small objects with a very low total foreground/background ratio. Additionally, we ran all experiment on this dataset under two different network architectures: additional to U-Net (results shown as WM17) we used DeepMedic~\cite{Kamnitsas2017} (results shown as WM17\textsuperscript{DM}). Furthermore, these runs were now executed in a patch-wise setting due to the larger image sizes. More details can be found under Sec.~\ref{sec:empirical_setup} and the results are listed in Tables~\ref{tab:l1vsl2}, \ref{tab:results} \& \ref{tab:tversky_results}.}
We added an analysis of the Tversky loss (Eqs.~\eqref{eq:TverskySetDefinition} \& \eqref{eq:tversky}, Prop.~\ref{prop:TverskyDiceApproximation}, and experiments in Sec.~\ref{sec:experiments}),
a comparison between the $L^1$ and $L^2$-based soft Dice variants (Sec.~\ref{sec:soft_dice_variants}),
visualizations of segmentation results and qualitative discussion in Sec.~\ref{sec:experiments},
multi-class segmentation experiments on the BRATS dataset,
additional experiments analyzing the effect of object size on the POLYPS dataset in Sec.~\ref{sec:multiclass},
and extended discussions.

\section{Risk minimization with Dice and related similarities}\label{sec:2}

When performing discriminative training of machine learning methods, such as stochastic gradient descent for a CNN~\cite{Goodfellow-et-al-2016}, we are performing risk minimization.
To learn a mapping $f$ from an observed input $x$ to a hidden variable $y$, empirical risk minimization optimizes the expectation of a loss function over a finite training set:
{ 
\setlength{\abovedisplayskip}{3pt}
\setlength{\belowdisplayskip}{5pt}
\begin{align} \label{eq:empiricalRisk}
\arg\min_{f\in\mathcal{F}} \underbrace{\frac{1}{n} \sum_{i=1}^n \ell(f(x_i),y_i)}_{=:\hat{\mathcal{R}}(f)} , 
\end{align}
}%
where $\ell$ is a loss function and $\mathcal{F}$ is a function class of interest, e.g.\ the set of functions that can be represented by a neural network with a given topology. 
We will denote the \review{empirical} distribution arising from a sample $\mathcal{S}:=\{(x_i,y_i)\}_{1\leq i \leq n}$ of size $n$ as $P_n$, and we may equivalently denote $\hat{\mathcal{R}}(f) = \mathbb{E}_{(x,y)\sim P_n}[\ell(x,y)]$.  

In binary medical image segmentation, $y$ can be thought of as a set of pixels labeled as foreground. It is therefore well defined to consider set theoretic notions such as $y\cap \tilde{y}$ for two different segmentations.
This motivates the use of multiple set theoretic similarity measures  between the ground truth segmentation $y$ and the predicted segmentation $\tilde{y}$ including
the Dice score $D$, the Jaccard index $J$, the Hamming similarity $H$, the weighted Hamming similarity $H_\gamma$, and the Tversky index $T_{\alpha,\beta}$:
\begin{align}
D(y,\tilde{y}) :=& \frac{2 |y\cap \tilde{y}|}{|y| + |\tilde{y}|}, \label{eq:DiceSetDefinition} \\
J(y,\tilde{y}) :=& \frac{|y\cap \tilde{y}|}{|y\cup \tilde{y}|}, \label{eq:JaccardSetDefinition} \\
H(y,\tilde{y}) :=& 1 - \frac{|y\setminus \tilde{y}| + |\tilde{y}\setminus y|}{d}, \label{eq:HammingSetDefinition} \\
H_\gamma(y,\tilde{y}) :=& 1 - \gamma \frac{|y\setminus \tilde{y}|}{|y|} - (1 - \gamma) \frac{|\tilde{y}\setminus y|}{d - |y|} , \label{eq:hammingbound} \\
T_{\alpha,\beta}(y,\tilde{y}) :=& \frac{|y\cap \tilde{y}|}{|y\cap \tilde{y}| + \alpha |\tilde{y}\setminus y| + \beta | y\setminus \tilde{y}| } , \label{eq:TverskySetDefinition}
\end{align}
where $d$ denotes the number of pixels, $0\leq \gamma \leq 1$, $\alpha>0$, and $\beta>0$.    We note that all these similarities are between 0 and 1, that $H_{\gamma}$ generalizes $H$ with equality when $\gamma = \frac{|y|}{d}$, and that the Tversky index is equivalent to the Dice score when $\alpha = \beta = \frac{1}{2}$ and equivalent to Jaccard when $\alpha = \beta = 1$. A further important relationship is that between the Jaccard index and the Dice score.  It is well known that
\begin{align} \label{eq:DiceJaccardMonotonicRelationship}
J(y,\tilde{y}) = \frac{D(y,\tilde{y})}{2-D(y,\tilde{y})} \text{ and } D(y,\tilde{y}) = \frac{2 J(y,\tilde{y})}{1+J(y,\tilde{y})} .
\end{align}
Indeed, in the risk minimization framework for medical image segmentation, there are numerous examples where each of these measures are optimized \cite{pmid30557049,Sudre2017,TverskyMLMI2017}.

\subsection{Differentiable loss surrogates\label{sec:dif_loss_sur}}

The similarities in Equations \eqref{eq:DiceSetDefinition}-\eqref{eq:TverskySetDefinition} are over a discrete space of binary segmentations, while the output of a predictive function, e.g.\ given by a neural network, is continuous.  In order to perform gradient descent by backpropagation~\cite{Rumelhart:1988:LRB:65669.104451,Goodfellow-et-al-2016} on the empirical risk (Equation~\eqref{eq:empiricalRisk}), we must replace a discrete loss
\begin{align}
1-S : \{0,1\}^d \times \{0,1\}^d \rightarrow \mathbb{R}_+ ,
\end{align}
where $S$ is any of the similarities in \eqref{eq:DiceSetDefinition}-\eqref{eq:TverskySetDefinition},
with a surrogate that is differentiable in the predicted label
\begin{align}
\ell : \{0,1\}^d \times \mathbb{R}^d \rightarrow \mathbb{R}_+ .
\end{align}

For the Hamming similarity (Equation~\eqref{eq:HammingSetDefinition}), cross-entropy loss and other convex surrogates are statistically consistent~\cite{Bartlett2016JASA}.
To optimize the weighted Hamming similarity, one may employ weighted cross-entropy \cite{Ling2010}.  In this section, we will encode sets by binary vectors $y,\tilde{y}\in \{0,1\}^d$, allowing us to write the weighted cross entropy loss as:
\begin{equation}
    \Delta_{H_{\gamma}}(y,\tilde{y})=-\gamma \langle y ,  \log(\tilde{y}) \rangle - (1-\gamma) \langle \mathbf{1}-y, \log(\mathbf{1}-\tilde{y}) \rangle ,
\end{equation}
where $\Delta_{S}$ denotes a differentiable surrogate for a given similarity $S$, $\langle \cdot, \cdot \rangle$ denotes the canonical inner product, the logarithm is taken element-wise over the vector, and $\mathbf{1}$ denotes a $d$-dimensional vector of ones.  To make the surrogate differentiable, one relaxes $\tilde{y}$ to take continuous values $\tilde{y}\in [0,1]^d$, where unbounded scores are constrained to lie in $[0,1]$ via a softmax operation \cite[Equation~(4.1)]{Goodfellow-et-al-2016}.

Similarly, differentiable surrogates have been proposed for the Dice score (e.g.\ soft Dice \cite{MilletariNA16,Sudre2017}), Jaccard index (e.g.\ soft Jaccard \cite{SoftJaccard2016} \review{or} Lov\'{a}sz-softmax \cite{Berman2018a}), and Tversky index (e.g.\ soft Tversky \cite{TverskyMLMI2017}).

The soft Dice loss \cite{MilletariNA16,Sudre2017} generalizes the Dice loss by noting that we have the identity $|y\cap \tilde{y}| = \langle y, \tilde{y}\rangle$.  Furthermore, $|y| = |y\cap y| = \langle y, y\rangle$.  We may now write the Dice loss as
\begin{align}
    \Delta_{D}(y,\tilde{y}) =& 1-\frac{2|y\cap\tilde{y}|}{|y|+|\tilde{y}|} \\
    =& 1-\frac{2 \langle y, \tilde{y} \rangle}{\langle y, y\rangle + \langle \tilde{y},\tilde{y}\rangle}.
    \label{eq:diceloss}
\end{align}
The soft Dice loss is simply the relaxation where $\tilde{y}\in [0,1]^d$.  We note that multiple generalizations are possible, e.g.\ by replacing $\langle \tilde{y},\tilde{y} \rangle = \|\tilde{y}\|_2^2$ with $\|\tilde{y}\|_p^p$ for any $L^p$ norm, and the statistical and optimization properties of the resulting surrogate are dependent on these choices.  We will see below that $\|\tilde{y}\|_1$ is used in related surrogate definitions.

The soft Jaccard loss \cite{SoftJaccard2016} is defined using similar ideas to the soft Dice loss:
\begin{align}
\Delta_{J}(y,\tilde{y}) = 1 - \frac{\langle y,\tilde{y} \rangle}{\|y\|_1 + \|\tilde{y}\|_1 - \langle y, \tilde{y} \rangle} ,
\end{align}
although as above, alternative relaxations of the discrete loss are possible, e.g.\ by replacing the $L^1$ norm with a squared $L^2$ norm.

The Lov\'{a}sz softmax is a strategy for constructing loss surrogates for submodular losses \cite{Berman2018a}.  The Jaccard loss has been proven to be submodular \cite[Proposition~11]{Yu2018a}, while the Dice loss is not \cite[Proposition 6]{Yu2016a}.  As with other surrogates described in this section, the Lov\'{a}sz softmax first projects pixel scores into $[0,1]$ using a softmax operation, and then computes the Lov\'{a}sz extension \cite{Lovasz1983} of the submodular loss function, which is convex and differentiable almost everywhere, giving it favorable properties as a loss surrogate.

We finally describe the soft Tversky loss function:
\begin{equation}\label{eq:tversky}
\begin{aligned}
    \Delta_{T_{\alpha,\beta}}(y,\tilde{y}) = 1- \frac{\langle y, \tilde{y} \rangle}{\langle y, \tilde{y} \rangle + \alpha \langle \mathbf{1}-y, \tilde{y} \rangle + \beta \langle y, \mathbf{1}-\tilde{y} \rangle} \\
    = 1- \frac{\langle y, \tilde{y} \rangle}{ (1-\alpha-\beta)\langle y, \tilde{y} \rangle +  \alpha \| \tilde{y} \|_1  +  \beta \| y \|_1  }.
\end{aligned}
\end{equation}


Next, we discuss the absolute and relative approximations between Dice, Jaccard, and Tversky, and inspect the existence of an approximation through a weighted Hamming similarity.

\subsection{Approximation bounds}\label{sec:ApproximationBounds}

\begin{definition}[Absolute approximation]
A similarity $S$ is absolutely approximated by $\tilde{S}$ with error $\varepsilon \geq 0$ if the following holds for all $y$ and $\tilde{y}$:
\begin{align}
    | S(y,\tilde{y}) - \tilde{S}(y,\tilde{y}) | \leq \varepsilon .
\end{align}
\end{definition}
\begin{definition}[Relative approximation]
A similarity $S$ is relatively approximated by $\tilde{S}$ with error $\varepsilon \geq 0$ if the following holds for all $y$ and $\tilde{y}$:
\begin{align}
    \frac{\tilde{S}(y,\tilde{y})}{1+\varepsilon} \leq S(y,\tilde{y}) \leq \tilde{S}(y,\tilde{y})(1+\varepsilon) .
\end{align}
\end{definition}
We note that both notions of approximation are symmetric in $S$ and $\tilde{S}$.

\begin{proposition}\label{thm:DiceJaccardApproximation}
$J$ and $D$ approximate each other with relative error of $1$ and absolute error of $3 - 2\sqrt{2}=0.17157\dots$. \end{proposition}
\begin{proof}
The relative error between $J$ and $D$ is given by (cf.\ Equation~\eqref{eq:DiceJaccardMonotonicRelationship})
\begin{align}
    \min_{\varepsilon\geq 0} \varepsilon ,\
    \text{s.t. } 
    x \leq \frac{x}{2-x} (1+\varepsilon),  \ \forall \ 0\leq x \leq 1. \\
    x \leq \frac{x}{2-x}(1+\varepsilon) 
    \implies 1-x \leq \varepsilon \implies \varepsilon = 1 .
\end{align}

The absolute error between $J$ and $D$ is given by 
\begin{align}
   \varepsilon = \sup_{0\leq x \leq 1} \left| x - \frac{x}{2-x} \right| = 3 - 2\sqrt{2} ,
\end{align}
which can be verified straightforwardly by first order conditions:
\begin{align}
\frac{\partial}{\partial x} \left( x - \frac{x}{2-x} \right) 
= 0 \implies & (2-x)^2 - 2 = 0 \\ 
\implies & x = 
2 - \sqrt{2}. 
\end{align}
\end{proof}

\begin{proposition}
$D$ and $H_{\gamma}$ (where $\gamma$ is chosen to minimize the approximation factor between $D$ and $H_{\gamma}$) do not relatively approximate each other,  and absolutely approximate each other with an error of $1$. We note that the absolute error bound is trivial as $D$ and $H_{\gamma}$ are both similarities in the range $[0,1]$.
\end{proposition}
\begin{proof}
For relative error, consider the case that $|y\setminus \tilde{y}|=0$, $|\tilde{y}\setminus y| = \alpha d$, and $|y \cap \tilde{y}| = \alpha^2 d$ for some $0\leq \alpha<\frac{\sqrt{5}-1}{2}$ and $d$ the number of pixels:
\begin{align}
    \inf_{\gamma} \sup_{y, \tilde{y}}\  & 1 - \gamma \frac{|y \setminus \tilde{y}|}{|y|} - (1-\gamma)\frac{|\tilde{y} \setminus y|}{d-|y|} \nonumber \\ & - \frac{2|y \cap \tilde{y}|}{|y \triangle \tilde{y}| +  2|y \cap \tilde{y}|} (1+\varepsilon) \leq 0  \\
    \implies \sup_{0\leq\alpha<\frac{\sqrt{5}-1}{2}} \  & 1 - \frac{\alpha }{1- \alpha^2 } -  \frac{2\alpha^2 }{\alpha  + 2\alpha^2 } (1+\varepsilon) \leq 0 .
\end{align}
If we let $\alpha \rightarrow 0$, it must be the case that $\varepsilon \rightarrow \infty$.

To show that the absolute approximation error is 1, we similarly take
\begin{align}
    \lim_{\alpha\rightarrow 0} 1 - \frac{\alpha }{1- \alpha^2 } -  \frac{2\alpha }{1  + 2\alpha} = 1. \rlap{$\qquad \qed$}
\end{align}
\end{proof}

\begin{corollary}\label{thm:DiceHammingNoApproximation}
$D$ and $H$ do not relatively approximate each other, and absolutely approximate each other with an error of $1$.
\end{corollary}
From these bounds, we see that a (weighted) binary loss can be an arbitrarily bad approximation for Dice when segmenting small objects, while the Jaccard loss gives multiplicative and additive approximation guarantees. Furthermore,  Eq.~\eqref{eq:DiceJaccardMonotonicRelationship} implies that
\begin{align}
&1-D(y,\tilde{y}) \leq 1-J(y,\tilde{y}) \implies \\
&\mathbb{E}_{(x,y)\sim P_n}[1-D(y,f(x))] \leq \mathbb{E}_{(x,y)\sim P_n}[1-J(y,f(x))]
\end{align}
and optimization with risk computed with the Jaccard loss minimizes an upper bound on risk computed with the Dice loss.  Similarly setting $\varphi(x) =  2x/(1+x)$, by application of Jensen's inequality we arrive at
\begin{align}
\mathbb{E}_{(x,y)\sim P_n}[1-J(y,f(x))] = \qquad \qquad \qquad \qquad  \\
\mathbb{E}_{(x,y)\sim P_n}[\varphi(1-D(y,f(x)))] \nonumber \\
\leq  \varphi(\mathbb{E}_{(x,y)\sim P_n}[1-D(y,f(x))])
\end{align}
and optimizing the Dice loss minimizes an upper bound on the Jaccard loss as $\varphi$ is a monotonic function over $[0,1]$.

\begin{proposition}\label{prop:TverskyDiceApproximation}
The Tversky index approximates the Dice score with an absolute error of $ \max\left( \left| \frac{\sqrt{2\alpha}-1}{\sqrt{2\alpha}+1}\right|, \left| \frac{\sqrt{2\beta}-1}{\sqrt{2\beta}+1}\right| \right)$ and a relative error of $\max(2\alpha,2\beta,0.5 \alpha^{-1},0.5 \beta^{-1})-1$.
\end{proposition}
\begin{proof}
The absolute error is given by
\begin{align}
   \sup_{y,\tilde{y}} |D(y,\tilde{y}) - T_{\alpha,\beta}(y,\tilde{y})| = \\ \sup_{y,\tilde{y}} \left| \frac{2 |y\cap \tilde{y}|}{|y| + |\tilde{y}|} - \frac{|y\cap \tilde{y}|}{|y\cap \tilde{y}| + \alpha |\tilde{y}\setminus y| + \beta | y\setminus \tilde{y}| } \right| = \\
   \sup_{a,b,c\geq 0,\ a+b+c = 1} \left| \frac{b}{b+\alpha a + \beta c} - \frac{b}{b + 0.5 a + 0.5 c} \right|. 
\end{align}
If $\alpha=\beta$, we can simplify this to an equation of one variable
\begin{align} \label{eq:TverskyDiceAbsApprox_oneVariableEq}
\sup_{0\leq b\leq 1} \left| \frac{b}{b+\alpha(1-b)} - \frac{b}{b+0.5(1-b)}\right|.
\end{align}
From first order conditions, this is maximized when $b=\frac{\alpha - \sqrt{2\alpha }}{\alpha-2}$, and the resulting error is therefore $ \left| \frac{\sqrt{2\alpha}-1}{\sqrt{2\alpha}+1}\right|$.

For the case that $\alpha\neq \beta$, we note that the directional derivative in the direction $a-c$ has constant sign (i.e.\ it is monotonic), indicating that at the maximum over the simplex, $a=0$ or $c=0$.  This indicates that the problem simplifies to taking the maximum over the special cases that $a=0$ or $c=0$, which in both cases simplifies to Equation~\eqref{eq:TverskyDiceAbsApprox_oneVariableEq}, and the absolute approximation error is $\max\left( \left| \frac{\sqrt{2\alpha}-1}{\sqrt{2\alpha}+1}\right| ,  \left| \frac{\sqrt{2\beta}-1}{\sqrt{2\beta}+1}\right| \right)$.

The relative error can similarly be computed as
\begin{align}
1+\varepsilon=& \max\left( \sup_{y,\tilde{y}} \frac{T_{\alpha,\beta}(y,\tilde{y})}{D(y,\tilde{y})}, \sup_{y,\tilde{y}} \frac{D(y,\tilde{y})}{T_{\alpha,\beta}(y,\tilde{y}} \right) \\
=& \max\left(
\sup_{a,b,c} \frac{b+0.5a +0.5c}{b+\alpha a + \beta c} , \sup_{a,b,c} \frac{b + \alpha a + \beta c}{b +0.5 a + 0.5 c}
\right)  ,
\end{align}
where in each $\sup$, $0<a+b+c\leq 1$ and $a,b,c\geq 0$.
The first $\sup$ is maximized by determining $\max( 0.5 \alpha^{-1}, 0.5 \beta^{-1})$ and setting $a = 1$ or $c=1$, respectively.  The second $\sup$ is similarly maximized by determining $\max(2\alpha, 2\beta)$ and setting either $a=1$ or $c=1$, respectively.  The approximation error can then be determined by maximizing over all 4 cases.
\end{proof}
The absolute and relative errors are shown in Fig.~\ref{fig:TverskyAbsRelErrors}.  The Tversky index does give non-trivial approximation errors for the Dice score for both absolute and relative error.  However, in both cases, the error bounds are increasing as $\alpha$ and $\beta$ deviate from 0.5.  This may provide the ability to trade off the performance on the Dice measure with that of other quality measures, but optimal performance on the Dice measure will be achieved at $\alpha=\beta=0.5$, discounting eventual small sample effects or biases introduced by differentiable surrogates.  Furthermore, if a trade-off between Dice and other quality measures is desired, one may simply take a weighted combination of loss functions for each of the desired metrics, directly encoding the relative importance of each.
\begin{figure}
    \centering
    \begin{tabularx}{\linewidth}{XX}
        \includegraphics[height=\linewidth]{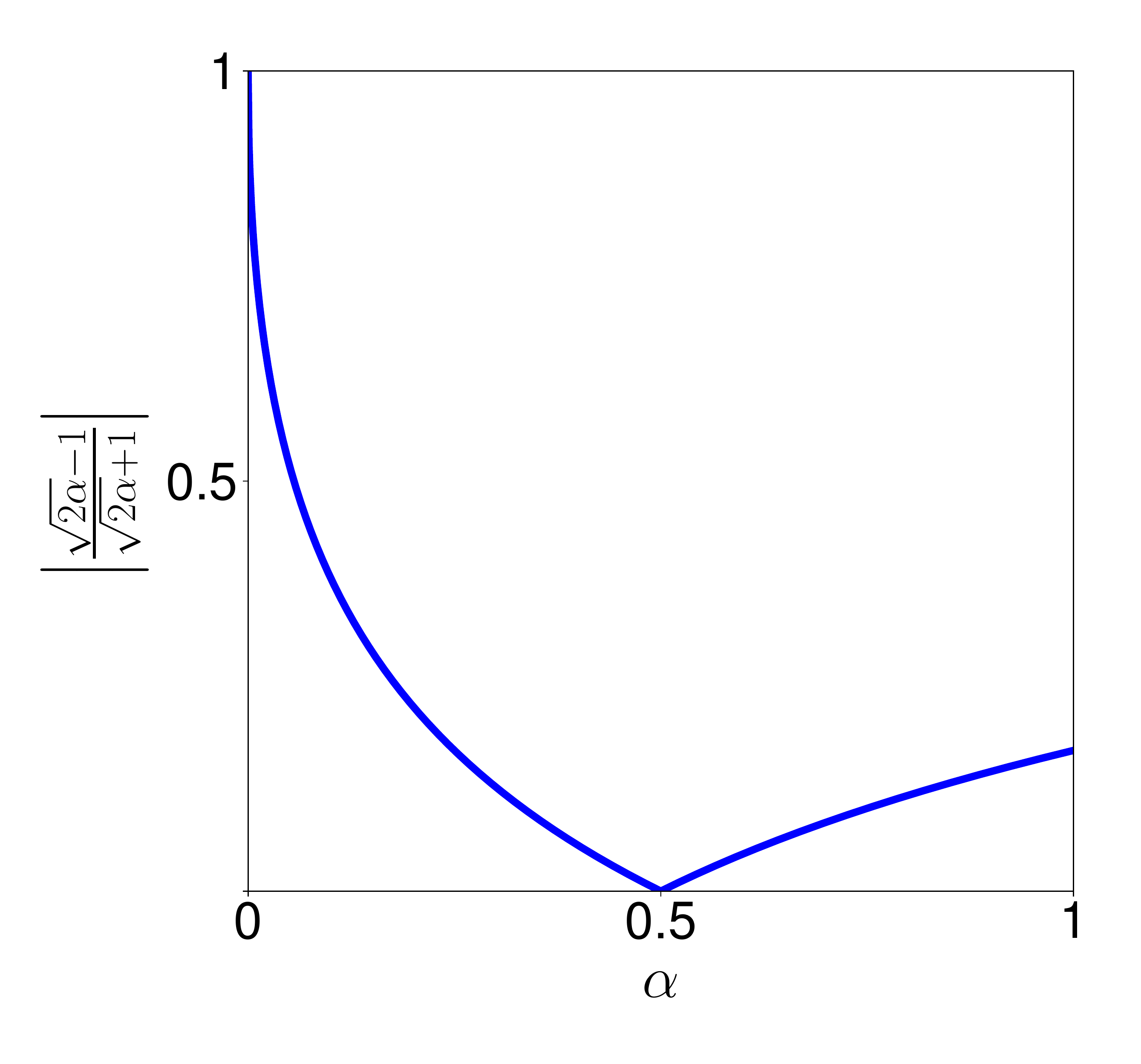}
        &\includegraphics[height=\linewidth]{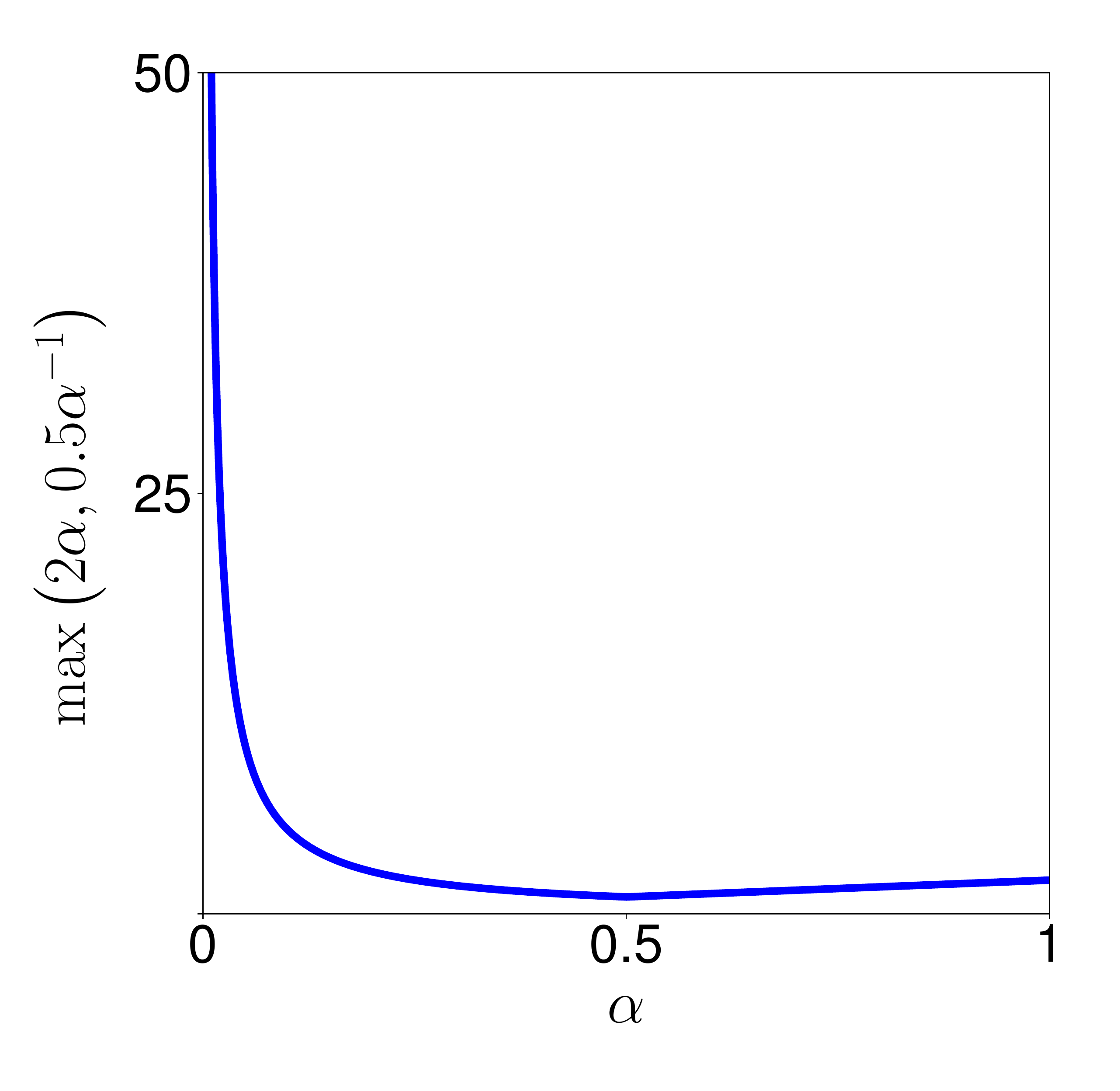} \\
    \end{tabularx}
    \caption{Absolute (left) and relative (right) errors for the Dice score and Tversky indices as a function of the Tversky parameter that maximally deviates from $0.5$ (see Proposition~\ref{prop:TverskyDiceApproximation}).}
    \label{fig:TverskyAbsRelErrors}
\end{figure}

\review{In this section, we have shown that Dice, Jaccard, and Tversky similarities give approximation bounds to each other, while cross-entropy and weighted cross-entropy do not approximate the Dice score.  This indicates that risk minimization with Dice, Jaccard, or Tversky will give some performance guarantees when measuring with the Dice score at test time.  On the other hand, training with (weighted) cross-entropy can yield arbitrarily poor results when evaluating with the Dice score. Approximation bounds from the Tversky index indicate increasing degradation in performance on the Dice score as parameters deviate from $\alpha=\beta=0.5$ suggesting best performance when training with the Dice loss, in correspondence with the risk minimization principle \cite{Vapnik:1995:NSL:211359}.
In the next section, we describe how we validate these theoretical results empirically in a range of medical image segmentation settings.}

\section{Empirical setup}
\label{sec:empirical_setup}
To investigate whether the aforementioned properties hold in practice, we investigate the performance of segmentation networks trained with different loss functions: cross-entropy (CE), weighted cross-entropy (wCE), soft Dice (sDice), soft Jaccard (sJaccard), Lov\'{a}sz-sigmoid (Lov\'{a}sz; this is the binary equivalent to the Lov\'{a}sz-softmax) and soft Tversky loss (sTversky). \review{The empirical validation is performed on the following six binary medical segmentation tasks:}
\begin{itemize}
	\item \textbf{BR18} - whole tumour segmentation on pre-operative MRI (BRATS 2018~\cite{Menze2015, Bakas2017, Bakas2018, brats2018}) - $N=285$
	\item \textbf{IS17} - final infarction after ischemic stroke treatment segmentation from acute MRI perfusion (ISLES 2017~\cite{isles2017}) - $N=43$
	\item \textbf{IS18} - ischemic core segmentation on acute CT perfusion (ISLES 2018~\cite{isles2018}) - $N=94$
	\item \textbf{MO17} - lower-left third molar segmentation on panoramic dental radiographs~\cite{DeTobel2017} - $N=400$
	\item \textbf{PO18} - colorectal polyp segmentation on colonoscopy images~\cite{eelbode2019tu1931} - $N=1166$
	\item \review{\textbf{WM17} - segmentation of white matter hyperintensities on MRI (WMH 2017~\cite{Kuijf2019})  - $N=60$}
\end{itemize}
The datasets BR18, IS17, IS18 and \review{WM17} are publicly available and 3D. The datasets MO17 and PO18 are in-house and 2D. An additional empirical validation for a multi-class segmentation task is given in Sec.~\ref{sec:multiclass} for BRATS 2018 with segmentation of different glioma sub-regions. The segmentation volumes considered for evaluation were the whole tumor (WT), tumor core (TC) and enhancing tumor (ET). \review{The inclusion of this allows us to empirically validate if the conclusions that will be drawn for binary image segmentation can still hold for individual structures, when the de facto optimization objective becomes the average loss across multiple structures.}\\

\subsection{Data preprocessing and network architectures}
\label{sec:data_preprocessing}
We use a 3D U-Net-like~\cite{Ronneberger2015a} architecture for the \review{four} 3D datasets BR18, IS17, IS18 and \review{WM17}. For MO17, we use the same architecture but with the 3D convolutions substituted by 2D convolutions. More specifically, we start from the No New-Net~\cite{Isensee2018} architecture, which was a top-ranked implementation during BRATS 2018. We have one level less and compensate with 3x3(x3) pooling and upsampling to keep sufficient field-of-view. The number of filters evolve similarly, starting with 20 in the first layer. For a richer comparison, a fully convolutional network with a VGG16 backbone and atrous convolutions, pretrained on ImageNet, i.e. DeepLab~\cite{Chen2016,long2015fully}, is used for PO18. \review{On WM17, additional runs with a DeepMedic-like~\cite{Kamnitsas2017} network architecture were done, further referred to as WM17\textsuperscript{DM}, depending on the context.}
In case of the multi-class segmentation task, we adapt the model with three sigmoid layers in the final layer. Thereby allowing a single voxel to be classified into one or more of the evaluation volumes.

We make use of all the available image modalities in each dataset to construct the input tensors (we exclude perfusion data for IS17 and IS18). In order to fit the memory constraints for whole-image processing, these input tensors are resized and cropped. The images of all public datasets are resampled to an isotropic voxel-size of 2 mm and\review{, except for WM17,} cropped to a fixed size of 136x136x82. \review{For WM17, we did patch-wise processing (with the centroids of the patches sampled uniformly inside the head region) instead with a patch size of 136x136x82 or 15x15x15 for U-Net or DeepMedic, respectively.} For MO17, we extract a 217x217 ROI around the geometrical center of the third molar from the panoramic radiograph at two times lower resolution. The images of PO18 are resampled to a fixed size of 384x288. All modalities are normalized according to its dataset's mean and standard deviation, except for PO18 where ImageNet statistics are used. We further use extensive data augmentation for all datasets, including Gaussian noise, lateral flipping and rigid transformations.\\

\begin{figure*}[htbp]%
\sbox\foursubbox{%
  \resizebox{\textwidth}{!}{%
    \includegraphics[height=3cm]{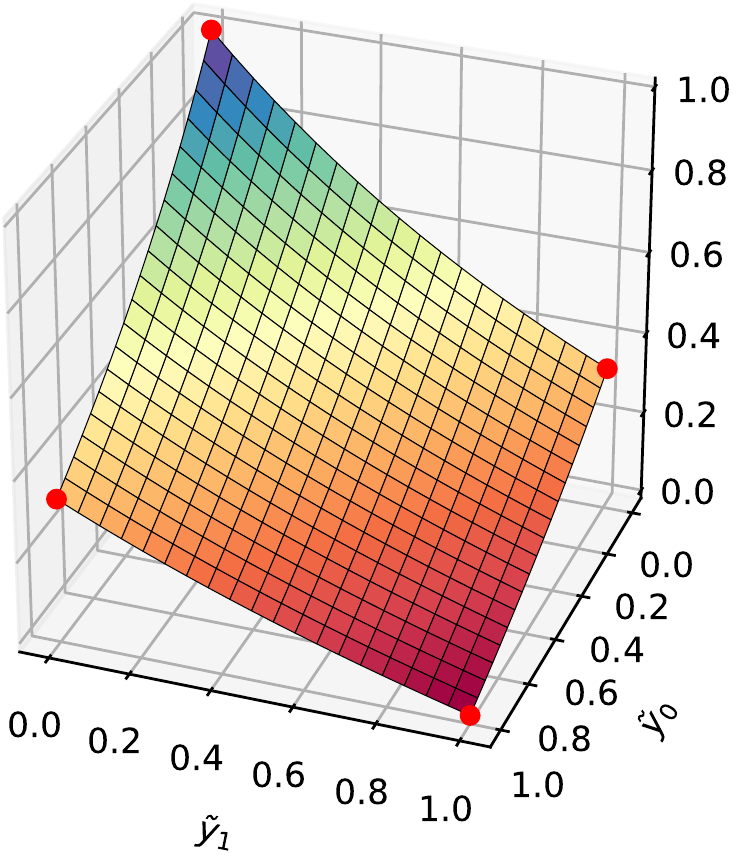}%
    \includegraphics[height=3cm]{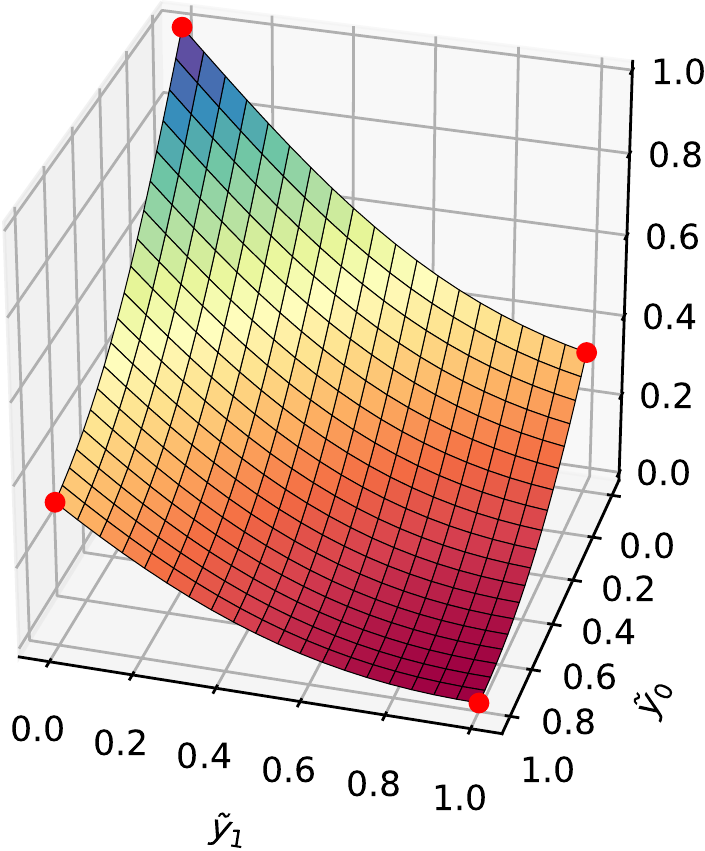}%
    \includegraphics[height=3cm]{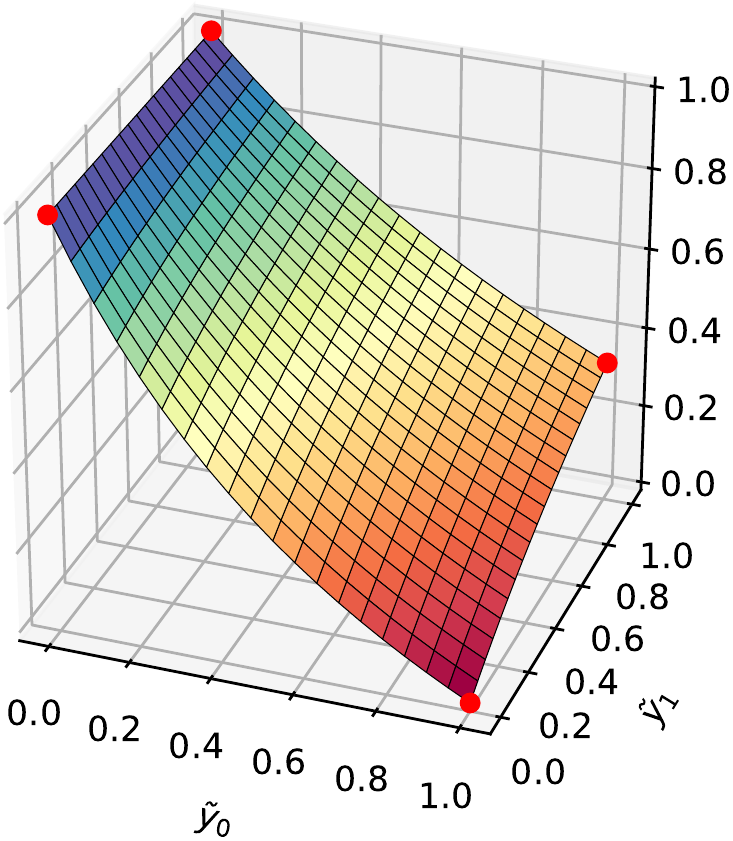}%
    \includegraphics[height=3cm]{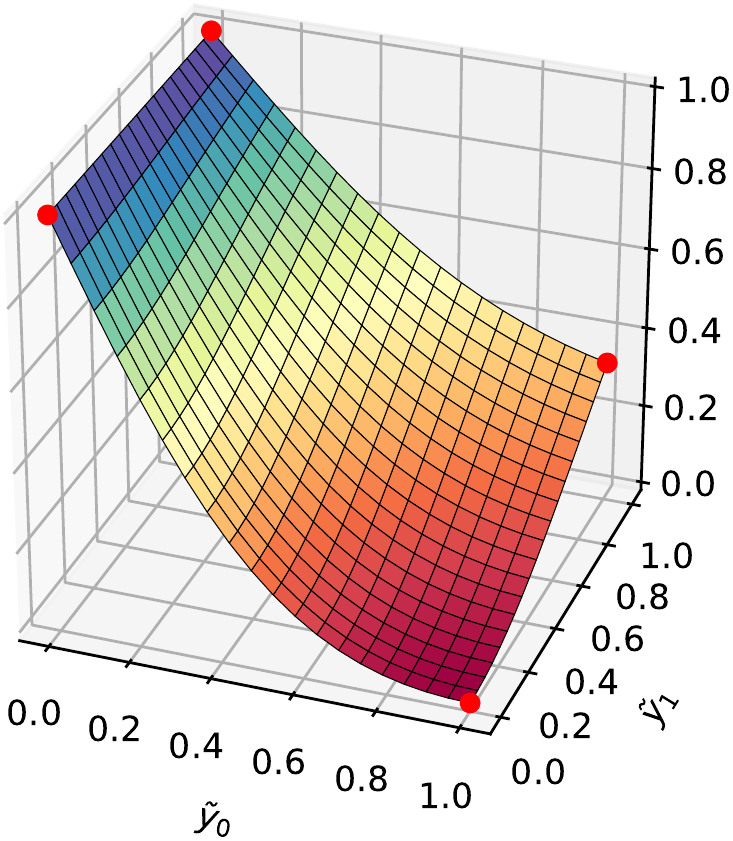}%
  }%
}%
\setlength{\foursubht}{\ht\foursubbox}%

  \centering%
\subcaptionbox{$y=[1, 1]$; $L^1$ soft Dice\label{fig:sdice11L1}}{%
  \includegraphics[height=\foursubht]{L1L2fig/sDice11}%
}%
\subcaptionbox{$y=[1, 1]$; $L^2$ soft Dice\label{fig:sdice11L2}}{%
  \includegraphics[height=\foursubht]{L1L2fig/sDice11L2}%
}%
\subcaptionbox{$y=[1, 0]$; $L^1$ soft Dice\label{fig:sdice10L1}}{%
  \includegraphics[height=\foursubht]{L1L2fig/sDice10}%
}%
\subcaptionbox{$y=[1, 0]$; $L^2$ soft Dice\label{fig:sdice10L2}}{%
  \includegraphics[height=\foursubht]{L1L2fig/sDice10L2}%
}%
  \caption{\review{Comparison between $L^1$ and $L^2$--based variants of the soft Dice loss for a two-pixel prediction $\tilde{y}=[\tilde{y}_0, \tilde{y}_1]$, and two different ground truth labels $y$. 
  The red dots indicate the values of the discrete Dice loss. 
  We see that the $L^2$ variants are flatter around the global minimum corresponding to the ground truth configuration.\label{fig:l1l2surfaces}}}%
\end{figure*}

\subsection{Training parameters}\label{sec:training_parameters}
The CNNs are pretrained with CE to help final convergence. In preliminary experiments we noticed this can speed-up the experiments and help metric-sensitive loss functions to be more robust.
\review{For U-Net and DeepLab,} we make use of the Adam optimizer~\cite{kingma2014adam} with the initial learning rate set at $10^{-3}$ (random initialization) or $10^{-4}$ (ImageNet initialization). \review{For DeepMedic, we do SGD with an initial learning rate $10^{-1}$}
The learning rate decreases by a factor of five when the validation loss on the complete images stops improving and the training stops when this loss increases.
The input batch size is $2$ for BR18, IS17 \review{and WM17}, $4$ for MO17 and IS18, $16$ for PO18 \review{and $64$ for WM17\textsuperscript{DM}}.

After CE-convergence, the training \reviewminor{continues} with CE, wCE, sDice, sJaccard, Lov\'{a}sz or Tversky \reviewminor{with a reset optimizer state.}
Our theoretical analysis suggests there is no optimal Dice or Jaccard approximation for wCE that can be derived before training starts (see Sect. \ref{sec:2}).
We therefore set the weights via the common heuristic of balancing out the number of foreground and background pixels/voxels~\cite{Sudre2017}, using a weight applied to the foreground class of $1/(2p)$ and a weight applied to the background class of $1/(2-2p)$. Here, $p$ represents the foreground prior.
For sTversky, two weighting factors need to be chosen. Our experiments are run for the case $\alpha+\beta=1.0$ analogously to~\cite{TverskyMLMI2017} and on a range of weight values with $\alpha$ ranging from 0.1 to 0.9 with a step size of 0.1. For $\alpha=0.5$, sTversky is equivalent to sDice.
We use similar optimization parameters as for CE pretraining, with the initial learning rate lowered to $10^{-4}$ for all datasets, except for MO17, which we found appropriate for convergence of the losses under study.\\

\subsection{\review{Model choice and statistical testing}}\label{sec:statistical_testing}
\review{In all experiments, we perform five-fold cross-validation and report the results aggregated over all the left-out subjects. For each fold we choose the model with the best validation loss (on the full images in case of patch-wise training) to test and report metrics. Final comparisons were done in a pair-wise setting between losses and tests for statistical significance were performed using non-parametric bootstrapping \reviewminor{as in the BRATS and ISLES challenges~\cite{brats2018,isles2017}. Implementation details are in \cite{bakas2018identifying}, with ranks substituted with scores (in line with the hypothesis).} To assess inferiority or superiority between pairs of optimization methods, significance level $p<0.05$ is used. In all tables, cells highlighted in gray point to the top-ranked losses (i.e. not significantly inferior compared to the loss with the highest performance). Values in italic denote inferiority of one loss compared to all others. \reviewminor{Boxplots for all measurements are reported separately in Supplementary Materials, Appendix~\ref{sec:boxplots}}}.

\section{Results and discussion}\label{sec:results}
This section is divided in four parts: First, we compare two common versions of the soft Dice loss present in the literature. Second and most importantly, the performance of all the losses is evaluated on \review{six} binary medical segmentation tasks. Thirdly, we also perform this evaluation in a multi-class setting and in the final part, the influence of the class-imbalance on the performance of these respective losses is investigated.
\review{We distinguish two groups of loss functions. The first group are the~\emph{CE-based losses}, which contains the (weighted) Hamming loss surrogates, CE and wCE. The second group are the~\emph{metric-sensitive losses}. It contains the sDice, sJaccard and Lov\'{a}sz losses, which are surrogates either for the Dice score or Jaccard index. The second group also contains the sTversky losses, which are surrogates for their respective Tversky indices.}

\subsection{Soft Dice variants: $L^1$  vs. $L^2$\label{sec:soft_dice_variants}}
As noted in Sec.~\ref{sec:dif_loss_sur}, we note that multiple soft generalizations for the Dice loss are possible by replacing $\langle \tilde{y},\tilde{y} \rangle = \|\tilde{y}\|_2^2$ in the denominator of Eq.~\eqref{eq:diceloss} by $\|\tilde{y}\|_p^p$ for any $L^p$ norm, and this choice can influence the optimization properties.
Both the $L^2$~\cite{MilletariNA16} and the $L^1$-based~\cite{Sudre2017} versions are present in the literature.
In a preliminary experiment, we implemented the $L^1$ and $L^2$ version of the soft Dice loss and compare their performance when evaluating on Dice score and Jaccard index in Table~\ref{tab:l1vsl2}.
There is a small favorable trend of the former definition and it significantly outperforms the latter on the IS18, PO18 \review{and WM17\textsuperscript{DM}} dataset.
Figure~\ref{fig:l1l2surfaces} shows a visualization of the loss surfaces for the two soft dice variants, in the case of a simple two-pixel prediction.
The fact that the $L^2$-based variants are flatter around the minimum could indicate that the $L^2$-based surrogates do not favour integer solutions, i.e. solutions that are close to an indicator vector -- which is where the soft surrogate of the Dice metric meets the target discrete Dice metric (red dots on Fig.~\ref{fig:l1l2surfaces}).
This would explain why we found that the $L^1$-based surrogate may outperform the $L^2$-based surrogate in some cases.

In all further experiments, we restrict our analysis to the $L^1$-based relaxation, and refer to it simply as the soft Dice loss.

\begin{table}[htb]
    \caption{\review{Dice scores and Jaccard indexes obtained for each dataset, using the $L^1$ and $L^2$ generalization for soft Dice loss. Values highlighted in grey are significantly superior compared to the other loss.}}
    \begin{tabularx}{\linewidth}{llXYYYYYYY}
    \toprule
         & Loss & \multicolumn{1}{l}{\lapbox[0.5\width]{0.5em}{\emph{Data} $\rightarrow$}} & BR18 & IS17 & IS18 & MO17 & PO18 & \review{WM17} & \review{WM17\textsuperscript{DM}} \\
    \midrule
    \parbox[t]{1mm}{\centering\multirow{2}{*}{\rotatebox[origin=c]{90}{Dice}}}
      & $L^1$   & &  0.870 &  0.362 &  \cellcolor{gray!50} 0.538 &  0.944 &  \cellcolor{gray!50} 0.656 & \review{0.712} & \cellcolor{gray!50} \review{0.717} \\
      & $L^2$   & &  0.869 &  0.358 &  \textit{0.525} &  0.943 &  \textit{0.639} & \review{0.708} & \review{\textit{0.703}}\\
    \midrule
    \parbox[t]{1mm}{\centering\multirow{2}{*}{\rotatebox[origin=c]{90}{Jacc.}}}
      & $L^1$   & &  0.780 &  0.244 & \cellcolor{gray!50} 0.407 &  0.897 &  \cellcolor{gray!50} 0.598 & \review{0.571} & \cellcolor{gray!50} \review{0.577} \\
      & $L^2$   & &  0.780 &  0.245 & \textit{0.395} &  0.896 &  \textit{0.567} & \review{0.566} & \review{\textit{0.563}} \\
         
    \bottomrule
\end{tabularx}
    \label{tab:l1vsl2}
\end{table}

\subsection{Comparison of losses for binary segmentation~\label{sec:experiments}}
The average Dice scores and Jaccard indexes for each dataset and evaluated loss function are shown in table~\ref{tab:results}. The results for sTversky are shown in a different Table~\ref{tab:tversky_results} for a better overview. Multiple observations can be made from these results:

\begin{table}[htb]
    \caption{\review{Dice scores and Jaccard indexes obtained for each dataset, using a CE-based loss (CE, wCE) or a metric-sensitive loss (sDice, sJaccard, Lov\'{a}sz). Cells in gray highlight the top-ranked losses. Values in italic point to a significant lower value compared to all other losses.}}
\begin{tabularx}{\linewidth}{lsXYY|YYY}
    \toprule
    & & \multicolumn{1}{r}{\lapbox[\width]{1em}{\emph{group} $\rightarrow$}} & \multicolumn{2}{c}{CE-based} & \multicolumn{3}{c}{Metric-sensitive} \\
    \midrule
    & Dataset & \multicolumn{1}{r}{\lapbox[\width]{1em}{\emph{loss} $\rightarrow$}} & CE & wCE & sDice & sJaccard & Lov\'{a}sz \\
    \midrule
    \parbox[t]{7mm}{\centering\multirow{5}{*}{\rotatebox[origin=c]{90}{Dice}}} 
    & BR18 & & 0.848 & \textit{0.808} & \cellcolor{gray!50} 0.870 & \cellcolor{gray!50}0.868 & \cellcolor{gray!50}0.871 \\
    & IS17 & & 0.281 & 0.310 & \cellcolor{gray!50} 0.362 & \cellcolor{gray!50} 0.363 & \cellcolor{gray!50} 0.357 \\
    & IS18 & & 0.463 & 0.487 & \cellcolor{gray!50} 0.538 & 0.528 & 0.508 \\
    & MO18 & & \cellcolor{gray!50} 0.942 & \textit{0.903} & \cellcolor{gray!50} 0.944 & \cellcolor{gray!50} 0.942 & \cellcolor{gray!50} 0.944 \\
    & PO18 & & 0.635 & \textit{0.602} & \cellcolor{gray!50} 0.656 & \cellcolor{gray!50} 0.651 & \cellcolor{gray!50} 0.649 \\
    & \review{WM17} & & \review{0.672} & \review{\textit{0.492}} & \cellcolor{gray!50} \review{0.712} & \cellcolor{gray!50} \review{0.712} & \review{0.700} \\
    & \review{WM17\textsuperscript{DM}} & & \review{0.671} & \review{\textit{0.383}} & \cellcolor{gray!50} \review{0.717} & \cellcolor{gray!50} \review{0.714} & \cellcolor{gray!50} \review{0.711} \\
    \midrule
    \parbox[t]{7mm}{\centering\multirow{5}{*}{\rotatebox[origin=c]{90}{Jaccard}}}
    & BR18 & & 0.751 & \textit{0.695} & \cellcolor{gray!50} 0.780 & \cellcolor{gray!50} 0.780 & \cellcolor{gray!50} 0.784 \\
    & IS17 & & 0.191 & 0.212 & \cellcolor{gray!50} 0.244 & \cellcolor{gray!50} 0.246 & \cellcolor{gray!50} 0.245 \\
    & IS18 & & 0.345 & 0.357 & \cellcolor{gray!50} 0.407 & 0.399 & 0.382 \\
    & MO18 & & \cellcolor{gray!50} 0.894 & \textit{0.835} & \cellcolor{gray!50} 0.897 & \cellcolor{gray!50}  0.896 & \cellcolor{gray!50} 0.896 \\
    & PO18 & & 0.541 & \textit{0.488} & \cellcolor{gray!50} 0.559 & \cellcolor{gray!50} 0.554 & \cellcolor{gray!50} 0.553 \\
    & \review{WM17} & & \review{0.530} & \review{\textit{0.355}} & \cellcolor{gray!50} \review{0.571} & \cellcolor{gray!50} \review{0.573} & \review{0.558} \\
    & \review{WM17\textsuperscript{DM}} & & \review{0.529} & \review{\textit{0.273}} & \cellcolor{gray!50} \review{0.577} & \cellcolor{gray!50} \review{0.576} & \cellcolor{gray!50} \review{0.572} \\
    \bottomrule
\end{tabularx}

    \label{tab:results}
\end{table}

\subsubsection{Equivalence of $J$ and $D$}
Theory indicates that Dice score and Jaccard index are equivalent ~(Prop.~\ref{thm:DiceJaccardApproximation}) and that they should provide the same ranking of the different loss functions. This also shows empirically in our results. If a loss $\ell_A$ performs better than loss $\ell_B$ in terms of Dice, it will perform better in terms of Jaccard index, and vice-versa.

\subsubsection{Performance of the surrogates}
The results show clearly that use of the metric-sensitive losses provides superior Dice scores and Jaccard indexes than (w)CE. There is only one dataset, MO17, for which the performance of CE is not inferior to that of the metric-sensitive losses. This is probably because the class imbalance is less pronounced in this dataset.
These results are what we expected since we already saw that CE and the metric-sensitive losses are theoretically divergent. A second observation that we make is that there is no significant difference between sDice, sJacard and Lov\'{a}sz as a loss function. This implies a free choice within the group of metric-sensitive losses.

For the Tversky loss, we observe that in no situation, $\alpha=0.5$-weighting is significantly inferior to any of the unequal weightings. With $\alpha$ going further away from the 0.5 baseline, the performance w.r.t. Dice score drops significantly once the weighting is too extreme. Therefore we can conclude that in our experiments $\alpha=\beta=0.5$ (equivalent to Dice loss) is the optimal weighting for all cases when evaluating on Dice score. \\

\begin{table*}[htb]
    \centering
    \caption{\review{Dice scores and Jaccard indexes obtained for all datasets using the range of sTversky losses with varying alpha/beta. Cells in gray highlight the top-ranked losses. There is no weighting scheme producing significantly higher values than sTversky 0.5/0.5.}}

\begin{tabularx}{\linewidth}{lsXYYYYYYYYY|YY}
    \toprule
    & Dataset & \multicolumn{1}{r}{\lapbox[\width]{1em}{\emph{$\alpha/\beta$} $\rightarrow$}} & 0.1/0.9 & 0.2/0.8 & 0.3/0.7 & 0.4/0.6 & 0.5/0.5 & 0.6/0.4 & 0.7/0.3 & 0.8/0.2 & 0.9/0.1 & \review{0.75/0.75} & \review{1.0/1.0} \\
    \midrule
    \parbox[t]{7mm}{\centering\multirow{5}{*}{\rotatebox[origin=c]{90}{Dice}}} 
    &BR18 &        &        0.801 &        0.844 &        0.859 &        0.863 & \cellcolor{gray!50}0.870 &   \cellcolor{gray!50}     0.867 &        0.857 &        0.831 &        \textit{0.787} &     \cellcolor{gray!50}   \review{0.871} & \cellcolor{gray!50}\review{0.868} \\
    &IS17 &        &    \cellcolor{gray!50}    0.352 &        0.345 &    \cellcolor{gray!50}    0.374 &   \cellcolor{gray!50}     0.374 &\cellcolor{gray!50} 0.362 &        0.346 &    \cellcolor{gray!50}    0.356 &    \cellcolor{gray!50}    0.345 &        0.293 &      \cellcolor{gray!50}  \review{0.371} & \cellcolor{gray!50}\review{0.363} \\
    &IS18 &        &        0.481 &        0.522 &        0.533 &     \cellcolor{gray!50}   0.540 & \cellcolor{gray!50}0.538 &    \cellcolor{gray!50}    0.527 &        0.519 &        0.490 &        \textit{0.445} &        \review{0.528} & \review{0.528} \\
    &MO18 &        &        0.907 &        0.928 &        0.936 &        0.941 & \cellcolor{gray!50}0.944 &    \cellcolor{gray!50}    0.942 &        0.938 &        0.932 &        0.902 &    \cellcolor{gray!50}   \review{0.944} & \cellcolor{gray!50}\review{0.942} \\
    &PO18 &        &        0.615 &        0.628 &        0.638 &        0.647 & \cellcolor{gray!50}0.656 &   \cellcolor{gray!50}     0.651 &        0.647 &        0.633 &        0.611 &          \review{0.646} & \cellcolor{gray!50}\review{0.651} \\
    &\review{WM17} &        &        \review{0.581} &        \review{0.649} &        \review{0.689} &        \review{0.700} & \cellcolor{gray!50}\review{0.712} &        \review{0.706} &        \review{0.693} &        \review{0.660} &        \review{0.565} &      \cellcolor{gray!50}  \review{0.712} & \cellcolor{gray!50}\review{0.712} \\
    &\review{WM17\textsuperscript{DM}} &        &        \review{0.599} &        \review{0.668} &        \review{0.696} &        \review{0.709} & \cellcolor{gray!50}\review{0.717} &        \review{0.706} &        \review{0.675} &        \review{0.628} &        \review{\textit{0.496}} &        \review{0.706} & \cellcolor{gray!50}\review{0.714} \\
    \midrule
    \parbox[t]{7mm}{\centering\multirow{5}{*}{\rotatebox[origin=c]{90}{Jaccard}}}
    &BR18 &        &        0.680 &        0.740 &        0.765 &        0.772 & \cellcolor{gray!50}0.780 &        0.775 &        0.764 &        0.726 &        \textit{0.666} &     \cellcolor{gray!50}   \review{0.783} & \cellcolor{gray!50}\review{0.780} \\
    &IS17 &        &        0.239 &        0.232 &     \cellcolor{gray!50}   0.257 &   \cellcolor{gray!50}     0.257 &\cellcolor{gray!50} 0.244 &        0.231 &    \cellcolor{gray!50}    0.238 &    \cellcolor{gray!50}    0.230 &        \textit{0.188} &      \cellcolor{gray!50} \review{0.250} &\cellcolor{gray!50} \review{0.246} \\
    &IS18 &        &        0.349 &        0.389 &     \cellcolor{gray!50}   0.401 &    \cellcolor{gray!50}    0.407 & \cellcolor{gray!50}0.407 &        0.398 &        0.390 &        0.361 &      \textit{0.321} &        \review{0.399} & \review{0.399} \\
    &MO18 &        &        0.834 &        0.870 &        0.884 &        0.892 & \cellcolor{gray!50}0.897 &    \cellcolor{gray!50}    0.895 &        0.886 &        0.876 &        0.826 &    \cellcolor{gray!50}    \review{0.896} & \cellcolor{gray!50}\review{0.896} \\
    &PO18 &        &        0.504 &        0.523 &        0.536 &        0.548 &\cellcolor{gray!50} 0.559 &   \cellcolor{gray!50}     0.556 &   \cellcolor{gray!50}     0.554 &        0.539 &        0.511 &          \review{0.550} & \cellcolor{gray!50}\review{0.554} \\
    &\review{WM17} &        &        \review{0.422} &        \review{0.496} &        \review{0.541} &        \review{0.556} &\cellcolor{gray!50} \review{0.571} &        \review{0.565} &        \review{0.548} &        \review{0.510} &        \review{0.410} &      \cellcolor{gray!50}  \review{0.571} & \cellcolor{gray!50}\review{0.573} \\
    &\review{WM17\textsuperscript{DM}} &        &        \review{0.441} &        \review{0.516} &        \review{0.551} &        \review{0.566} & \cellcolor{gray!50}\review{0.577} &        \review{0.566} &        \review{0.532} &        \review{0.481} &        \review{\textit{0.360}} &        \review{0.565} & \cellcolor{gray!50}\review{0.576} \\
    \bottomrule
\end{tabularx}

    \label{tab:tversky_results}
\end{table*}

\subsubsection{Weighting of cross-entropy}
We can see from our results that wCE is typically underperforming w.r.t. CE. Choosing a different weight might increase the performance of wCE, but it is apparent that the choice of weighting is non-trivial and very application-dependent. Finding an optimal weight is therefore not easy and will probably require additional parameter tuning as compared to the metric-sensitive losses.
Moreover, we show in section~\ref{sec:classimbalance} that performance of wCE is better for a smaller range of object sizes. Therefore, it is highly unlikely that one specific weighting could provide an appropriate surrogate for all the object sizes and datasets which is in agreement with our theory as well.

\begin{sidewaysfigure*}[p!]
    \centering

    \begin{tabularx}{0.75\textheight}{cYYYYYYY}
        & Original & Ground truth & CE & wCE & sDice & sJaccard & Lov\'{a}sz
    \end{tabularx}

    \def\arraystretch{0}
    \setlength{\tabcolsep}{0pt}
    \begin{tabularx}{0.75\textheight}{l @{\hspace{2pt}} XXXXXXX}
    
        \rotatebox{90}{\hspace{20pt} BR18} &
        \includegraphics[width=\linewidth]{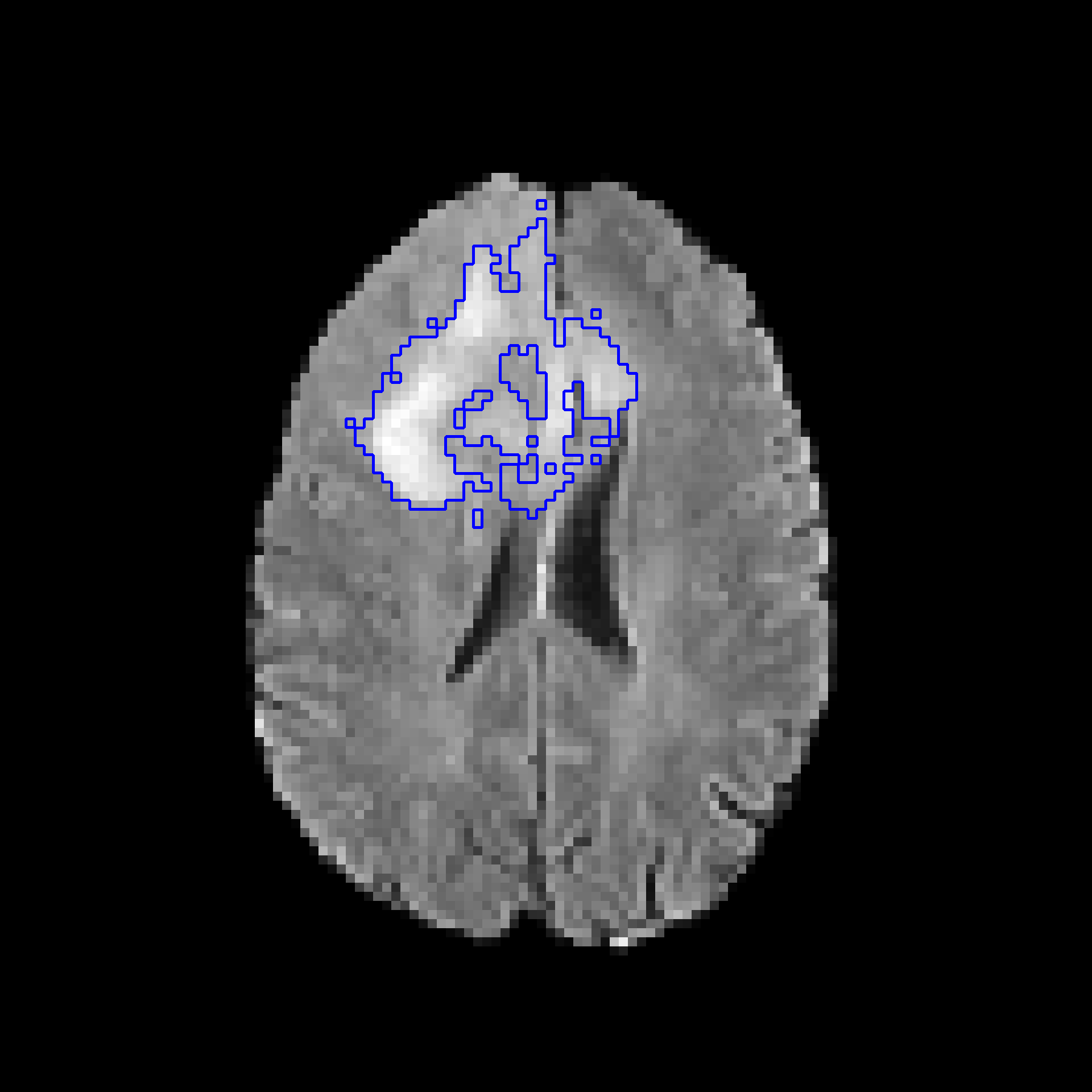} & 
        \includegraphics[width=\linewidth]{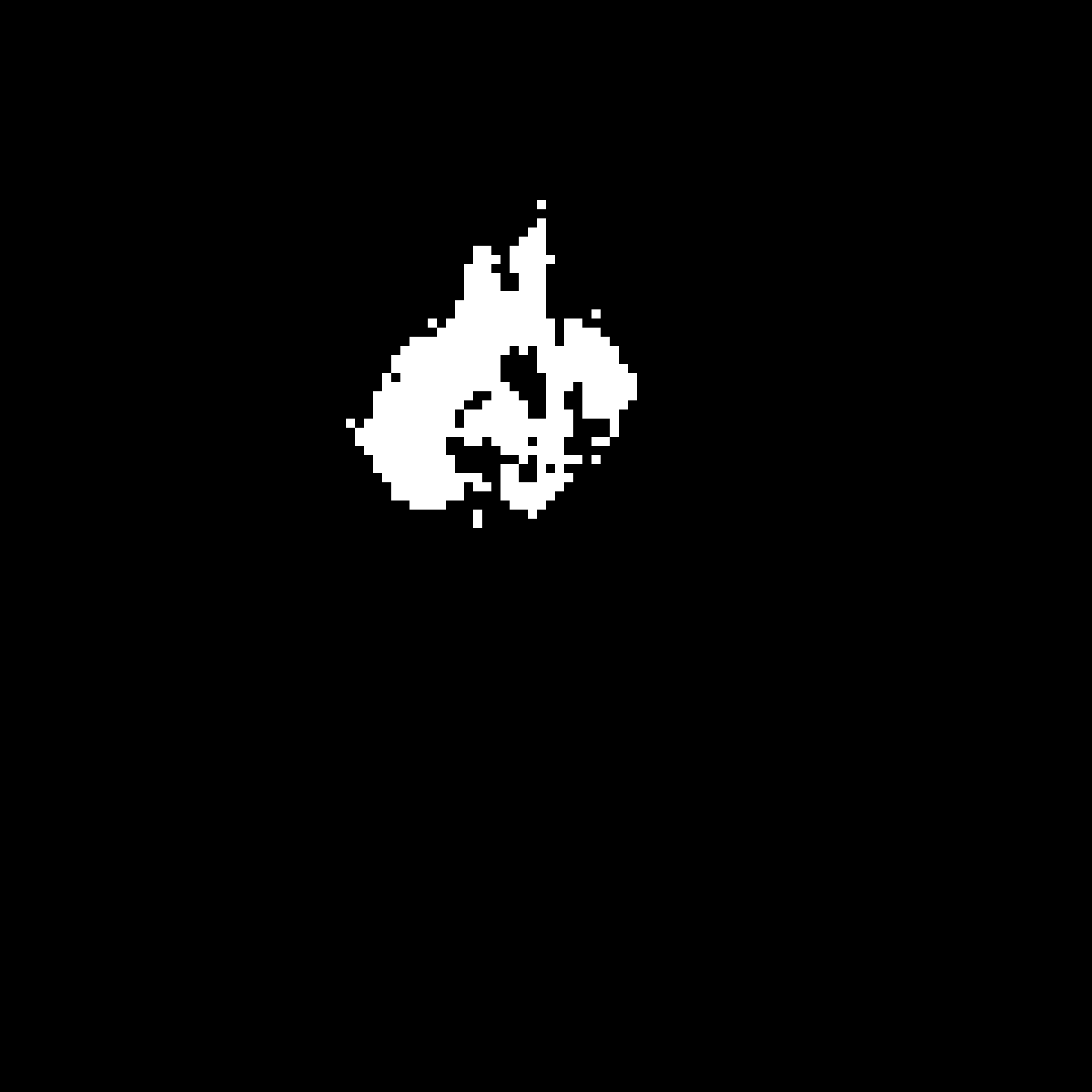} & 
        \includegraphics[width=\linewidth]{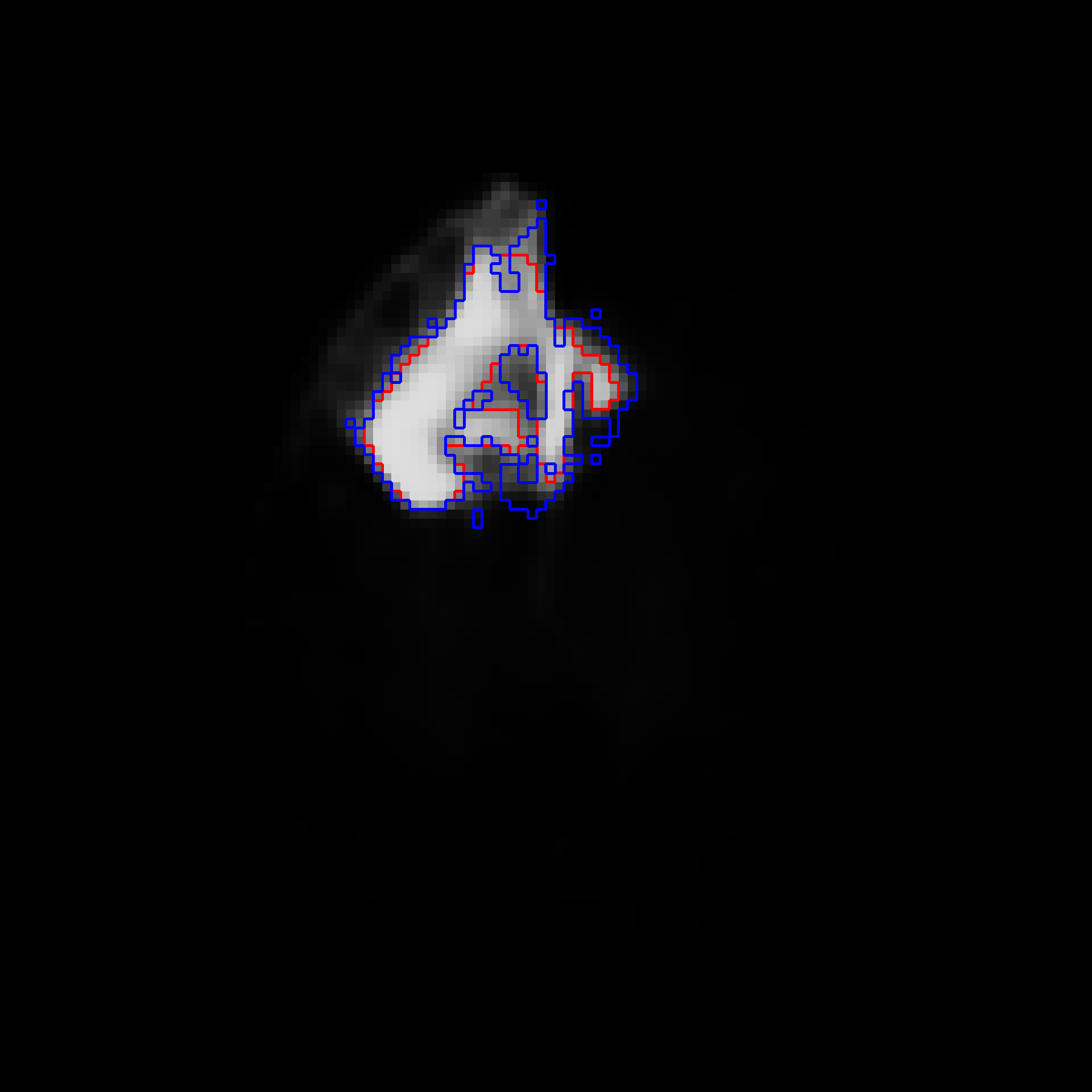} &
        \includegraphics[width=\linewidth]{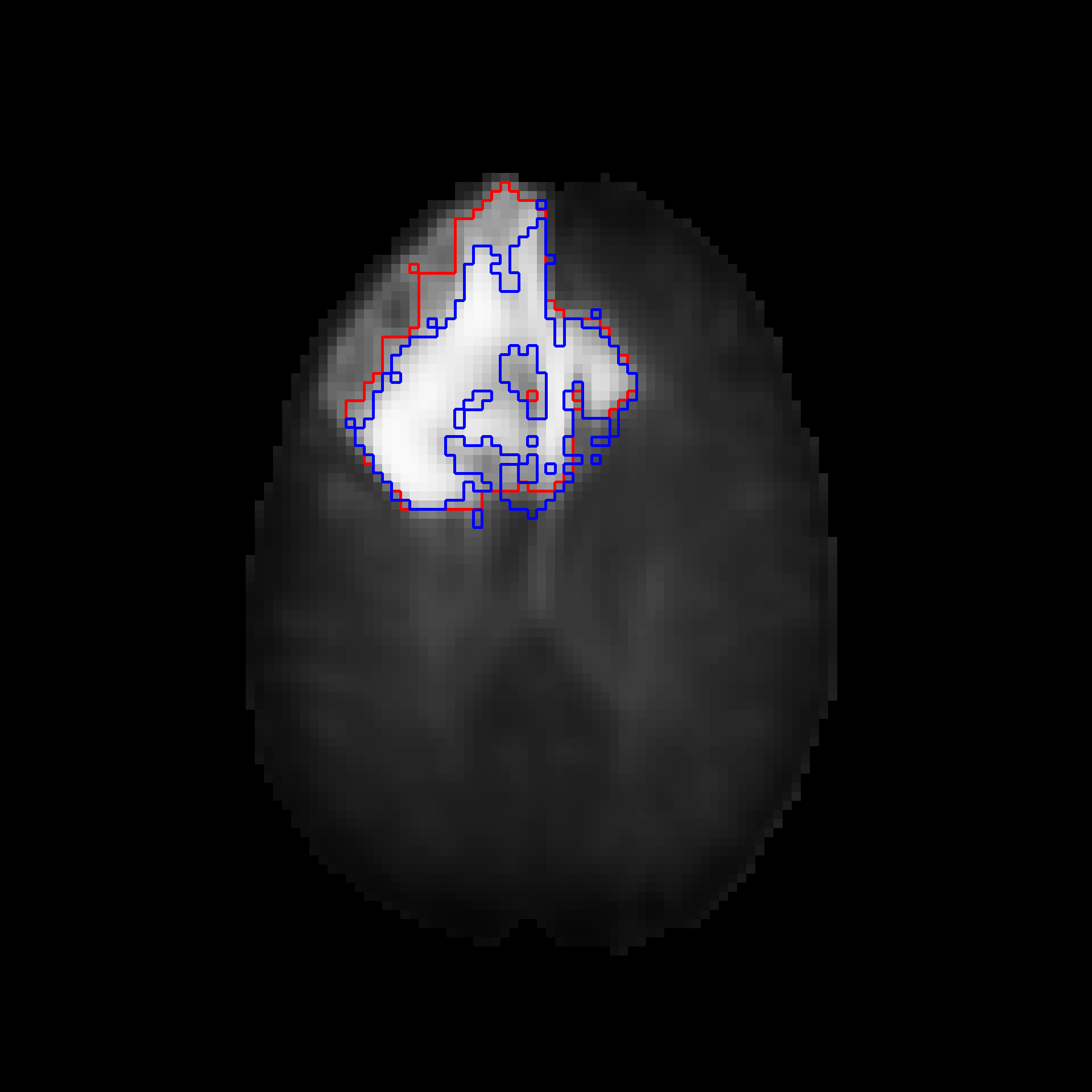} &
        \includegraphics[width=\linewidth]{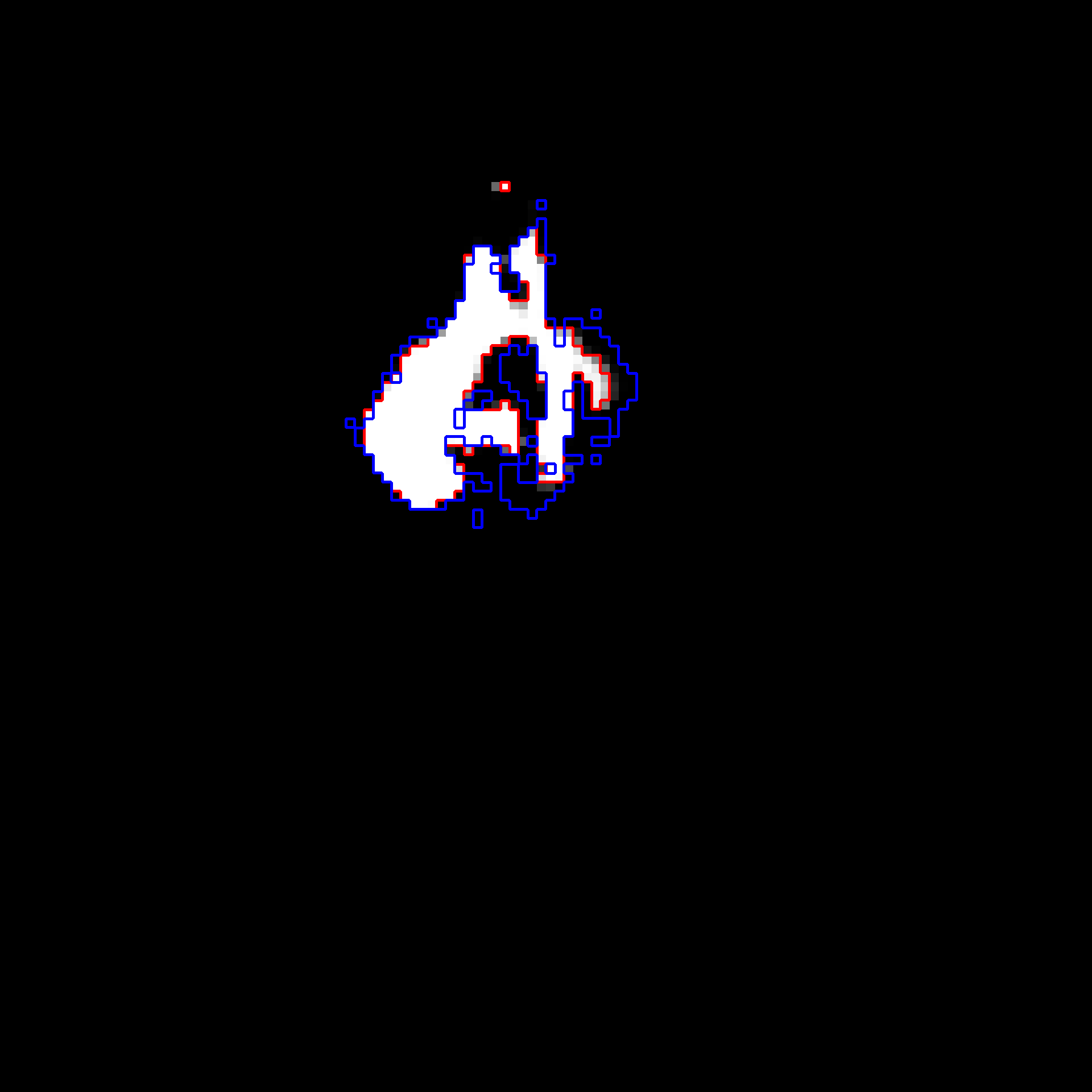} & 
        \includegraphics[width=\linewidth]{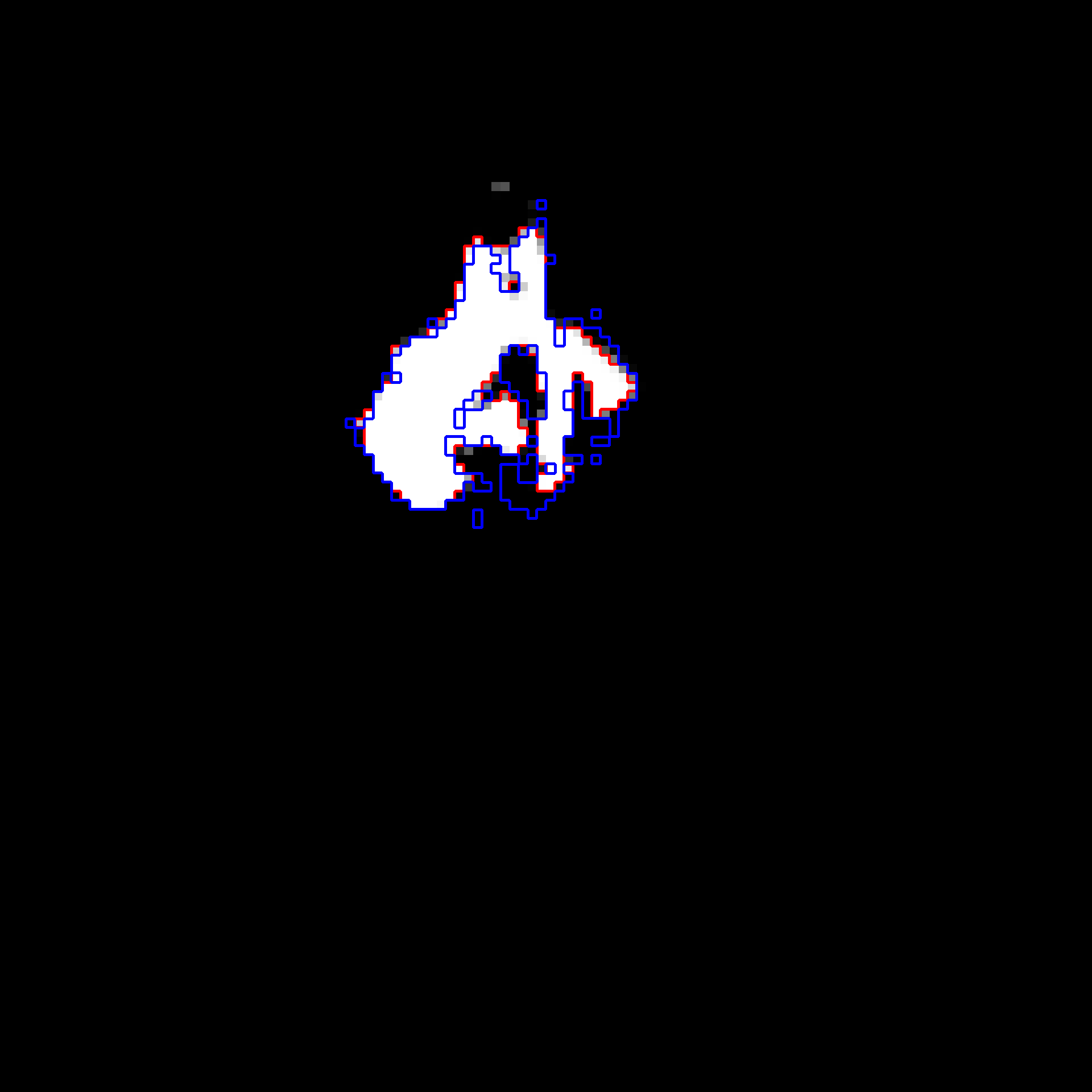} &
        \includegraphics[width=\linewidth]{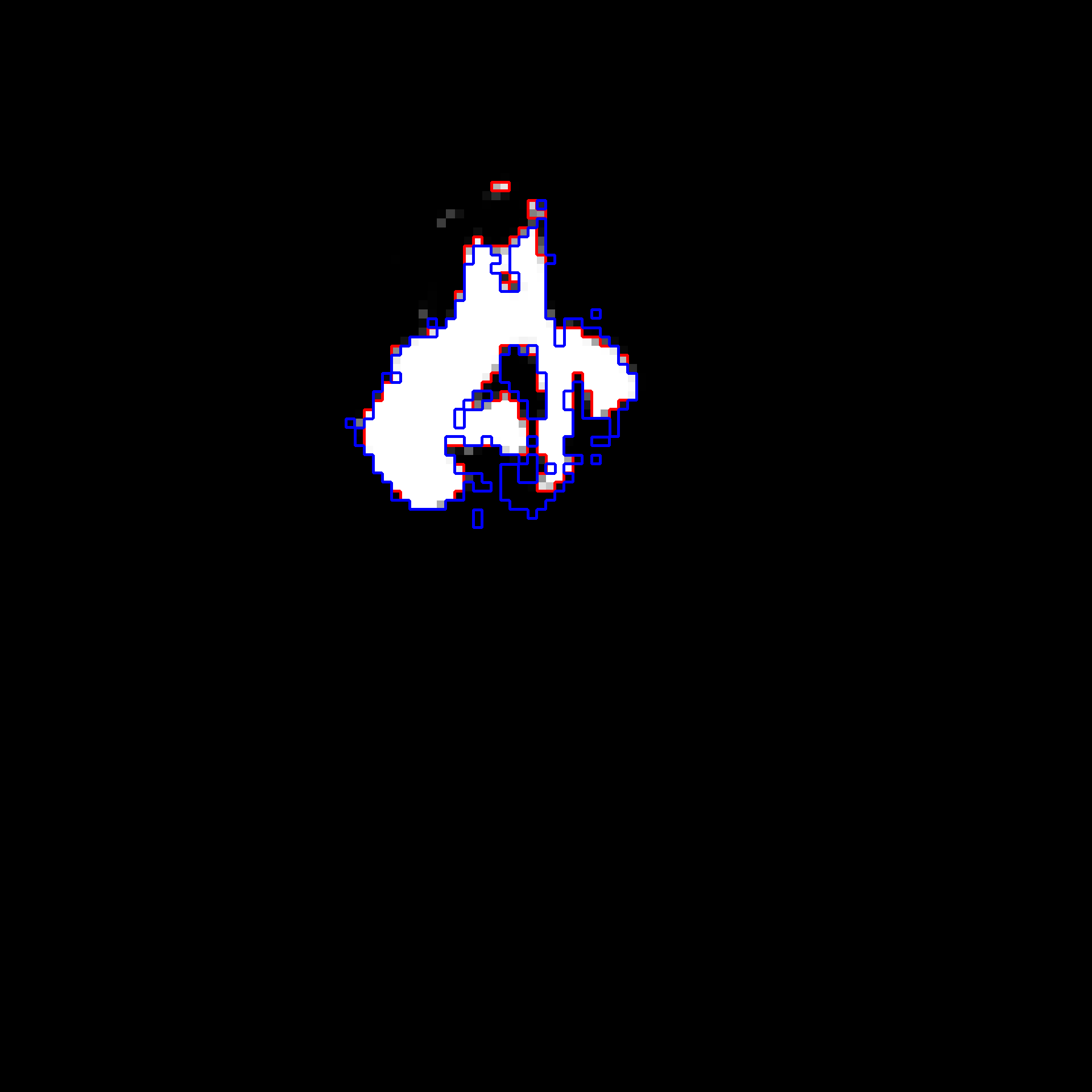} \\
       
        \rotatebox{90}{\hspace{25pt} IS17} &
        \includegraphics[width=\linewidth]{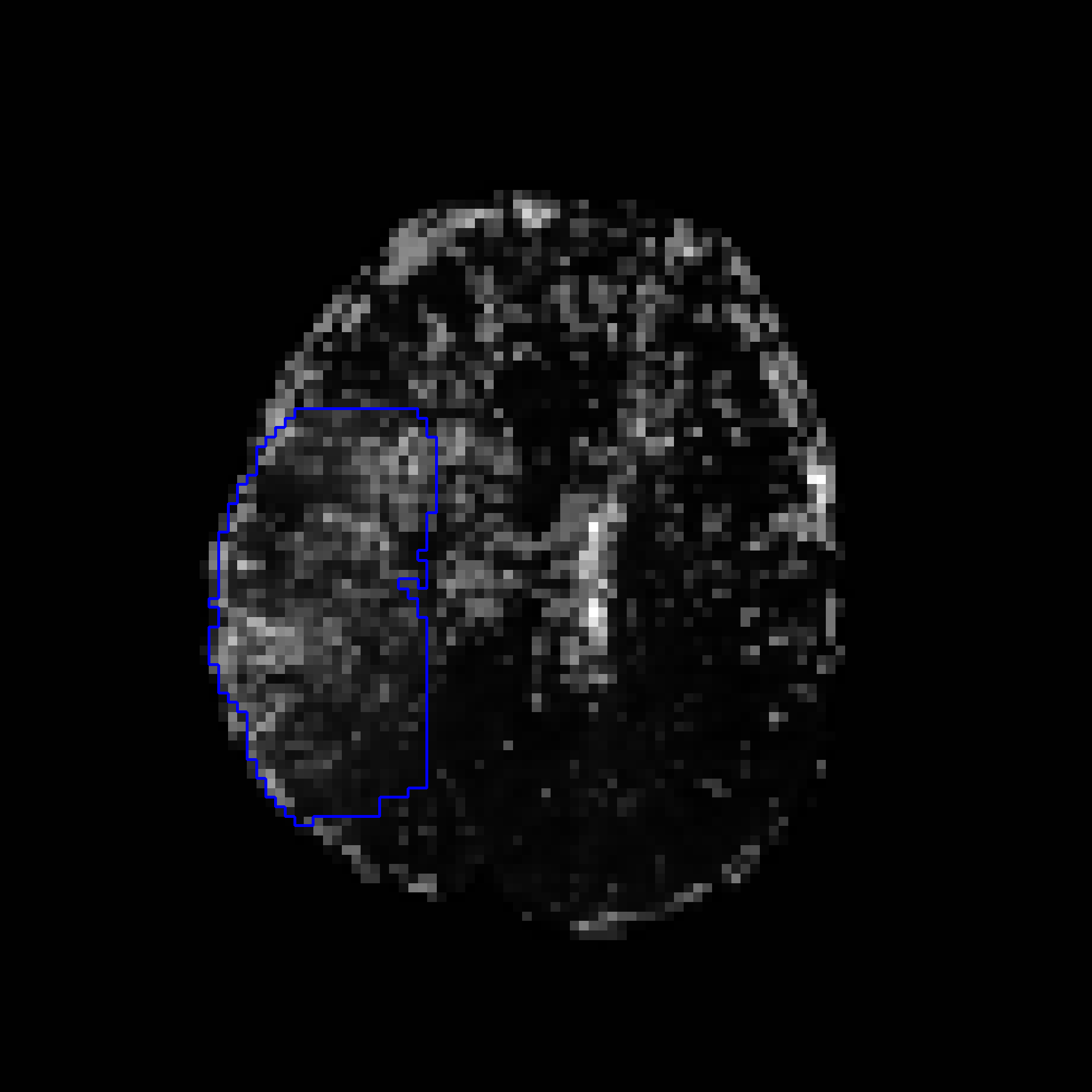} & 
        \includegraphics[width=\linewidth]{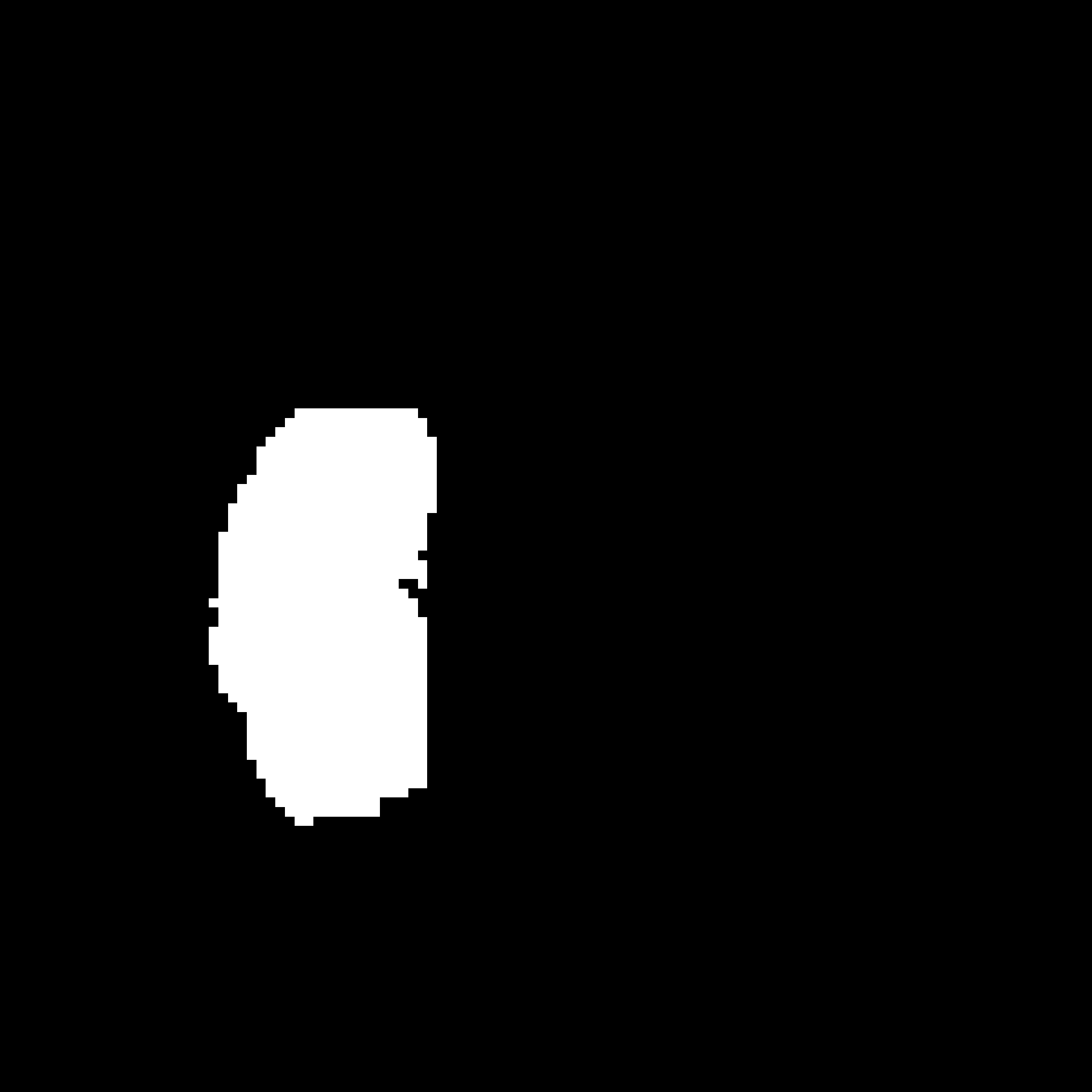} &
        \includegraphics[width=\linewidth]{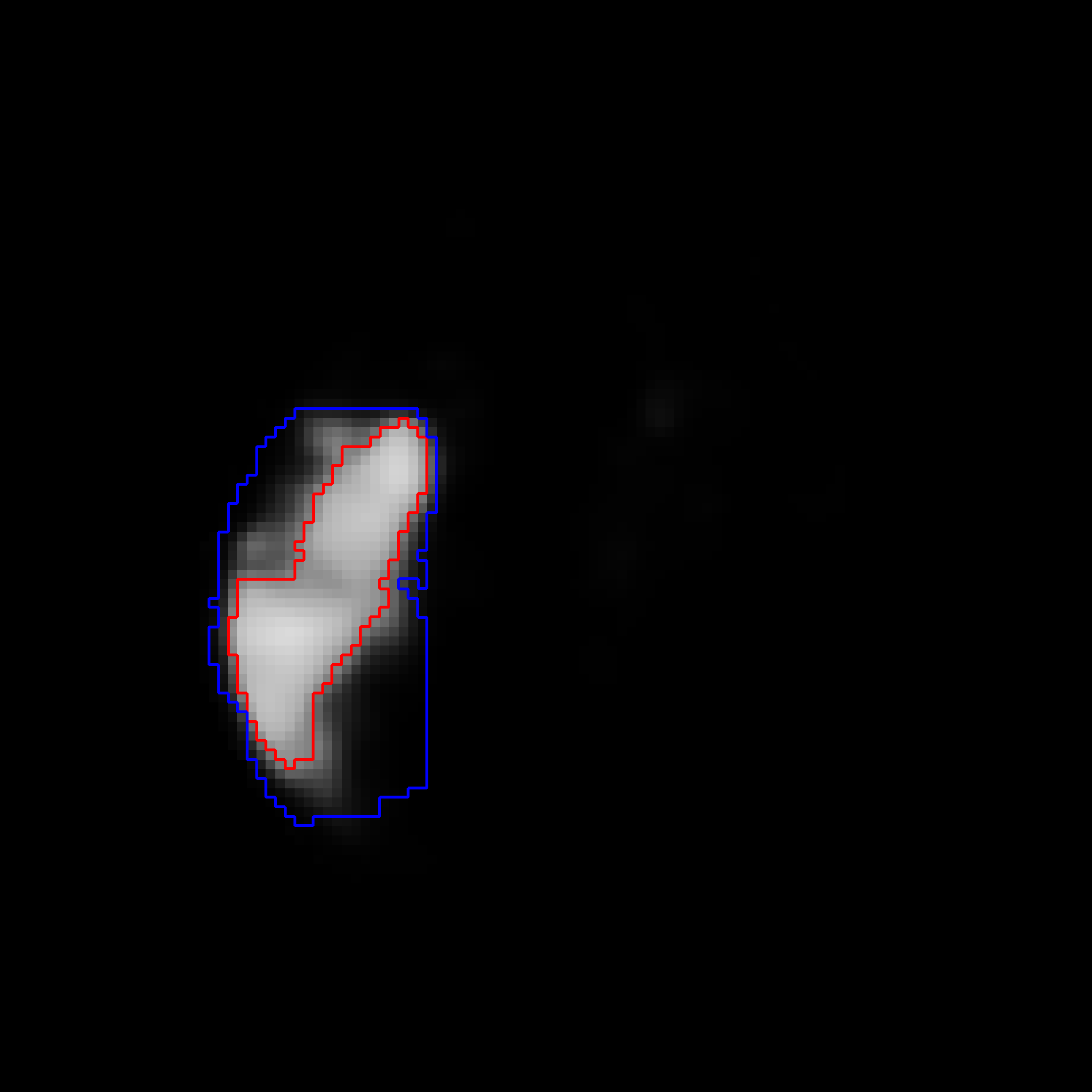} & 
        \includegraphics[width=\linewidth]{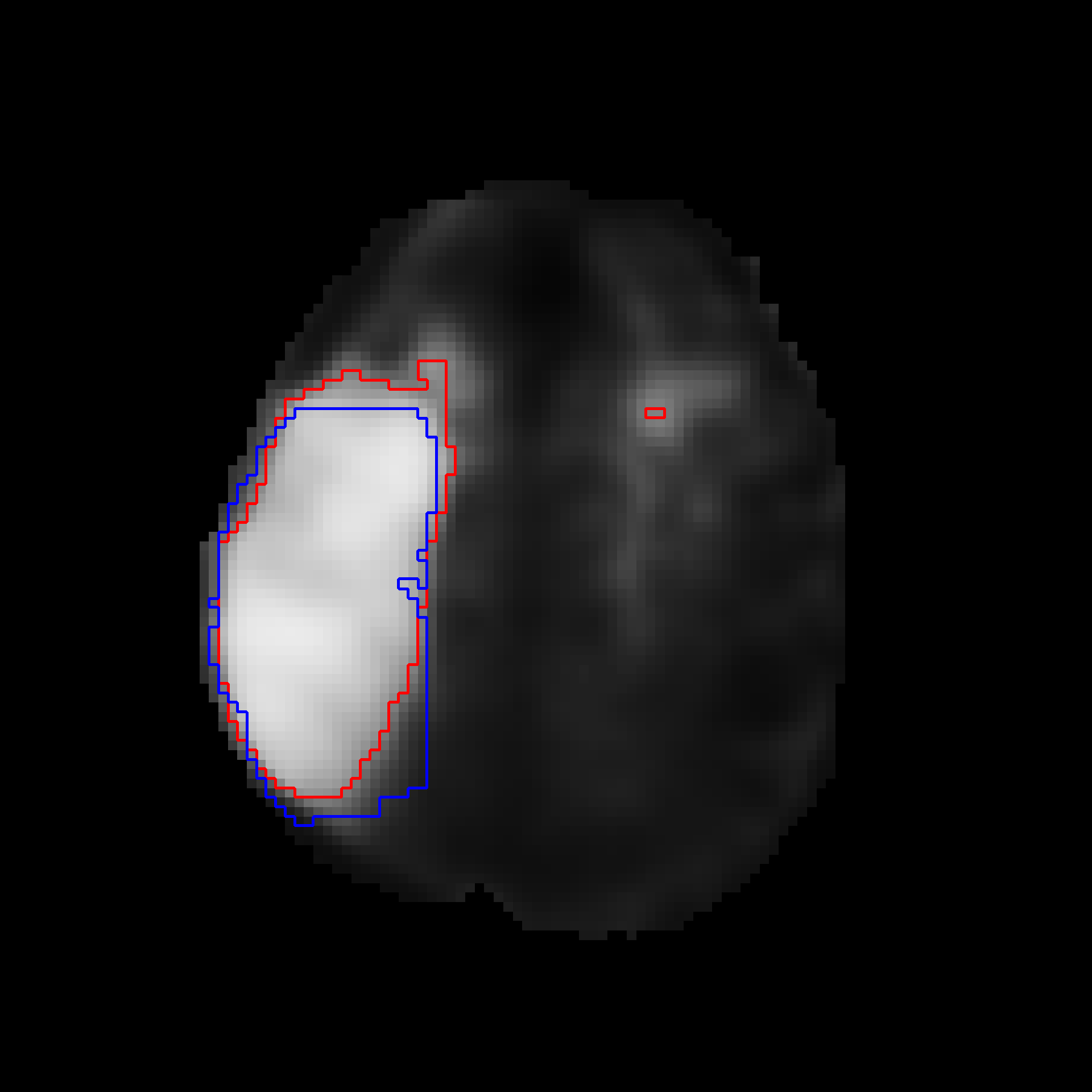} &
        \includegraphics[width=\linewidth]{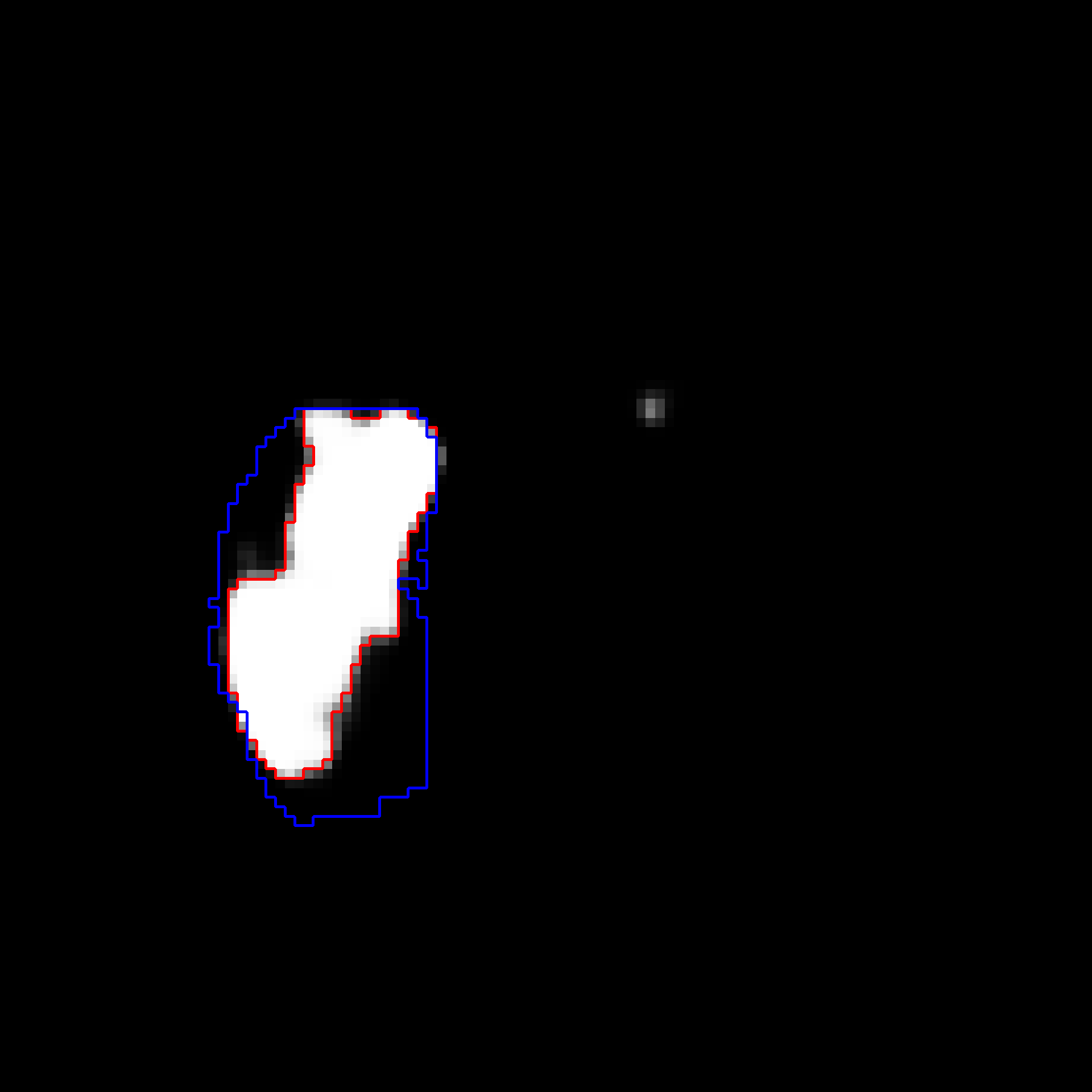} & 
        \includegraphics[width=\linewidth]{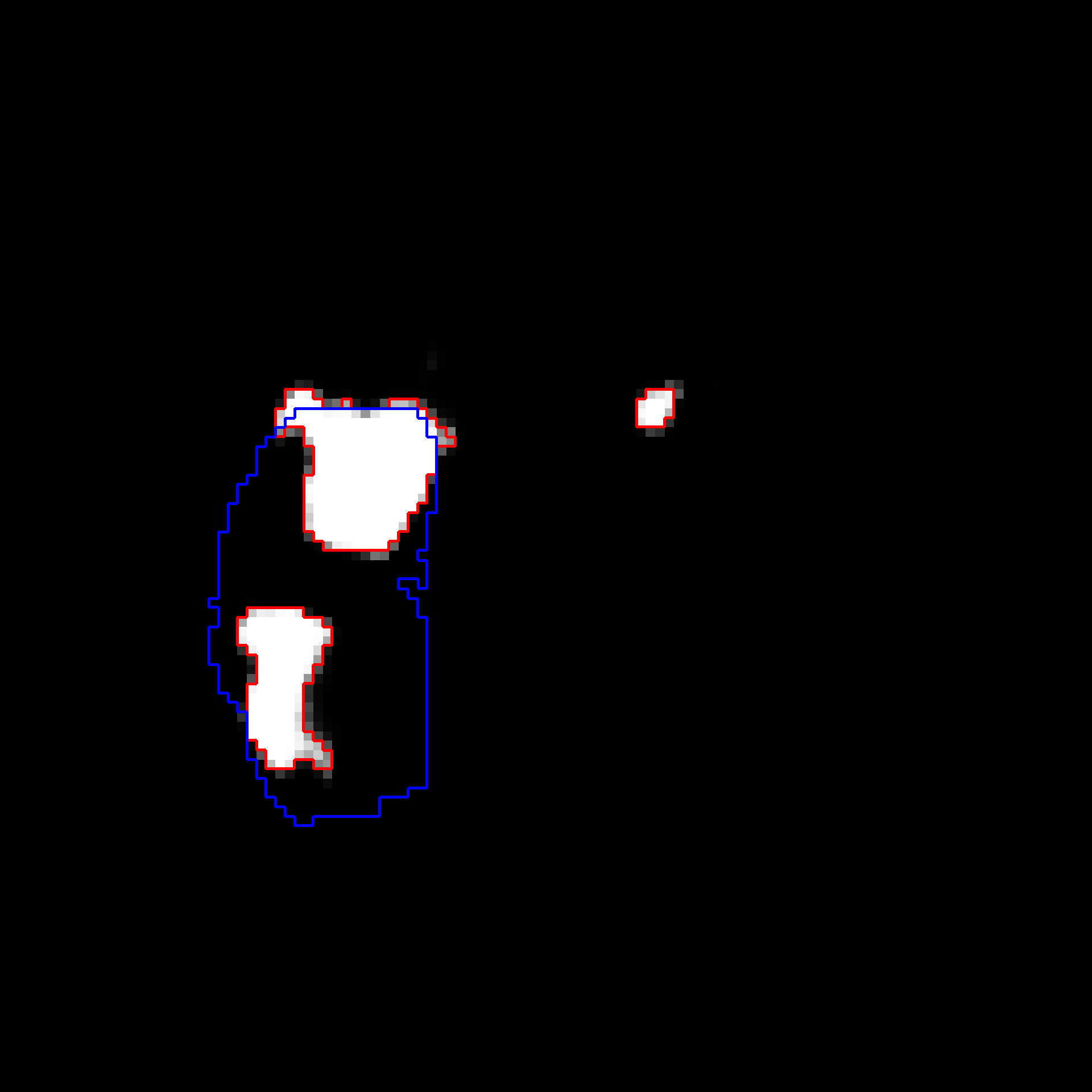} &
        \includegraphics[width=\linewidth]{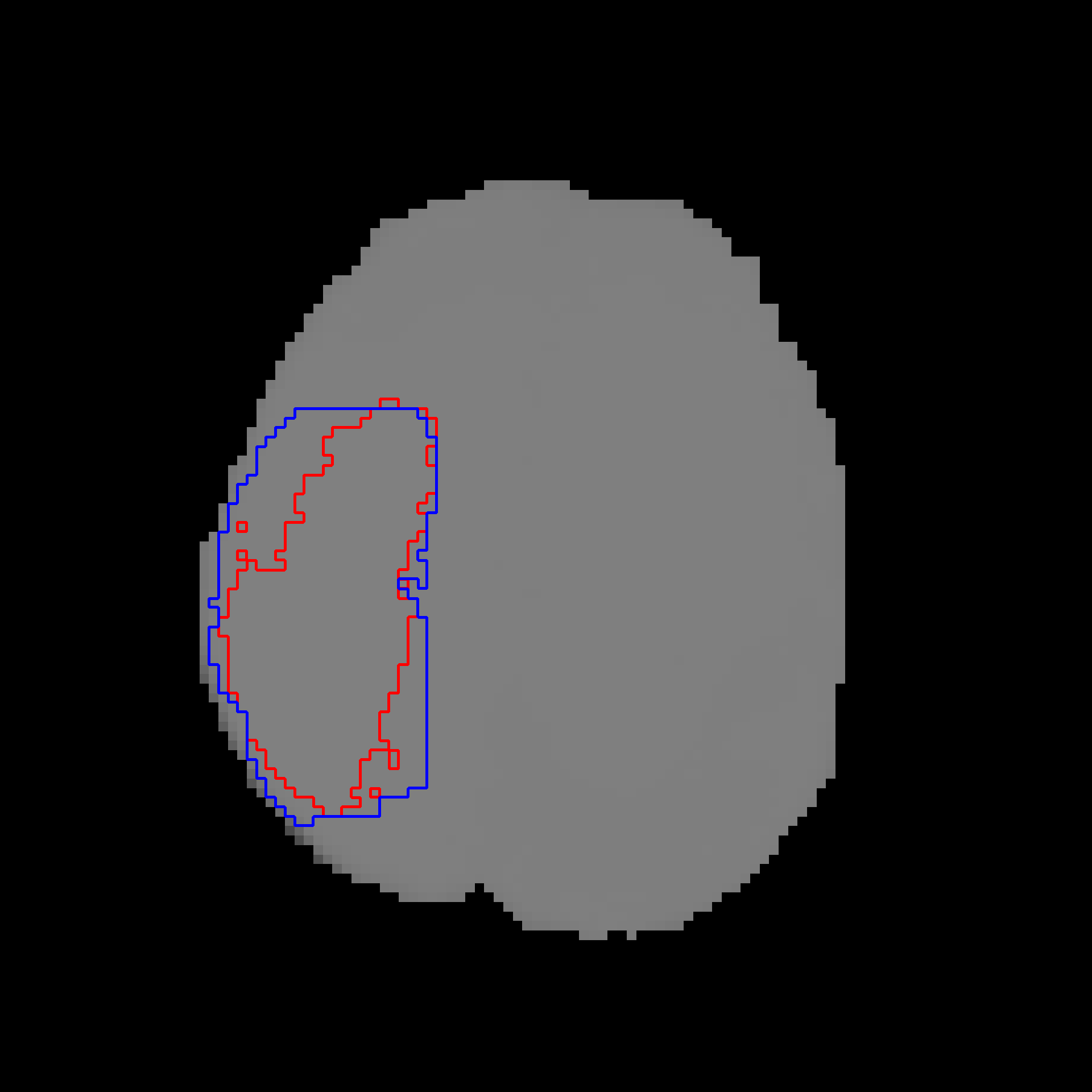} \\
  
        \rotatebox{90}{\hspace{25pt} IS18} &
        \includegraphics[width=\linewidth]{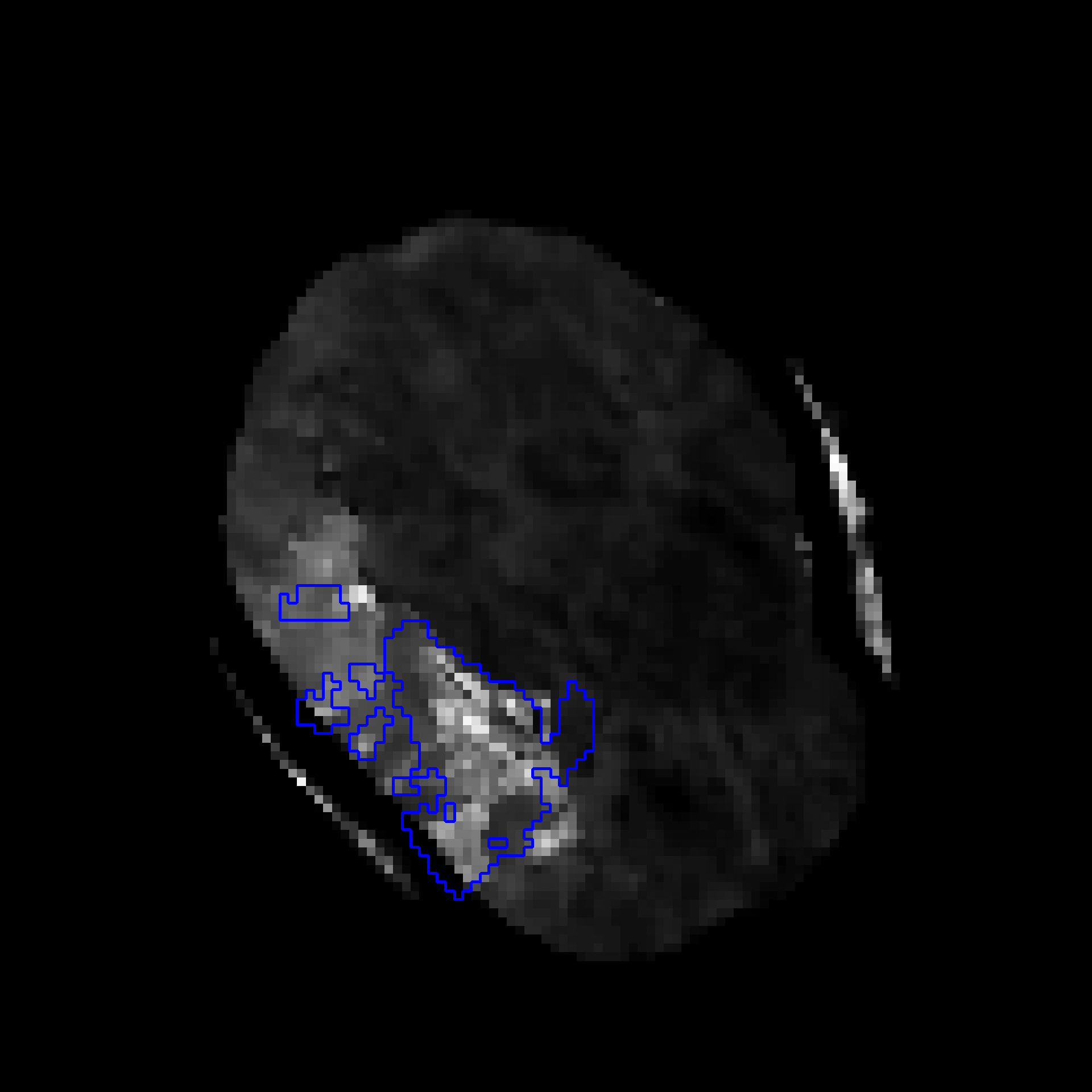} & 
        \includegraphics[width=\linewidth]{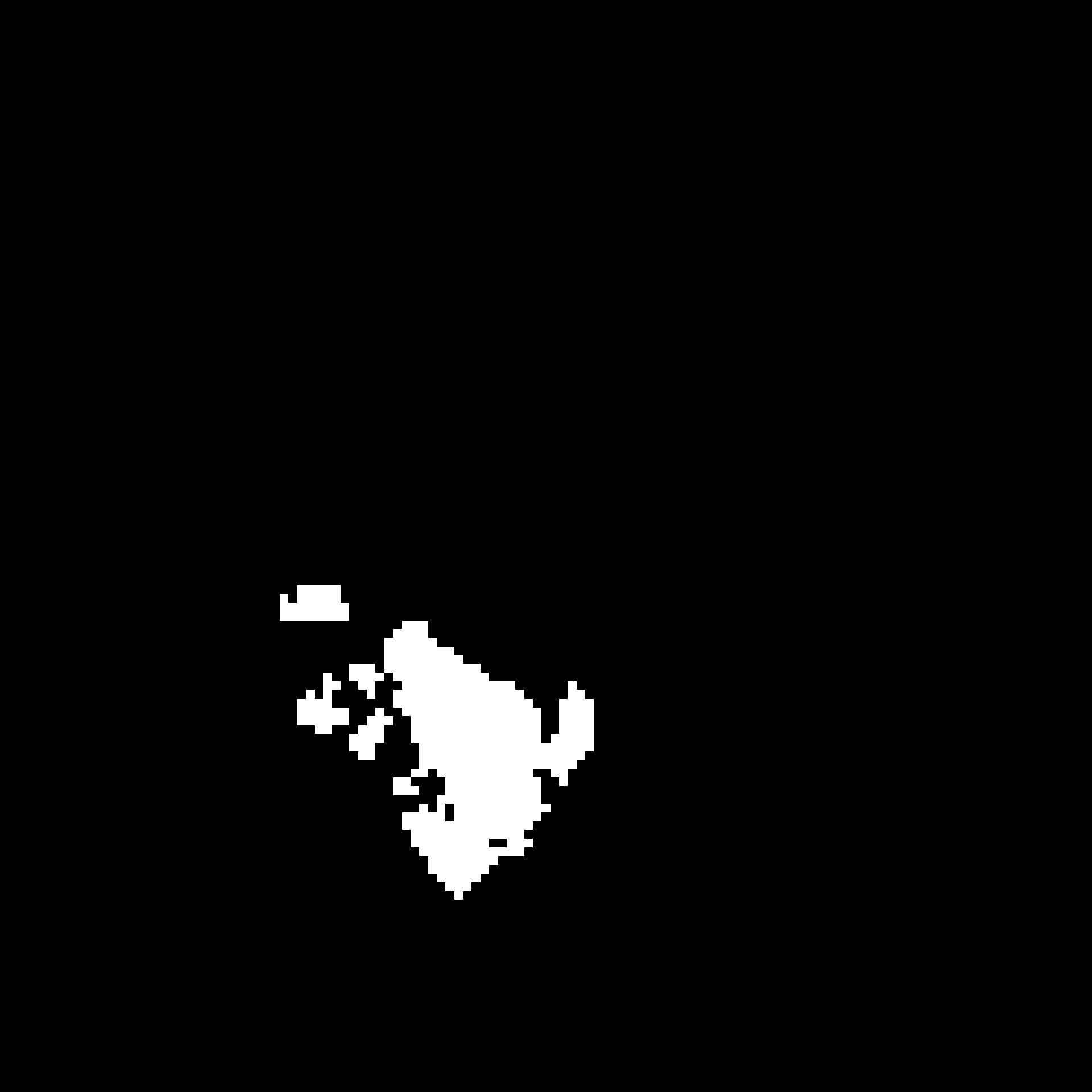} &
        \includegraphics[width=\linewidth]{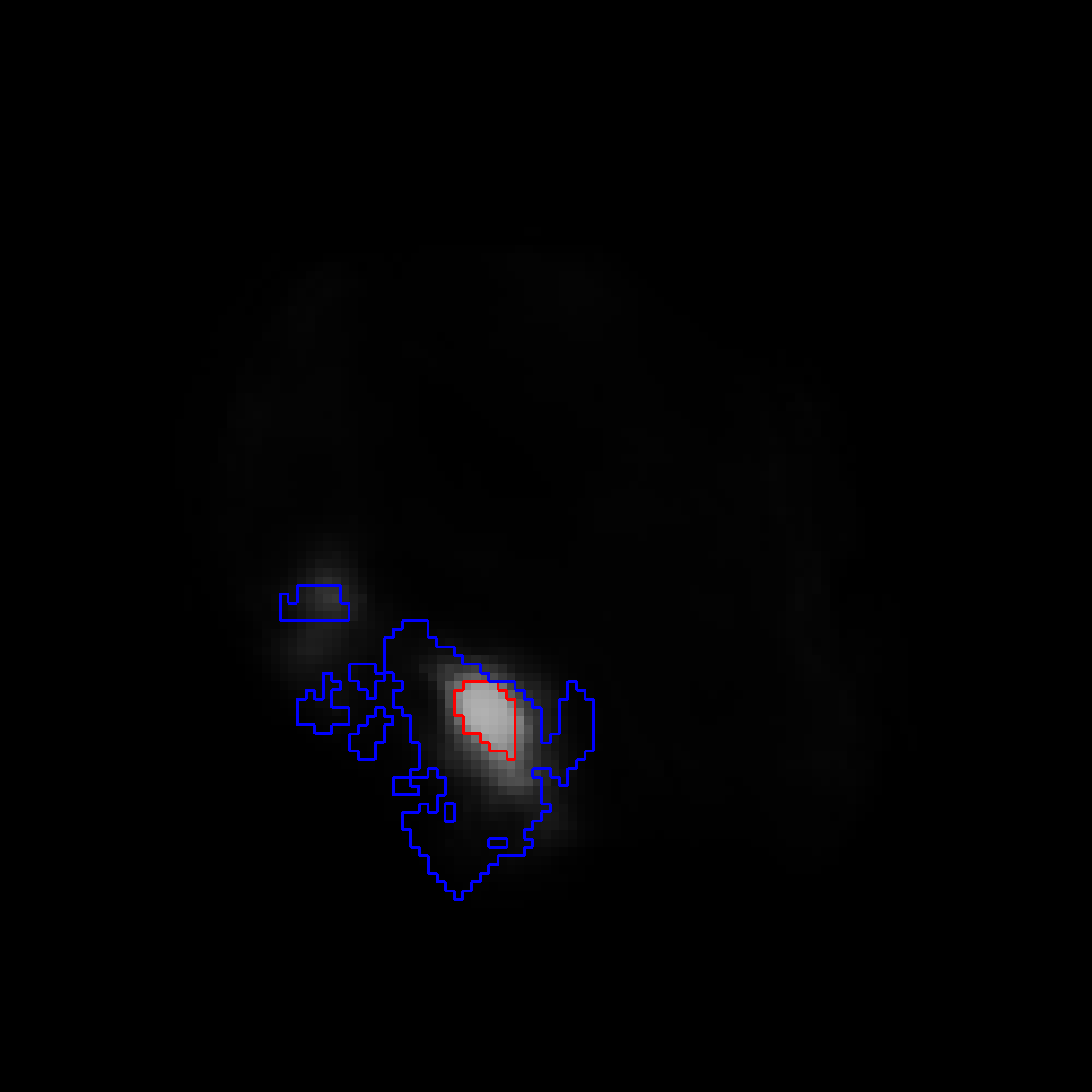} & 
        \includegraphics[width=\linewidth]{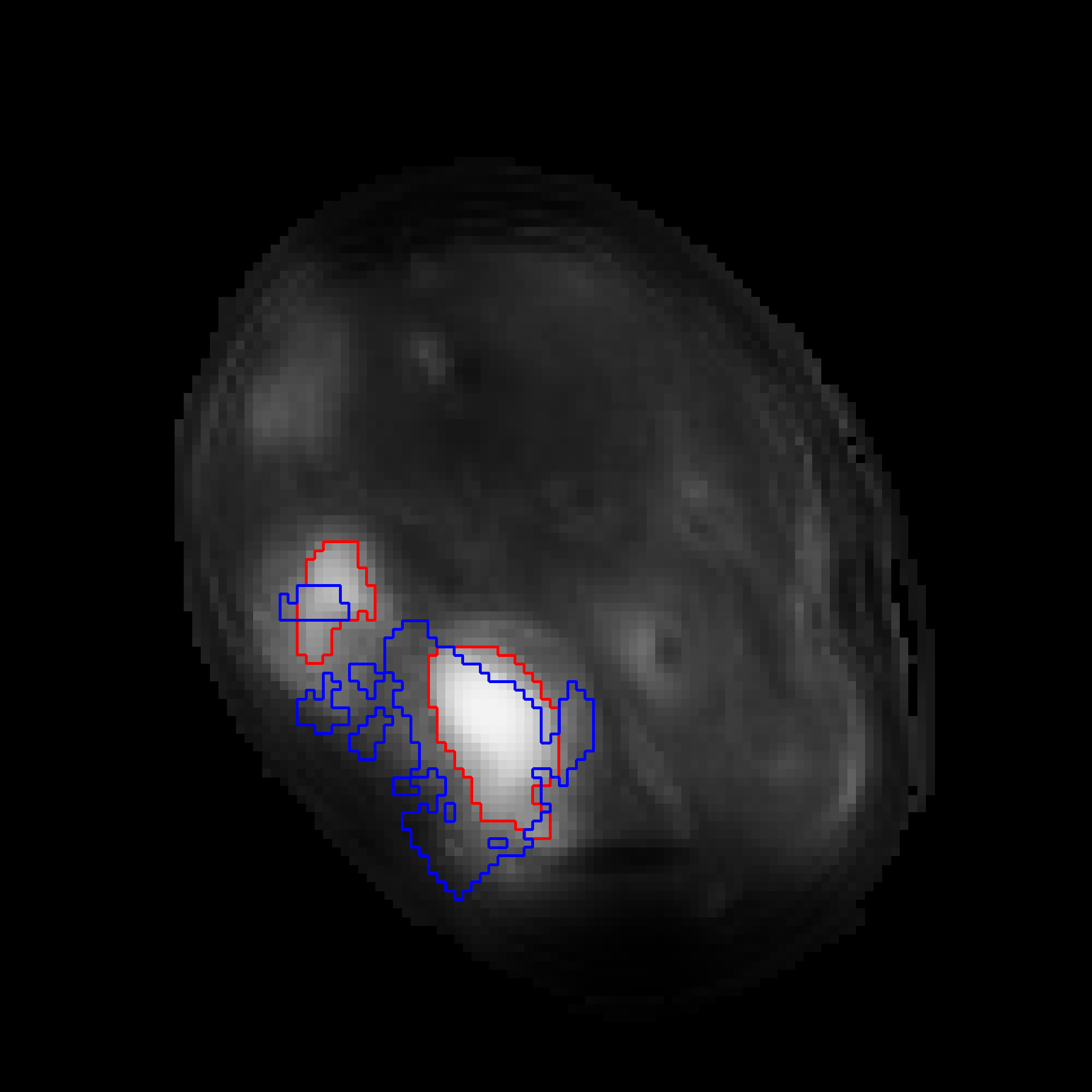} &
        \includegraphics[width=\linewidth]{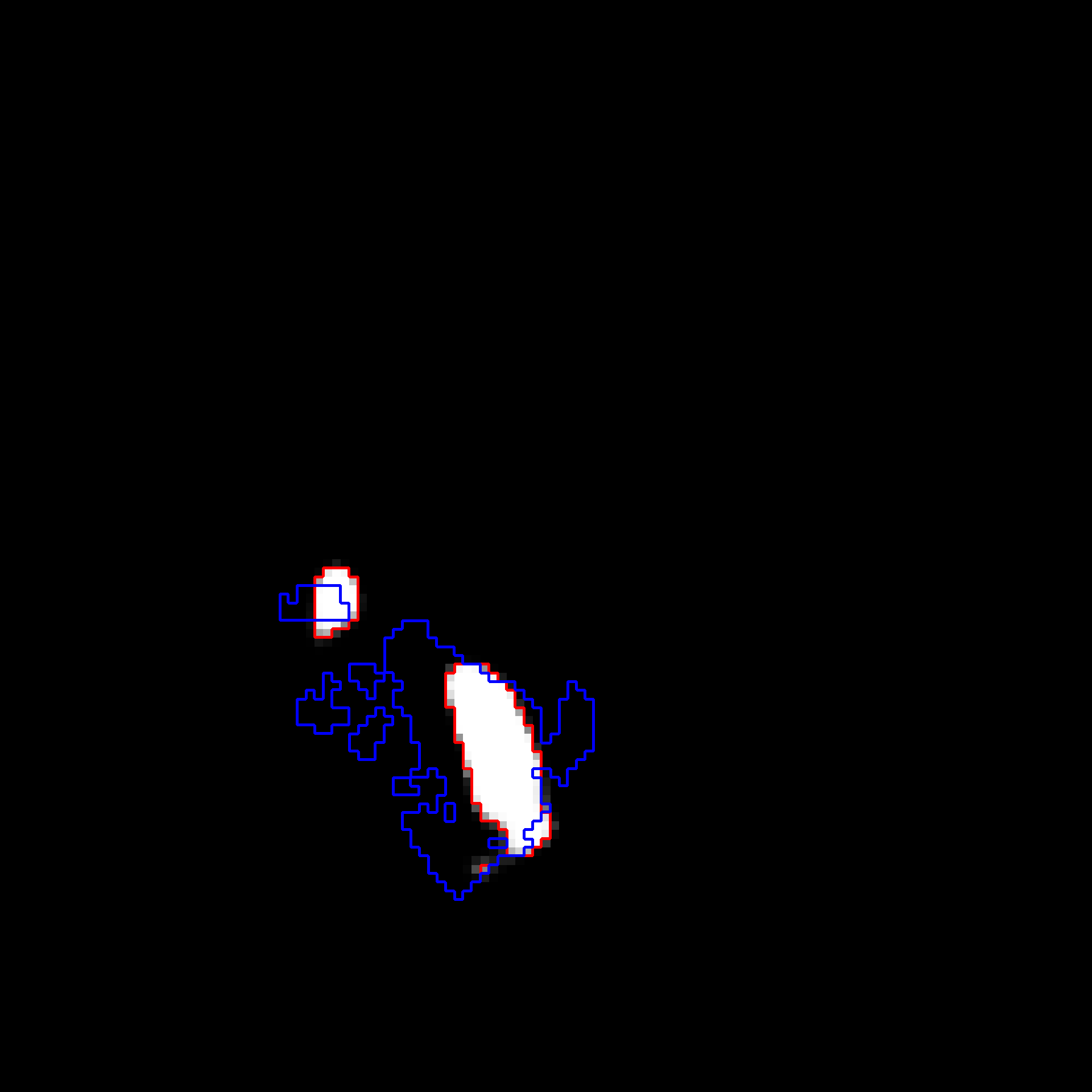} & 
        \includegraphics[width=\linewidth]{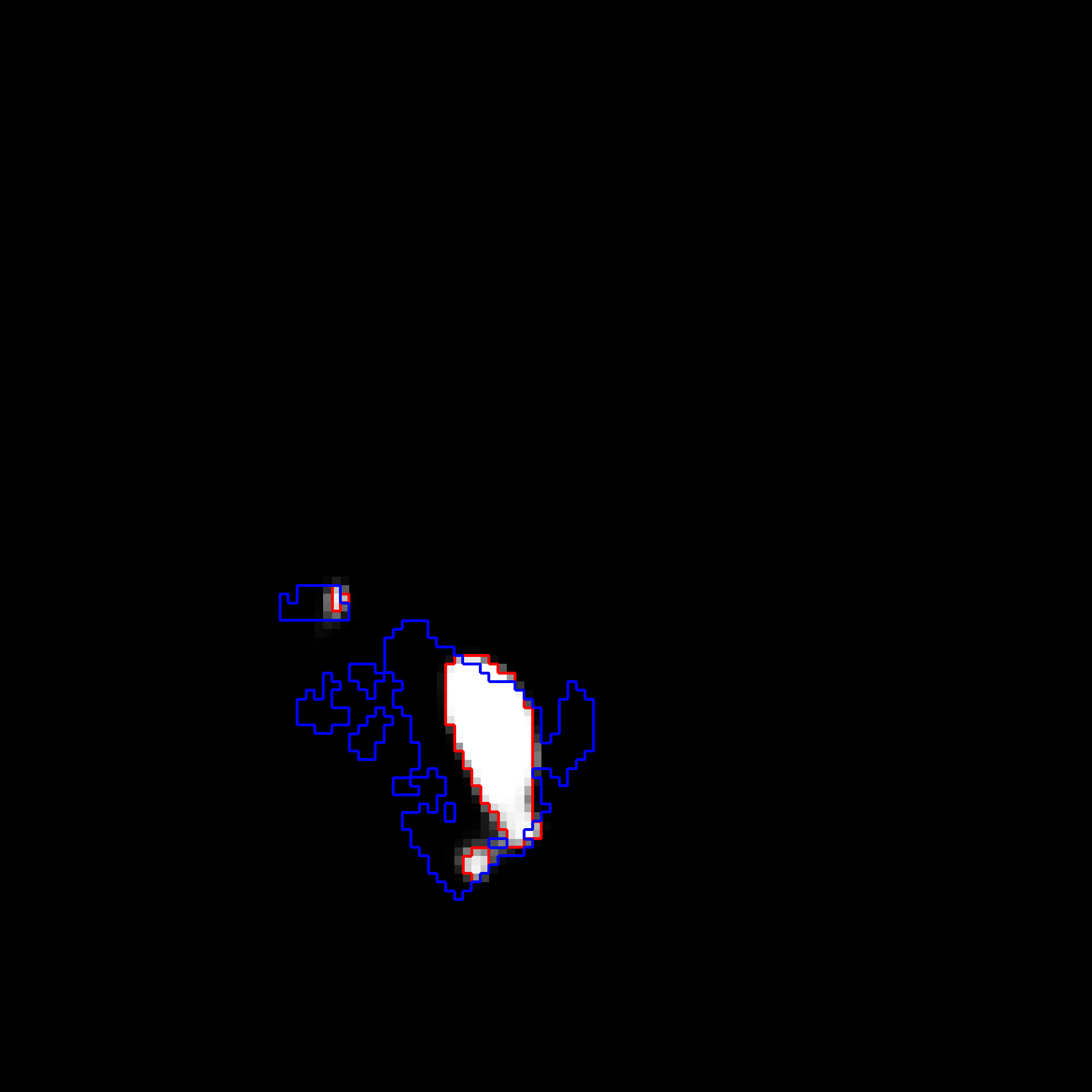} &
        \includegraphics[width=\linewidth]{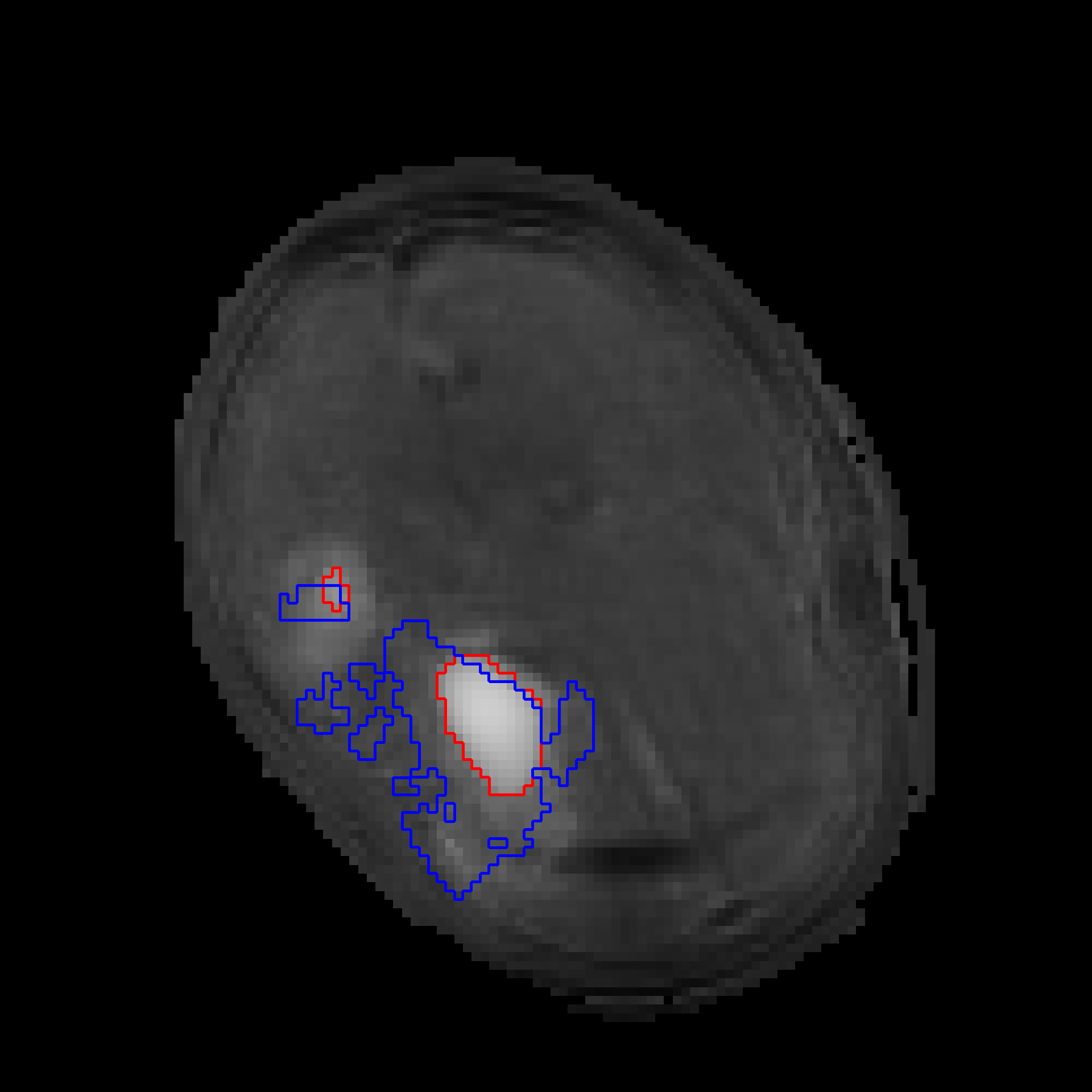} \\ 
        
        \rotatebox{90}{\hspace{20pt} MO17} &
        \includegraphics[width=\linewidth]{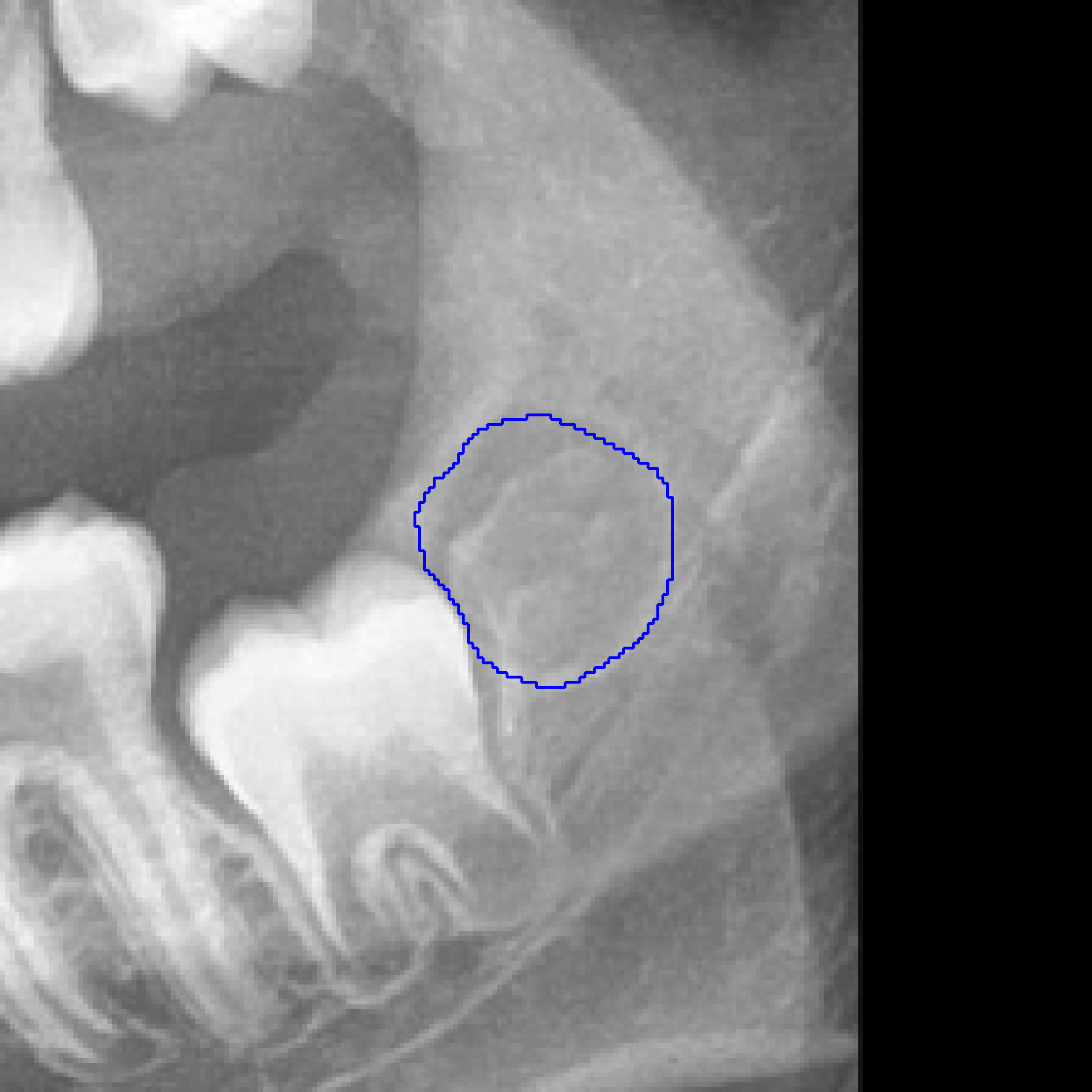} & 
        \includegraphics[width=\linewidth]{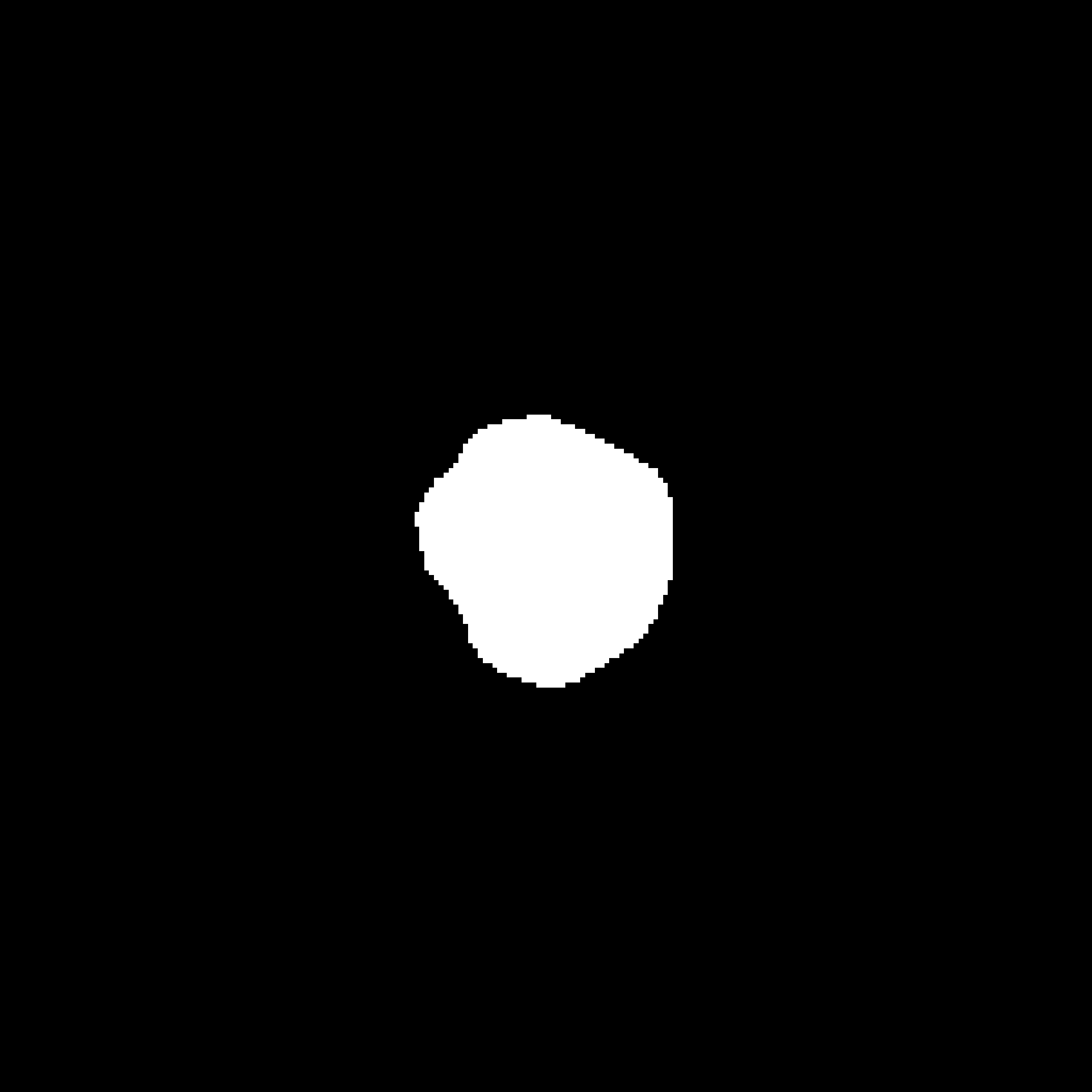} &
        \includegraphics[width=\linewidth]{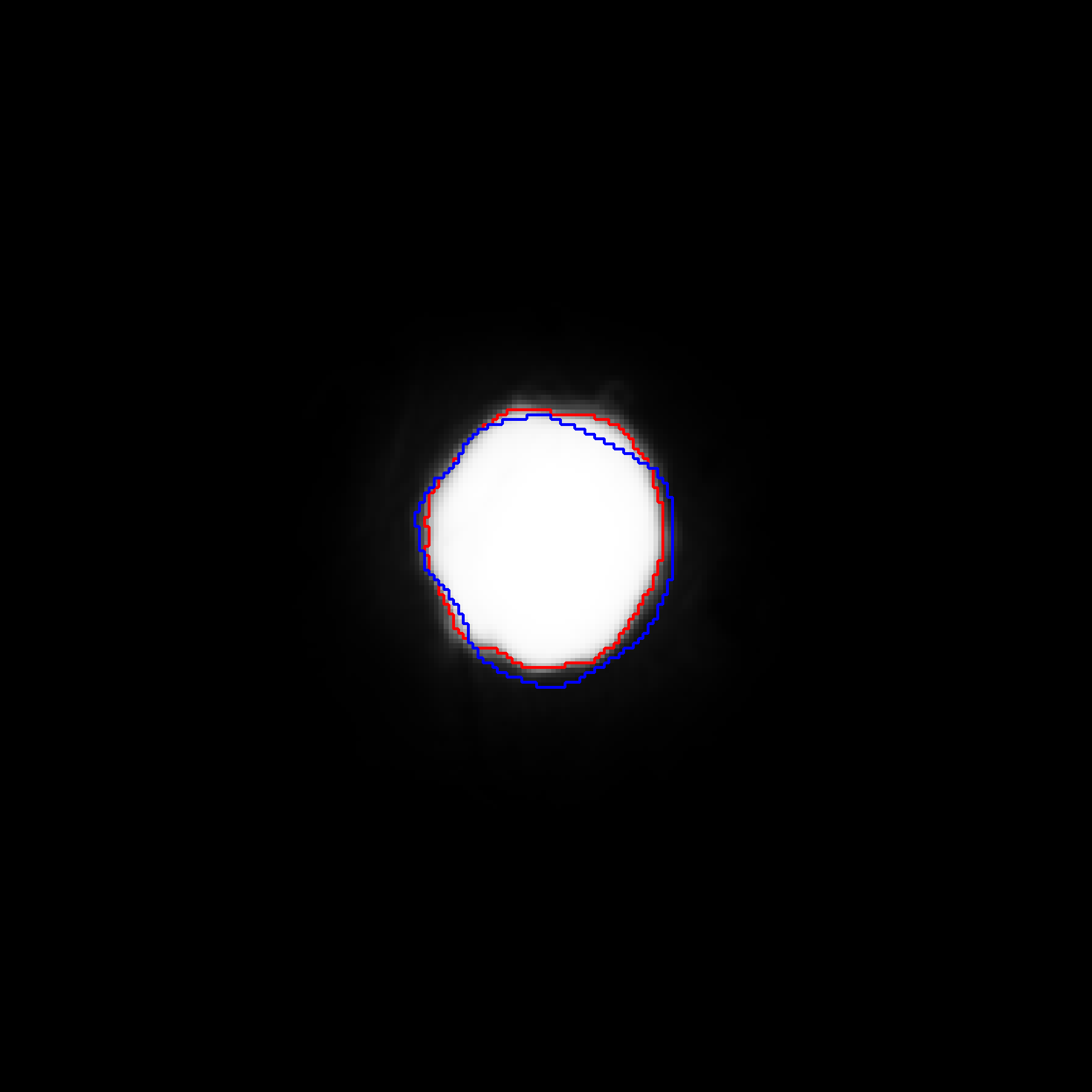} & 
        \includegraphics[width=\linewidth]{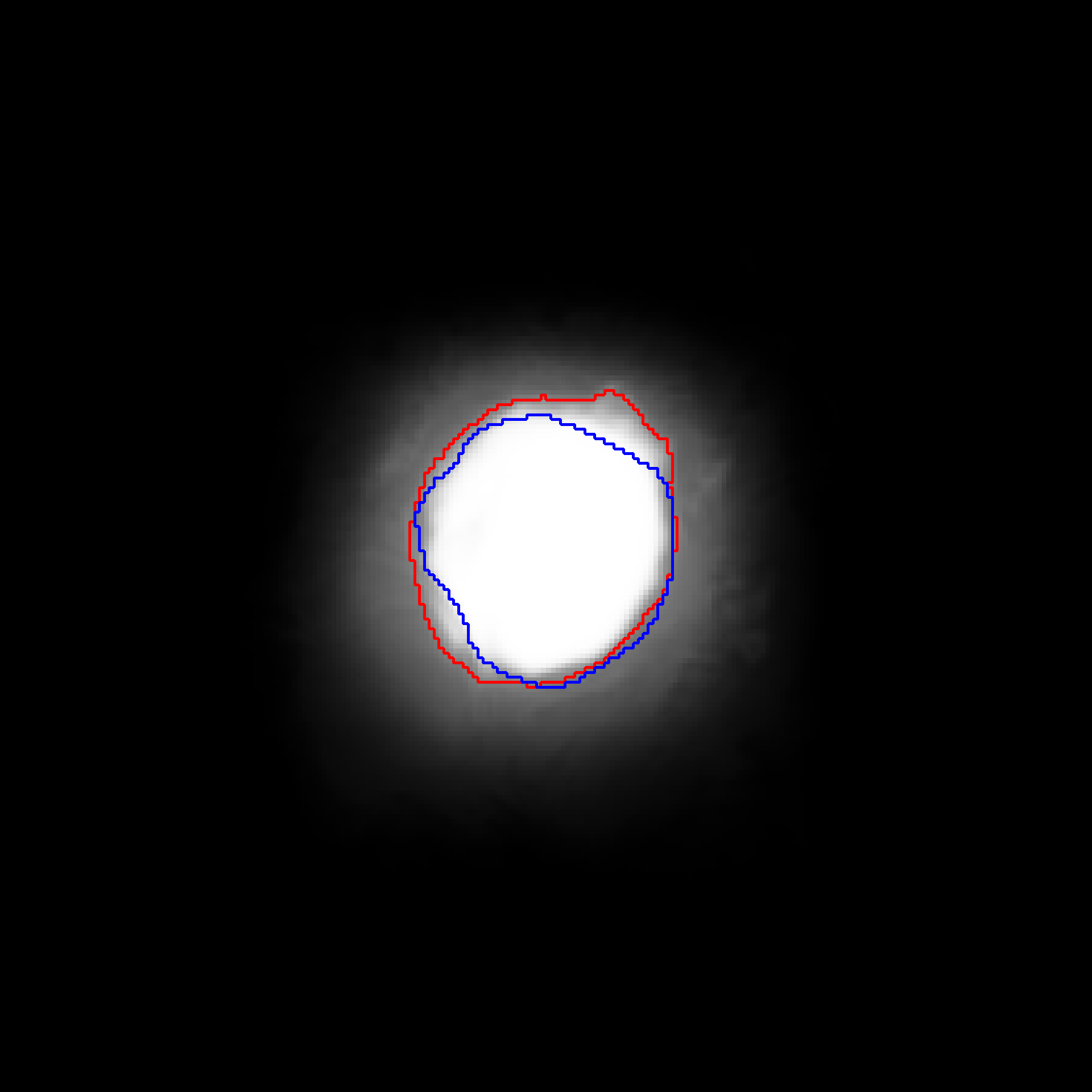} &
        \includegraphics[width=\linewidth]{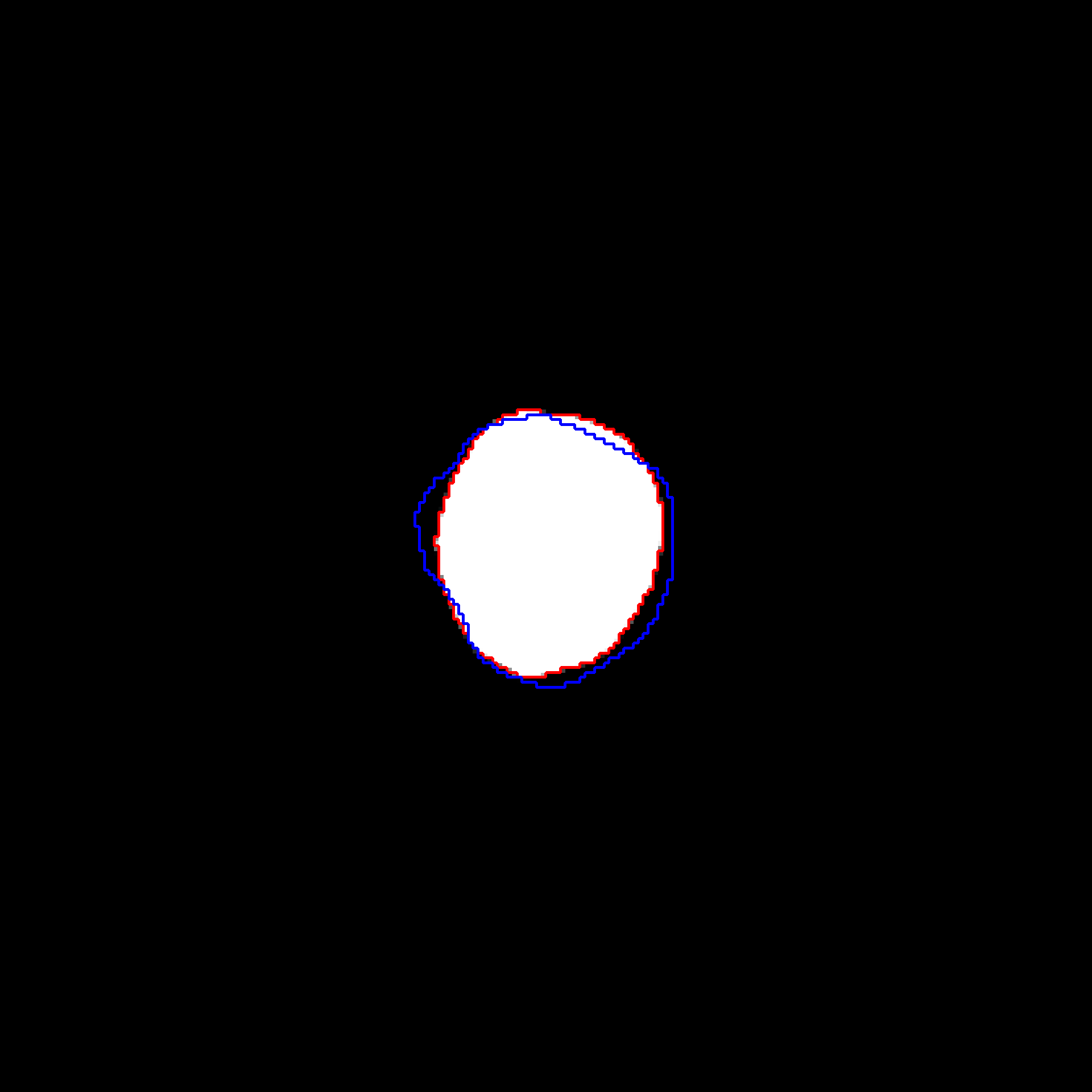} & 
        \includegraphics[width=\linewidth]{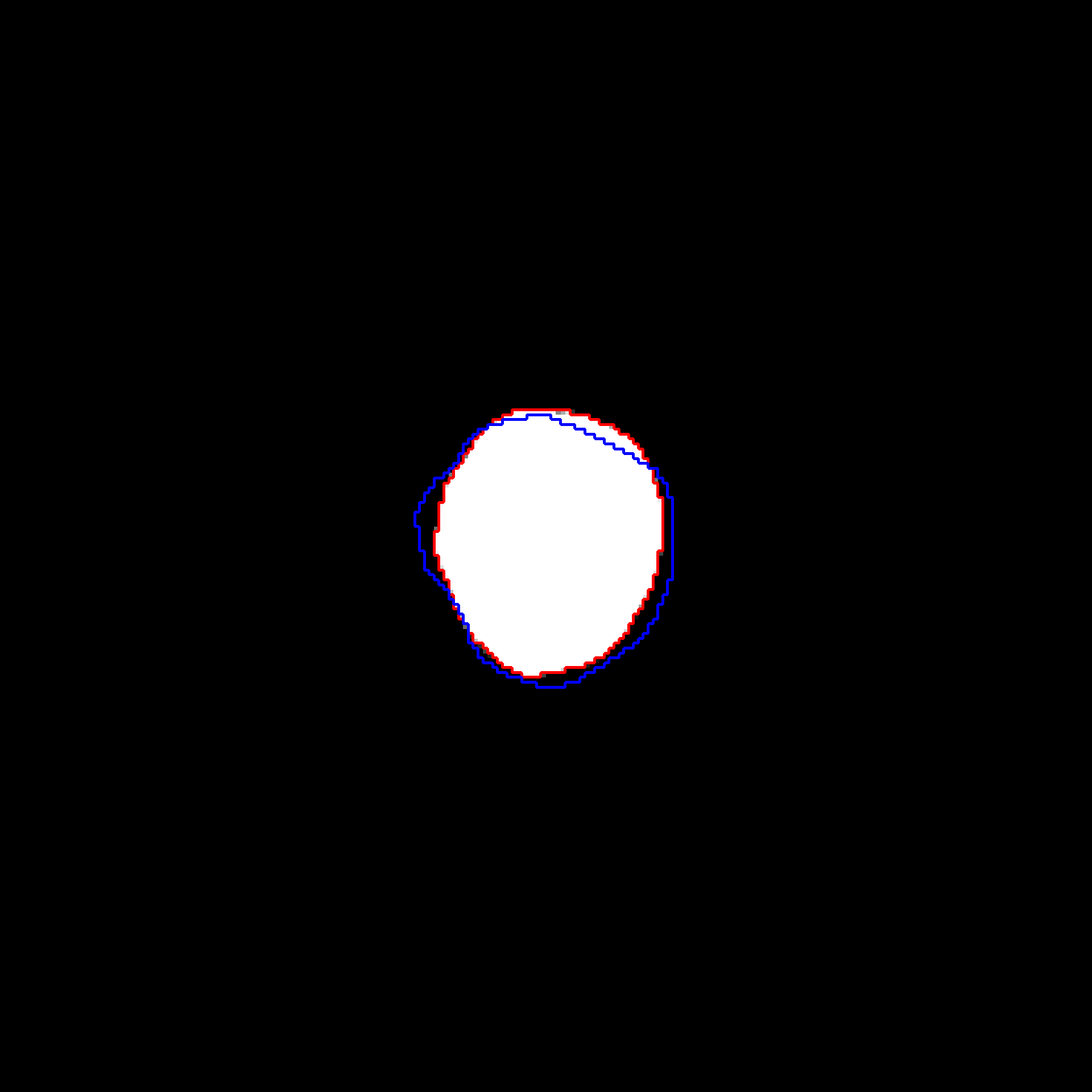} &
        \includegraphics[width=\linewidth]{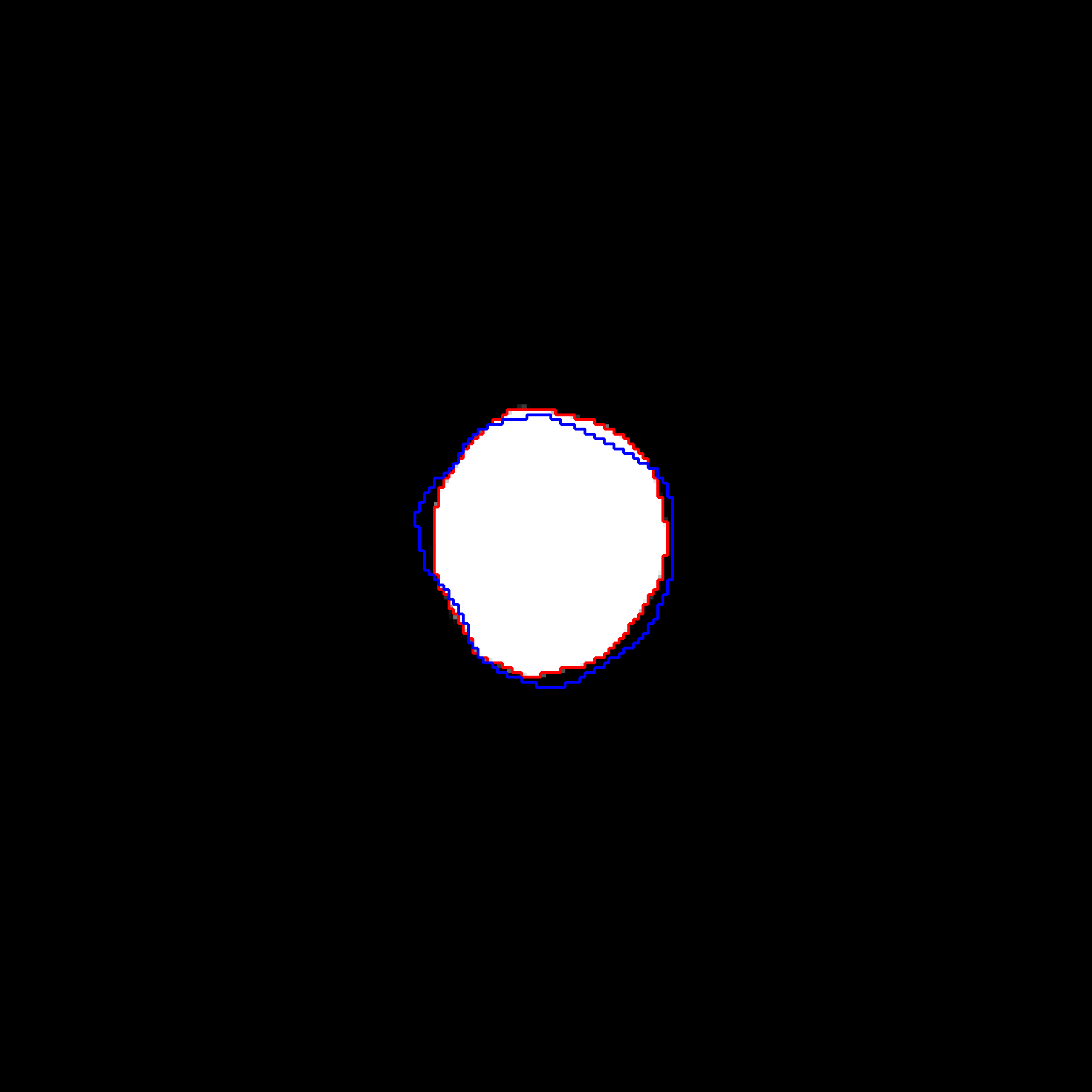} \\
        
        \rotatebox{90}{\hspace{5pt} PO18} &
        \includegraphics[width=\linewidth]{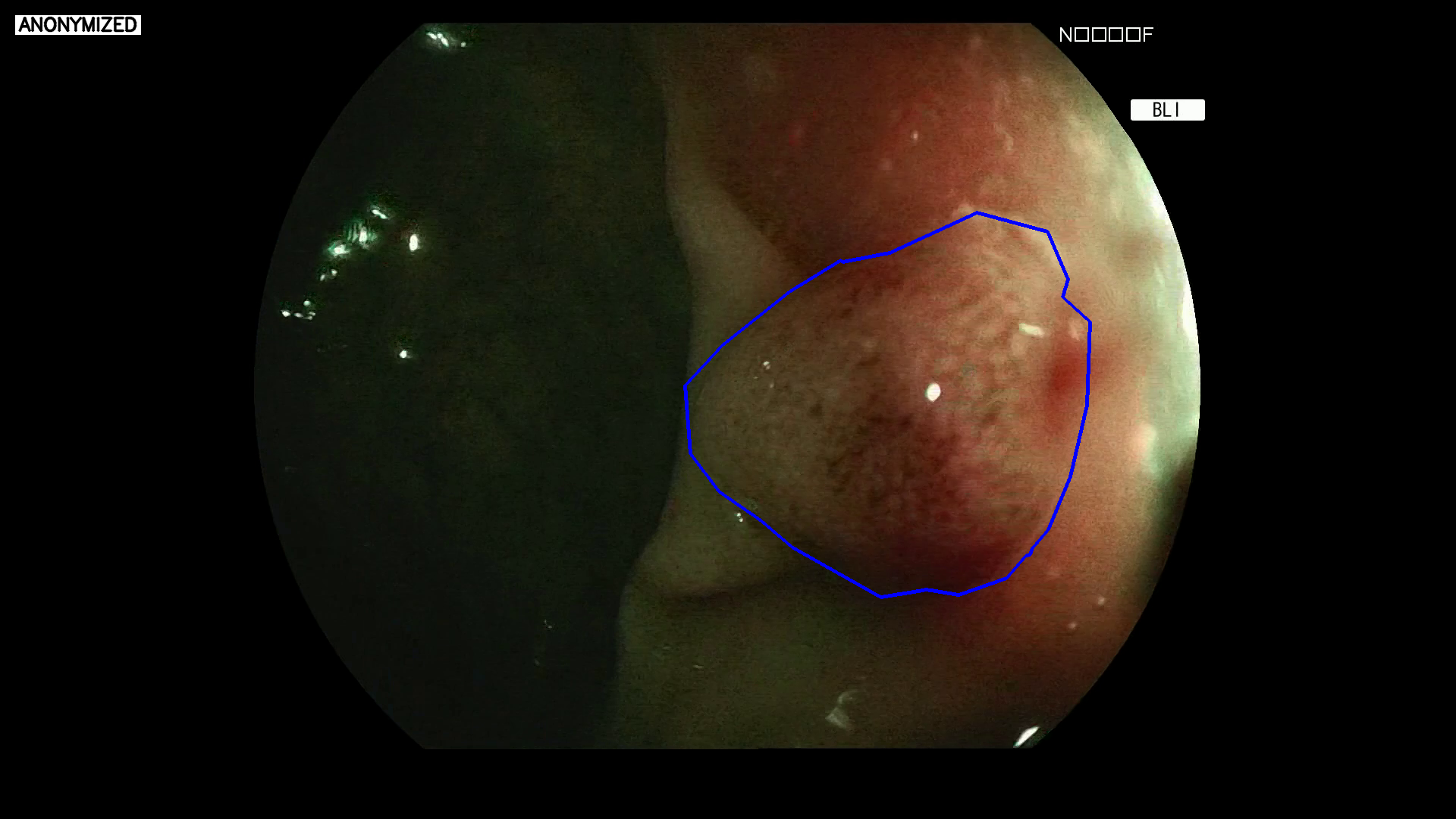} &
        \includegraphics[width=\linewidth]{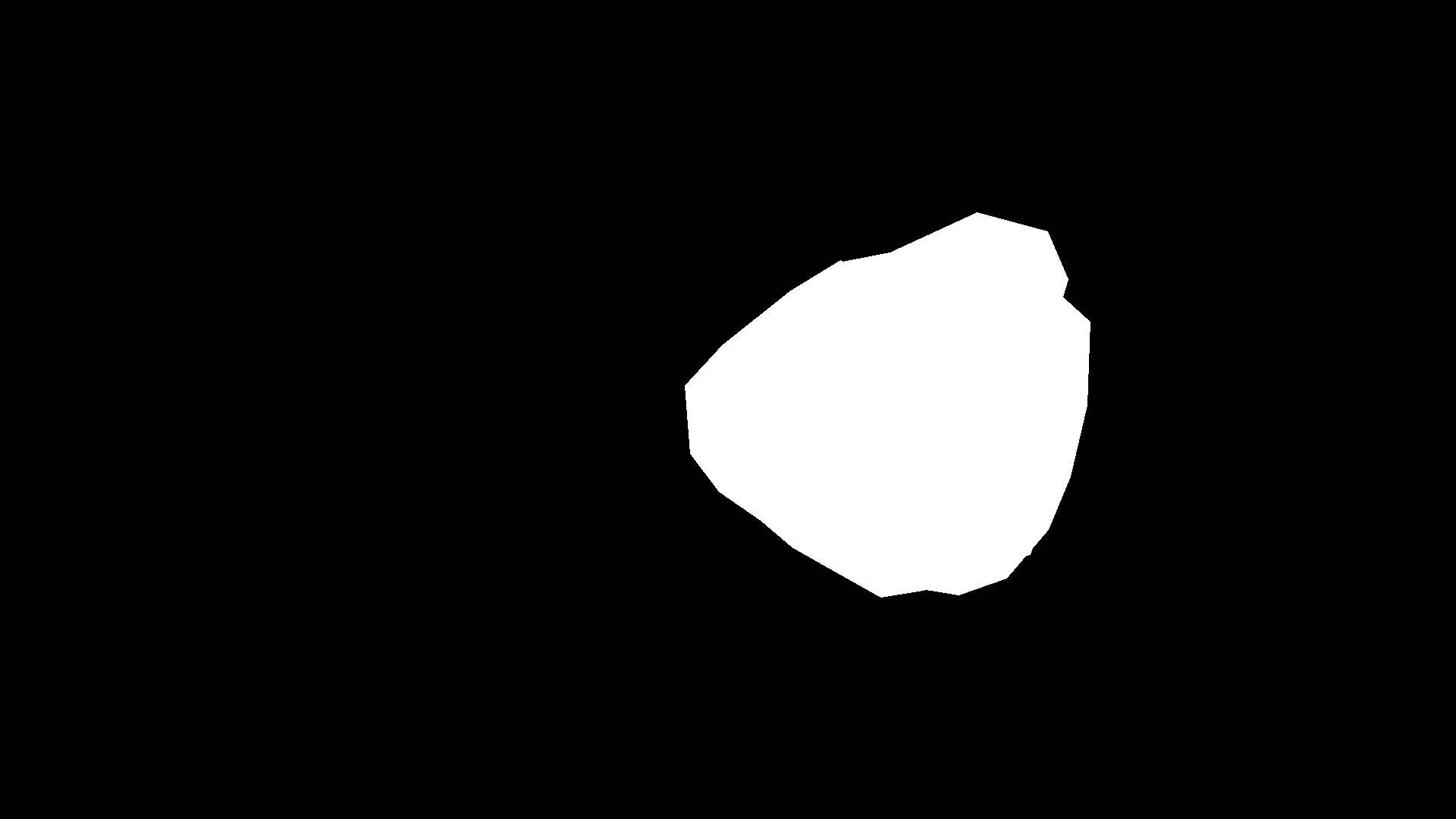} &
        \includegraphics[width=\linewidth]{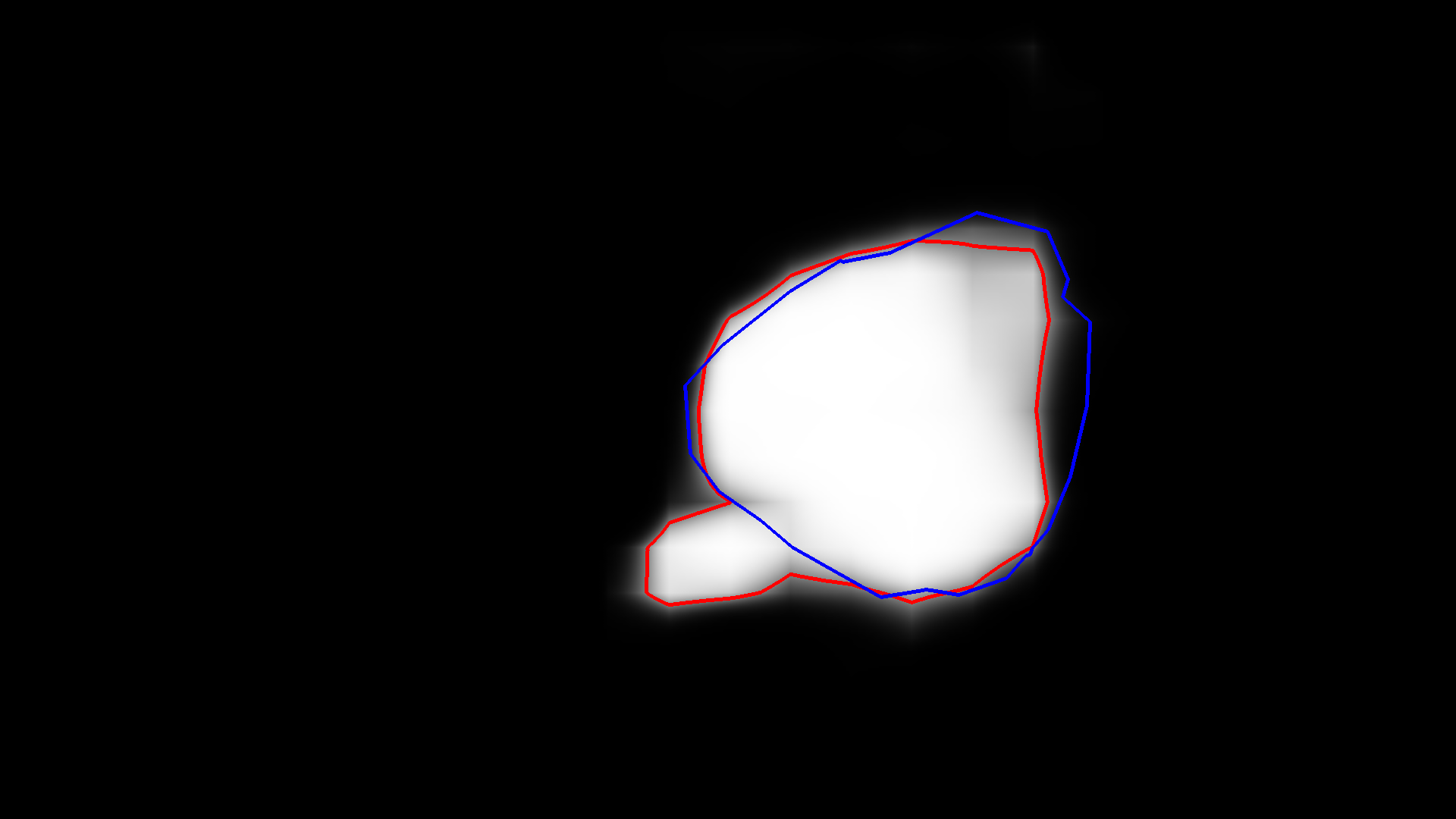} & \includegraphics[width=\linewidth]{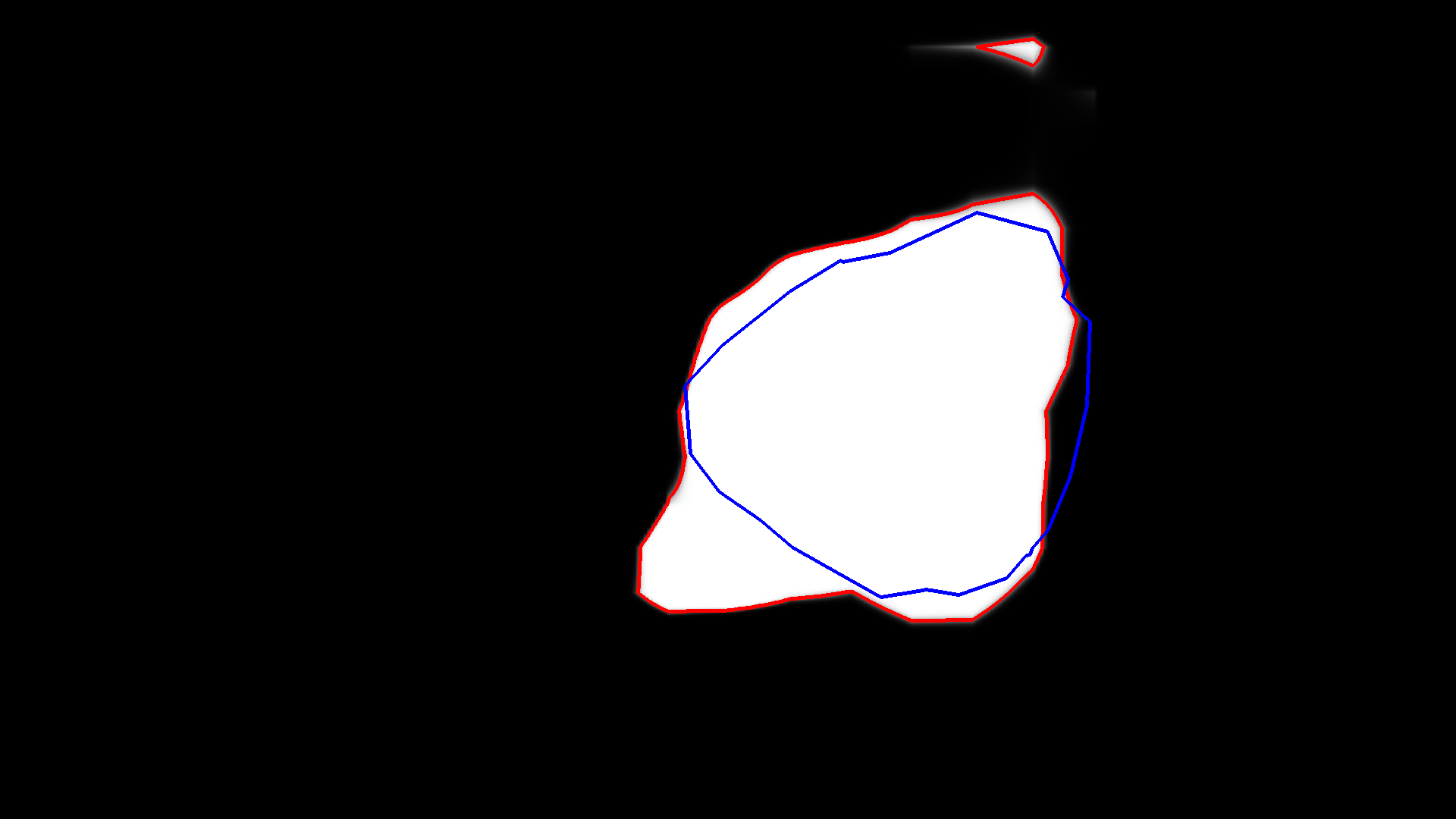} &
        \includegraphics[width=\linewidth]{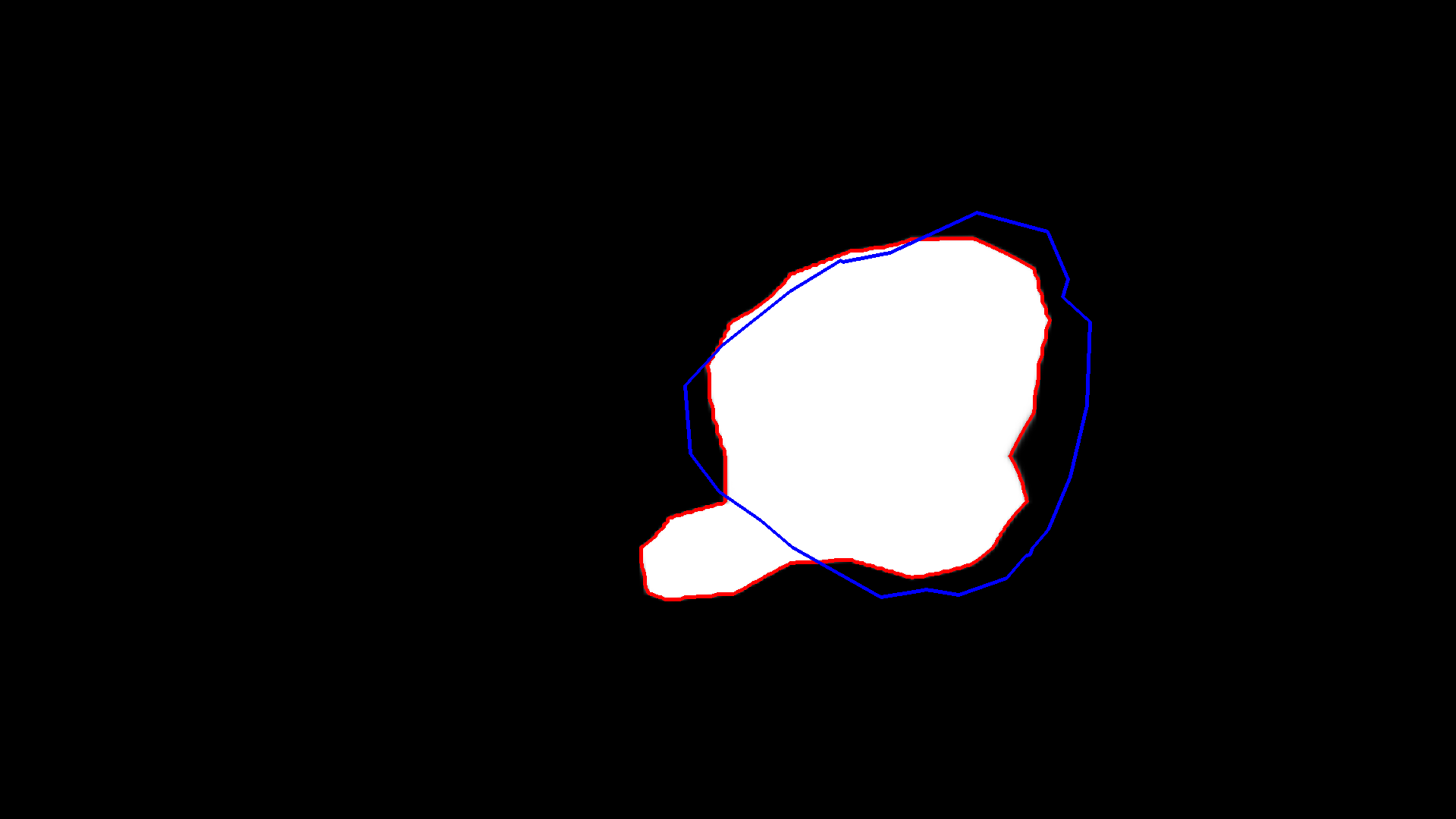} & \includegraphics[width=\linewidth]{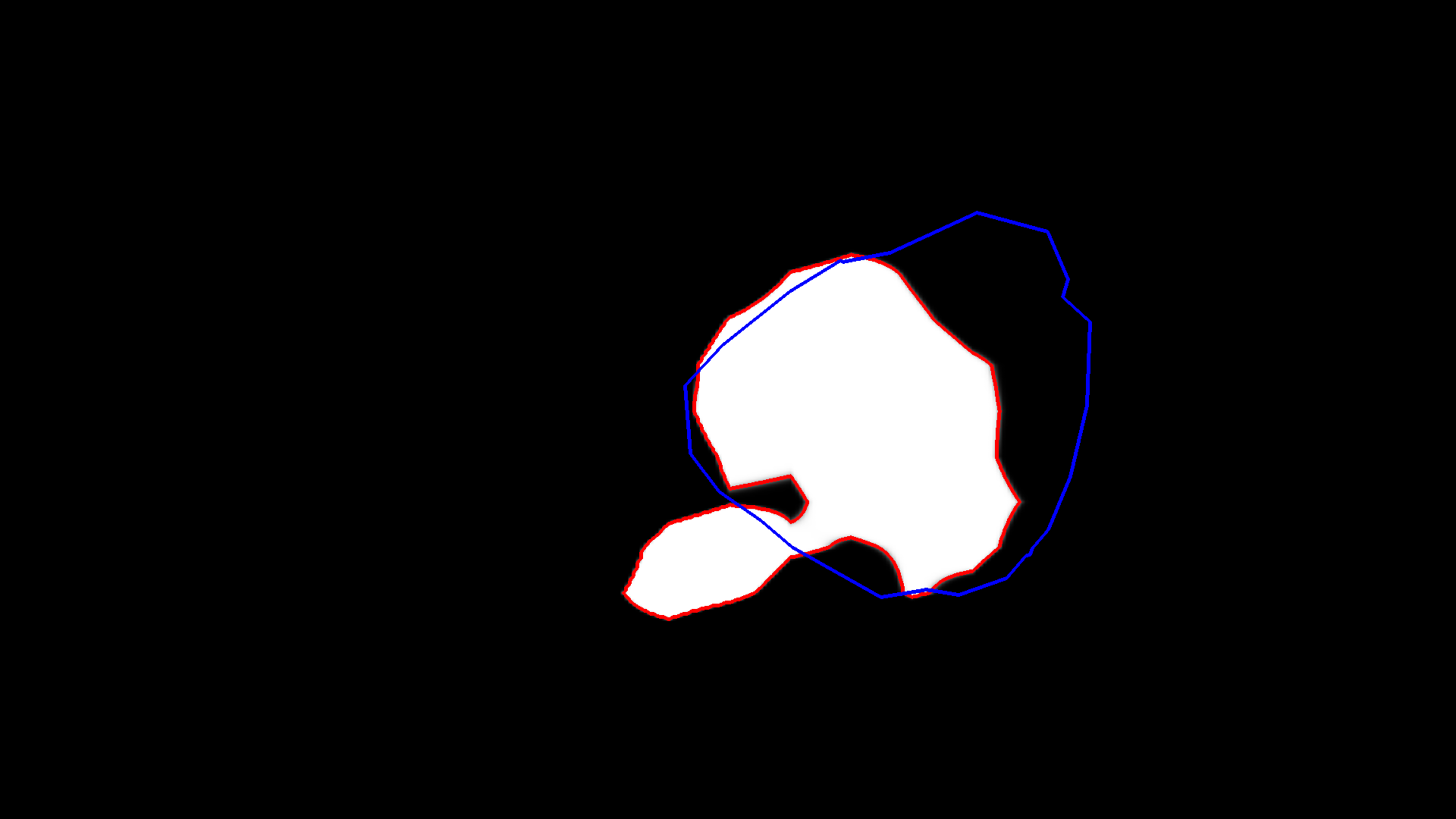} &
        \includegraphics[width=\linewidth]{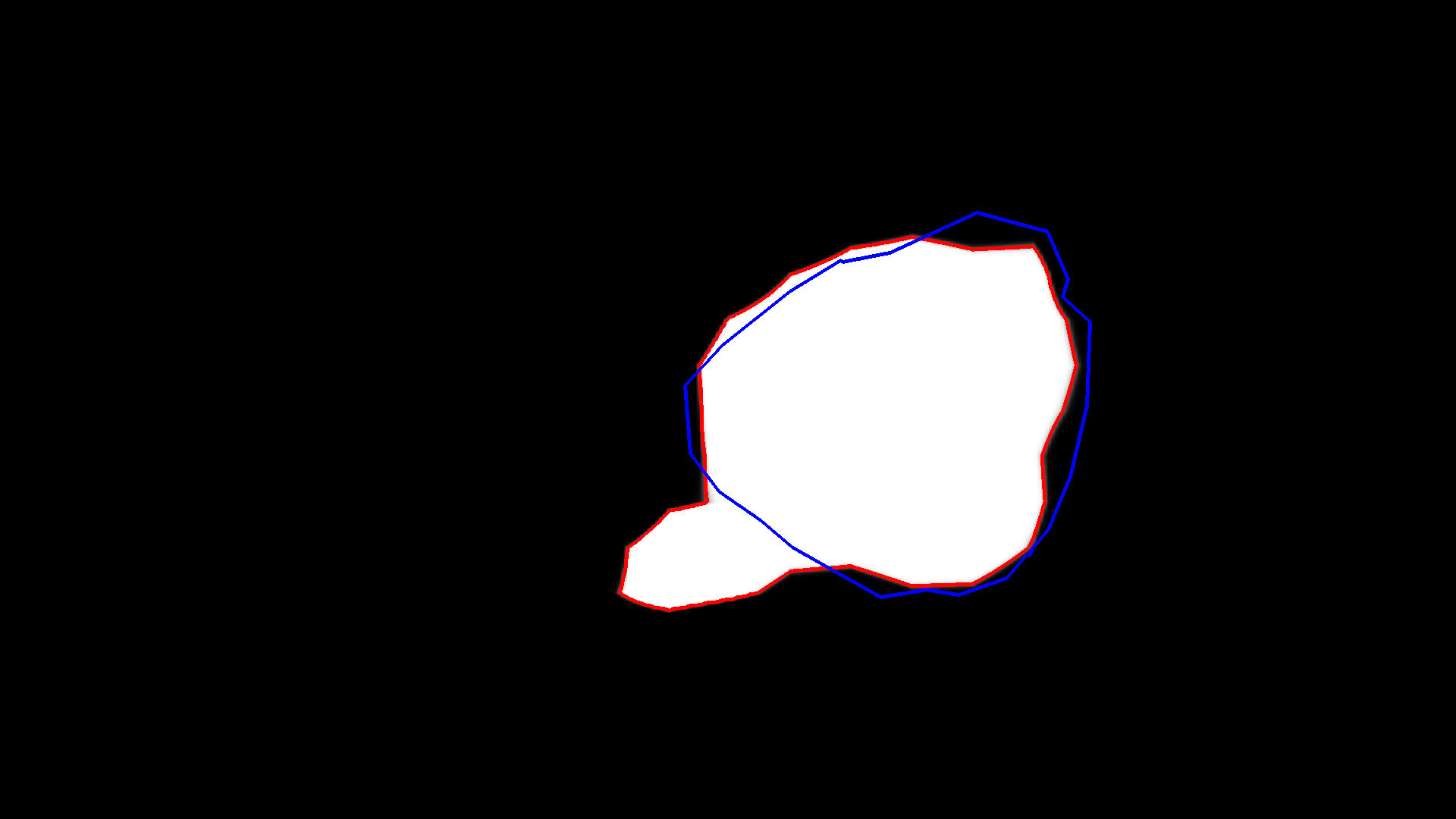} \\ 
        
        \rotatebox{90}{\hspace{45pt} \review{WM17}} &
        \includegraphics[width=\linewidth]{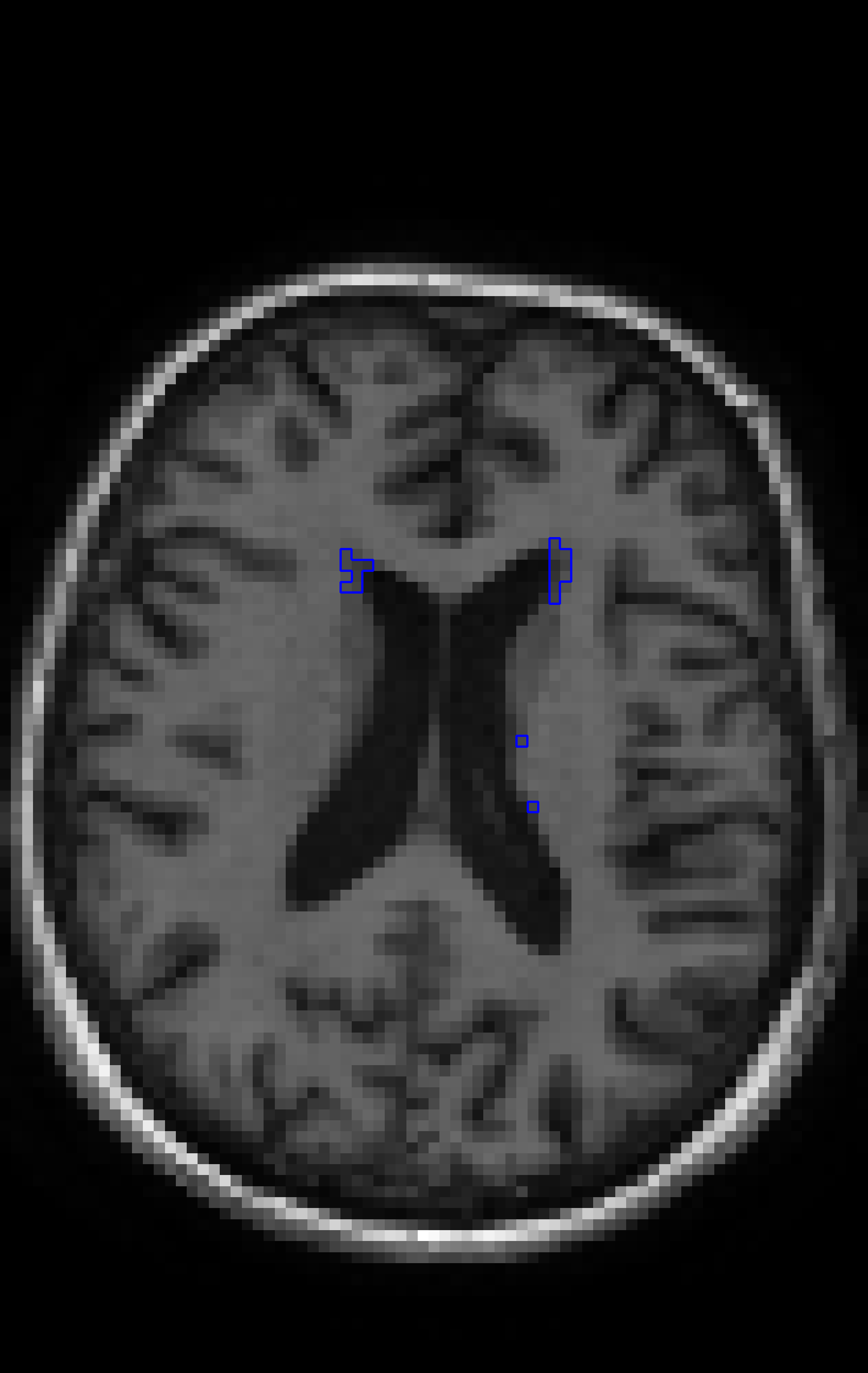} & 
        \includegraphics[width=\linewidth]{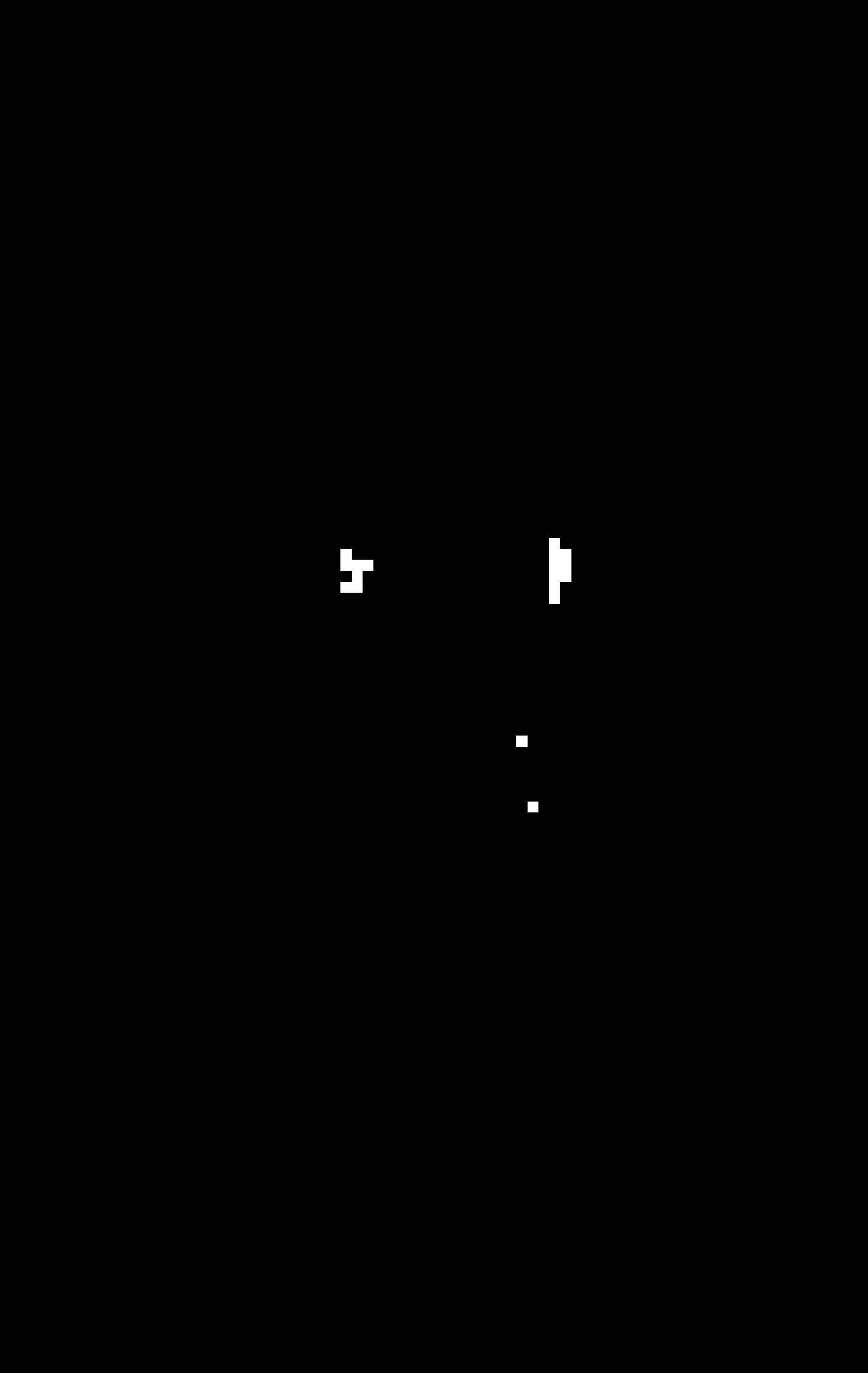} &
        \includegraphics[width=\linewidth]{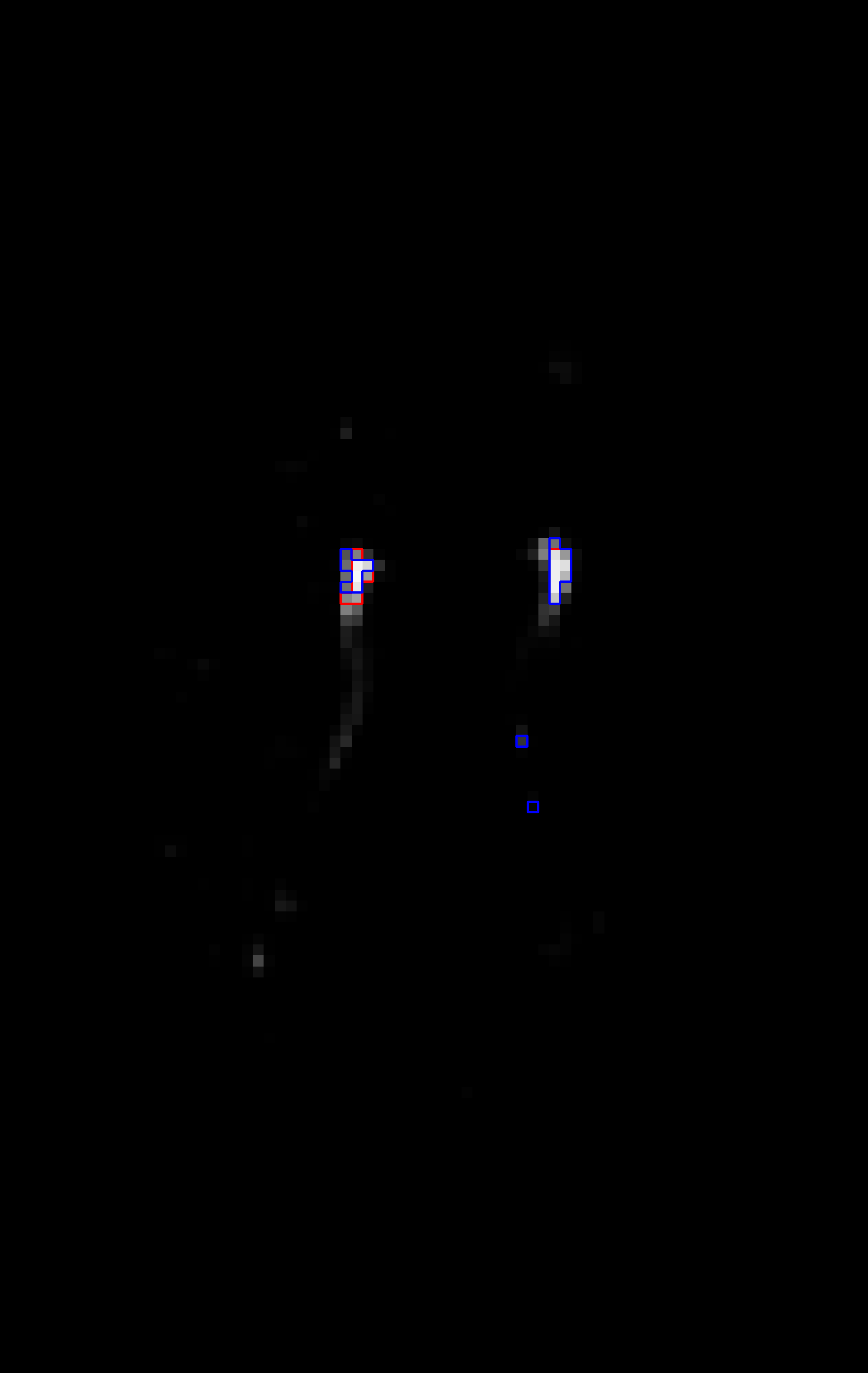} & 
        \includegraphics[width=\linewidth]{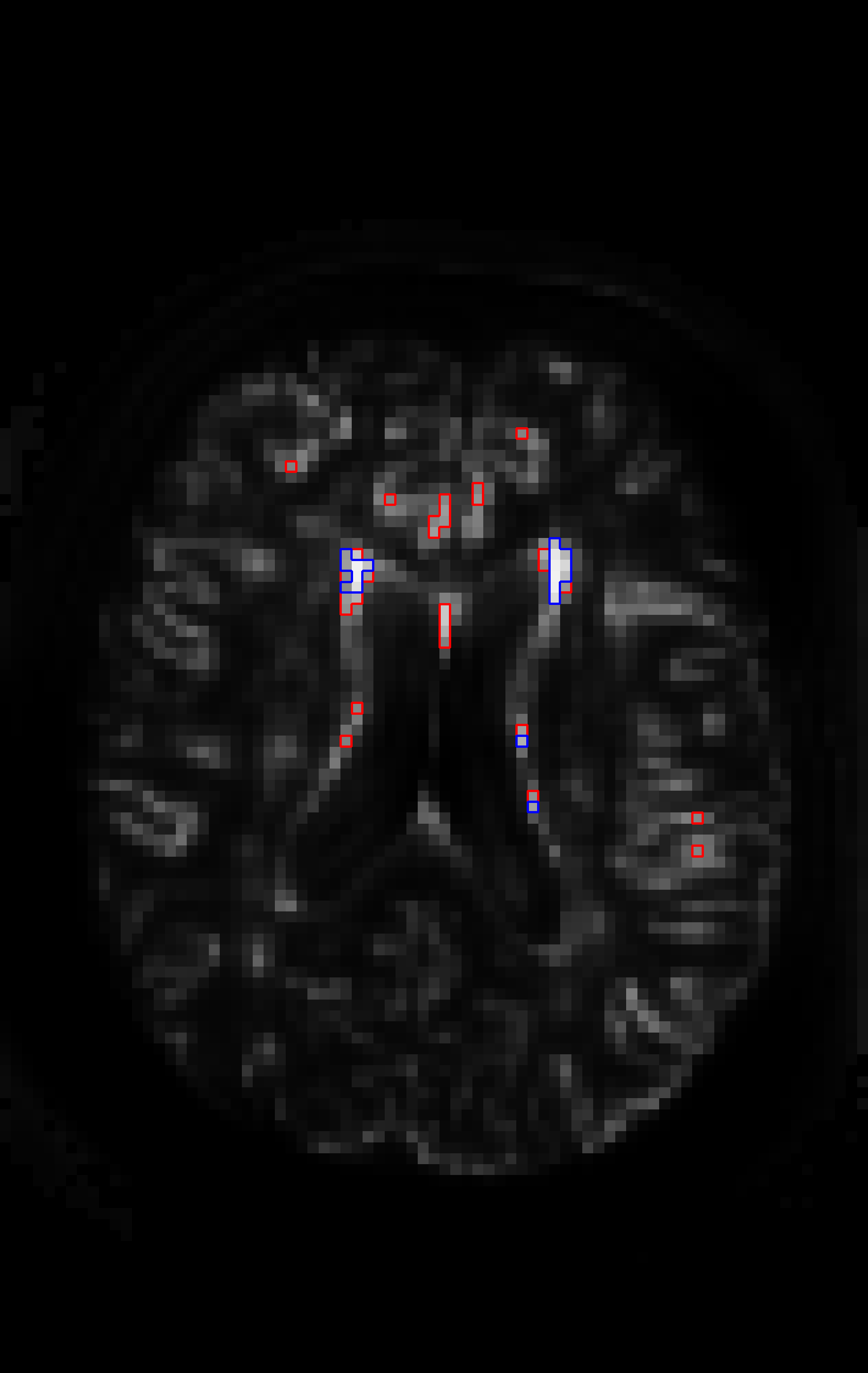} &
        \includegraphics[width=\linewidth]{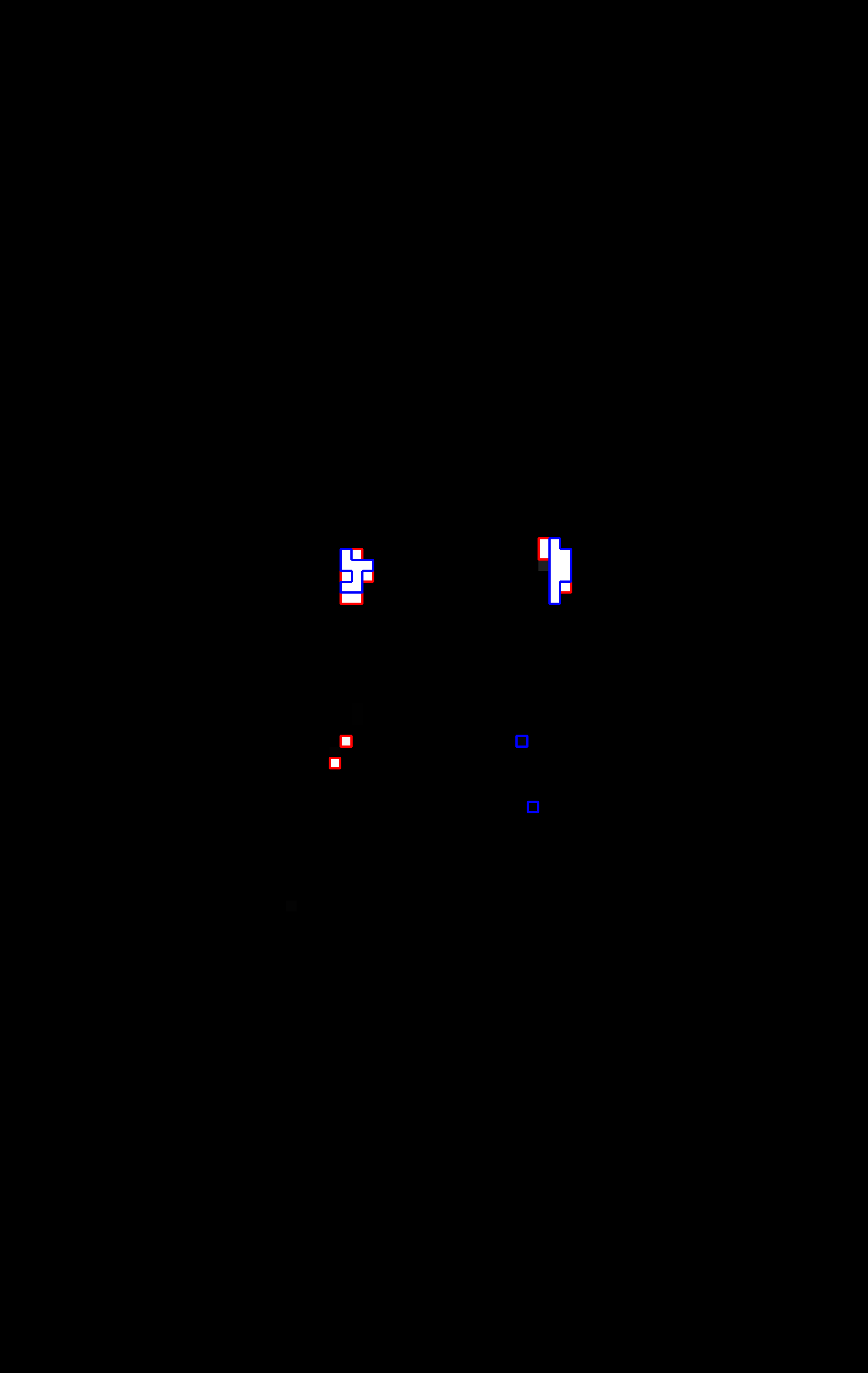} & 
        \includegraphics[width=\linewidth]{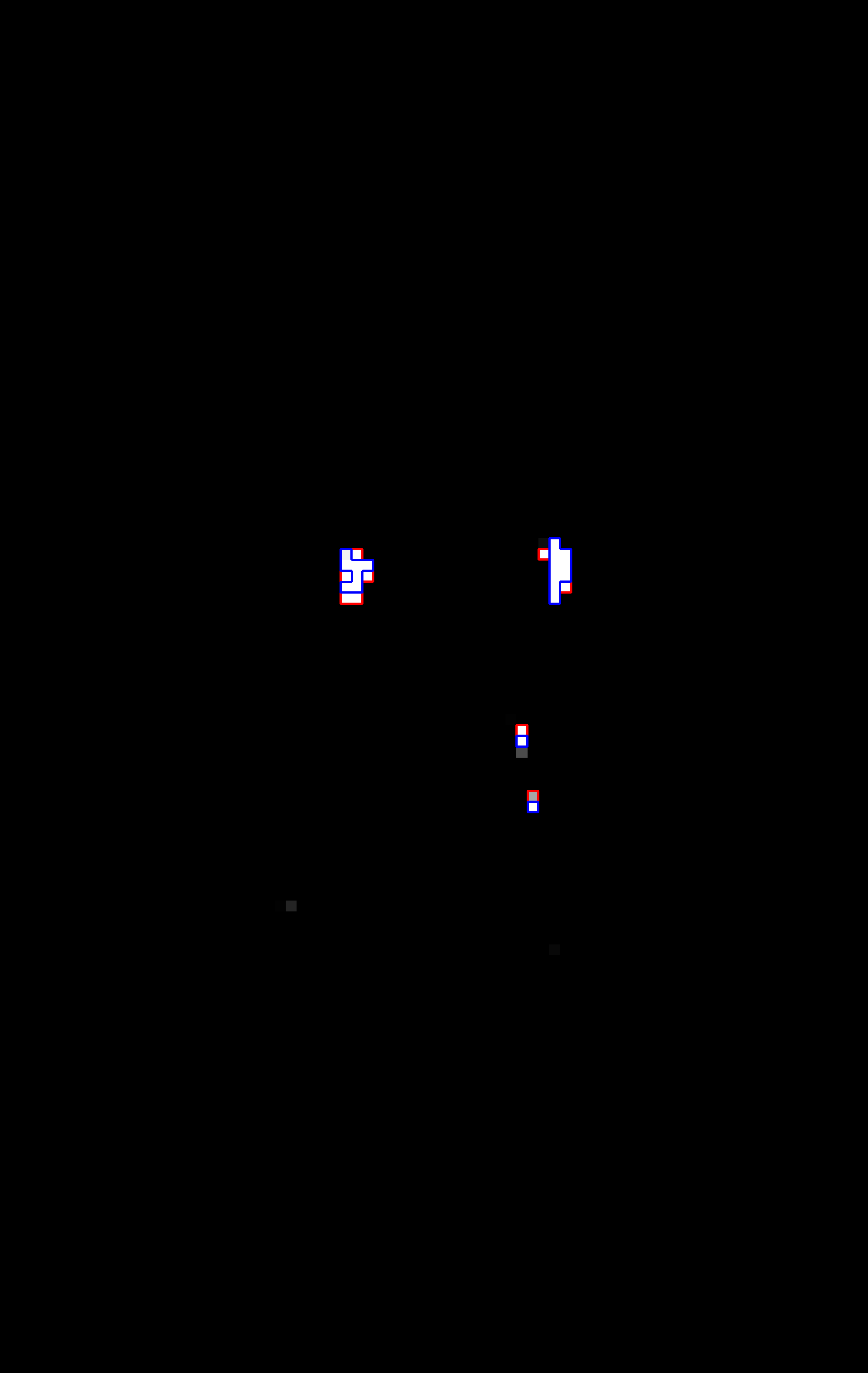} &
        \includegraphics[width=\linewidth]{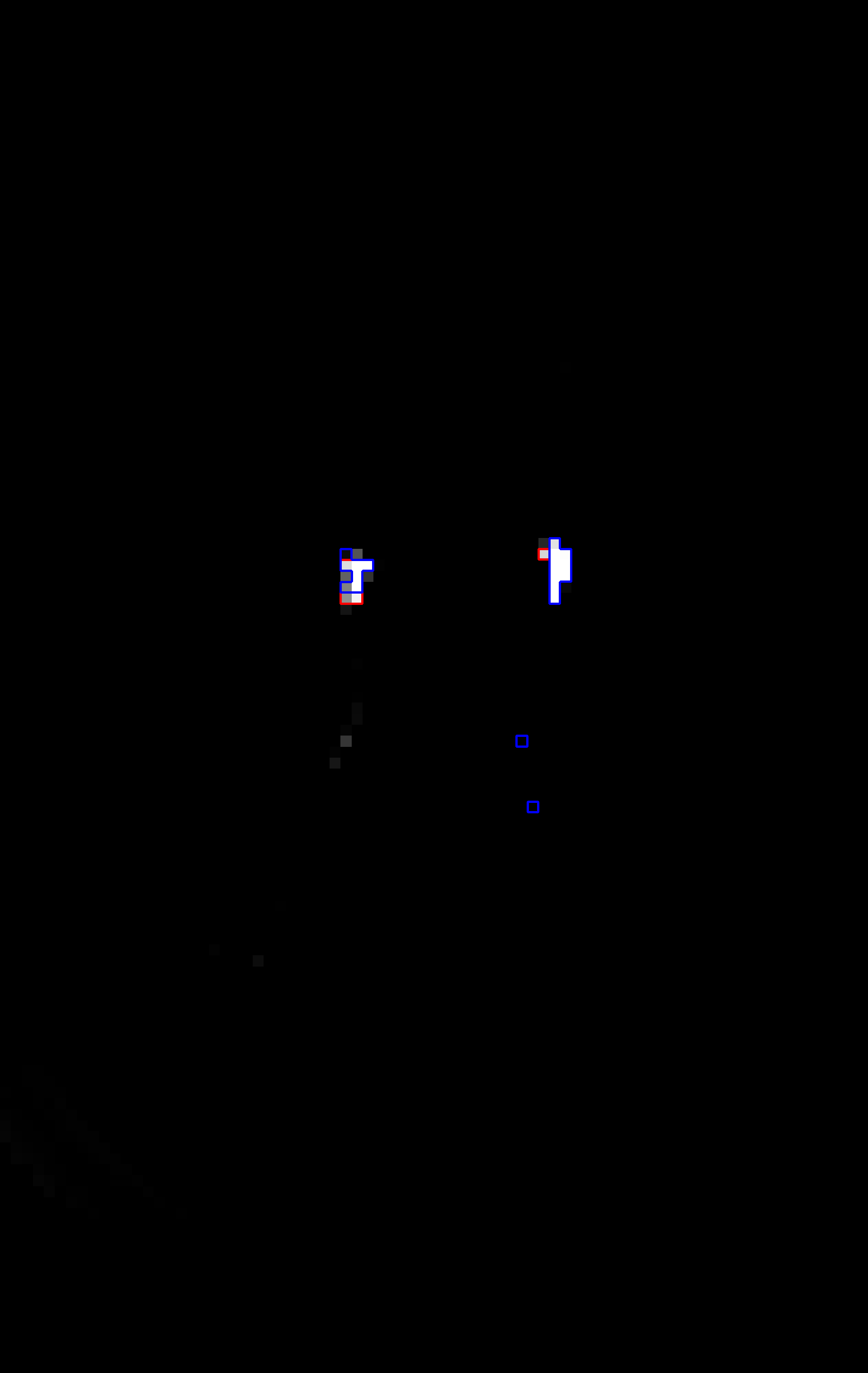} \\
    \end{tabularx}

    \caption{\review{A qualitative evaluation of the resulting predictions from networks trained with CE, wCE, sDice, sJaccard and Lov\'{a}sz, columns 3-7, respectively. The predicted delineations are obtained after thresholding at 0.5 and are overlayed in red. The original images inputted to the network (or part of) are shown on the left. The ground truths are shown in the second column and their delineations are overlayed in blue on the predictions. The slices shown had a median Dice, both across cases and within one case (having a minimum of one foreground pixel), for the CE model and were checked visually for being representative.}}
    \label{fig:losses_qual}
\end{sidewaysfigure*}

\subsubsection{Qualitative inspection\label{par:qualitative_inspection}}
In Fig. \ref{fig:losses_qual} we present example slices for each dataset. From left to right the images represent the input (with ground truth delineations in blue), the ground truth segmentation and the predictions using each of the original loss functions (with ground truth delineations in blue and predicted delineations in red). The slices shown had a median Dice, both across cases and within one case (having a minimum of one foreground pixel), for the CE model and were checked for being representative. We note a wider value range for the outputs of the cross-entropy based losses. In particular, at a threshold of 0.5, the CE model seems to systematically under-segment the structures. We further note that the outputs of the metric-sensitive losses are highly similar and have a narrow value range close to 1.0. Only for both ISLES datasets, the Lov\'{a}sz predictions have a value range close to 0.5. However, after thresholding at 0.5 the predicted delineations are similar to sDice and sJaccard.

\subsubsection{\review{Network architecture and patch-wise training}}
\review{It is clear that results are consistent over different state-of-the-art network architectures: U-Net (used in BR18, IS17, IS18, MO17 and WM17), DeepLab (used in PO18) and DeepMedic (used in WM17\textsuperscript{DM}). All architectures were shown to be optimized according to the implemented loss functions. This was somehow expected since these architectures are used almost interchangeably in recent biomedical segmentation challenges~\cite{isles2017,brats2018}.}

\review{Interestingly, the results are also consistent for patch-wise training, whereas it is unclear how the effective, patch-wise training loss relates to the actual test metric that is evaluated on full images.}

\subsubsection{\reviewminor{Generalizability to absolute and distance-based metrics}}
\label{sec:other_metrics}
\reviewminor{The results above show that the use of metric-sensitive losses provides superior Dice scores and Jaccard indices compared to (w)CE. However, often multiple metrics are of interest for evaluation such as volume errors or distance-based metrics~\cite{brats2018,Kuijf2019}. In this light, we report three additional metrics for each of the public datasets: accuracy (ACC), Hausdorff distance (HDD) and absolute volume difference (AVD). The performance on these metrics is compared for the experiments with CE and sDice in table~\ref{tab:other_metrics} (equivalence within each group of loss functions was already shown). For HDD and AVD, sDice is never inferior to CE and in certain cases significantly better. For ACC, CE is often superior, which is to be expected since it directly optimizes the evaluation metric.}

\begin{table}[htbp]
    \caption{\reviewminor{Accuracy (ACC), Hausdorff distance (HDD) and absolute volume difference (AVD) for each of the public datasets and for networks trained with CE and sDice. Values highlighted in grey point to a significantly better value compared to the other loss. Italic formatting denotes significantly inferior values.}}
    \begin{tabularx}{\linewidth}{llXYYYYY}
    \toprule
         & Loss & \multicolumn{1}{r}{\lapbox[\width]{1em}{\emph{Dataset} $\rightarrow$}} & BR18 & IS17 & IS18 & WM17 & WM17\textsuperscript{DM} \\
    \midrule
    \parbox[t]{3mm}{\centering\multirow{2}{*}{\rotatebox[origin=c]{90}{ACC}}}
      & CE    & & \textit{0.997} & \cellcolor{gray!50} 0.995 & \cellcolor{gray!50} 0.989 & 0.999 & 0.999 \\
      & sDice & & \cellcolor{gray!50} 0.998 & \textit{0.993} & \textit{0.988} & 0.999 & 0.999 \\
    \midrule
    \parbox[t]{3mm}{\centering\multirow{2}{*}{\rotatebox[origin=c]{90}{HDD}}}
      & CE    & & 22.673 & \textit{73.102} & \textit{42.994} & 34.210 & 30.526 \\
      & sDice & & 22.911 & \cellcolor{gray!50} 47.043 & \cellcolor{gray!50} 33.366 & 34.806 & 28.695 \\
    \midrule
    \parbox[t]{3mm}{\centering\multirow{2}{*}{\rotatebox[origin=c]{90}{AVD}}}
      & CE    & & \textit{15.579} & 21.042 & 12.029 & \textit{5.628} & \textit{5.865} \\
      & sDice & & \cellcolor{gray!50} 13.375 & 35.895 & 11.954 & \cellcolor{gray!50} 3.895 & \cellcolor{gray!50} 4.064 \\
    \bottomrule
\end{tabularx}
    \label{tab:other_metrics}
\end{table}

\subsection{Influence of class imbalance}\label{sec:classimbalance}
Most datasets for image segmentation consist of objects of different sizes. The metric-sensitive losses are typically expected to help most with segmentation of small-sized objects because they are invariant to scale~\cite{Berman2018a}. The approximation bound for the Hamming loss in Eq.~\eqref{eq:hammingbound} is dependent on the number of foreground pixels in the ground truth and this makes us think that optimization with cross-entropy will lead to a lower Dice score for small objects.

\begin{figure*}[htbp]%
\sbox\foursubbox{%
  \resizebox{\textwidth}{!}{%
    \includegraphics[height=3cm]{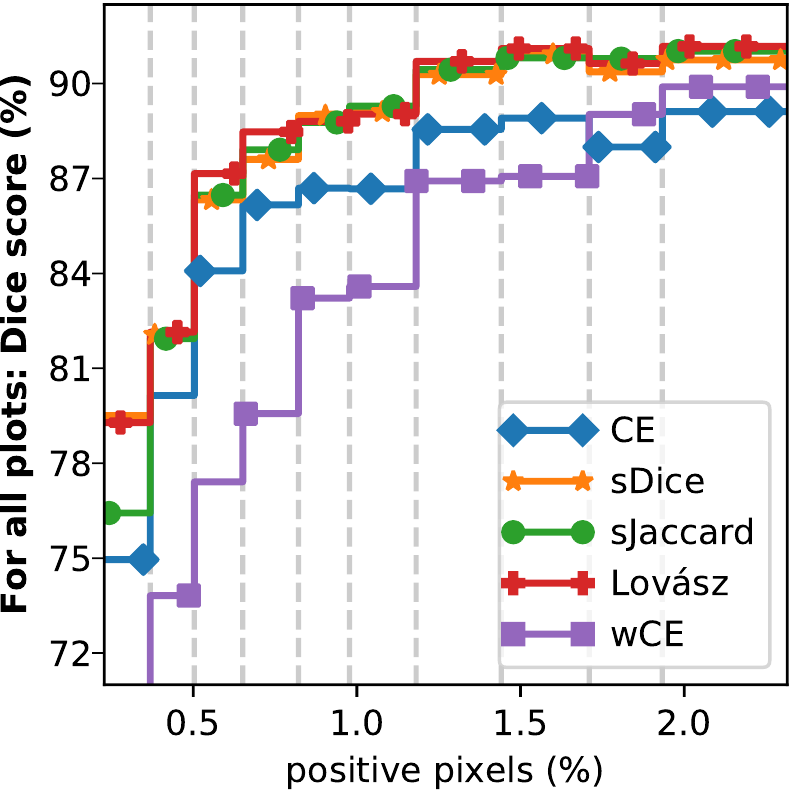}%
    \includegraphics[height=3cm]{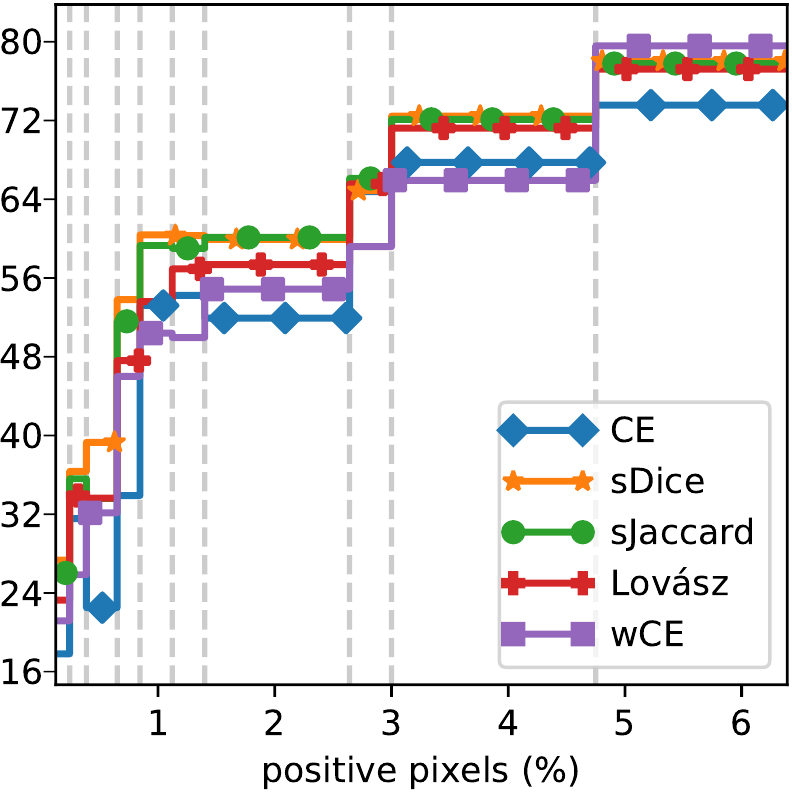}%
    \includegraphics[height=3cm]{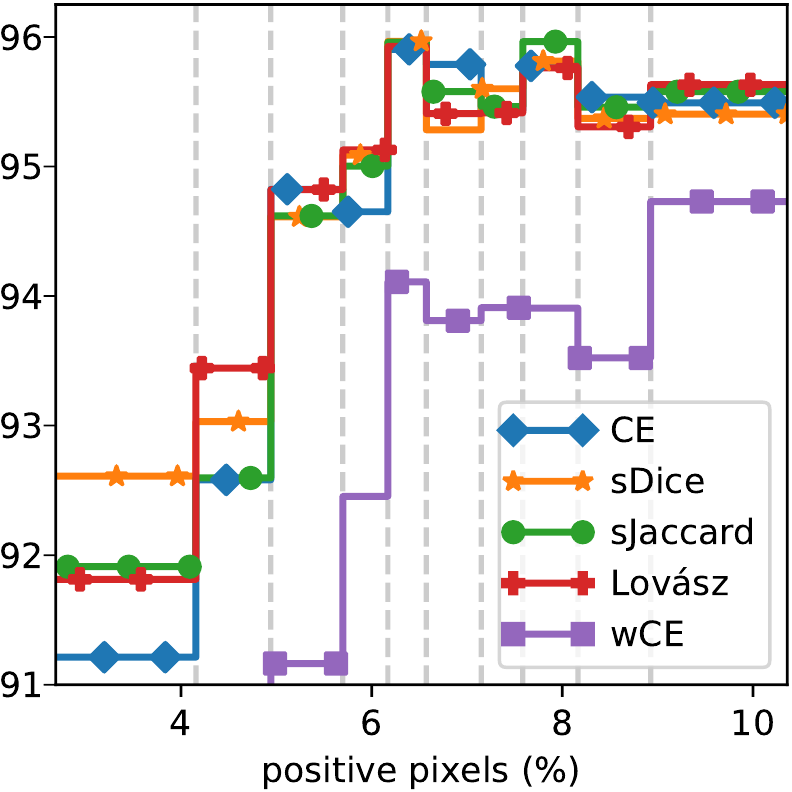}%
    \includegraphics[height=3cm]{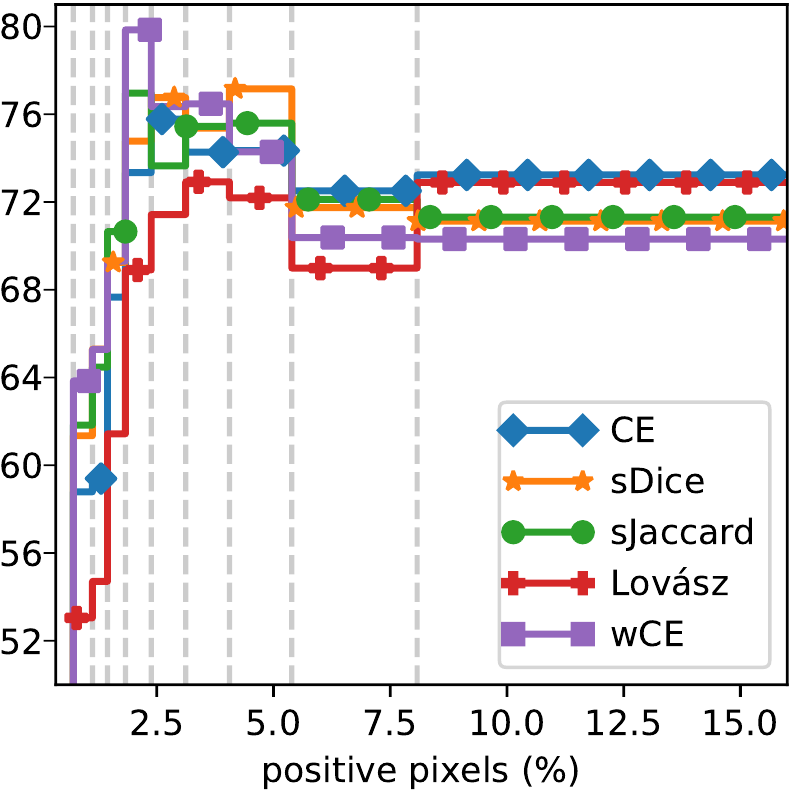}%
  }%
}%
\setlength{\foursubht}{\ht\foursubbox}%
  \centering%
\subcaptionbox{BR18\label{fig:areaBR}}{%
  \includegraphics[height=\foursubht]{histfig/BRATS_2018_decomp_new_model}%
}%
\subcaptionbox{IS18\label{fig:areaIS}}{%
  \includegraphics[height=\foursubht]{histfig/ISLES_2018_decomp}%
}%
\subcaptionbox{MO17\label{fig:areaMO}}{%
  \includegraphics[height=\foursubht]{histfig/38_decomp_new_model}%
}%
\subcaptionbox{PO18\label{fig:areaPO}}{%
  \includegraphics[height=\foursubht]{histfig/POLYPS_decomp}%
}%
  \caption{Dice score as a function of object size (\# foreground pixels in \%) for BR18 (a), IS18 (b), MO17 (c) and PO18 (d). The results for IS17 \review{and WM17} are omitted due to the small dataset size and lack of statistical relevance. 
  For visualization purposes, the Dice scores are averaged within every 10th percentile (bordered by the dashed lines). 
  The superior performance for metric-sensitive losses holds across the entire range of object sizes.\label{fig:dice_size}}%
\end{figure*}

We study the performance of the different loss functions w.r.t. the object size in Fig.~\ref{fig:dice_size}. This figure shows the average Dice scores in function of the number of foreground pixels in the ground truth.
We show that optimization with one of the metric-sensitive losses is beneficial for all object sizes and that CE can provide poor Dice scores for almost the entire range of scales. CE is underperforming the other losses for all datasets, but especially in BR18 and IS18. \review{We further note that the performance of wCE as a loss function can greatly depend on the lesion size (which is not the case for CE). In fact, wCE does improve performance for some sample-size ranges, but never across the entire range and it can also perform much worse than the metric-sensitive losses.} This shows again that simply giving a weight to cross-entropy will not be able to surrogate the target metric across all sample scales. \\

In order to further investigate the influence of \reviewminor{the} foreground to background ratio (fg/bg ratio) on the performance of the previously described loss functions \reviewminor{at the dataset level}, we virtually create new datasets from PO18. \reviewminor{Following Sect.~\ref{sec:data_preprocessing}, the apparent fg/bg ratios of the datasets are 0.00518, 0.0264, 0.0348, 0.0424, 0.0671 and 0.0705 for WM17, IS17, IS18, PO18, BR18 and MO17, respectively. With the lower end of the fg/bg spectrum covered by WM17 and the converging trend of the losses at the case level (i.e. object size) observed in Fig.~\ref{fig:dice_size}, we question the hypothesis for higher fg/bg ratios. For each of the ratios we want to analyse (i.e. 0.05, 0.1, 0.2, 0.3, 0.4, 0.5),} a new dataset is created as follows: a binary rectangular mask is created for each image as shown in Fig.~\ref{fig:cropfig}. This rectangle has the same aspect ratio as the original image and contains the entire polyp (or part of it if it doesn't fit). The size of this rectangle is the same for all images in each virtual dataset and is determined such that the average fg/bg ratio over all images is equal to the desired ratio. This binary mask is then used to mask the output of the final layer of the network during training and testing. This way, only the pixels inside the rectangle are used to update the weights of the network, but the field-of-view (FOV) doesn't change between the virtual datasets. Otherwise, if we would just crop the input images before feeding them to the network, segmentation could become harder when increasing the fg/bg ratio because the FOV decreases.

\begin{figure}[htbp]
    \centering
    \resizebox{\linewidth}{!}{
    \begin{tabular}{ccc}
        \includegraphics{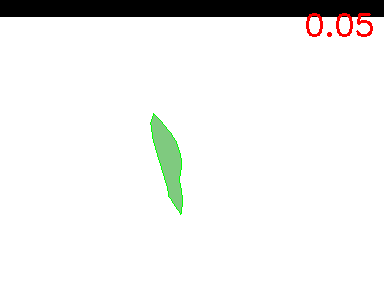} & \includegraphics{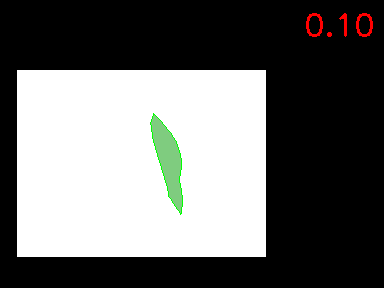} & \includegraphics{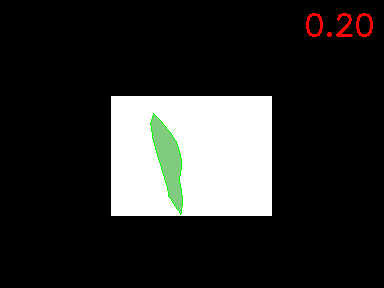} \\
        \includegraphics{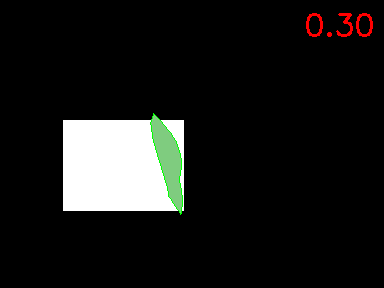} & \includegraphics{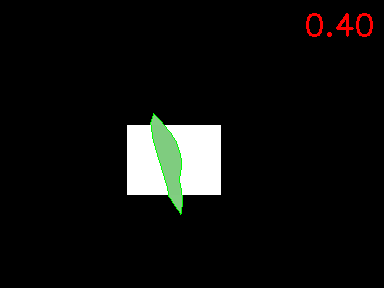} & \includegraphics{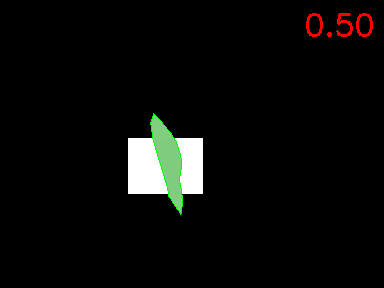}
    \end{tabular}
    }
    \caption{\reviewminor{For each foreground to background ratio in red, a rectangular mask is created with the same aspect ratio as the original image and which contains the entire polyp (or part of it if not possible) in green.}}
    \label{fig:cropfig}
\end{figure}

We then analysed the loss functions' performance for each of the \reviewminor{six} virtual datasets analogously to the original PO18 dataset in Section~\ref{sec:empirical_setup}. The same network, training procedure and training images from PO18 are used with one difference: the output of the final layer is masked with the corresponding mask for a given fg/bg ratio. The results are shown in Table~\ref{tab:crop_results}.
We only post-trained with CE and sDice since we already showed that there is no significant difference within each group of loss functions.
We see no clear trend when going from \reviewminor{0.05} to 0.5, indicating that the superiority of the metric-sensitive loss functions is not dependent on the fg/bg ratio. Even for a ratio of 0.5, post-training with Dice loss shows a significantly higher performance w.r.t. Dice score than CE.

\begin{table}[htbp]
    \caption{\review{Dice scores and Jaccard indexes obtained for each virtual PO18 dataset and for networks trained with CE or sDice. Values highlighted in grey point to a significantly higher value compared to the other loss. Italic formatting denotes significantly inferior values.}}
    \begin{tabularx}{\linewidth}{llXYYYYYY}
    \toprule
         & Loss & \multicolumn{1}{r}{\lapbox[\width]{1em}{\emph{Dataset} $\rightarrow$}} & \reviewminor{0.05} & 0.1 & 0.2 & 0.3 & 0.4 & 0.5 \\
    \midrule
\parbox[t]{3mm}{\centering\multirow{2}{*}{\rotatebox[origin=c]{90}{Dice}}}
      & CE    & & \reviewminor{\textit{0.627}} & \textit{0.645} &  \textit{0.673} &  \textit{0.691} &  \textit{0.730} &  \textit{0.757} \\
      & sDice & & \cellcolor{gray!50} \reviewminor{0.639} & \cellcolor{gray!50} 0.662 & \cellcolor{gray!50} 0.683 & \cellcolor{gray!50} 0.701 & \cellcolor{gray!50} 0.740 & \cellcolor{gray!50} 0.765 \\
    \midrule
    \parbox[t]{3mm}{\centering\multirow{2}{*}{\rotatebox[origin=c]{90}{Jacc.}}}
      & CE    & & \reviewminor{\textit{0.534}} & \textit{0.552} &  \textit{0.579} &  \textit{0.597} &  \textit{0.639} &  \textit{0.672} \\
      & sDice & & \cellcolor{gray!50} \reviewminor{0.541} & \cellcolor{gray!50} 0.564 & \cellcolor{gray!50} 0.584 & \cellcolor{gray!50} 0.604 & \cellcolor{gray!50} 0.646 & \cellcolor{gray!50} 0.676 \\
         
    \bottomrule
\end{tabularx}
    \label{tab:crop_results}
\end{table}

\subsection{Comparison of losses for multi-class segmentation}\label{sec:multiclass}
The results for the original loss functions are presented in Table \ref{tab:multiclass_results}. Again, we observe a consistent, superior performance of the metric-sensitive losses compared to the cross-entropy losses w.r.t. the Dice score and Jaccard index. Given foreground-background ratios of 0.065, 0.028 and 0.012 for WT, TC and ET, respectively, we observe a similar trend as in Fig. \ref{fig:dice_size}. The metric-sensitive losses obtain higher Dice scores and Jaccard indexes compared to the cross-entropy losses over multiple foreground-background ratios.

\begin{table}[htbp]
    \caption{\review{Dice scores and Jaccard indexes obtained for the multi-class BRATS dataset using the original loss functions. Results are presented for the three volumes that were considered for evaluation during the challenge: whole tumor (WT), tumor core (TC) and enhancing tumor (ET). We also report the average across the three types (all). Values in gray are the significantly best performing losses. There are no values in italic since no single loss performed significantly worse than all the others.}}

\begin{tabularx}{\linewidth}{llXYY|YYY}
    \toprule
    & & \multicolumn{1}{r}{\lapbox[\width]{1em}{\emph{group} $\rightarrow$}} & \multicolumn{2}{c}{CE-based} & \multicolumn{3}{c}{Metric-sensitive} \\
    \midrule
    & Type & \multicolumn{1}{r}{\lapbox[\width]{1em}{\emph{loss} $\rightarrow$}} & CE & wCE & sDice & sJaccard & Lov\'{a}sz \\
    \midrule
    \parbox[t]{7mm}{\centering\multirow{4}{*}{\rotatebox[origin=c]{90}{Dice}}} 
    & All   & & 0.639 & 0.665 & 0.740 & \cellcolor{gray!50} 0.747 & 0.716 \\
    & WT    & & 0.798 & 0.799 & 0.849 & \cellcolor{gray!50}0.859 & \cellcolor{gray!50}0.858 \\
    & TC    & & 0.620 & 0.636 & 0.725 & \cellcolor{gray!50}0.737 & 0.692 \\
    & ET    & & 0.498 & 0.560 & \cellcolor{gray!50}0.646 & \cellcolor{gray!50}0.645 & 0.598 \\
    \midrule
    \parbox[t]{7mm}{\centering\multirow{4}{*}{\rotatebox[origin=c]{90}{Jaccard}}}
    & All & & 0.522 & 0.554 & 0.632 & \cellcolor{gray!50}0.643 & 0.606 \\
    & WT  & & 0.683 & 0.688 & 0.753 & \cellcolor{gray!50}0.767 & \cellcolor{gray!50}0.766 \\
    & TC  & & 0.487 & 0.518 & 0.609 & \cellcolor{gray!50}0.627 & 0.569 \\
    & ET  & & 0.395 & 0.456 & \cellcolor{gray!50}0.534 & \cellcolor{gray!50}0.534 & 0.484 \\
    \bottomrule
\end{tabularx}
    \label{tab:multiclass_results}
\end{table}

The results for the range of sTversky losses are presented in Table \ref{tab:multiclass_tversky_results}.
\begin{table*}[htbp]
    \centering
    \caption{\review{Dice scores and Jaccard indexes obtained for the multi-class BRATS dataset using the range of sTversky losses with varying alpha/beta. Results are presented for the three volumes that were considered for evaluation during the challenge: whole tumor (WT), tumor core (TC) and enhancing tumor (ET). We also report the average across the three types (all). Cells in gray highlight the top-ranked losses and italic values are significantly inferior to all others in the row. There is no weighting scheme producing significantly higher values than sTversky 0.5/0.5.}}
\begin{tabularx}{\linewidth}{lsXYYYYYYYYY|YY}
    \toprule
    & Type & \multicolumn{1}{r}{\lapbox[\width]{1em}{\emph{$\alpha/\beta$} $\rightarrow$}} & 0.1/0.9 & 0.2/0.8 & 0.3/0.7 & 0.4/0.6 & 0.5/0.5 & 0.6/0.4 & 0.7/0.3 & 0.8/0.2 & 0.9/0.1 & \review{0.75/0.75} & \review{1.0/1.0} \\
    \midrule
    \parbox[t]{7mm}{\centering\multirow{4}{*}{\rotatebox[origin=c]{90}{Dice}}} 
    &All &        &        0.649 &        0.702 &        0.725 &        0.739 & 0.740 &        0.738 &        0.709 &        0.704 &        0.649 &   \cellcolor{gray!50}     \review{0.741} & \cellcolor{gray!50}\review{0.747} \\
    &WT &        &        0.771 &        0.822 &        0.844 &        0.853 & 0.849 &        0.848 &        0.826 &        0.816 &        0.757 &   \cellcolor{gray!50}     \review{0.858} & \cellcolor{gray!50}\review{0.859} \\
    &TC &        &        0.652 &        0.703 &        0.724 &     \cellcolor{gray!50}   0.735 & 0.725 &        0.725 &        0.687 &        0.686 &        0.623 &        \cellcolor{gray!50}\review{0.731} & \cellcolor{gray!50}\review{0.737} \\
    &ET &        &        \textit{0.525} &        0.582 &        0.606 &        0.631 & \cellcolor{gray!50}0.646 &   \cellcolor{gray!50}     0.642 &        0.614 &        0.611 &        0.568 &     \cellcolor{gray!50}   \review{0.634} &\cellcolor{gray!50} \review{0.645} \\
    \midrule
    \parbox[t]{7mm}{\centering\multirow{4}{*}{\rotatebox[origin=c]{90}{Jaccard}}}
    &All &        &        0.517 &        0.580 &        0.610 &        0.630 & 0.632 &        0.629 &        0.595 &        0.586 &        0.521 &        \review{0.634} & \cellcolor{gray!50}\review{0.643} \\
    &WT &        &        0.640 &        0.711 &        0.742 &        0.757 & 0.753 &        0.750 &        0.719 &        0.704 &        0.627 &       \cellcolor{gray!50} \review{0.765} & \cellcolor{gray!50}\review{0.767} \\
    &TC &        &        0.513 &        0.572 &        0.602 &        0.617 & 0.609 &        0.609 &        0.567 &        0.563 &        0.493 &        \review{0.615} &\cellcolor{gray!50} \review{0.627} \\
    &ET &        &        0.397 &        0.457 &        0.487 &        0.515 & \cellcolor{gray!50}0.534 &  \cellcolor{gray!50}      0.528 &        0.500 &        0.491 &        0.444 &        \review{0.522} &\cellcolor{gray!50} \review{0.534} \\
    \bottomrule
\end{tabularx}
    \label{tab:multiclass_tversky_results}
\end{table*}

Similar to the results regarding the range of sTversky losses in Sect. \ref{sec:results}, there is no weighting scheme that is significantly superior to the Dice equivalent weighting of $\alpha=\beta=0.5$. There is a clear trend of sub-optimal weighting when $\alpha$ and/or $\beta$ deviate from 0.5, which relates back to Fig. \ref{fig:TverskyAbsRelErrors} where relative and absolute errors are bound to increase.

\section{Conclusion}
More and more metric-sensitive loss functions find their way into the optimization of CNNs for image segmentation, both in the context of medical imaging and classical computer vision. Nonetheless, we saw a great mass of research in the MICCAI 2018 proceedings still using per-pixel losses while their evaluation was based primarily on the Dice score or Jaccard index. In this work we questioned the latter approach from both theoretical and empirical point of views.\\
On the one hand, theory suggests that the Dice score and Jaccard index approximate each other relatively and absolutely. On the other hand, we found that no such approximations exist for a weighted Hamming similarity. For the Tversky index, the approximation gets monotonically worse when weighting differently from the soft Dice setting. We were able to confirm these findings in an extensive empirical validation on five binary medical segmentation tasks. Further experiments reaffirmed their superior performance across different object sizes, foreground/background ratios and in a multi-class setting.\\
We conclude that segmentation tasks may benefit from a wider adoption of metric-sensitive loss-functions when evaluation is performed using \review{only} the Dice score or the Jaccard index.


\appendices



\section{MICCAI 2018 proceedings}\label{sec:MICCAI2018methodology}

In the introduction, it is stated that only 30\% of all MICCAI 2018 proceedings actually use a metric-sensitive loss function. This, even though the Dice score is reported as their evaluation metric. We categorized all MICCAI 2018 proceedings \cite{frangi2018medical} as indicated by the flowchart in Figure~\ref{fig:flowchart_proceedings_count}. From the total number of MICCAI 2018 proceedings (n = 372), there are 103 papers dealing with learning-based segmentation. Almost all of them (n=96) use the Dice score or Jaccard index to evaluate the performance of their algorithm, but only 28 proceedings report to use soft Dice or any of the other metric-sensitive loss functions described in this work. Many methods use cross-entropy or its weighted variant (n=38) and the rest have used a variety of other losses such as adverserial loss, focal loss, mean-squared-error, etc.

\begin{figure}[htb]
    \centering
    \def\svgwidth{\columnwidth}
    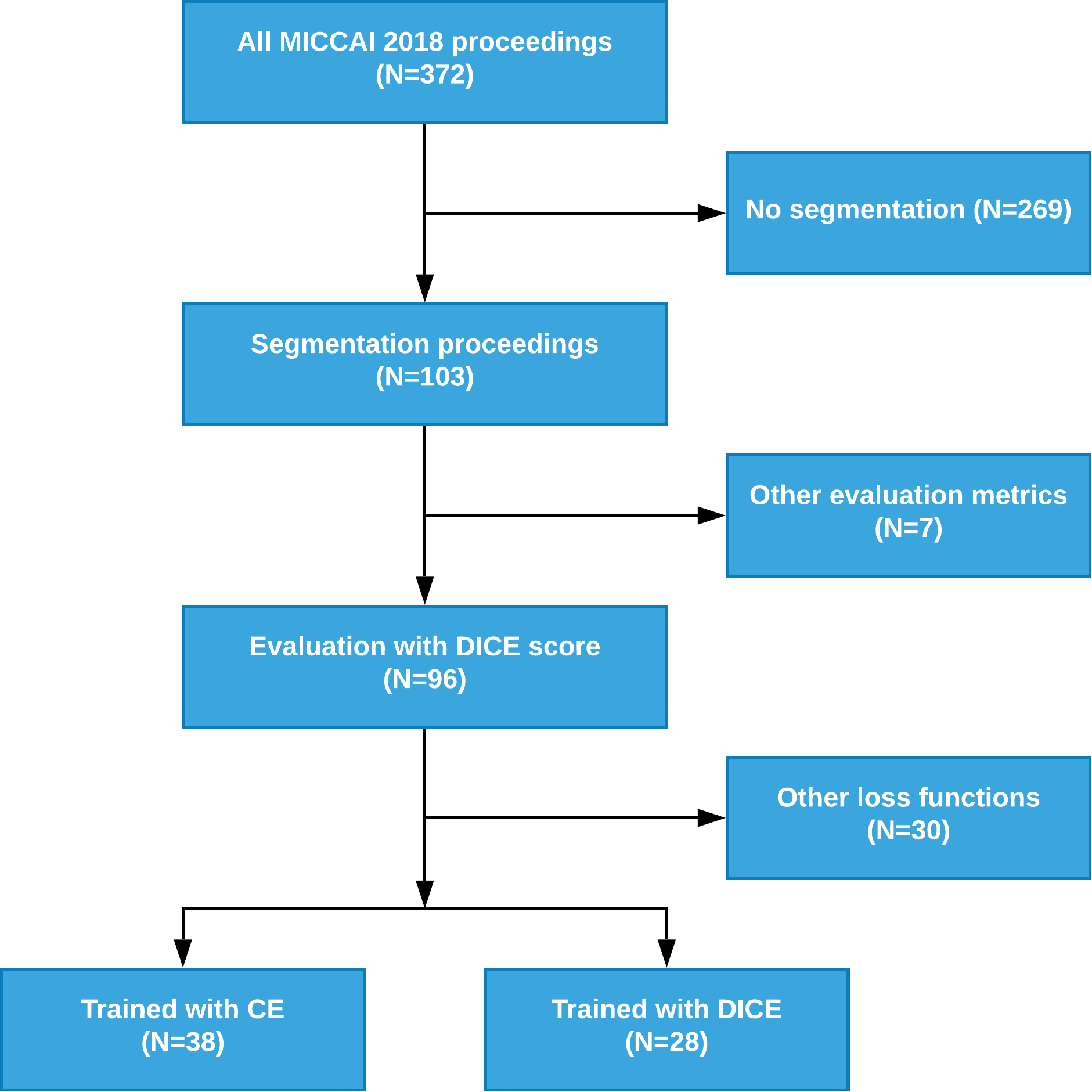
    \caption{Flowchart used for categorization of MICCAI 2018 proceedings. All segmentation papers are selected and from these, we retain all the papers that used DICE score for evaluation. We then analysed which type of loss function was used.}  \label{fig:flowchart_proceedings_count}
\end{figure}


\ifCLASSOPTIONcaptionsoff
  \newpage
\fi

\newpage



\bibliographystyle{IEEEtran}
\bibliography{IEEEabrv,biblio}
%

\FloatBarrier
\clearpage
\counterwithin{definition}{section}
\counterwithin{theorem}{section}
\counterwithin{proposition}{section}

\counterwithin{figure}{section}
\counterwithin{table}{section}
\counterwithin{equation}{section}

\reviewpar
\textbf{{\Large Supplementary material}} \\
\pagenumbering{Roman}  

\setcounter{section}{0} 
\section{F measures}\label{sec:fmeasures}
\reviewpar
We show in Section~\ref{sec:results} that the weighting $\alpha=\beta=0.5$ for the Tversky loss (equivalent to Dice loss) is the optimal weighting when evaluation is performed on Dice score. Alternative weightings of Tversky loss can however be beneficial in case other metrics are of interest. In Table~\ref{tab:tversky_fmeasures}, we show that a different weighting scheme can be very effective to achieve superior performance when evaluation is done with alternative F-measures ($F\textsuperscript{0.5}$, $F\textsuperscript{1.0}$, $F\textsuperscript{1.5}$ and $F\textsuperscript{2.0}$). A relative higher weight for $\alpha$ gives better performance for the lower order F-measures, whereas the best weighting for F\textsuperscript{2.0} is found in the lower range of values for $\alpha$. In this table, results for the weightings $\alpha=\beta=0.5$ and $\alpha=\beta=1.0$ are also provided which correspond to sDice and sJaccard respectively as well as the balanced version between both, i.e. weighting $\alpha=\beta=1.0$. No clear difference can be observed between these three.

\reviewminorpar
\section{Boxplots}\label{sec:boxplots}
This section gives a complementary boxplot for all tables in the main manuscript providing more detailed spread information for all measurements. It can be observed that the IS17, IS18 and PO18 datasets have the highest variability which is to be expected due to their difficult nature and smaller sample size.

\begin{figure}[!htbp]
    \centering
    \resizebox*{\linewidth}{!}{
    
    \def\arraystretch{0}
    \setlength{\tabcolsep}{0pt}
    \begin{tabularx}{\linewidth}{l @{\hspace{2pt}} YYYYYYY}

        & BR18 &   IS17 &  IS18 & MO18 & PO18 & WM17 & WM17\textsuperscript{DM}  \\
        
        \rotatebox{90}{\hspace{5pt} DICE} &
        \includegraphics[width=\linewidth]{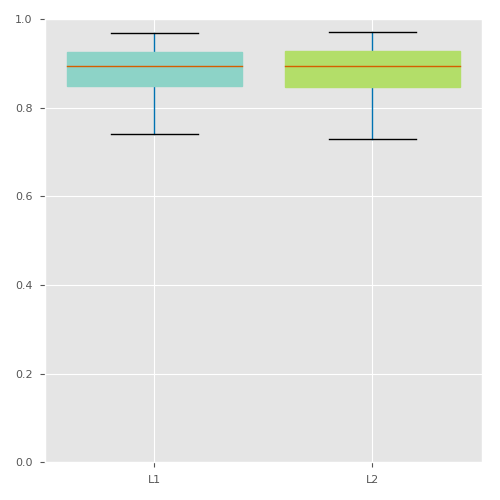} &  
        \includegraphics[width=\linewidth]{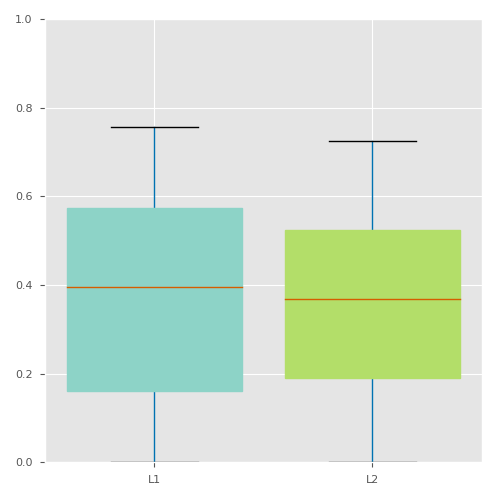} &  
        \includegraphics[width=\linewidth]{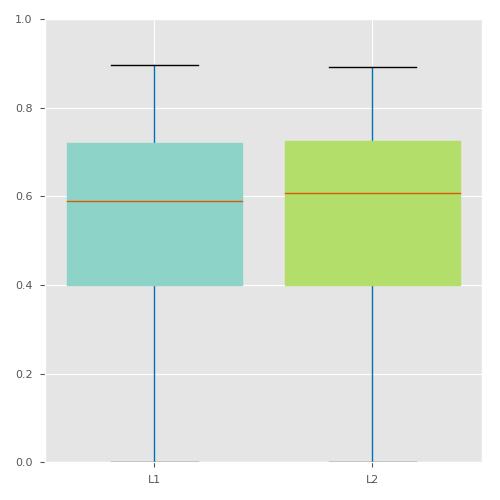} &  
        \includegraphics[width=\linewidth]{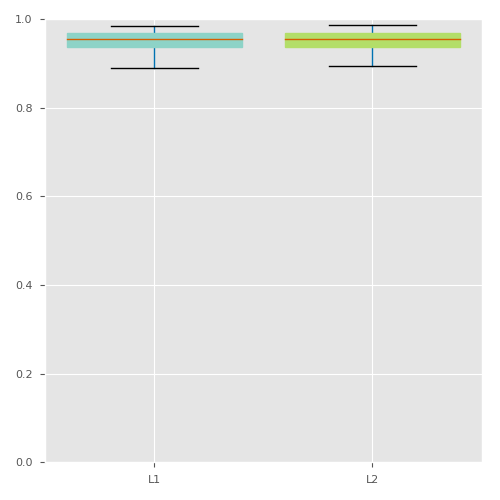} &  
        \includegraphics[width=\linewidth]{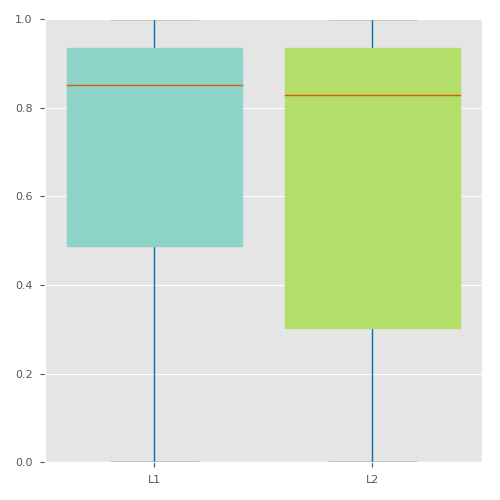} &  
        \includegraphics[width=\linewidth]{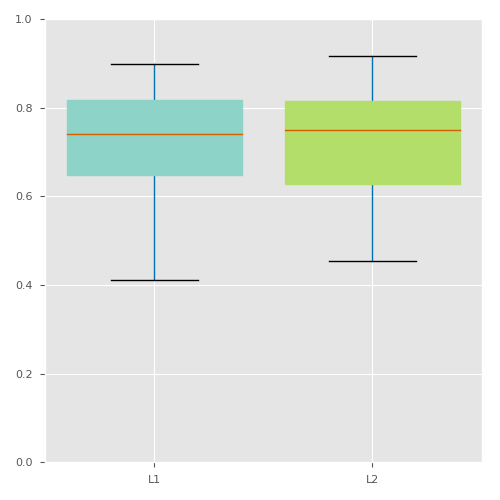} &  
        \includegraphics[width=\linewidth]{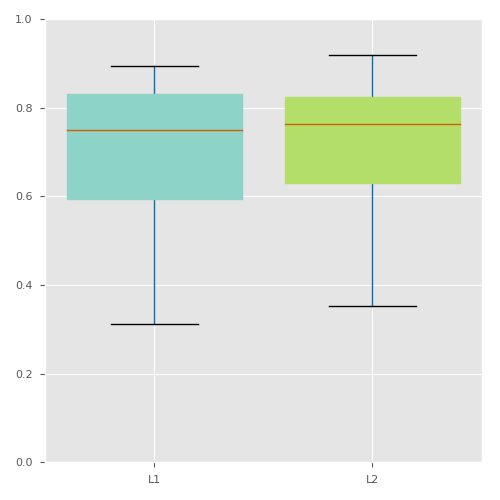}  \\ 
        
        \rotatebox{90}{\hspace{5pt} JACC} &
        \includegraphics[width=\linewidth]{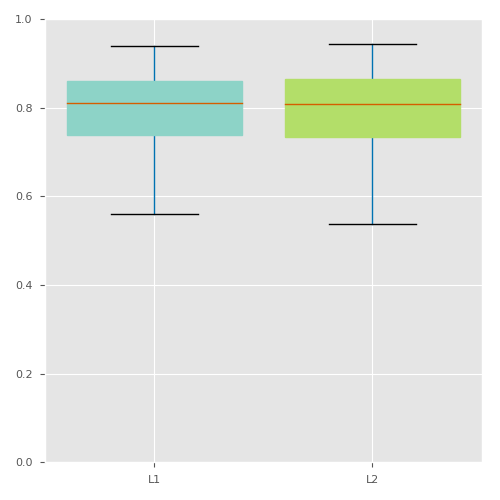} &  
        \includegraphics[width=\linewidth]{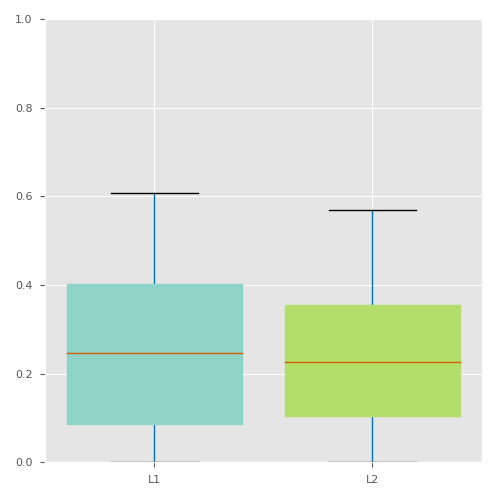} &  
        \includegraphics[width=\linewidth]{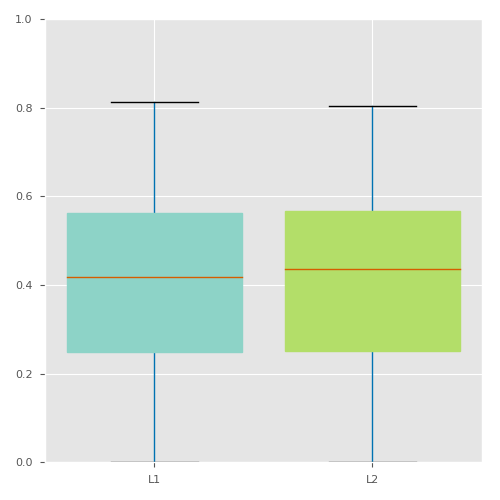} &  
        \includegraphics[width=\linewidth]{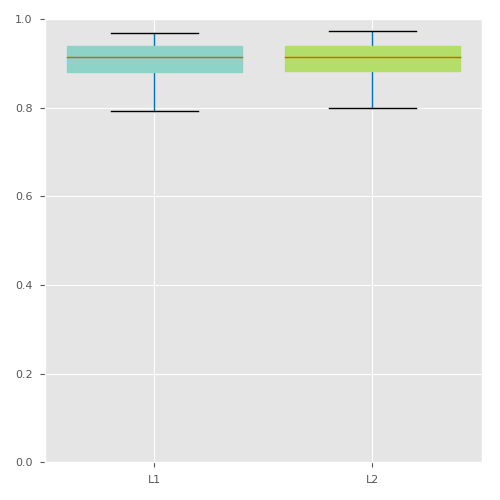} &  
        \includegraphics[width=\linewidth]{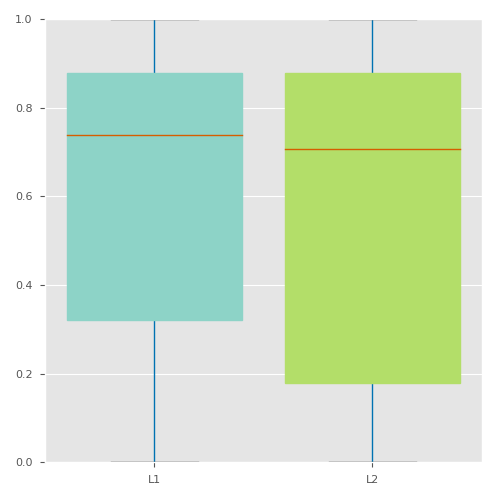} &  
        \includegraphics[width=\linewidth]{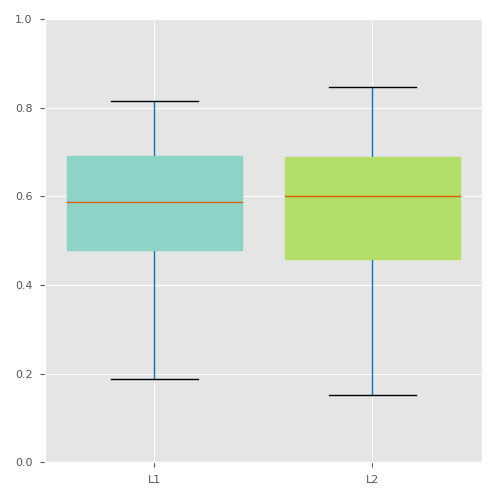} &  
        \includegraphics[width=\linewidth]{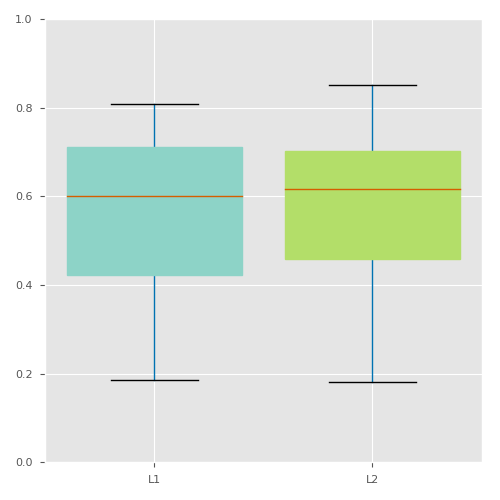}  \\ 

    \end{tabularx}
    }
        \caption{\reviewminorpar Boxplots for the Dice scores and Jaccard indexes obtained for each dataset, using the $L^1$ and $L^2$ generalization for soft Dice loss.}
    \label{fig:boxplots_L1vsL2}
\end{figure}

\begin{figure}[!htbp]
    \centering
    \resizebox*{\linewidth}{!}{
    
    \def\arraystretch{0}
    \setlength{\tabcolsep}{0pt}
    \begin{tabularx}{\linewidth}{l @{\hspace{2pt}} YYYYYY}

        & 0.05 &   0.1 &  0.2 & 0.3 & 0.4 & 0.5   \\
        
        \rotatebox{90}{\hspace{5pt} DICE} &
        \includegraphics[width=\linewidth]{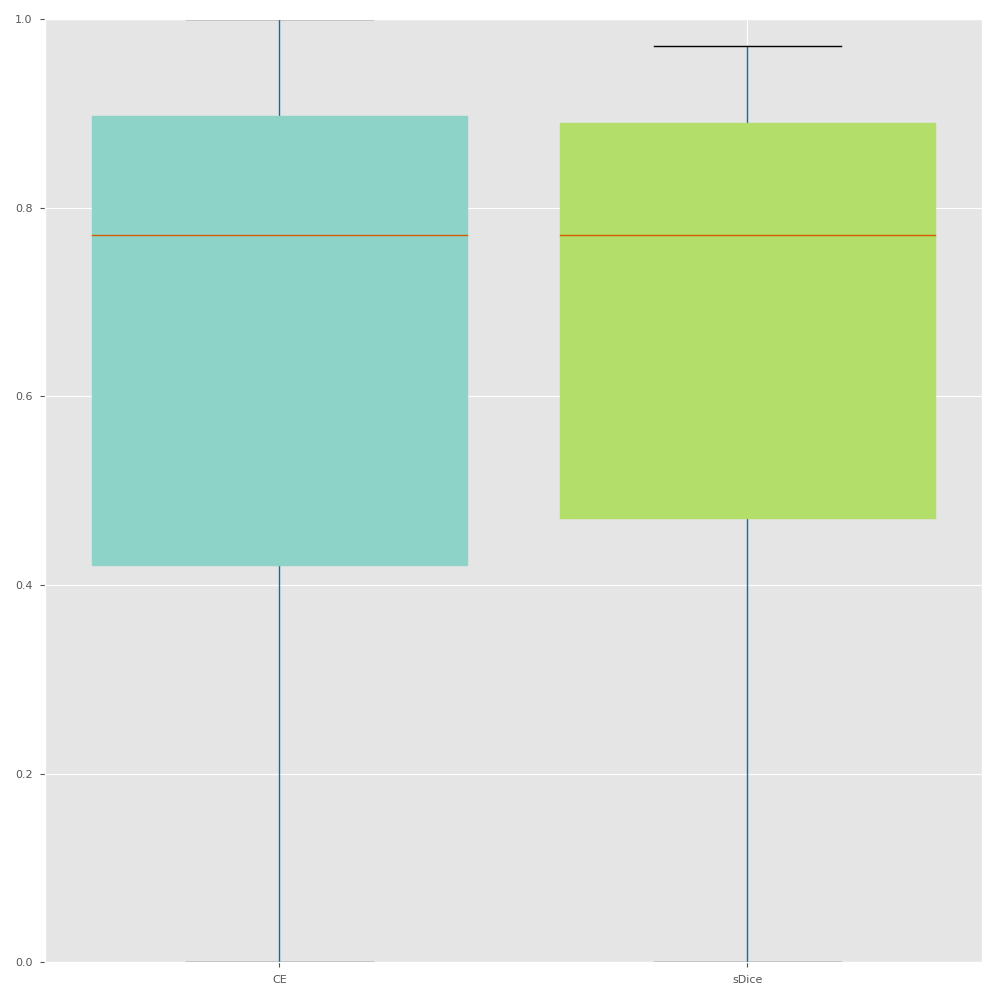} &  
        \includegraphics[width=\linewidth]{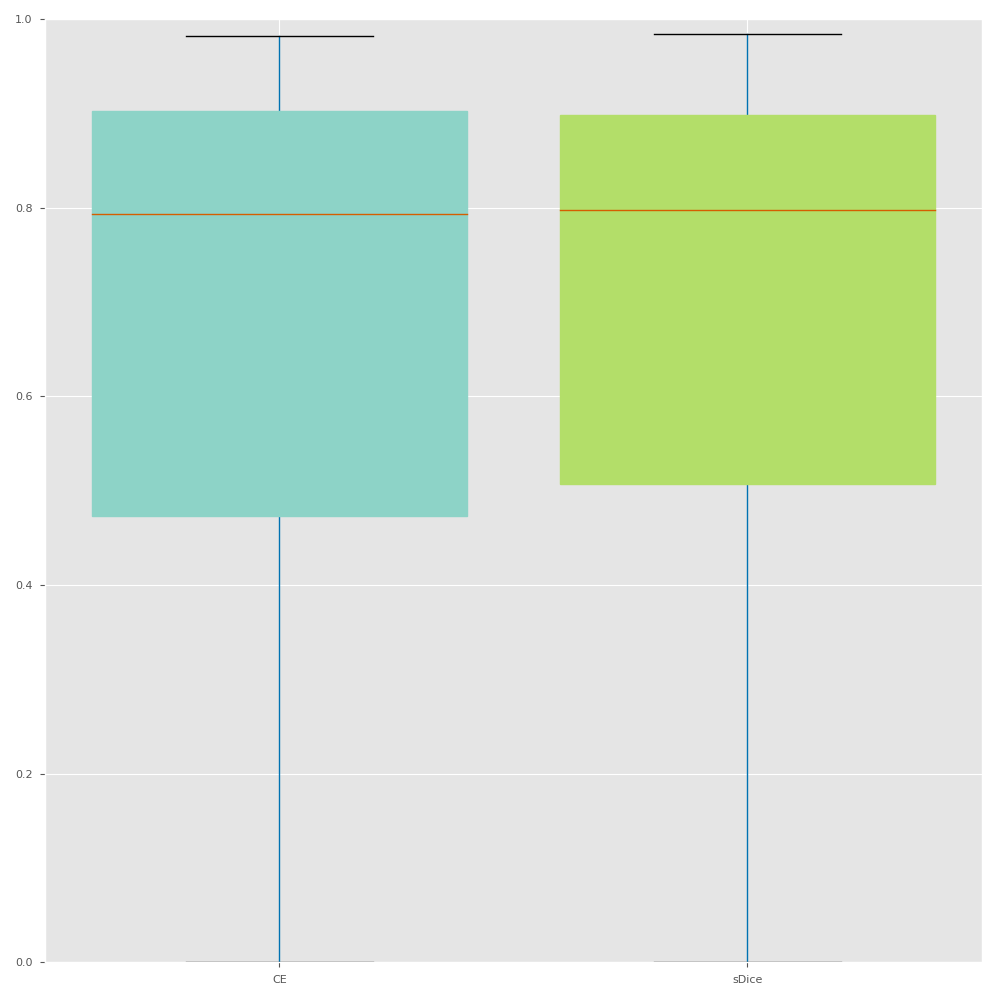} &  
        \includegraphics[width=\linewidth]{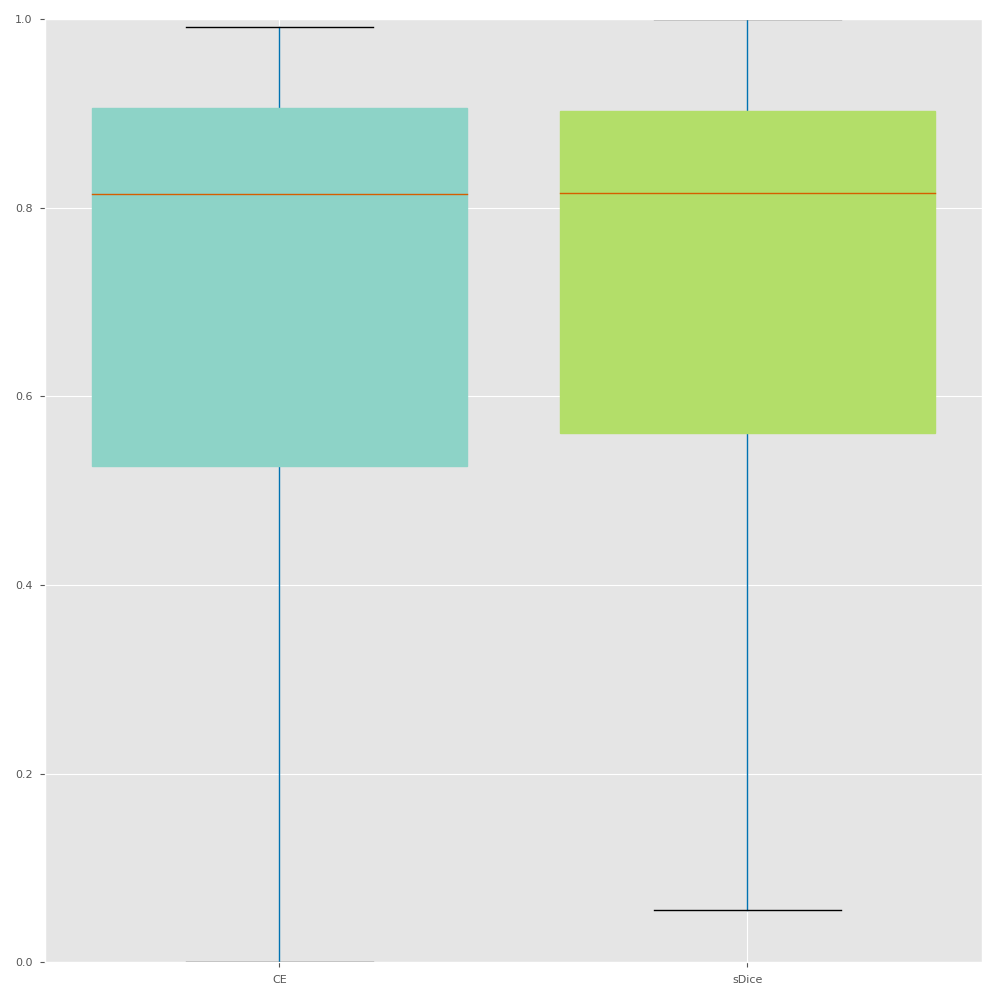} &  
        \includegraphics[width=\linewidth]{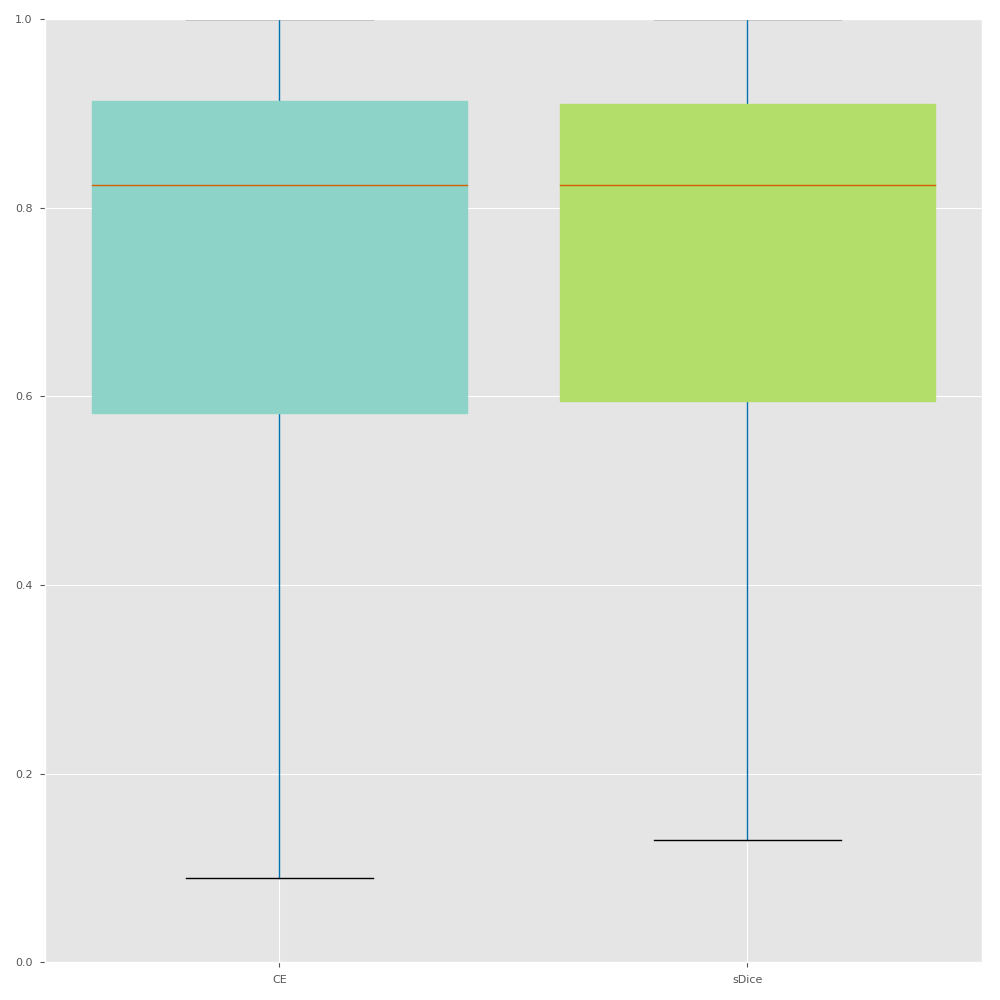} &  
        \includegraphics[width=\linewidth]{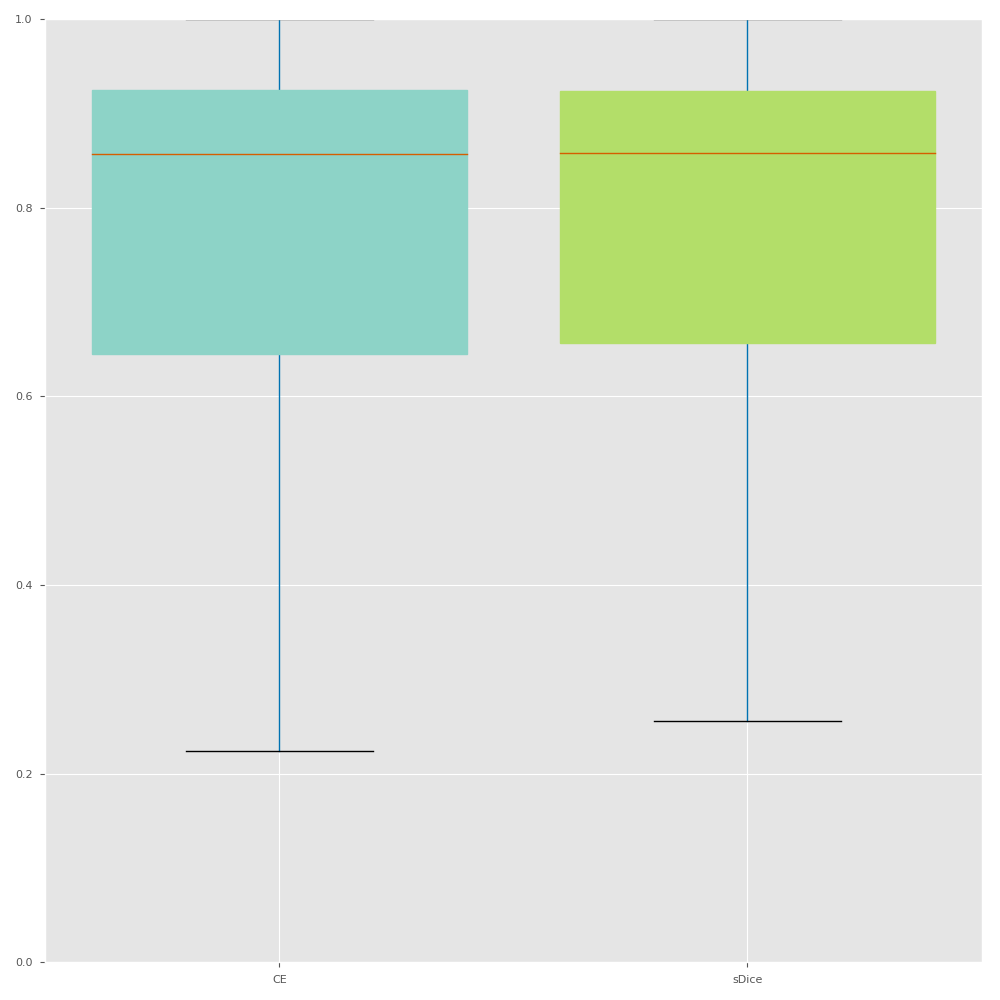} &  
        \includegraphics[width=\linewidth]{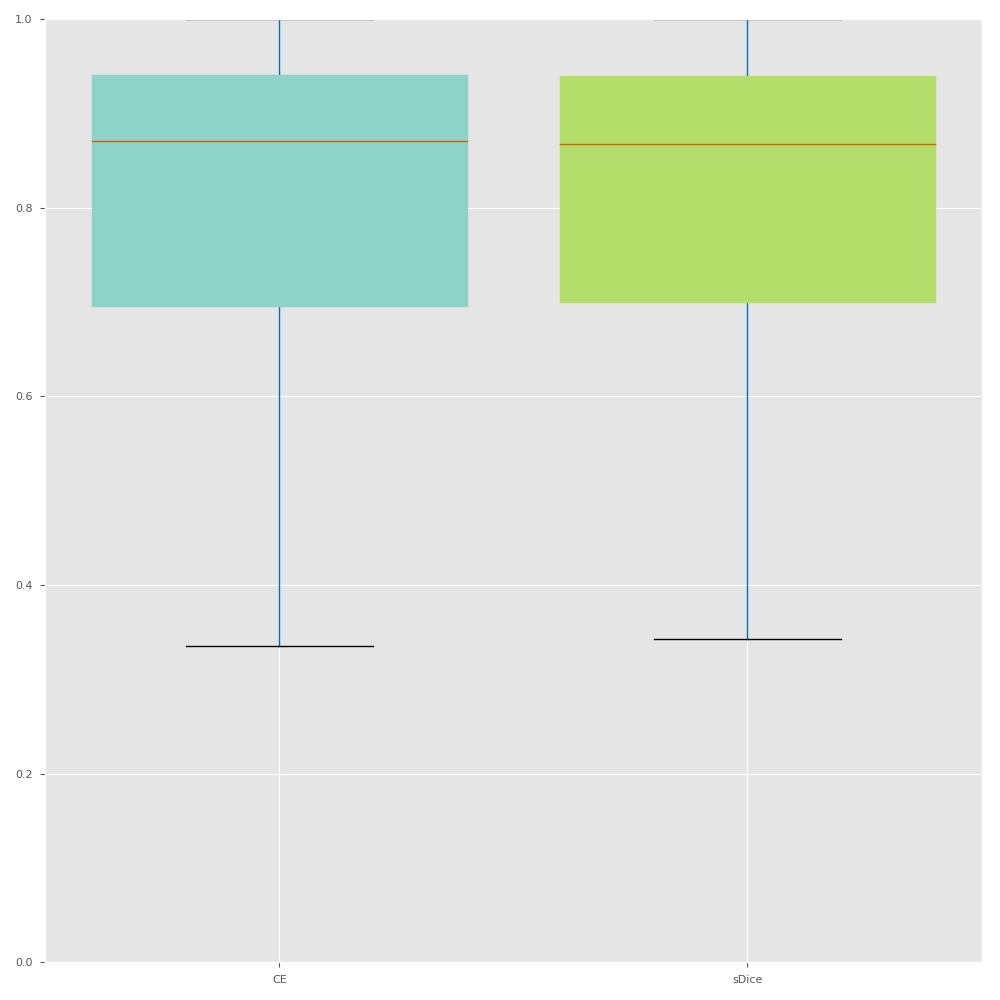}  \\ 
        
        \rotatebox{90}{\hspace{5pt} JACC} &
        \includegraphics[width=\linewidth]{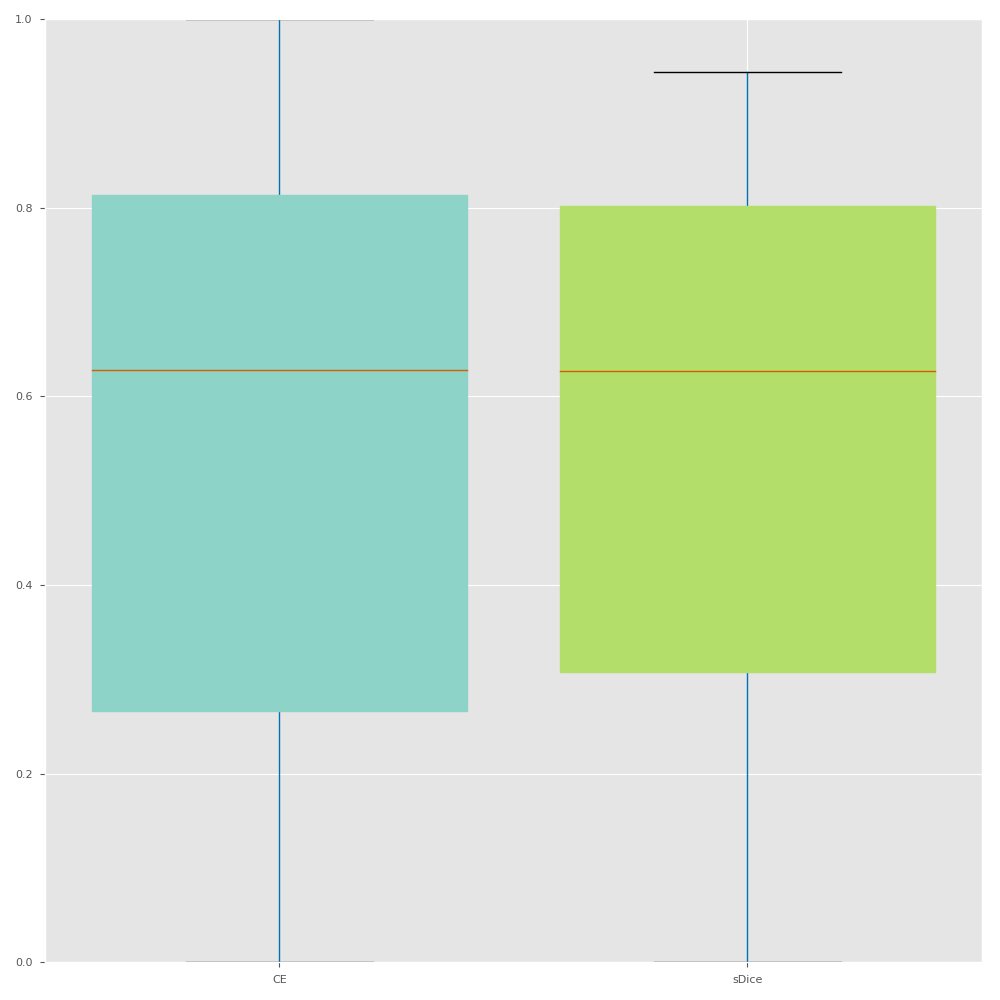} &  
        \includegraphics[width=\linewidth]{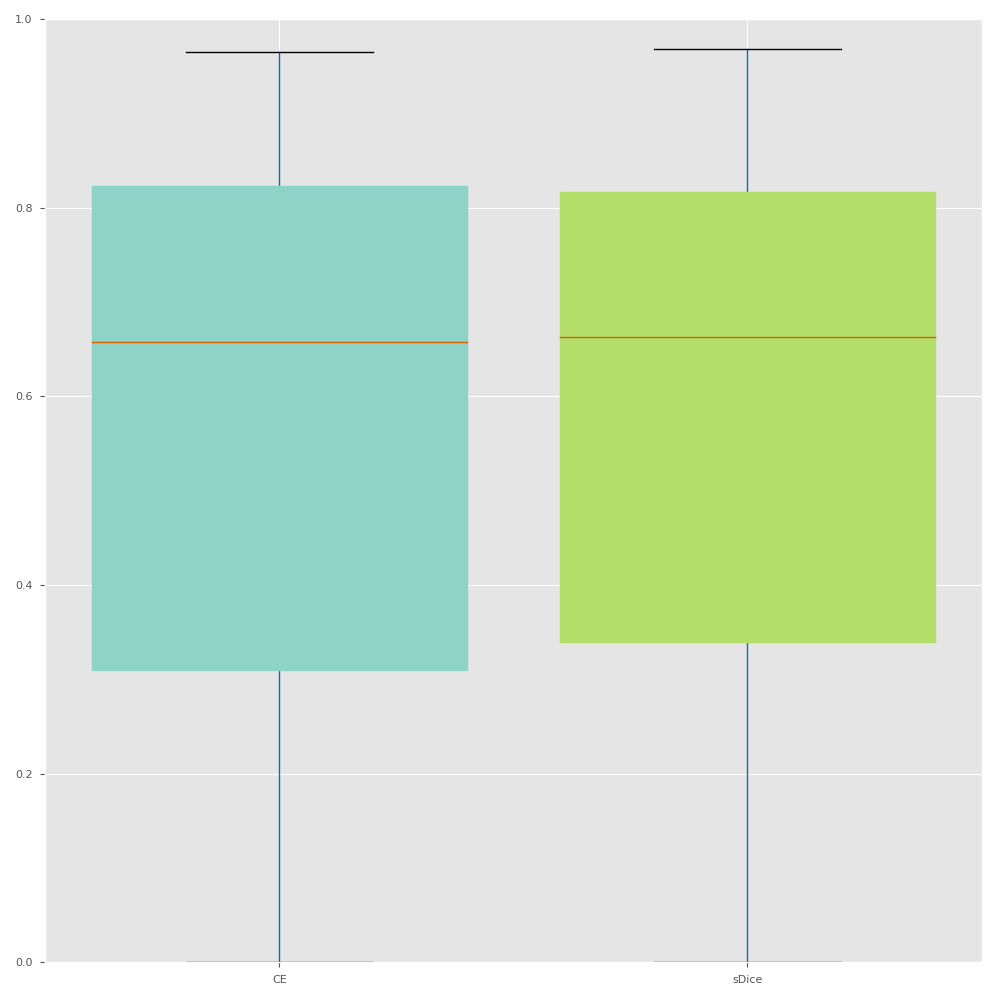} &  
        \includegraphics[width=\linewidth]{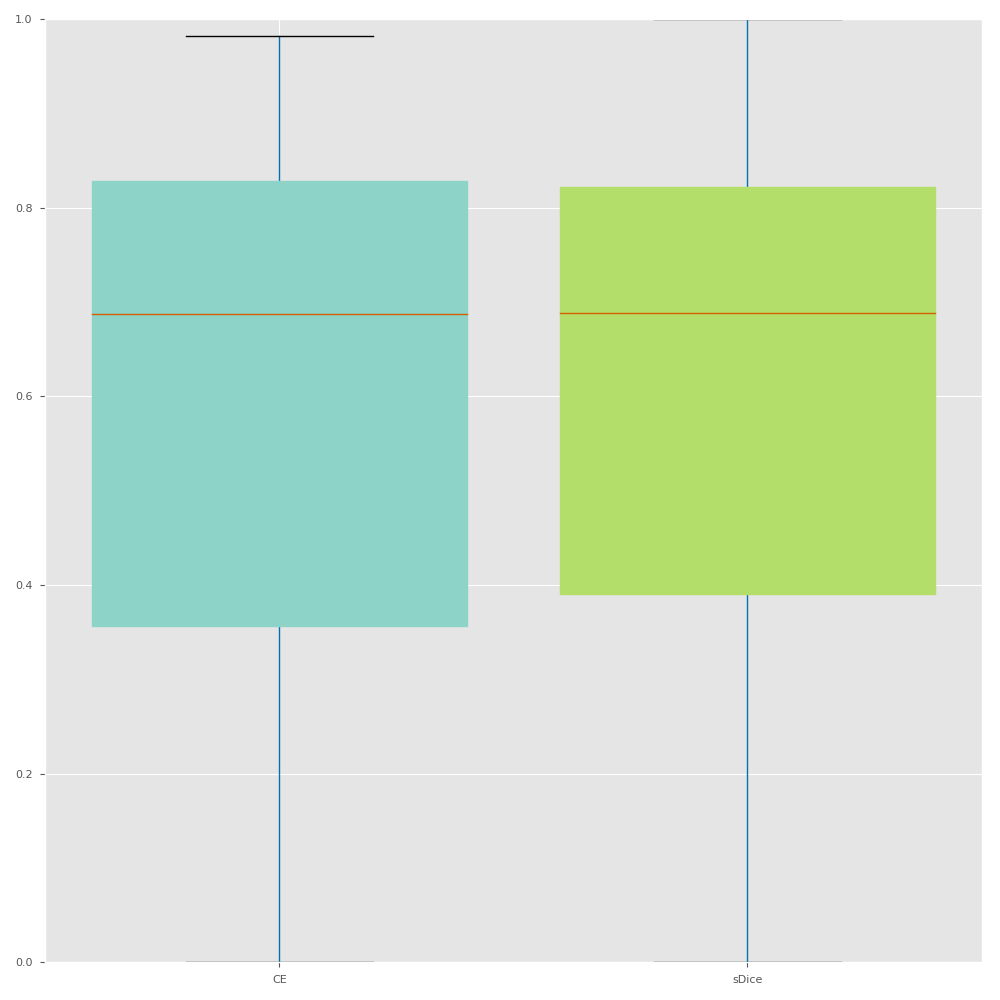} &  
        \includegraphics[width=\linewidth]{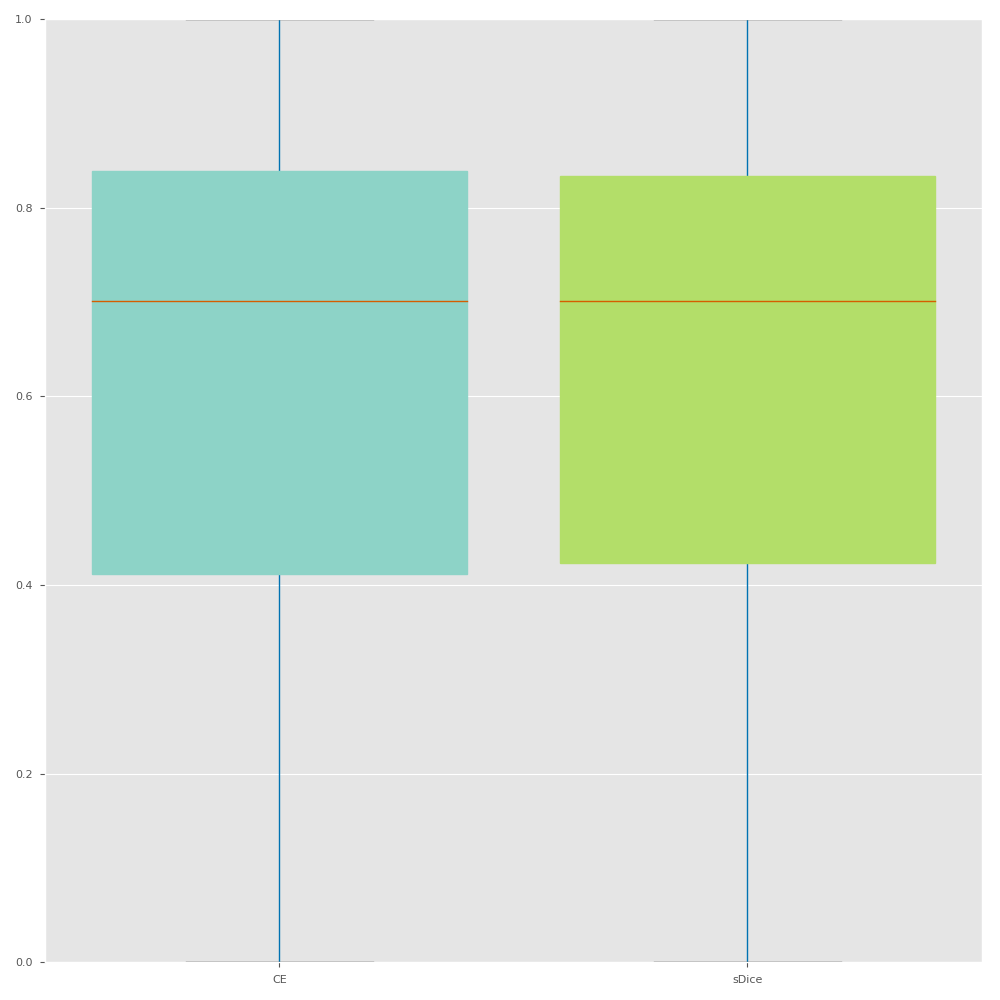} &  
        \includegraphics[width=\linewidth]{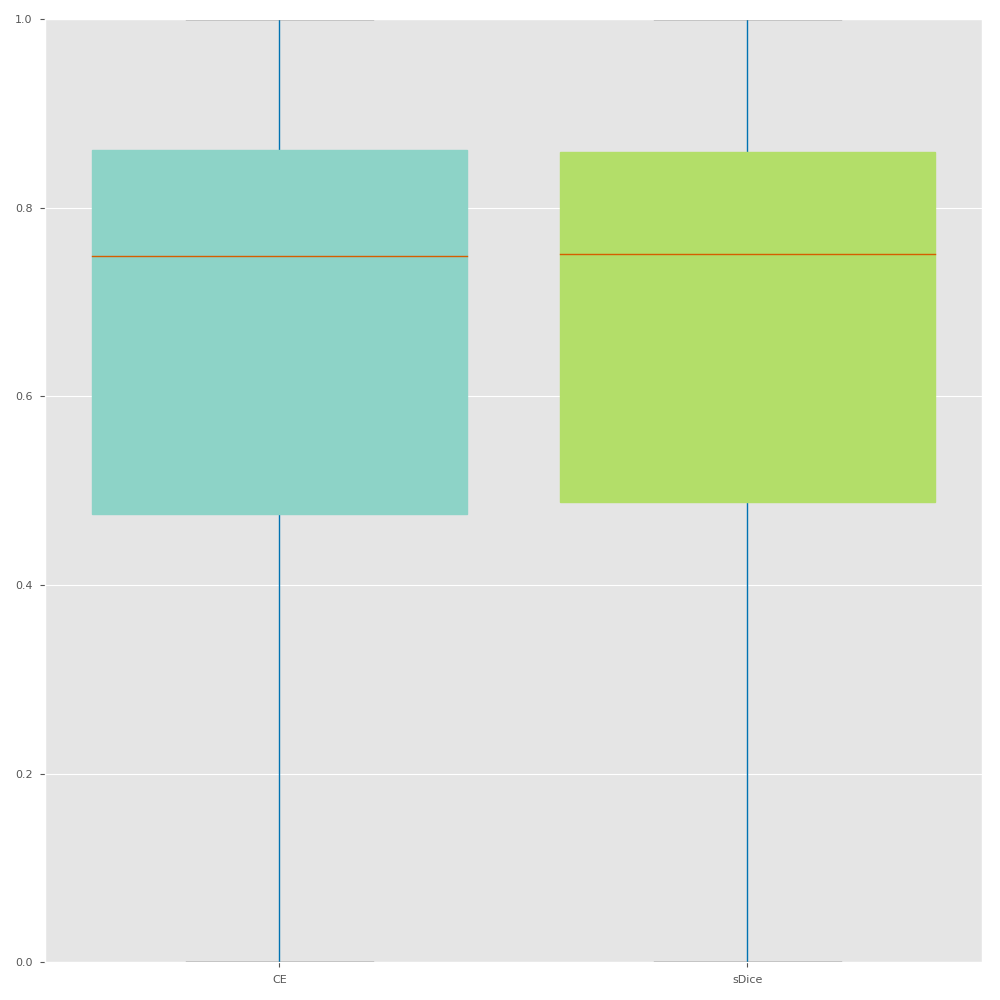} &  
        \includegraphics[width=\linewidth]{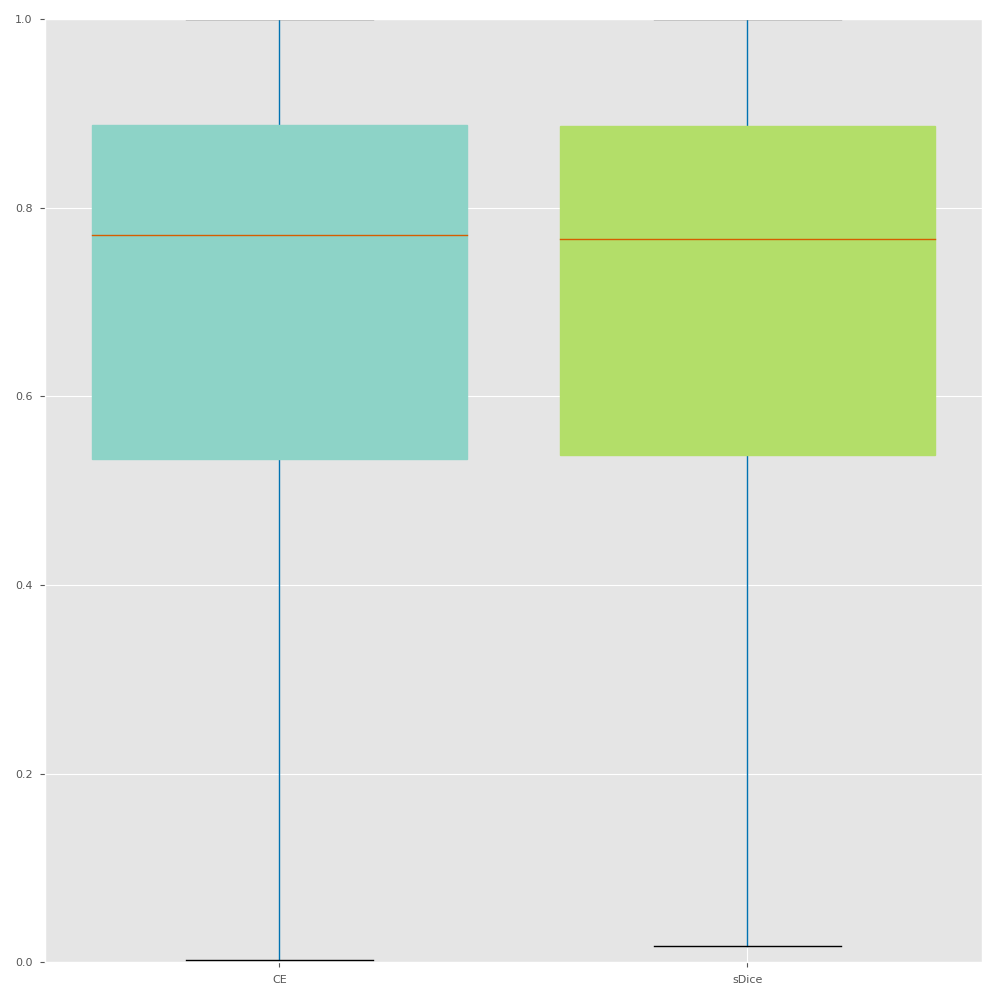}  \\ 

    \end{tabularx}
    }
        \caption{\reviewminorpar Boxplots for the Dice scores and Jaccard indexes obtained for each virtual PO18 dataset and for networks trained with CE or sDice.}
    \label{fig:boxplots_cropped}
\end{figure}

\begin{figure}[!htbp]
    \centering
    \resizebox*{\linewidth}{!}{
    
    \def\arraystretch{0}
    \setlength{\tabcolsep}{0pt}
    \begin{tabularx}{\linewidth}{l @{\hspace{2pt}} YYYY}

        & All &   WT &  TC & ET  \\
        
        \rotatebox{90}{\hspace{5pt} DICE} &
        \includegraphics[width=\linewidth]{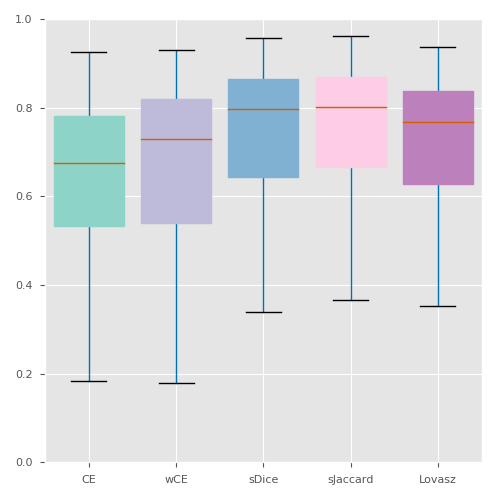} &  
        \includegraphics[width=\linewidth]{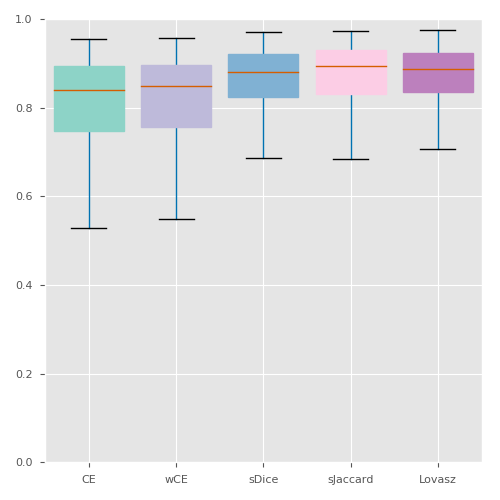} &  
        \includegraphics[width=\linewidth]{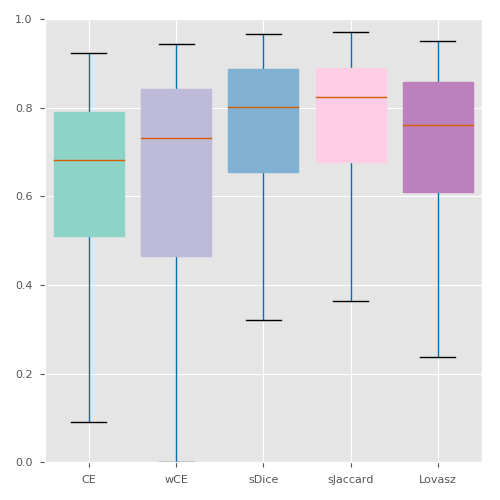} &  
        \includegraphics[width=\linewidth]{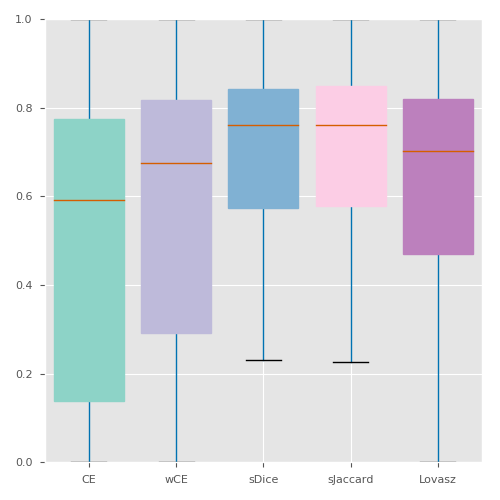}   \\ 
        
        \rotatebox{90}{\hspace{5pt} JACC} &
        \includegraphics[width=\linewidth]{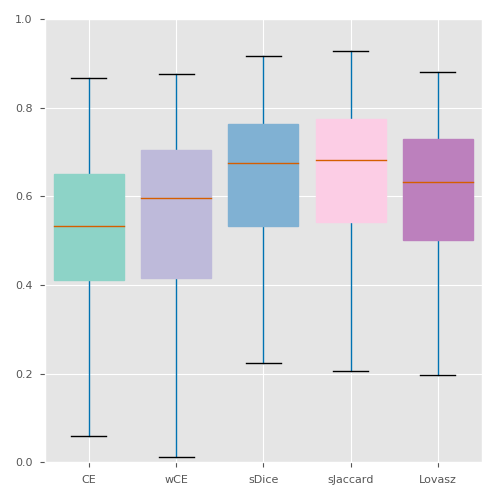} &  
        \includegraphics[width=\linewidth]{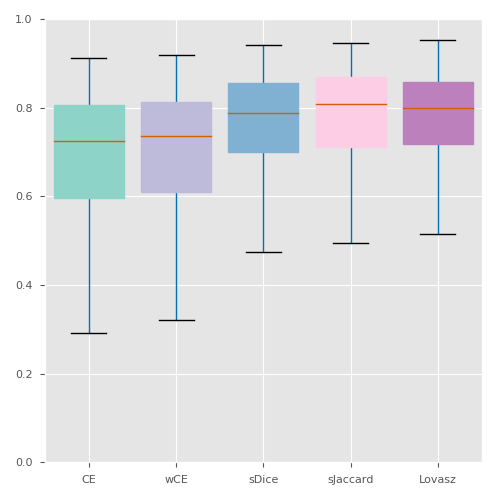} &  
        \includegraphics[width=\linewidth]{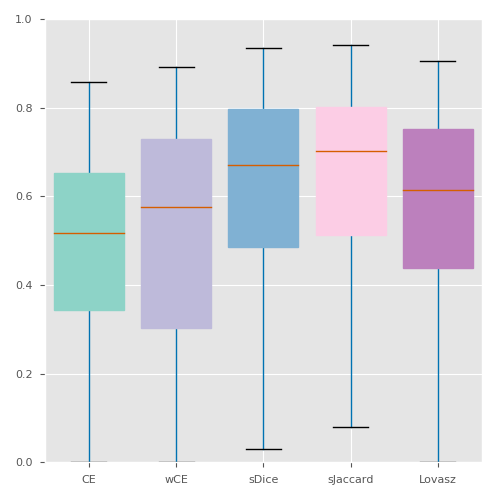} &  
        \includegraphics[width=\linewidth]{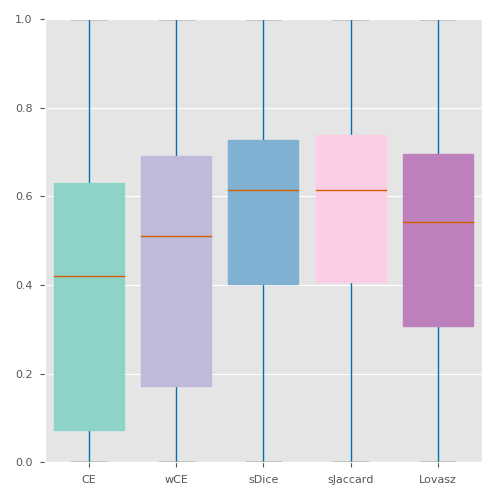} \\ 

    \end{tabularx}
    }
        \caption{\reviewminorpar Boxplots for the Dice scores and Jaccard indexes obtained for the multi-class BRATS dataset using the original loss functions. Results are presented for the three volumes that were considered for evaluation during the challenge: whole tumor (WT), tumor core (TC) and enhancing tumor (ET). We also report the average across the three types (all).}
    \label{fig:boxplots_multiclass}
\end{figure}

\begin{figure}[!htbp]
    \centering
    \resizebox*{\linewidth}{!}{
    
    \def\arraystretch{0}
    \setlength{\tabcolsep}{0pt}
    \begin{tabularx}{\linewidth}{l @{\hspace{2pt}} YYYY}

        & All &   WT &  TC & ET  \\
        
        \rotatebox{90}{\hspace{5pt} DICE} &
        \includegraphics[width=\linewidth]{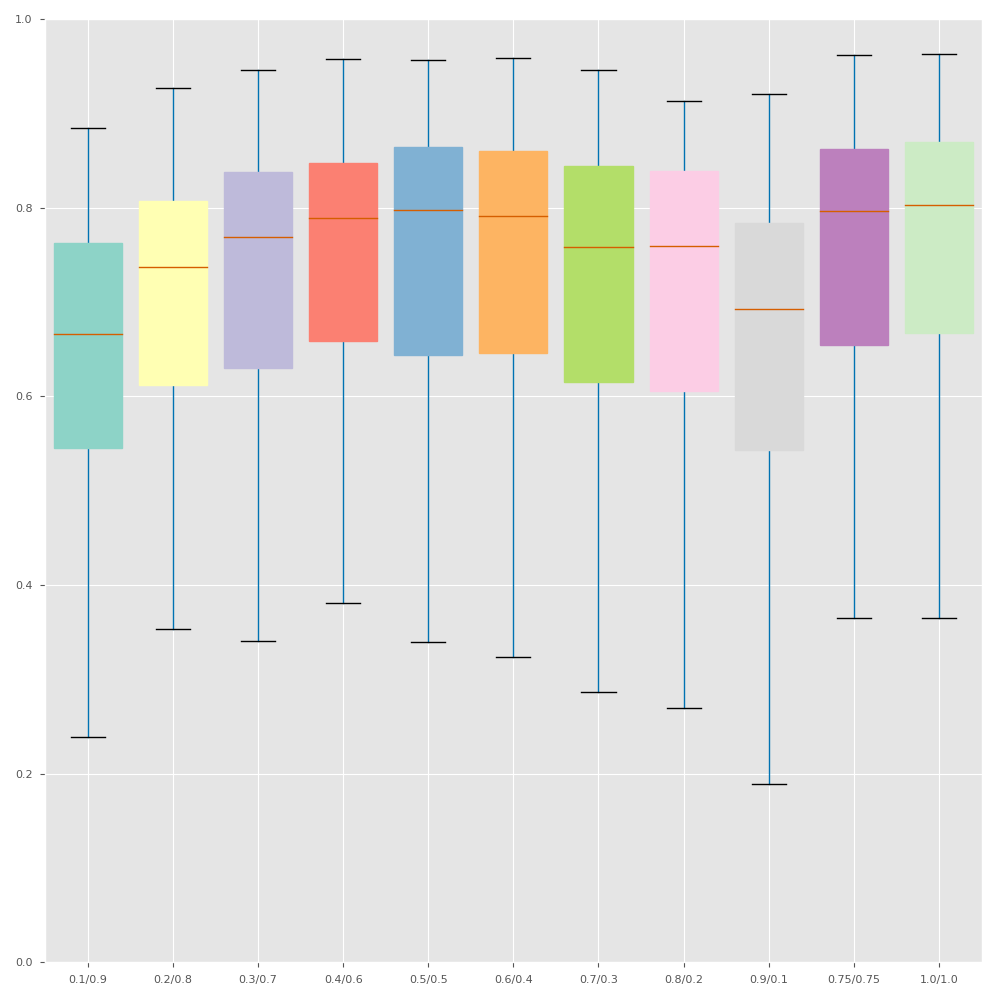} &  
        \includegraphics[width=\linewidth]{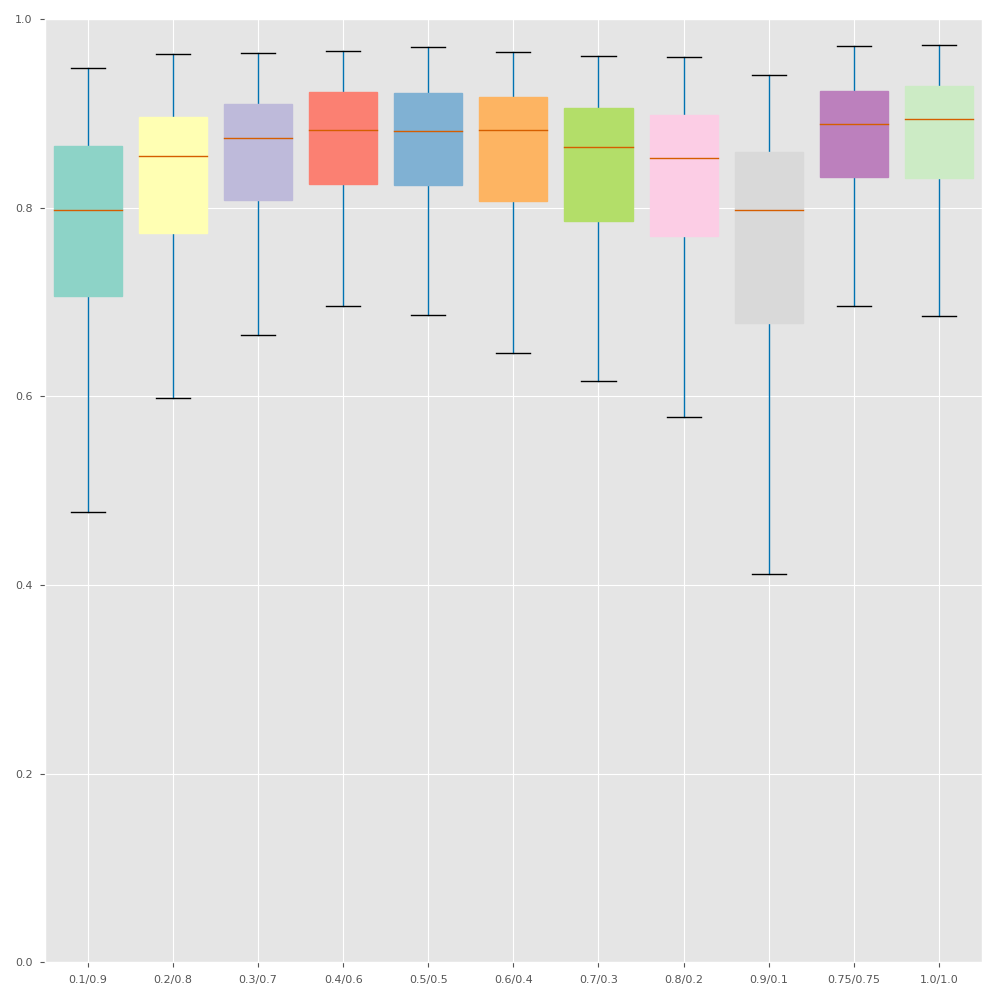} &  
        \includegraphics[width=\linewidth]{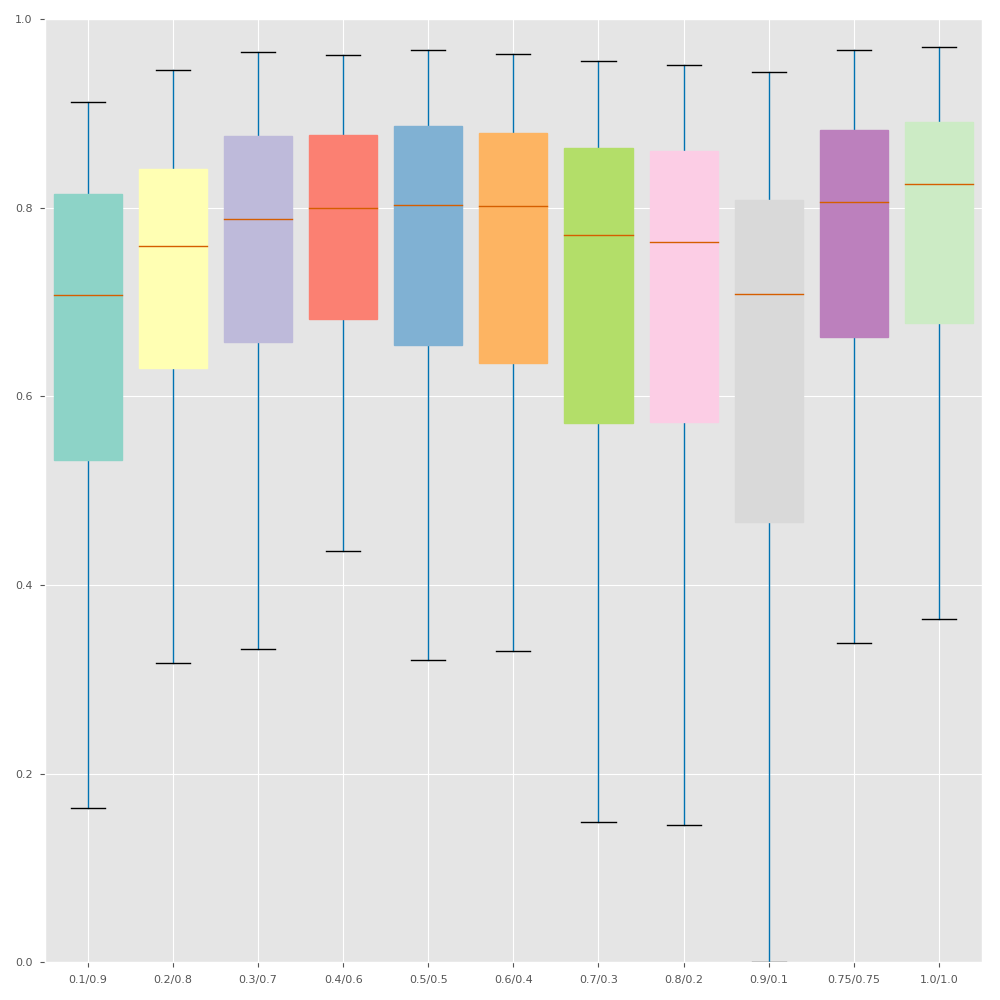} &  
        \includegraphics[width=\linewidth]{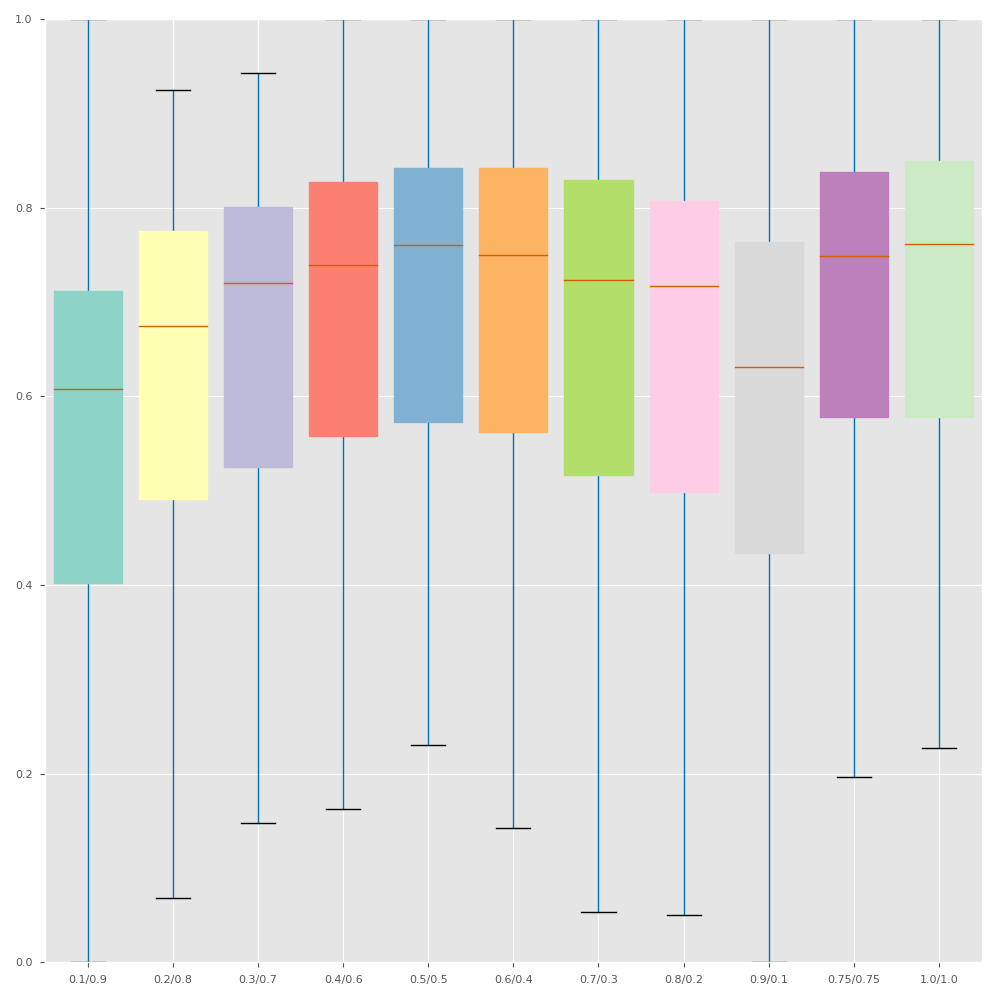}   \\ 
        
        \rotatebox{90}{\hspace{5pt} JACC} &
        \includegraphics[width=\linewidth]{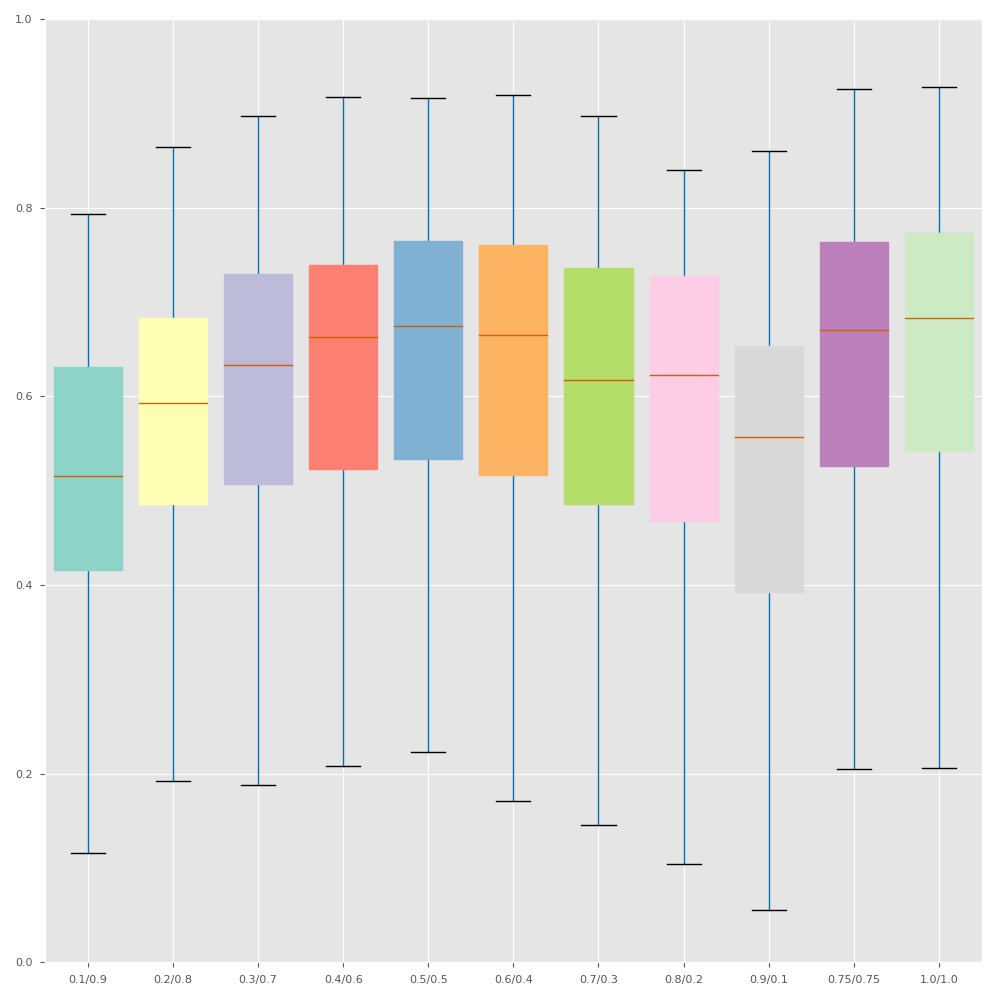} &  
        \includegraphics[width=\linewidth]{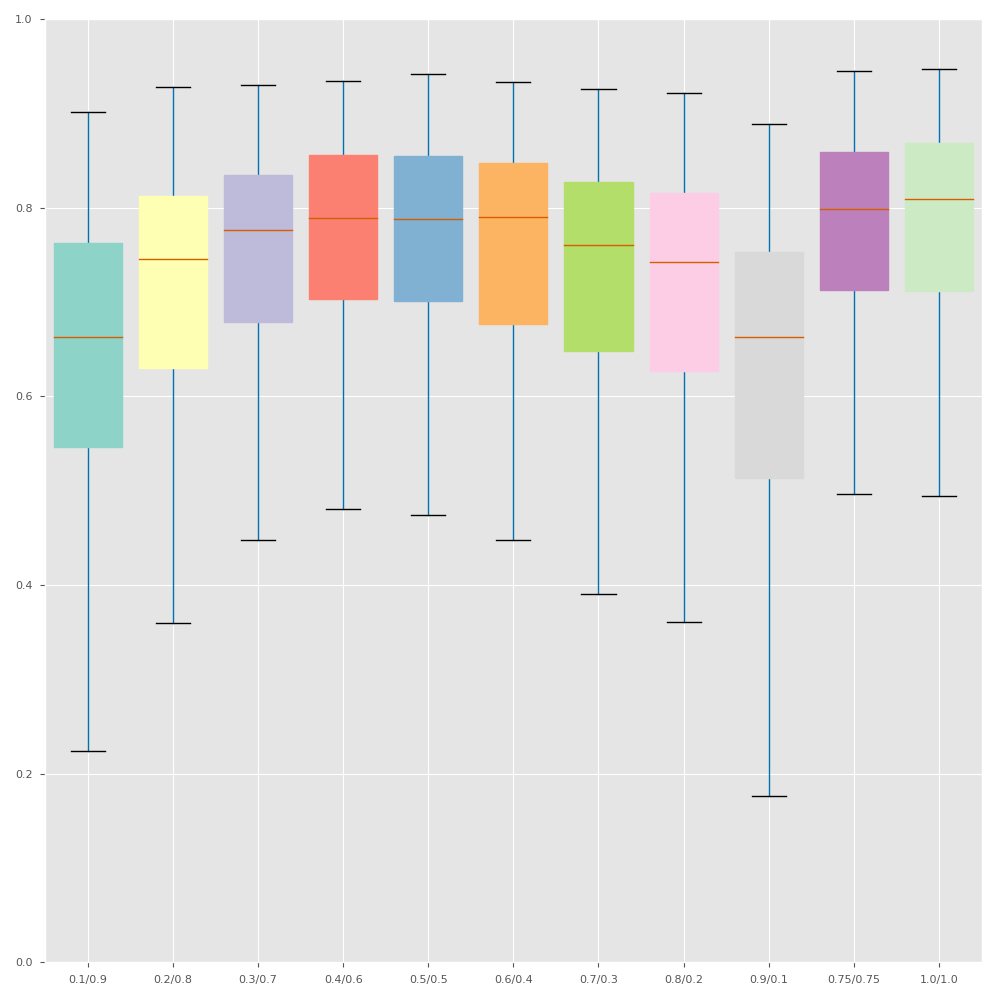} &  
        \includegraphics[width=\linewidth]{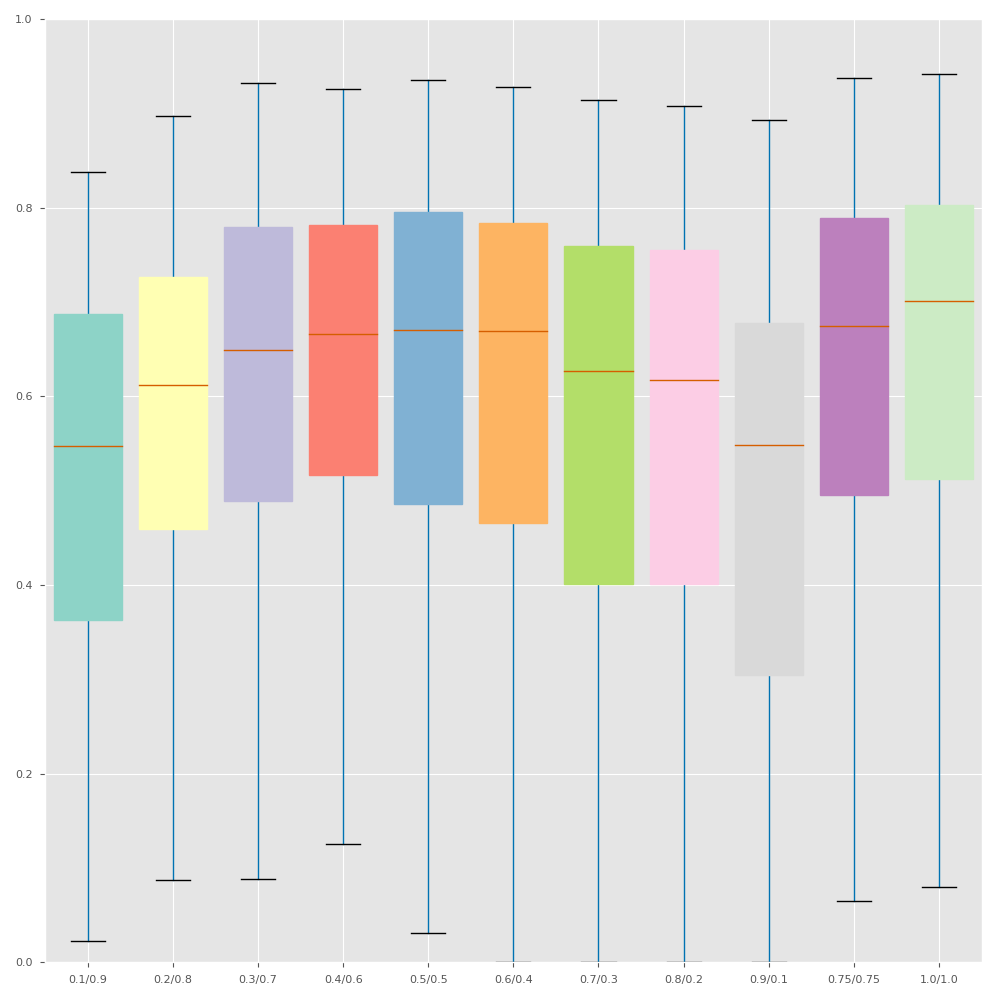} &  
        \includegraphics[width=\linewidth]{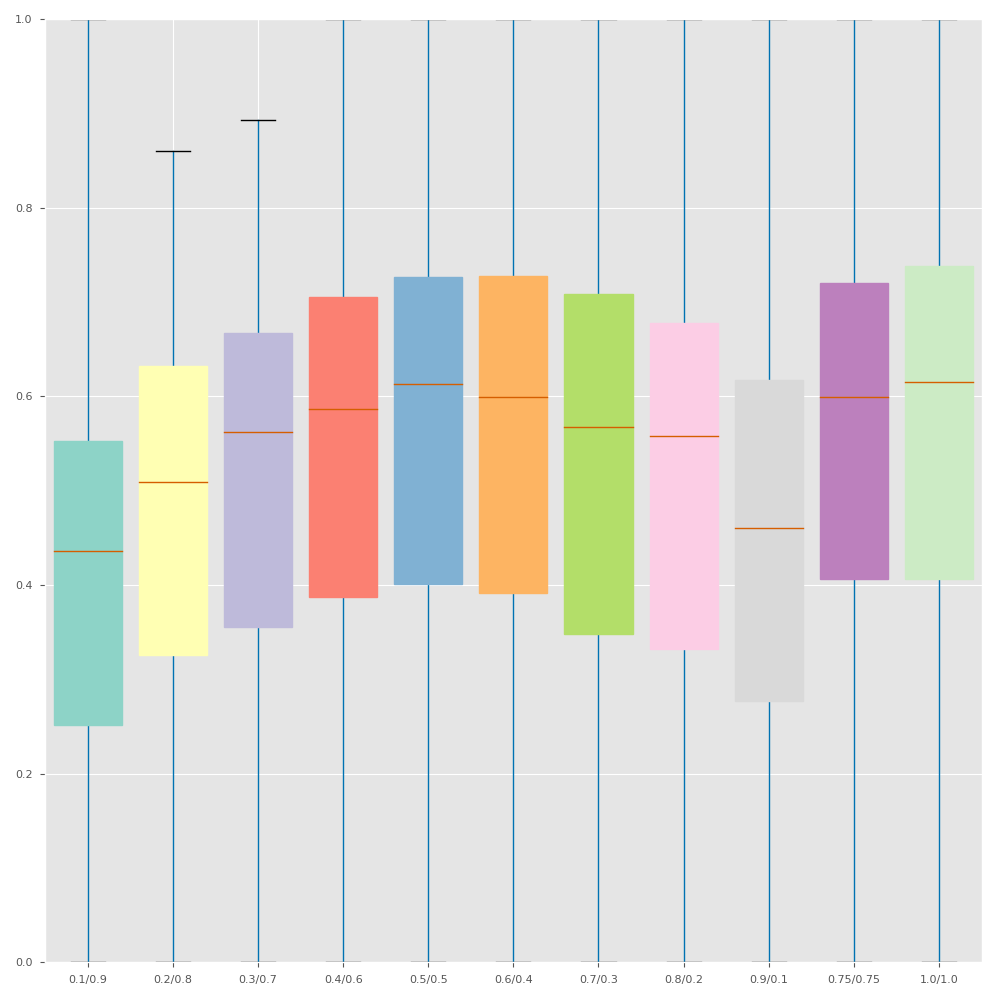} \\ 

    \end{tabularx}
    }
        \caption{\reviewminorpar Boxplots for the Dice scores and Jaccard indexes obtained for the multi-class BRATS dataset using the range of sTversky
        losses with varying alpha/beta. Results are presented for the three volumes that were considered for evaluation during the challenge: whole tumor (WT), tumor core (TC) and enhancing tumor (ET). We also report the average across the three types (all).}
    \label{fig:boxplots_multiclass_tversky}
\end{figure}

\begin{figure}[!htbp]
    \centering
    \resizebox*{\linewidth}{!}{
    
    \def\arraystretch{0}
    \setlength{\tabcolsep}{0pt}
    \begin{tabularx}{\linewidth}{l @{\hspace{2pt}} YYYYYYY}

        & BR18 &   IS17 &  IS18 & MO18 & PO18 & WM17 & WM17\textsuperscript{DM}  \\
        
        \rotatebox{90}{\hspace{5pt} ACC} &
        \includegraphics[width=\linewidth]{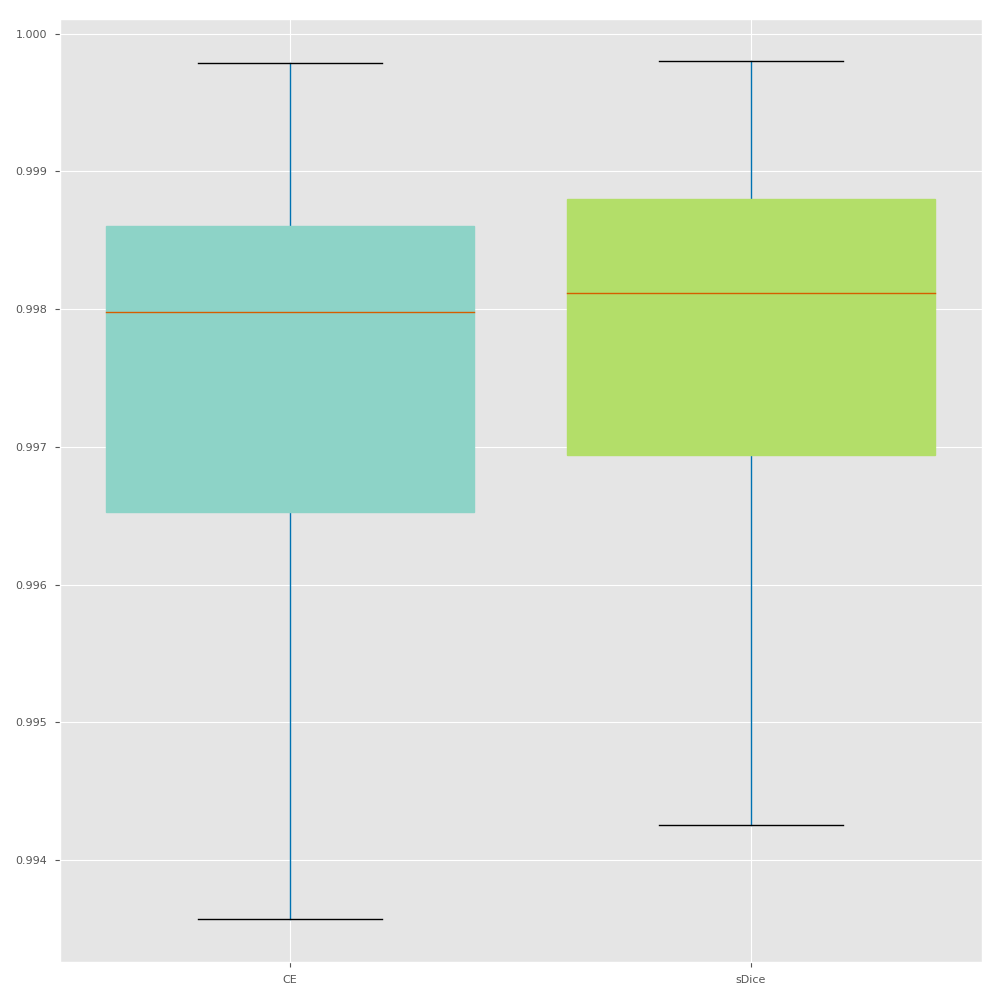} &  
        \includegraphics[width=\linewidth]{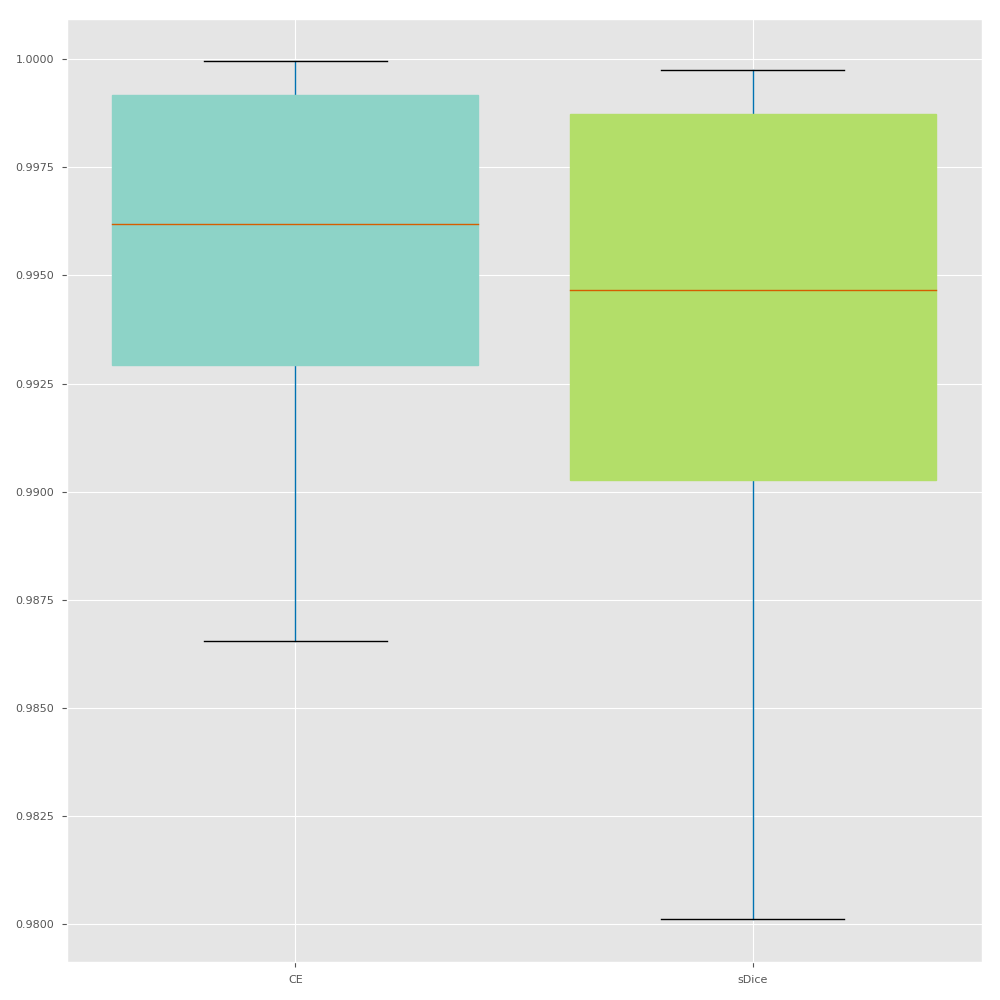} &  
        \includegraphics[width=\linewidth]{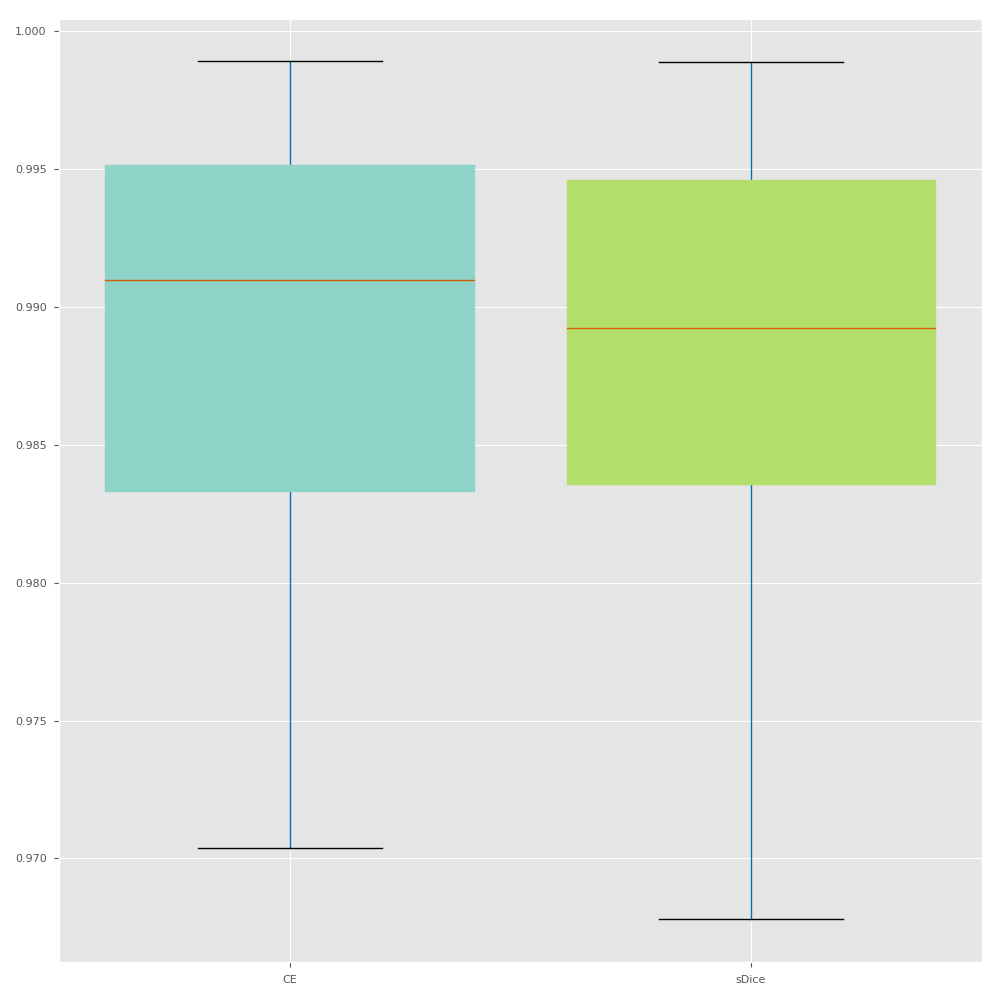} &  
        \includegraphics[width=\linewidth]{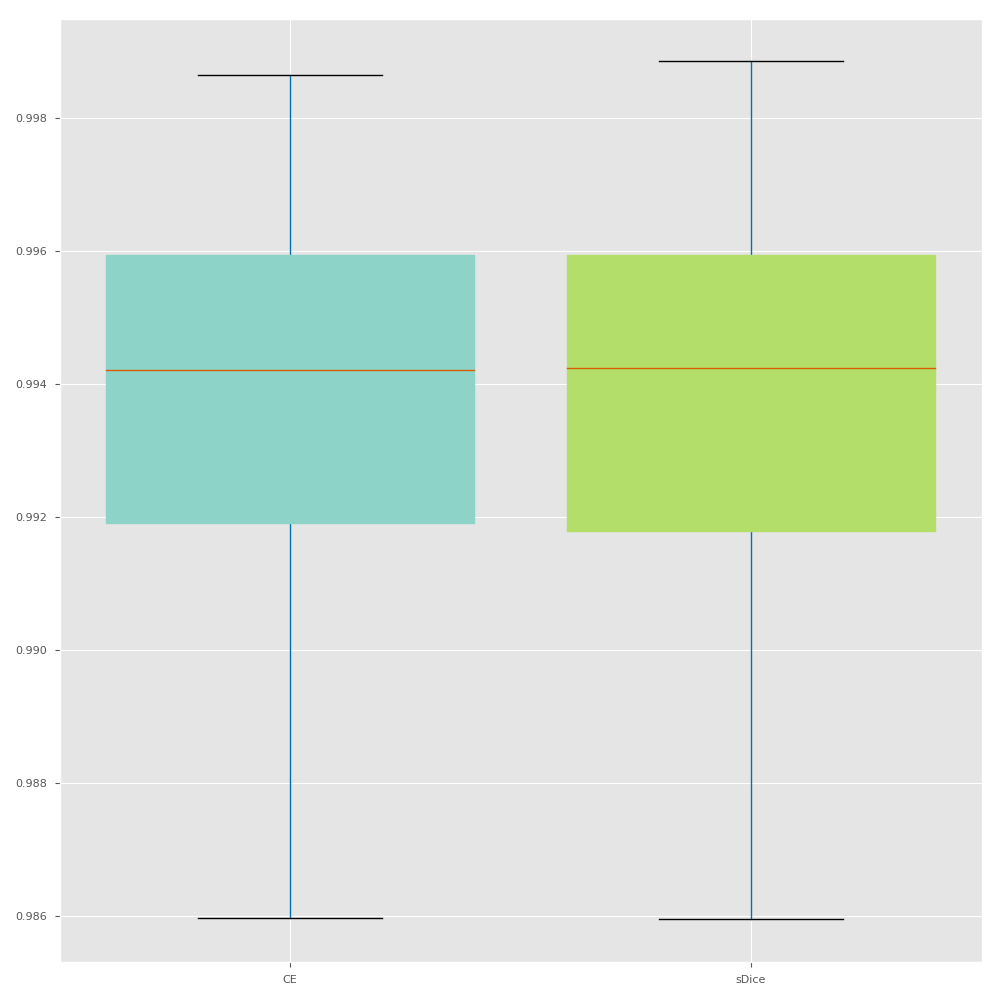} &  
        \includegraphics[width=\linewidth]{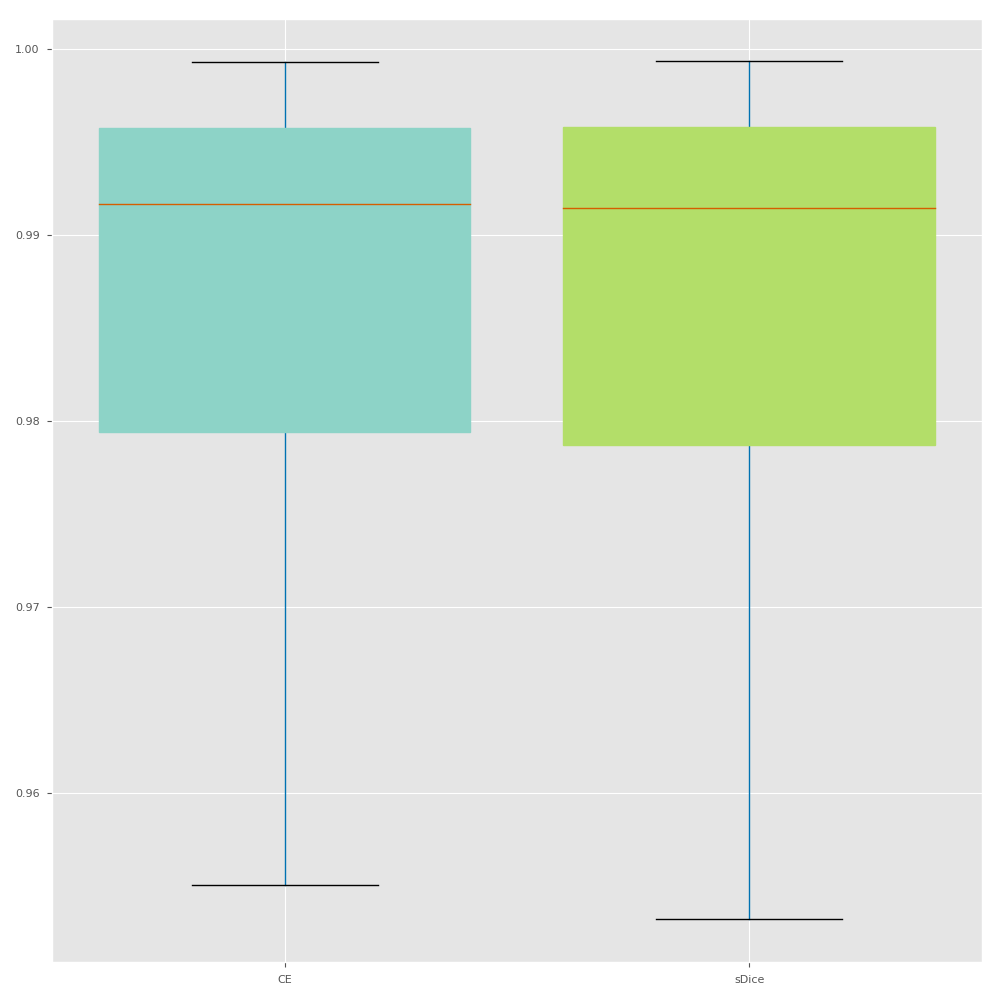} &  
        \includegraphics[width=\linewidth]{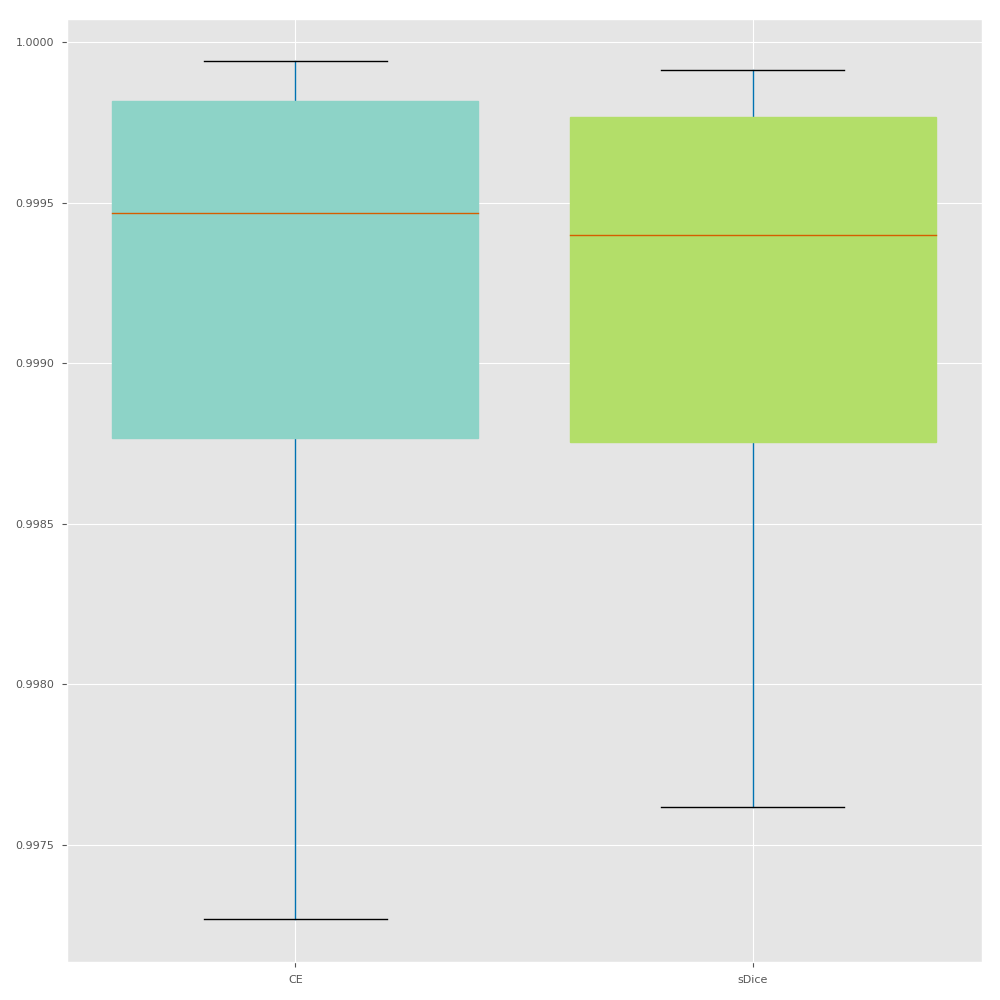} &  
        \includegraphics[width=\linewidth]{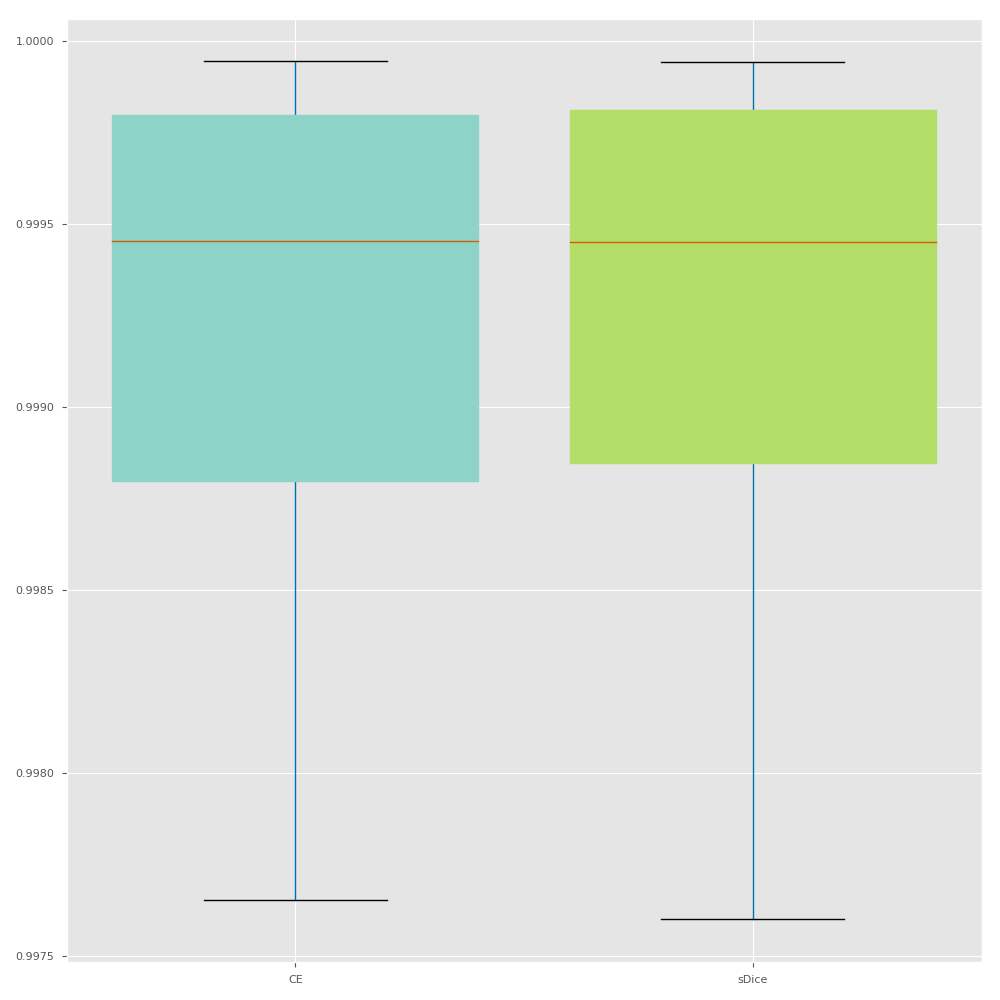}  \\ 
        
        \rotatebox{90}{\hspace{5pt} HAU} &
        \includegraphics[width=\linewidth]{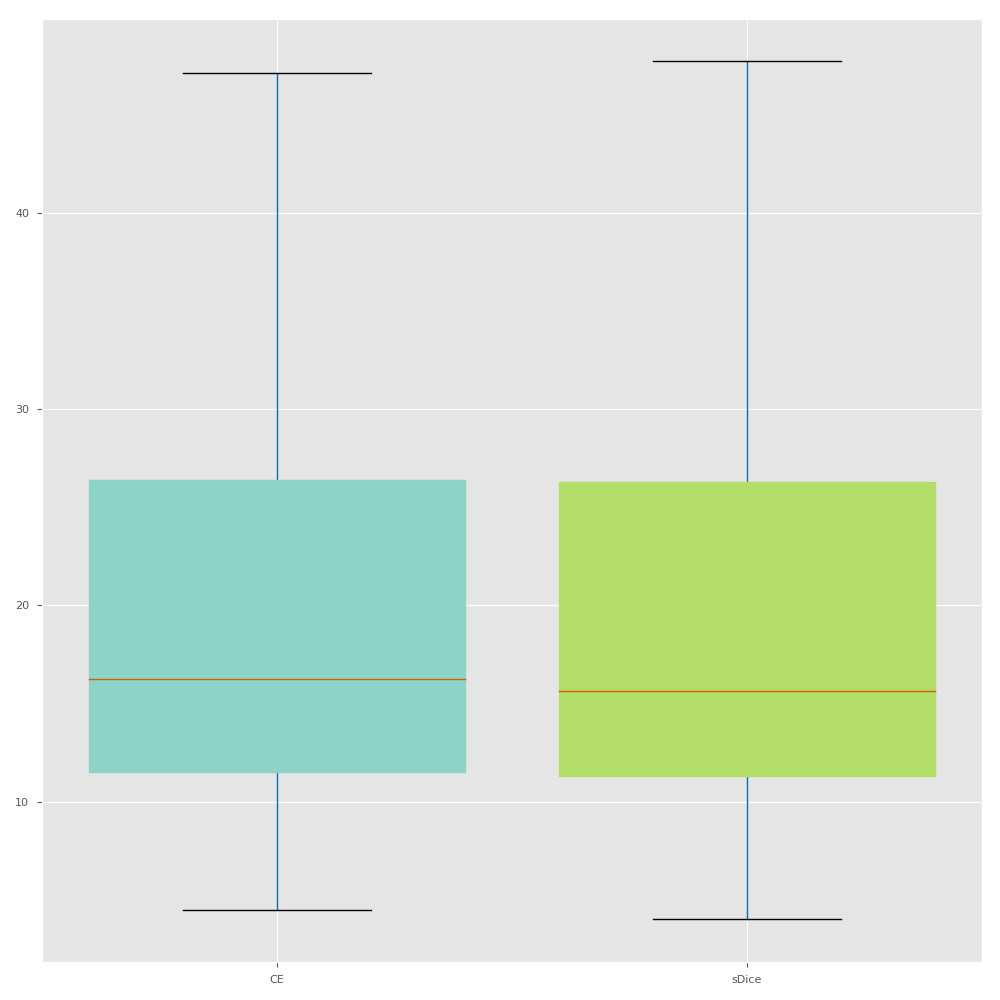} &  
        \includegraphics[width=\linewidth]{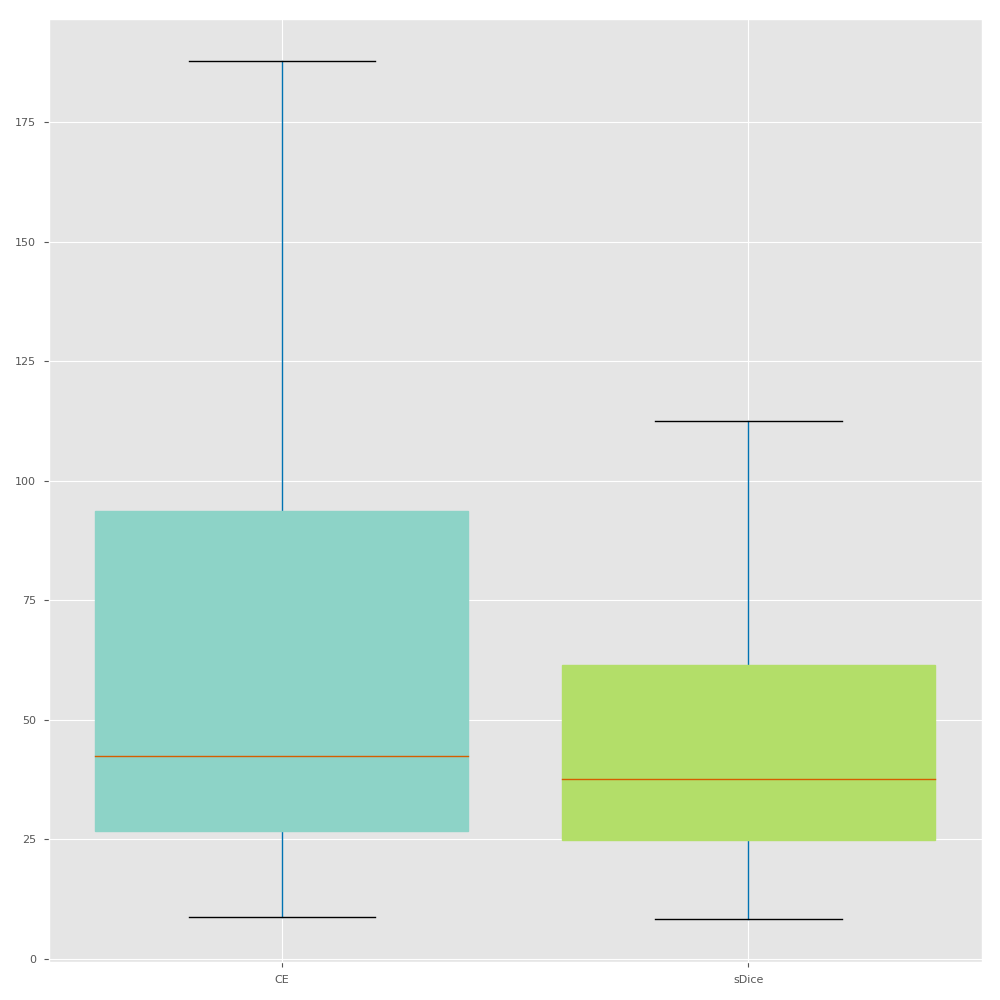} &  
        \includegraphics[width=\linewidth]{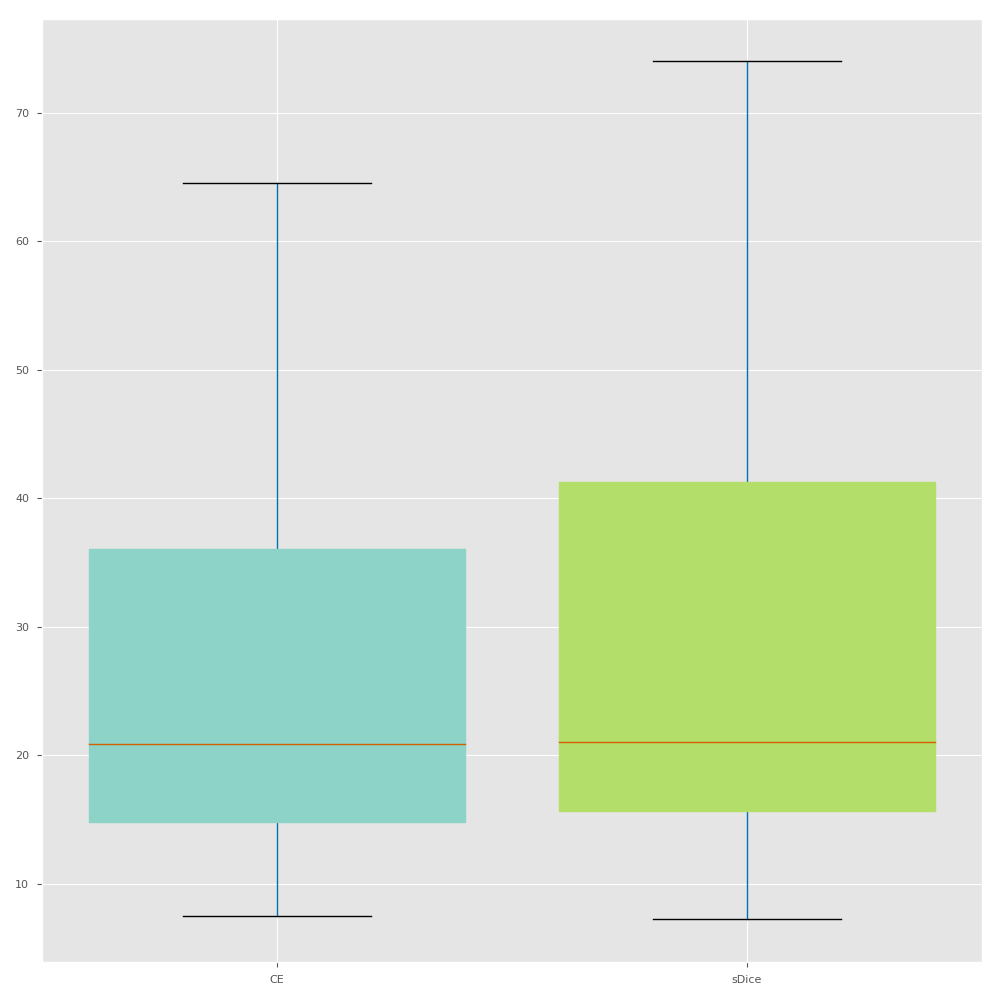} &  
        \includegraphics[width=\linewidth]{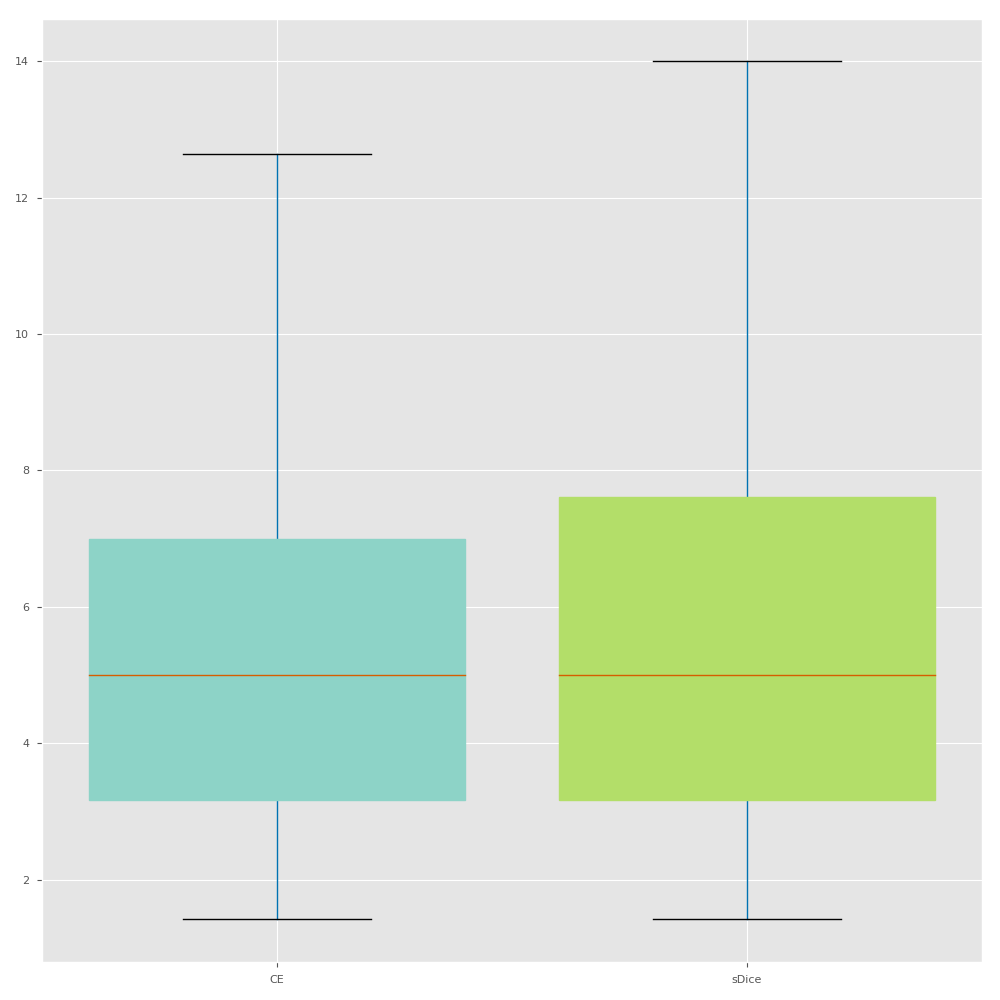} &  
        \includegraphics[width=\linewidth]{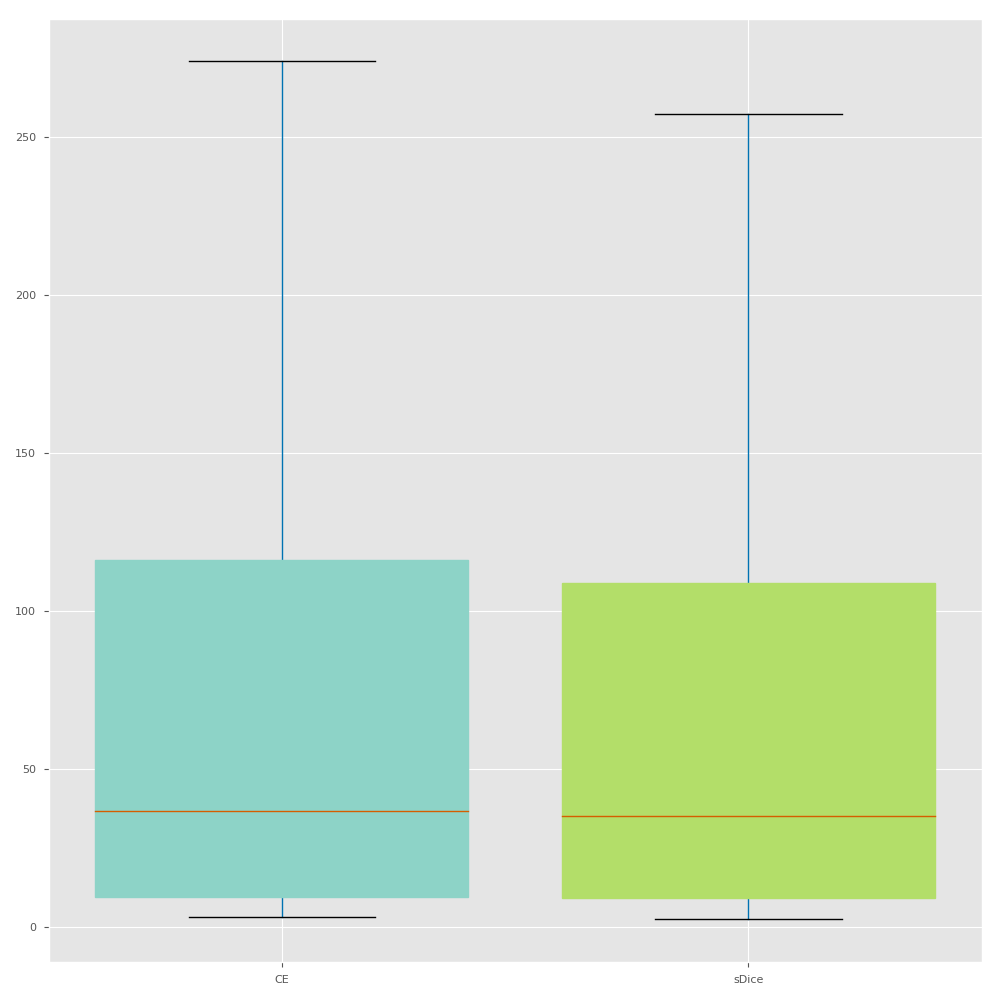} &  
        \includegraphics[width=\linewidth]{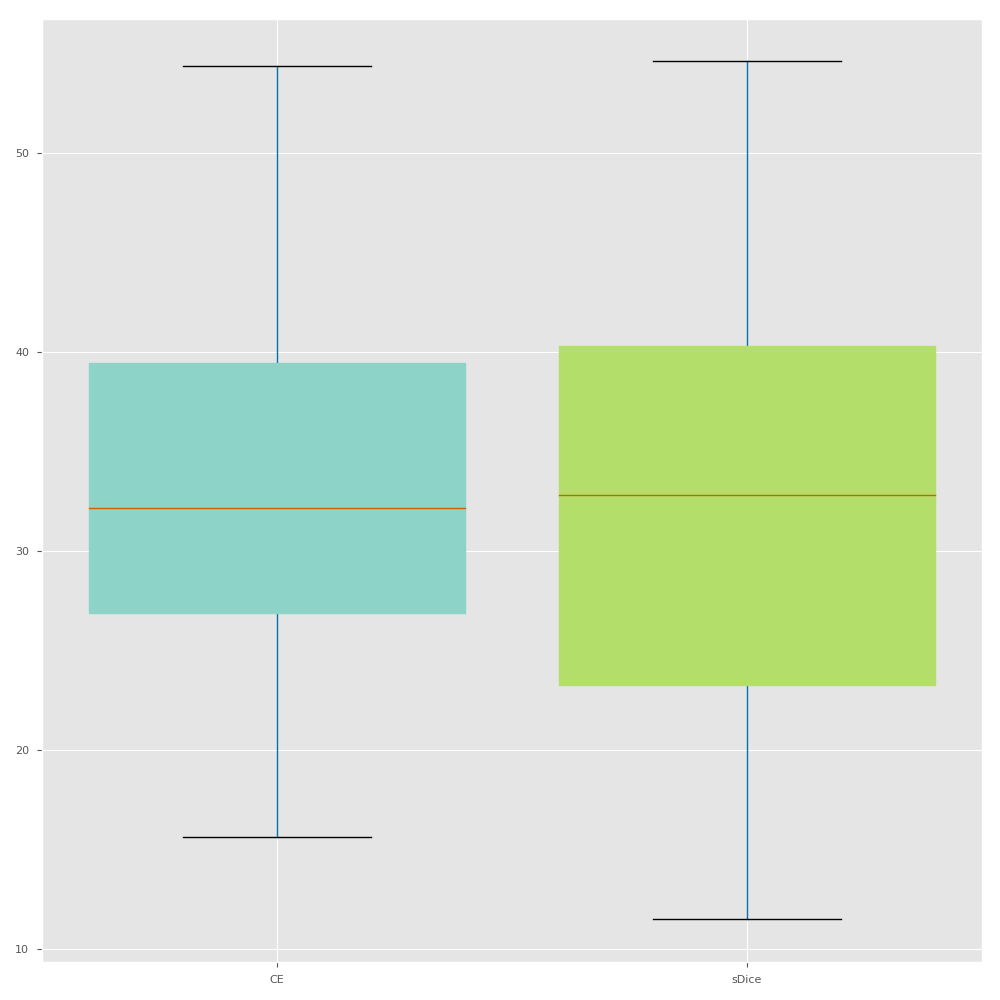} &  
        \includegraphics[width=\linewidth]{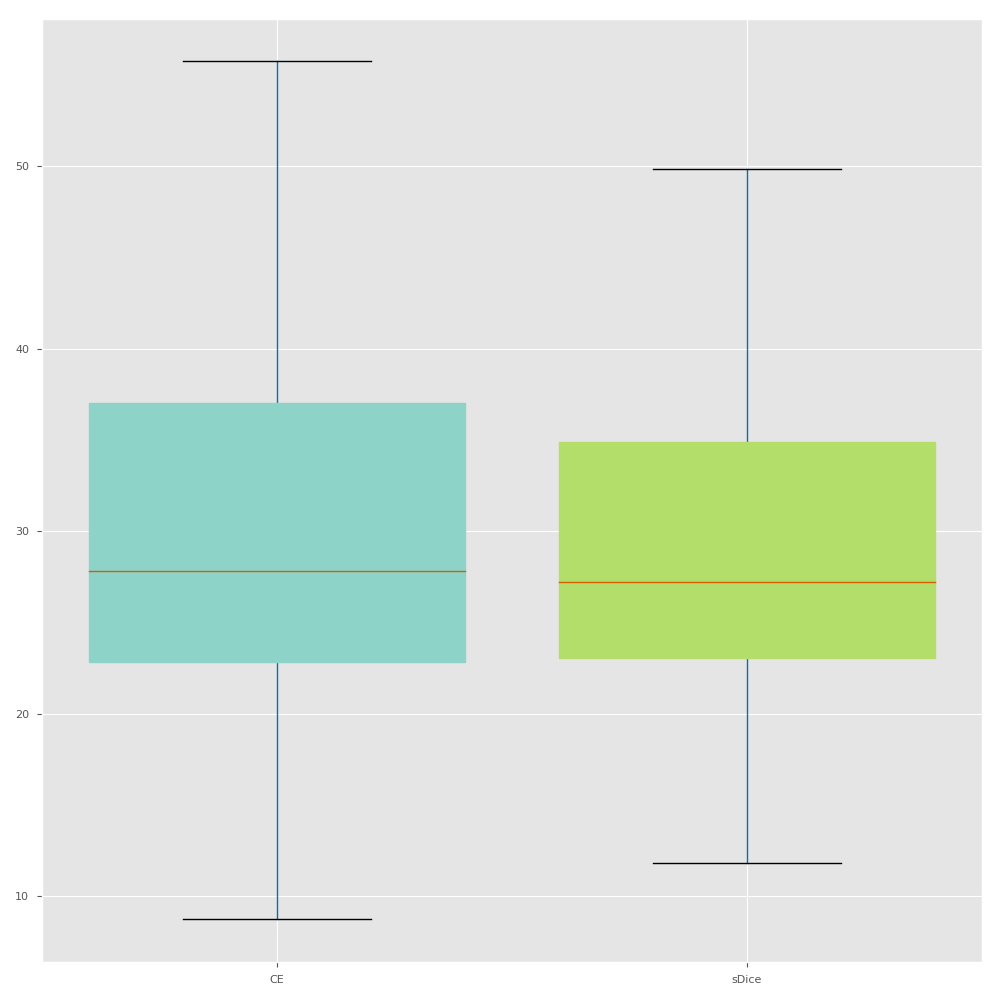}  \\ 
        
        \rotatebox{90}{\hspace{5pt} AVD} &
        \includegraphics[width=\linewidth]{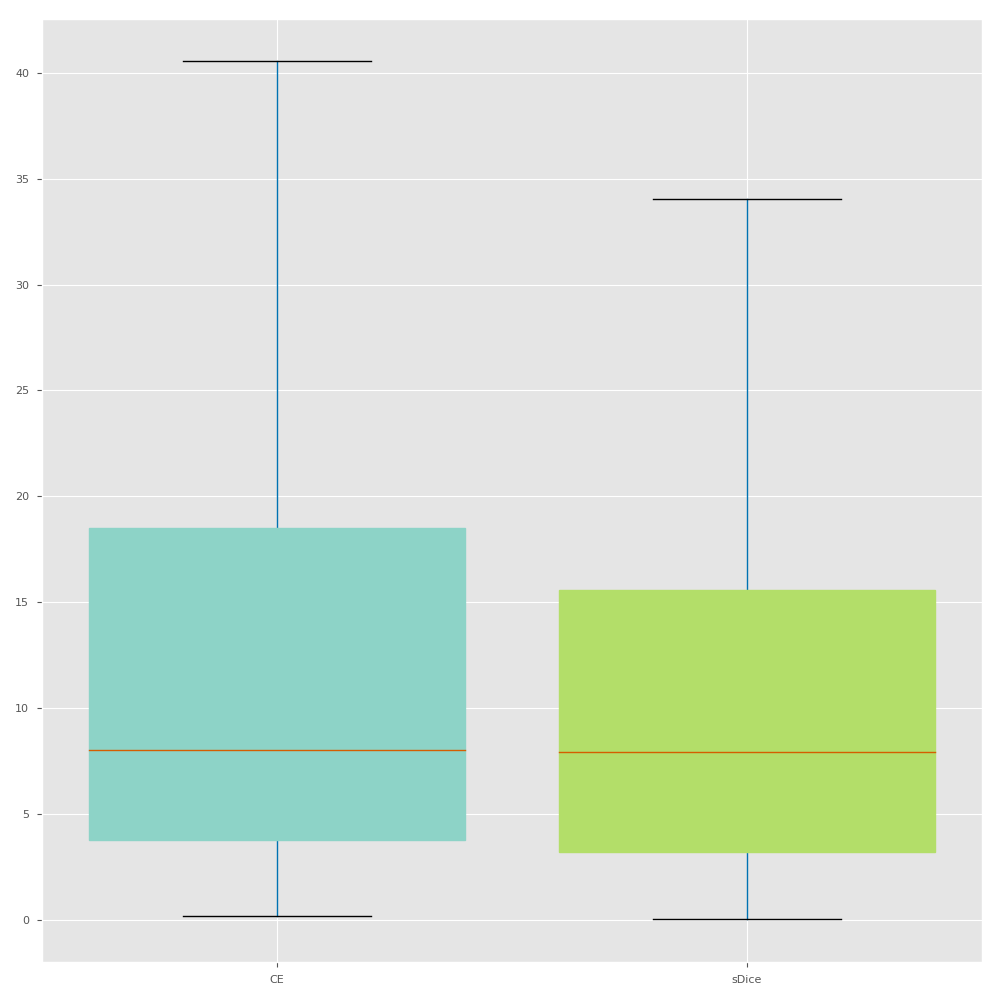} &  
        \includegraphics[width=\linewidth]{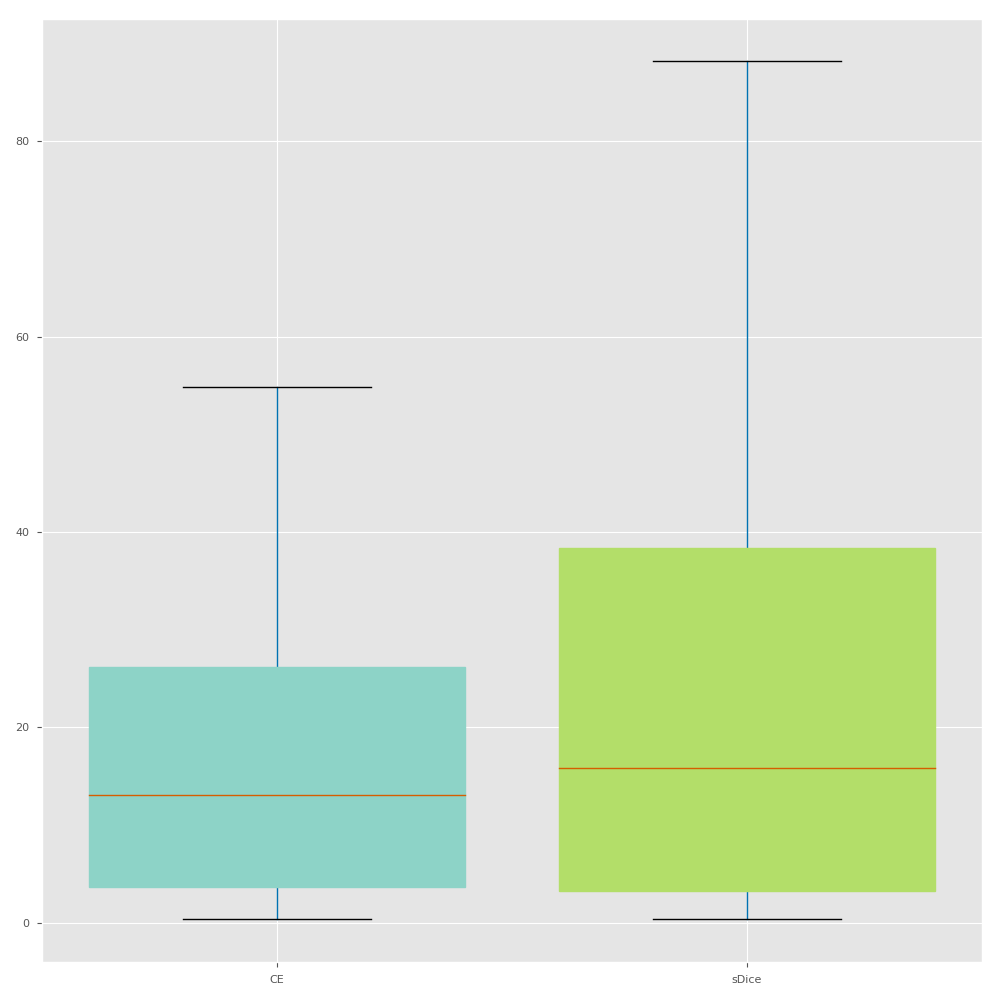} &  
        \includegraphics[width=\linewidth]{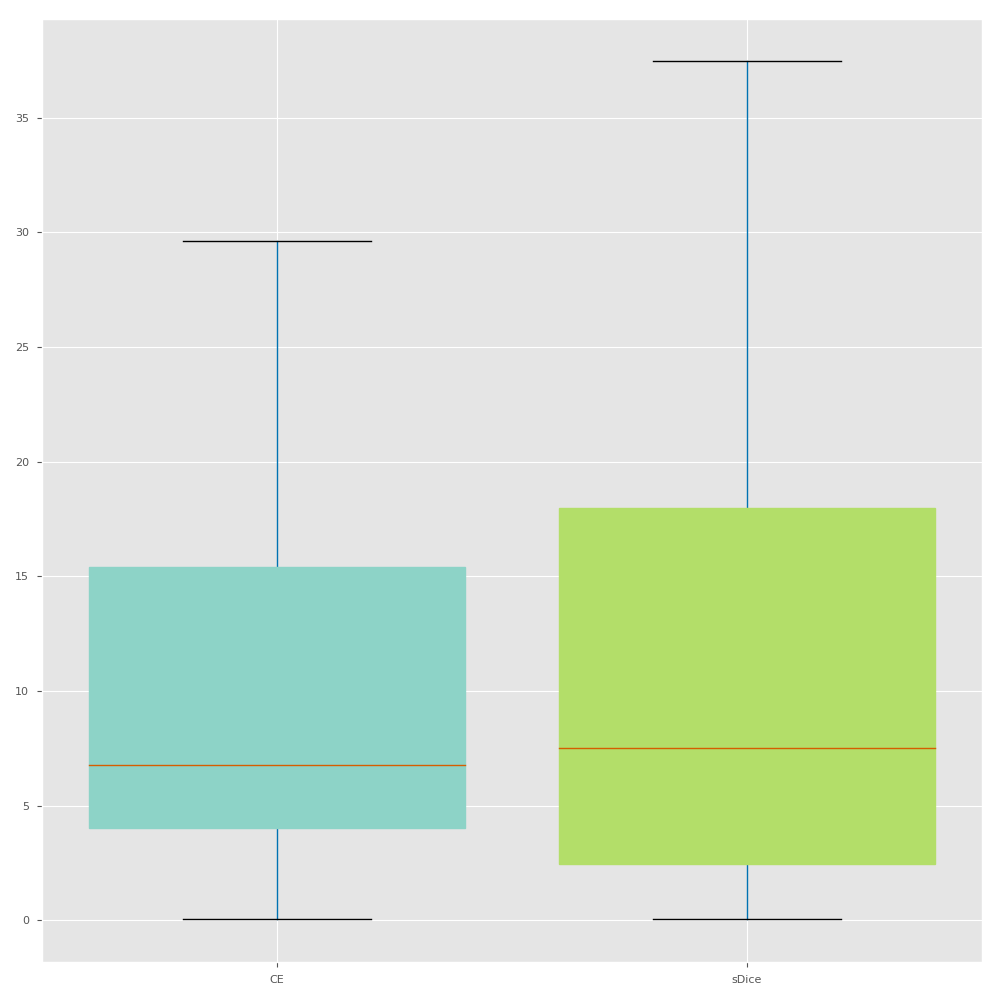} &  
        \includegraphics[width=\linewidth]{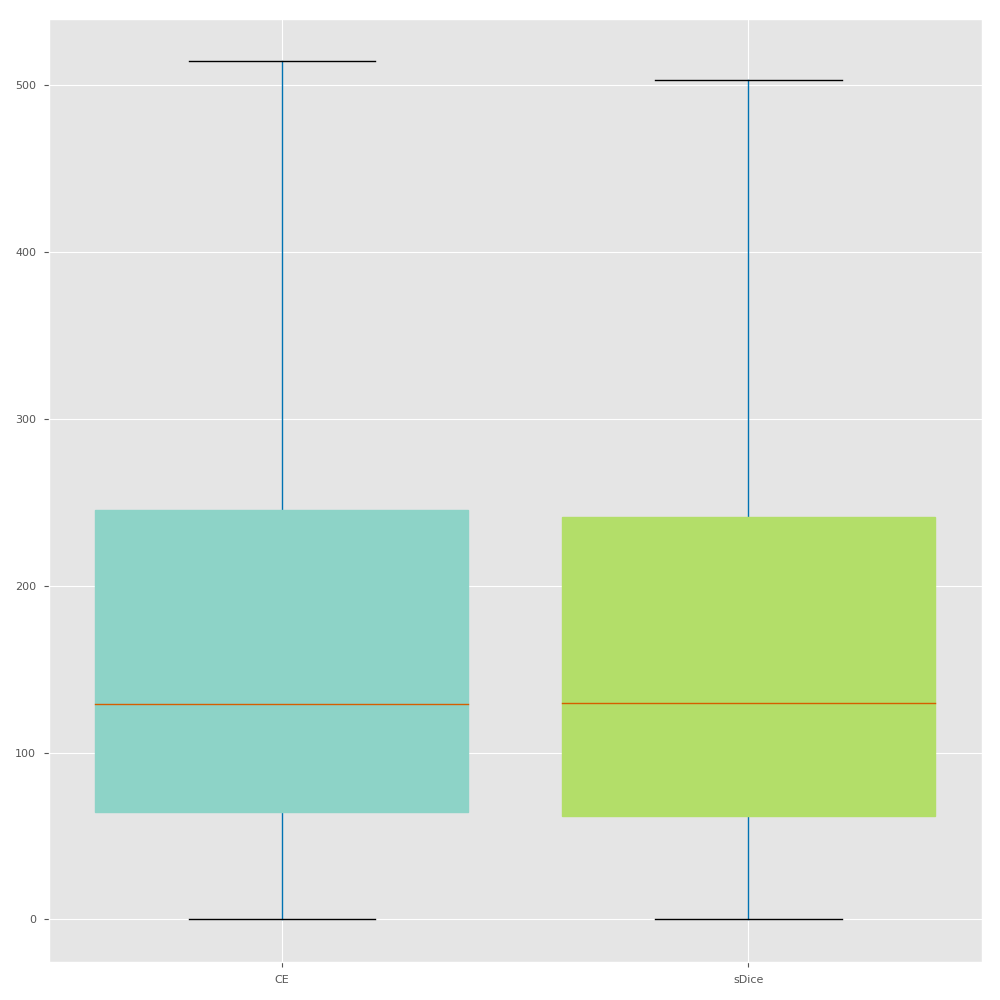} &  
        \includegraphics[width=\linewidth]{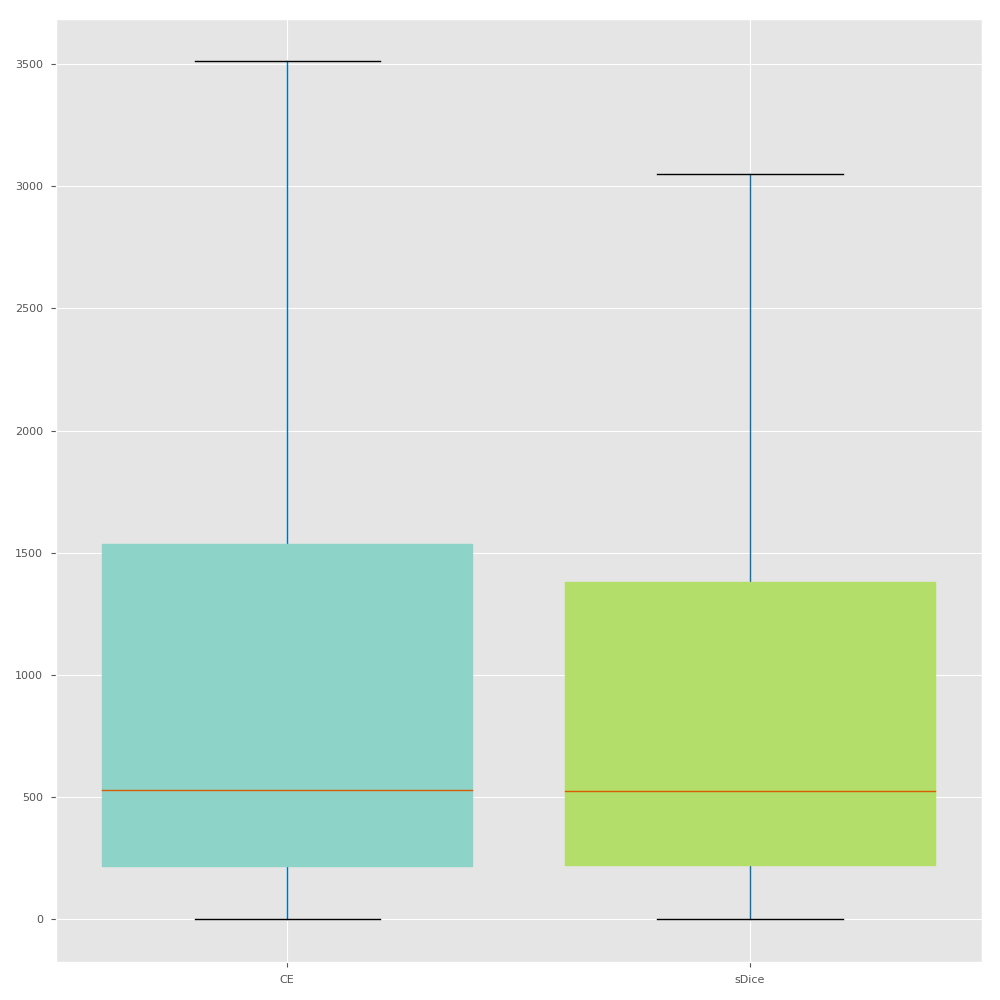} &  
        \includegraphics[width=\linewidth]{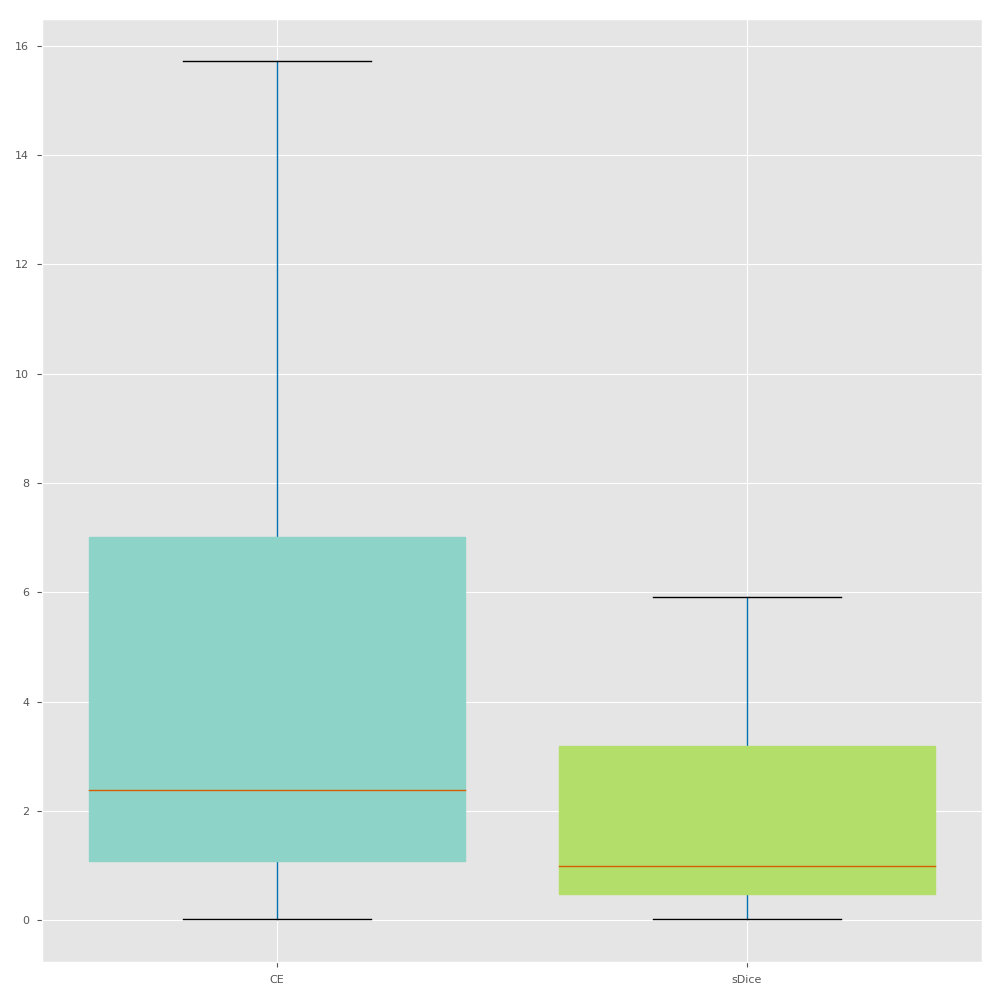} &  
        \includegraphics[width=\linewidth]{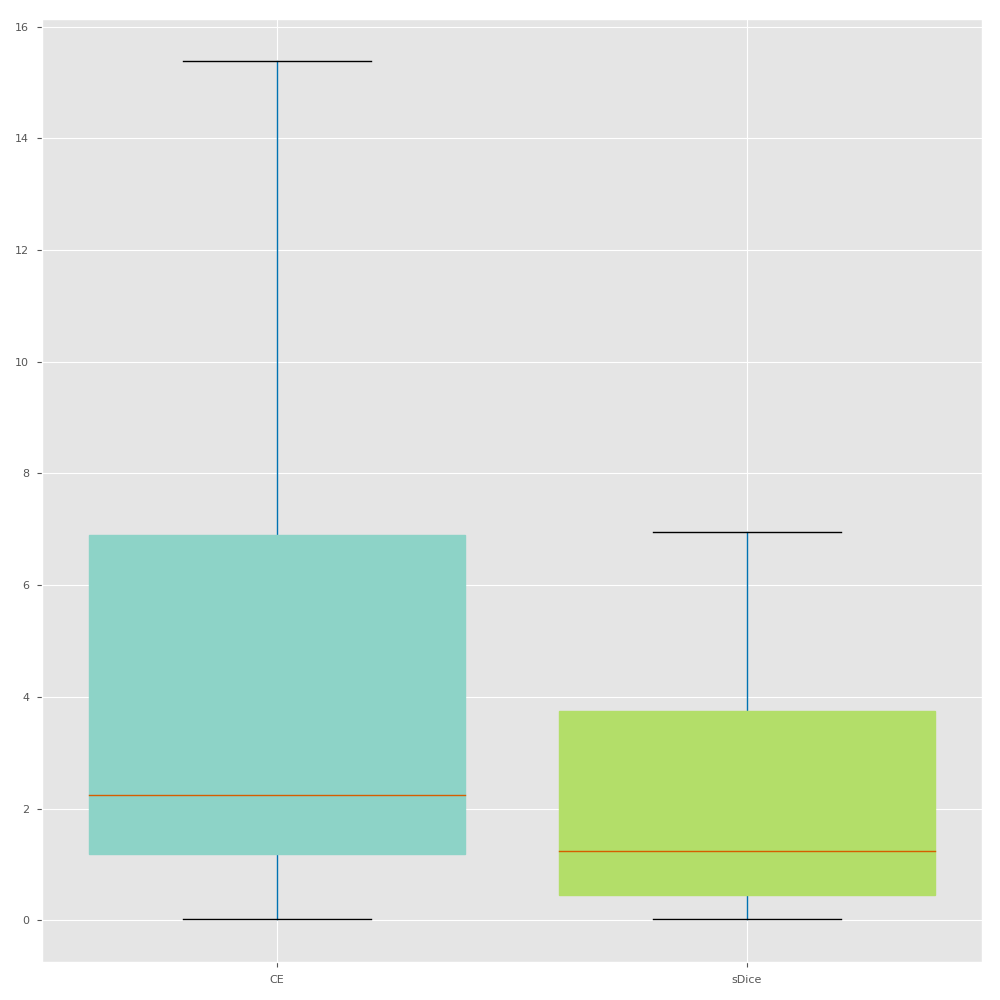}  \\ 

    \end{tabularx}
    }
        \caption{\reviewminorpar Boxplots for the accuracy, Hausdorff distance and absolute volume difference obtained for each dataset, using a cross-entropy loss (CE, wCE) or a metric-sensitive loss (sDice, sAVD, Lov\'{a}sz).}
    \label{fig:boxplots_othermetrics}
\end{figure}

\begin{table*}
    \centering
    \caption{\reviewpar Alternative F measures obtained for each dataset using the range of sTversky losses with varying alpha/beta. Cells in grey highlight the top-ranked losses and values that are significantly inferior to all others are shown in italic.}
    \begin{tabularx}{\linewidth}{lsXYYYYYYYYY|YY}
    \toprule
    & Measure & \multicolumn{1}{r}{\lapbox[\width]{1em}{\emph{$\alpha/\beta$} $\rightarrow$}} & 0.1/0.9 & 0.2/0.8 & 0.3/0.7 & 0.4/0.6 & 0.5/0.5 & 0.6/0.4 & 0.7/0.3 & 0.8/0.2 & 0.9/0.1 & 0.75/0.75 & 1.0/1.0 \\
    \midrule
    \parbox[t]{7mm}{\centering\multirow{5}{*}{\rotatebox[origin=c]{90}{BR18}}} 
    & F0.5 & & \textit{0.739} & 0.806 & 0.841 & 0.853 & 0.878 & 0.890 & \cellcolor{gray!50}0.896 & \cellcolor{gray!50}0.894 & 0.879 &  0.878 & 0.877 \\
    & F1.0 & & 0.801 & 0.844 & 0.859 & 0.863 & \cellcolor{gray!50}0.870 & \cellcolor{gray!50}0.867 & 0.857 & 0.831 & \textit{0.787} &  \cellcolor{gray!50}0.871 & \cellcolor{gray!50}0.868 \\
    & F1.5 & & 0.850 & \cellcolor{gray!50}0.874 & \cellcolor{gray!50}0.875 & \cellcolor{gray!50}0.873 & 0.867 & 0.855 & 0.837 & 0.798 & \textit{0.741} &  0.869 & 0.865 \\
    & F2.0 & & 0.882 & \cellcolor{gray!50}0.893 & 0.885 & 0.880 & 0.867 & 0.850 & 0.827 & 0.782 & \textit{0.719} &  0.869 & 0.865 \\
    \midrule
    \parbox[t]{7mm}{\centering\multirow{5}{*}{\rotatebox[origin=c]{90}{IS17}}} 
    & F0.5 & & 0.316 & 0.318 & \cellcolor{gray!50}0.358 & \cellcolor{gray!50}0.372 & \cellcolor{gray!50}0.368 & \cellcolor{gray!50}0.367 & \cellcolor{gray!50}0.387 & \cellcolor{gray!50}0.382 & \cellcolor{gray!50}0.363 &  0.367 & 0.362 \\
    & F1.0 & & \cellcolor{gray!50}0.352 & 0.345 & \cellcolor{gray!50}0.374 & \cellcolor{gray!50}0.374 & \cellcolor{gray!50}0.362 & 0.346 & \cellcolor{gray!50}0.356 & \cellcolor{gray!50}0.345 & 0.293 &  \cellcolor{gray!50}0.371 & \cellcolor{gray!50}0.363 \\
    & F1.5 & & \cellcolor{gray!50}0.398 & \cellcolor{gray!50}0.384 & \cellcolor{gray!50}0.402 & \cellcolor{gray!50}0.393 & \cellcolor{gray!50}0.378 & 0.351 & 0.360 & 0.342 & \textit{0.274} &  \cellcolor{gray!50}0.394 & \cellcolor{gray!50}0.384 \\
    & F2.0 & & \cellcolor{gray!50}0.439 & \cellcolor{gray!50}0.420 & \cellcolor{gray!50}0.430 & \cellcolor{gray!50}0.415 & \cellcolor{gray!50}0.396 & 0.362 & 0.372 & 0.348 & \textit{0.270} &  \cellcolor{gray!50}0.419 & \cellcolor{gray!50}0.408 \\
    \midrule
    \parbox[t]{7mm}{\centering\multirow{5}{*}{\rotatebox[origin=c]{90}{IS18}}} 
    & F0.5 & & \textit{0.411} & 0.468 & 0.495 & 0.516 & \cellcolor{gray!50}0.528 & 0.528 & \cellcolor{gray!50}0.539 & \cellcolor{gray!50}0.543 & \cellcolor{gray!50}0.530 &  0.521 & 0.520 \\
    & F1.0 & & 0.481 & 0.522 & 0.533 & \cellcolor{gray!50}0.540 & \cellcolor{gray!50}0.538 & \cellcolor{gray!50}0.527 & 0.519 & 0.490 & \textit{0.445} &  0.528 & 0.528 \\
    & F1.5 & & 0.551 & \cellcolor{gray!50}0.573 & \cellcolor{gray!50}0.571 & \cellcolor{gray!50}0.568 & 0.555 & 0.537 & 0.517 & 0.470 & \textit{0.413} &  0.544 & 0.544 \\
    & F2.0 & & \cellcolor{gray!50}0.607 & \cellcolor{gray!50}0.612 & 0.600 & 0.590 & 0.569 & 0.547 & 0.520 & 0.464 & \textit{0.399} &  0.558 & 0.558 \\
    \midrule
    \parbox[t]{7mm}{\centering\multirow{5}{*}{\rotatebox[origin=c]{90}{MO17}}} 
    & F0.5 & & \textit{0.866} & 0.900 & 0.918 & 0.930 & 0.943 & 0.948 & 0.951 & \cellcolor{gray!50}0.954 & 0.948 &  0.943 & 0.942 \\
    & F1.0 & & 0.907 & 0.928 & 0.936 & 0.941 & \cellcolor{gray!50}0.944 & \cellcolor{gray!50}0.942 & 0.938 & 0.932 & 0.902 &  \cellcolor{gray!50}0.944 & \cellcolor{gray!50}0.942 \\
    & F1.5 & & 0.936 & 0.947 & \cellcolor{gray!50}0.949 & \cellcolor{gray!50}0.948 & 0.945 & 0.940 & 0.930 & 0.919 & \textit{0.876} &  0.945 & 0.944 \\
    & F2.0 & & 0.954 & \cellcolor{gray!50}0.959 & 0.957 & 0.953 & 0.946 & 0.939 & 0.926 & 0.913 & \textit{0.862} &  0.946 & 0.945 \\
    \midrule
    \parbox[t]{7mm}{\centering\multirow{5}{*}{\rotatebox[origin=c]{90}{PO18}}} 
    & F0.5 & & \textit{0.588} & 0.609 & 0.625 & 0.644 & \cellcolor{gray!50}0.664 & \cellcolor{gray!50}0.664 & \cellcolor{gray!50}0.671 & \cellcolor{gray!50}0.669 & \cellcolor{gray!50}0.664 & 0.656 & 0.657 \\
    & F1.0 & & 0.614 & 0.628 & 0.637 & 0.647 & \cellcolor{gray!50}0.656 & \cellcolor{gray!50}0.650 & 0.647 & 0.633 & 0.611 & 0.646 & \cellcolor{gray!50}0.651 \\
    & F1.5 & & 0.643 & 0.651 & \cellcolor{gray!50}0.654 & \cellcolor{gray!50}0.658 & \cellcolor{gray!50}0.660 & 0.651 & 0.641 & 0.620 & \textit{0.591} & 0.650 & \cellcolor{gray!50}0.656 \\
    & F2.0 & & \cellcolor{gray!50}0.665 & \cellcolor{gray!50}0.669 & \cellcolor{gray!50}0.668 & \cellcolor{gray!50}0.668 & \cellcolor{gray!50}0.665 & 0.655 & 0.640 & 0.616 & \textit{0.582} & 0.655 & \cellcolor{gray!50}0.662 \\
    \midrule
    \parbox[t]{7mm}{\centering\multirow{5}{*}{\rotatebox[origin=c]{90}{WM17}}} 
    & F0.5 & & \textit{0.498} & 0.592 & 0.659 & 0.690 & 0.724 & 0.742 & \cellcolor{gray!50}0.761 & \cellcolor{gray!50}0.759 & 0.720 &  0.723 & 0.730 \\
    & F1.0 & & 0.581 & 0.649 & 0.689 & 0.700 & \cellcolor{gray!50}0.712 & 0.706 & 0.693 & 0.660 & 0.565 &  \cellcolor{gray!50}0.712 & \cellcolor{gray!50}0.712 \\
    & F1.5 & & 0.657 & 0.700 & \cellcolor{gray!50}0.718 & 0.713 & \cellcolor{gray!50}0.711 & 0.692 & 0.660 & 0.615 & \textit{0.501} &  0.712 & 0.708 \\
    & F2.0 & & 0.714 & \cellcolor{gray!50}0.736 & \cellcolor{gray!50}0.739 & 0.724 & 0.713 & 0.686 & 0.645 & 0.593 & \textit{0.472} &  0.714 & 0.708 \\
    \midrule
    \parbox[t]{7mm}{\centering\multirow{5}{*}{\rotatebox[origin=c]{90}{WM17\textsuperscript{DM}}}} 
    & F0.5 & & \textit{0.521} & 0.617 & 0.671 & 0.709 & 0.736 & \cellcolor{gray!50}0.758 & \cellcolor{gray!50}0.759 & \cellcolor{gray!50}0.752 & 0.647 &  0.731 & 0.736 \\
    & F1.0 & & 0.599 & 0.668 & 0.696 & 0.709 & \cellcolor{gray!50}0.717 & 0.706 & 0.675 & 0.628 & \textit{0.496} &  0.706 & \cellcolor{gray!50}0.714 \\
    & F1.5 & & 0.670 & \cellcolor{gray!50}0.713 & \cellcolor{gray!50}0.722 & \cellcolor{gray!50}0.717 & 0.712 & 0.682 & 0.636 & 0.574 & \textit{0.436} &  0.697 & 0.707 \\
    & F2.0 & & 0.721 & \cellcolor{gray!50}0.744 & \cellcolor{gray!50}0.739 & 0.723 & 0.712 & 0.672 & 0.618 & 0.548 & \textit{0.409} &  0.695 & 0.706 \\
    \bottomrule
\end{tabularx}
    \label{tab:tversky_fmeasures}
\end{table*}

\begin{figure*}[!htbp]
    \centering
    \resizebox*{\linewidth}{!}{
    
    \def\arraystretch{0}
    \setlength{\tabcolsep}{0pt}
    \begin{tabularx}{\linewidth}{l @{\hspace{2pt}} YYYYYYY}

        & BR18 &   IS17 &  IS18 & MO18 & PO18 & WM17 & WM17\textsuperscript{DM}  \\
        
        \rotatebox{90}{\hspace{20pt} DICE} &
        \includegraphics[width=\linewidth]{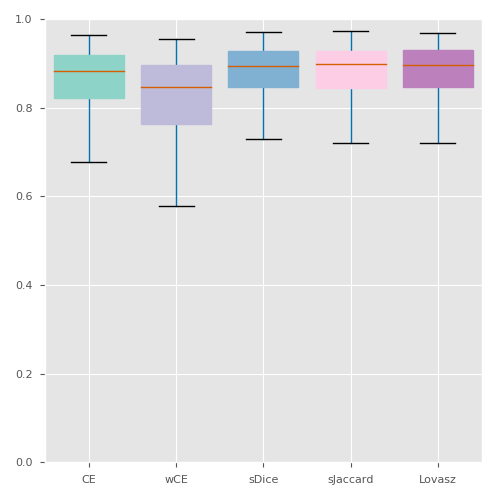} &  
        \includegraphics[width=\linewidth]{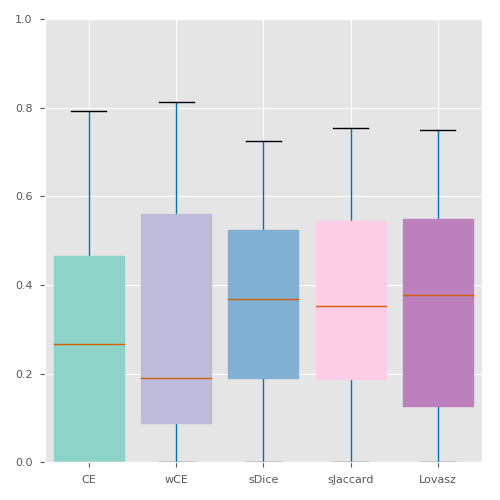} &  
        \includegraphics[width=\linewidth]{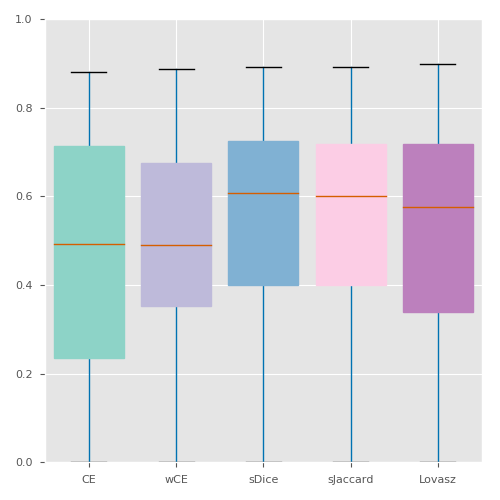} &  
        \includegraphics[width=\linewidth]{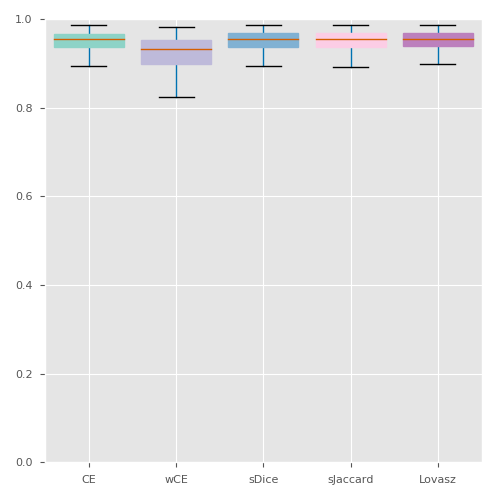} &  
        \includegraphics[width=\linewidth]{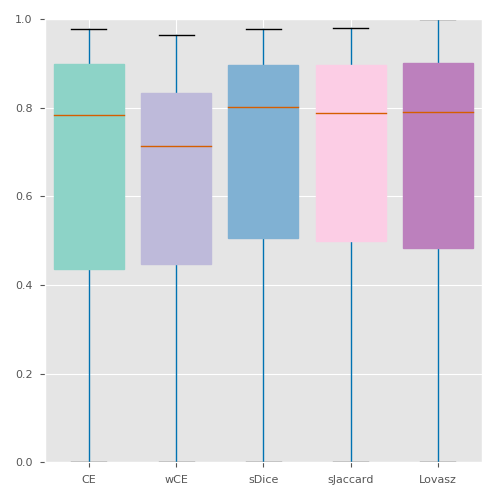} &  
        \includegraphics[width=\linewidth]{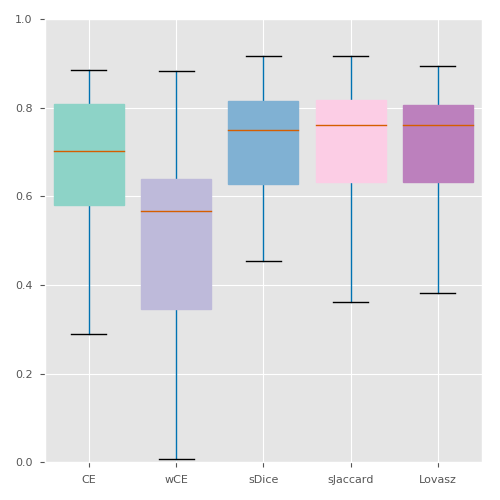} &  
        \includegraphics[width=\linewidth]{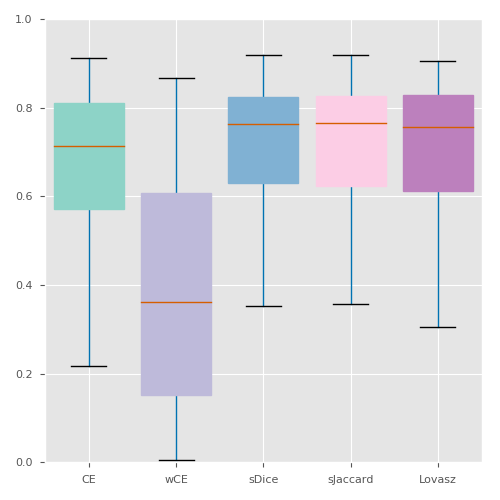}  \\ 
        
        \rotatebox{90}{\hspace{20pt} JACC} &
        \includegraphics[width=\linewidth]{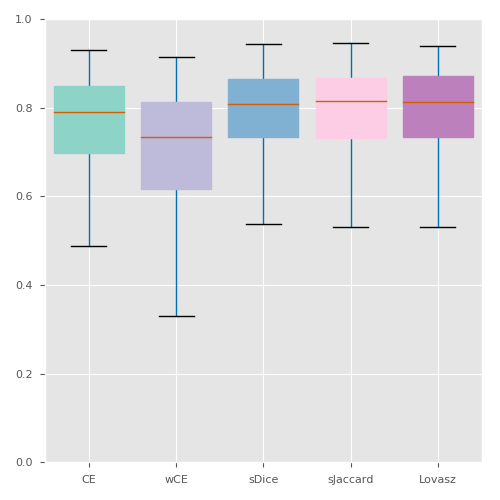} &  
        \includegraphics[width=\linewidth]{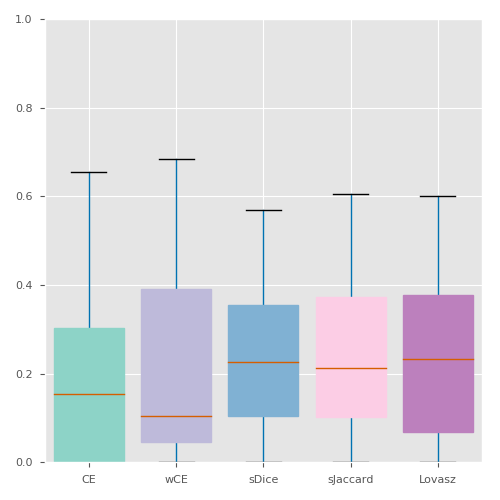} &  
        \includegraphics[width=\linewidth]{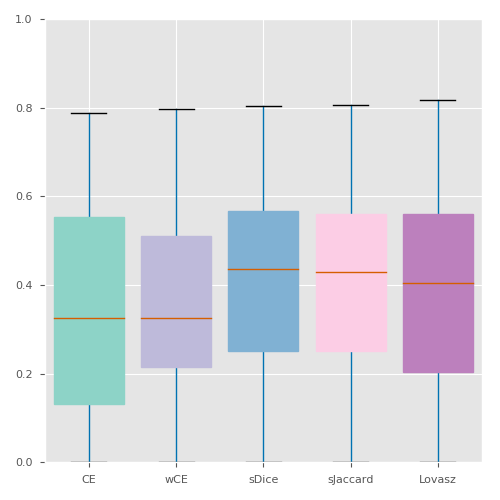} &  
        \includegraphics[width=\linewidth]{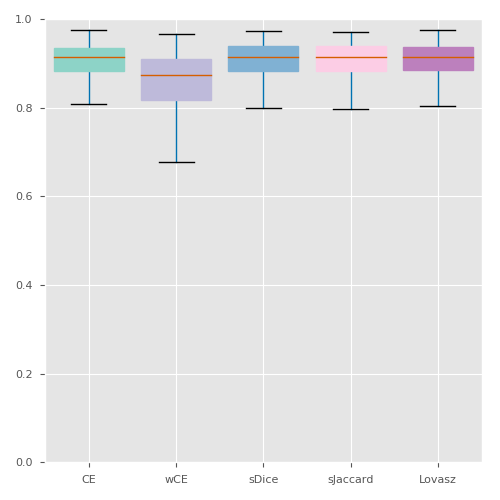} &  
        \includegraphics[width=\linewidth]{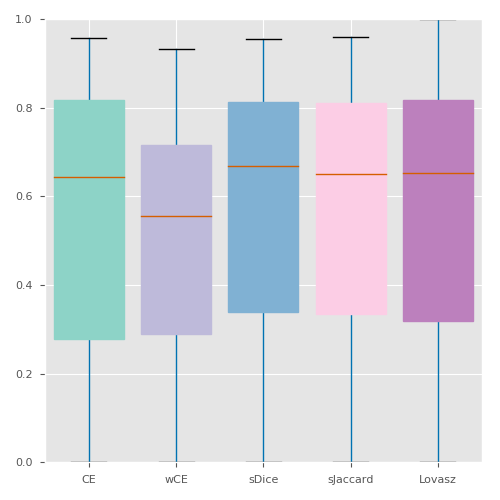} &  
        \includegraphics[width=\linewidth]{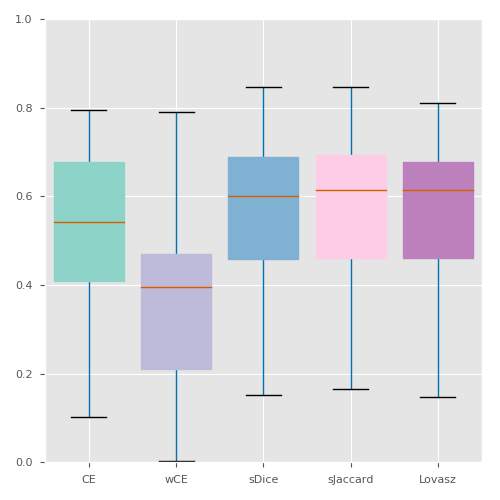} &  
        \includegraphics[width=\linewidth]{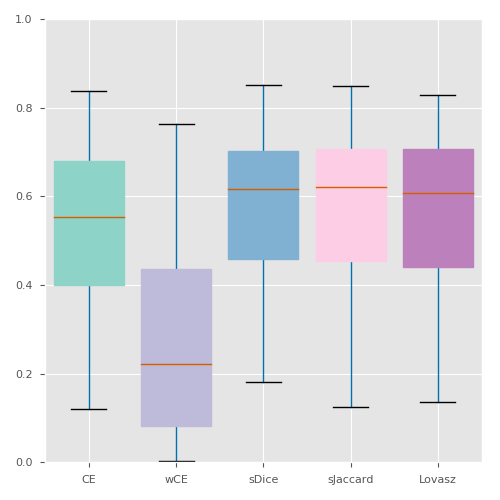}  \\ 

    \end{tabularx}
    }
        \caption{\reviewminorpar Boxplots for the Dice scores and Jaccard indexes obtained for each dataset, using a cross-entropy loss (CE, wCE) or a metric-sensitive loss (sDice, sJaccard, Lov\'{a}sz).}
    \label{fig:boxplots_main}
\end{figure*}

\begin{figure*}[!htbp]
    \centering
    \resizebox*{\linewidth}{!}{
    
    \def\arraystretch{0}
    \setlength{\tabcolsep}{0pt}
    \begin{tabularx}{\linewidth}{l @{\hspace{2pt}} YYYYYYY}

        & BR18 &   IS17 &  IS18 & MO18 & PO18 & WM17 & WM17\textsuperscript{DM}  \\
        
        \rotatebox{90}{\hspace{20pt} DICE} &
        \includegraphics[width=\linewidth]{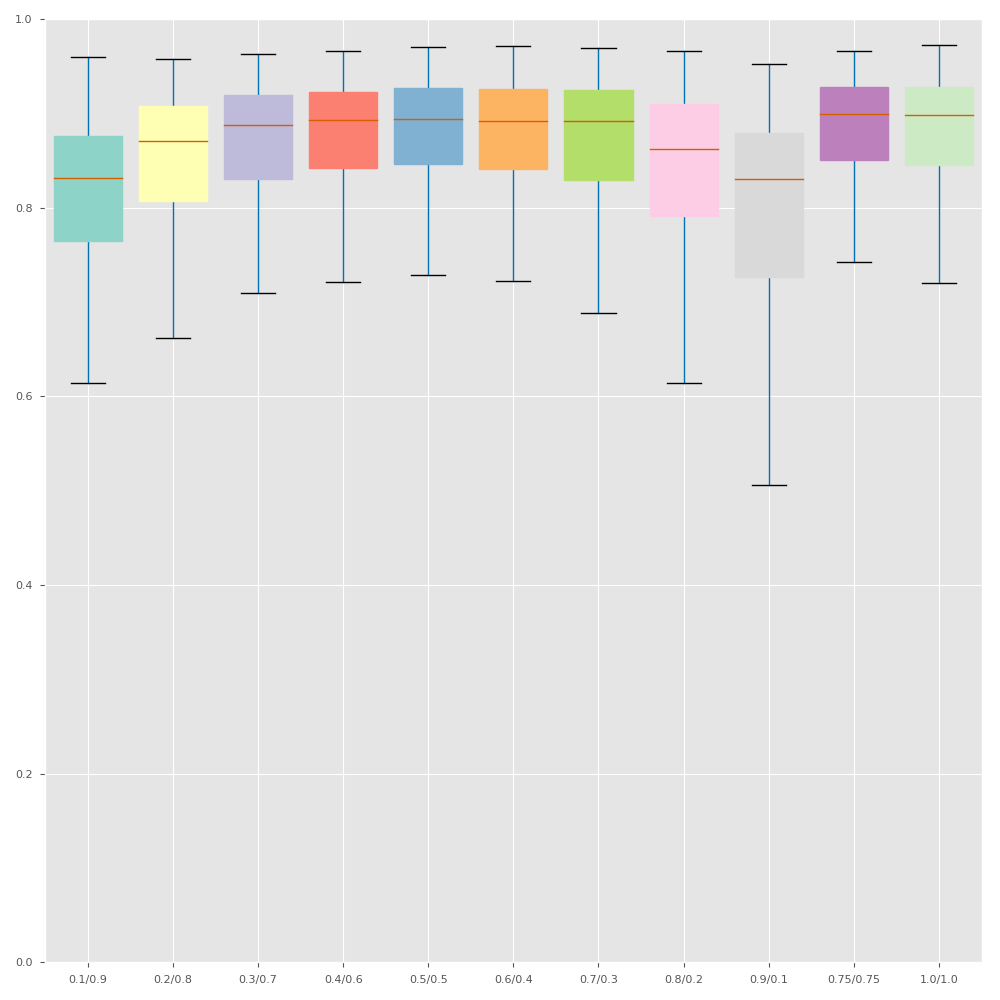} &  
        \includegraphics[width=\linewidth]{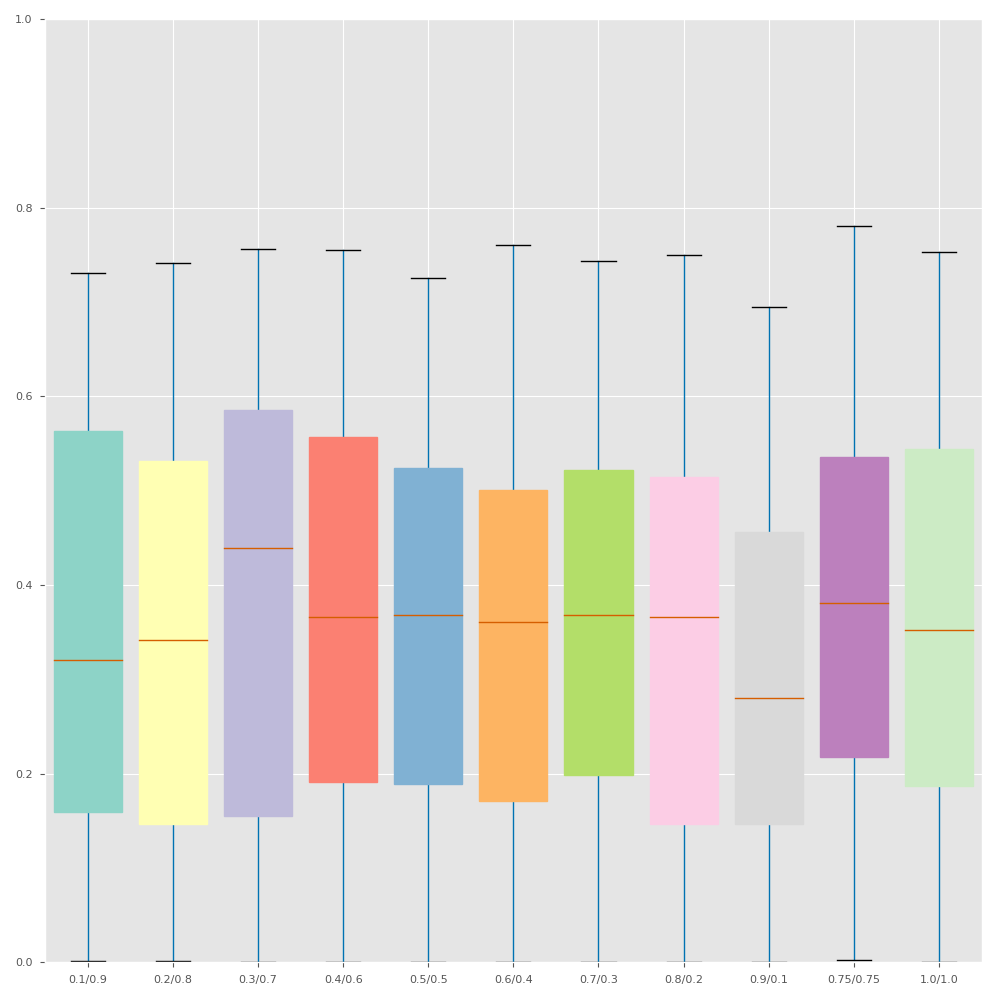} &  
        \includegraphics[width=\linewidth]{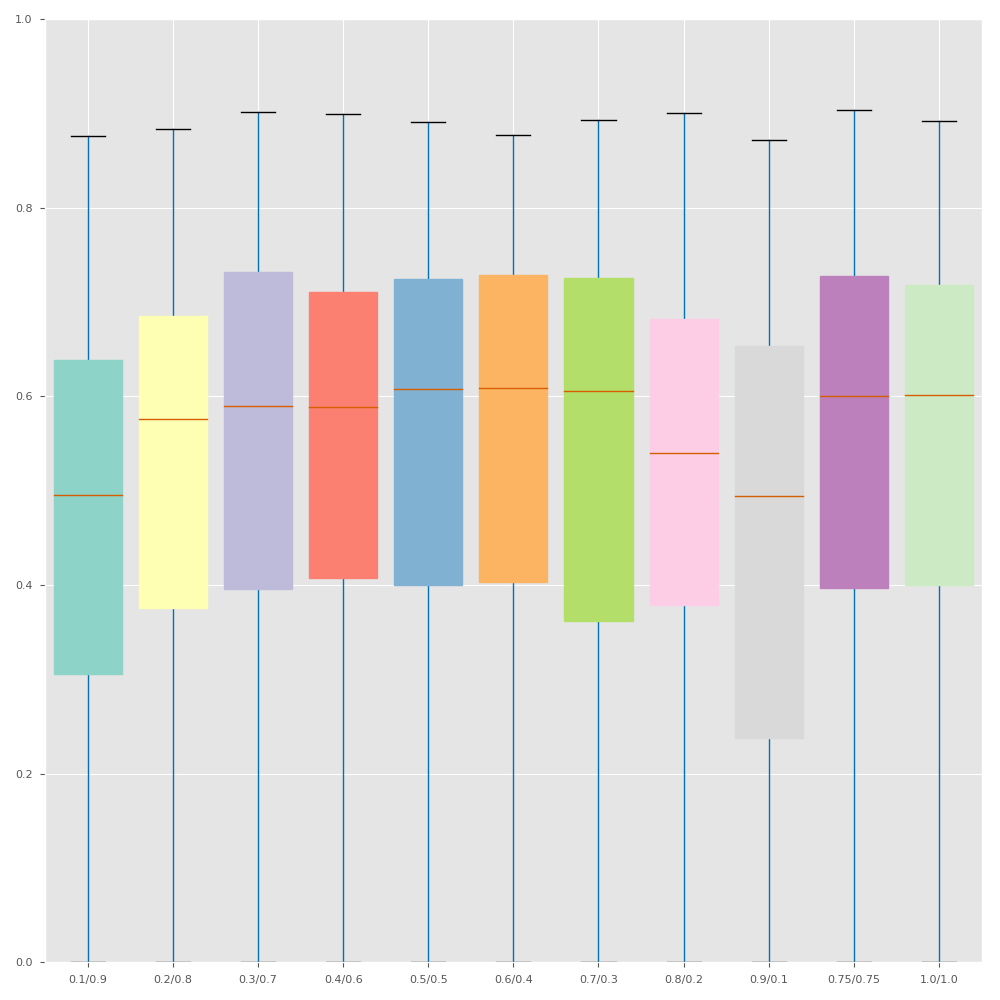} &  
        \includegraphics[width=\linewidth]{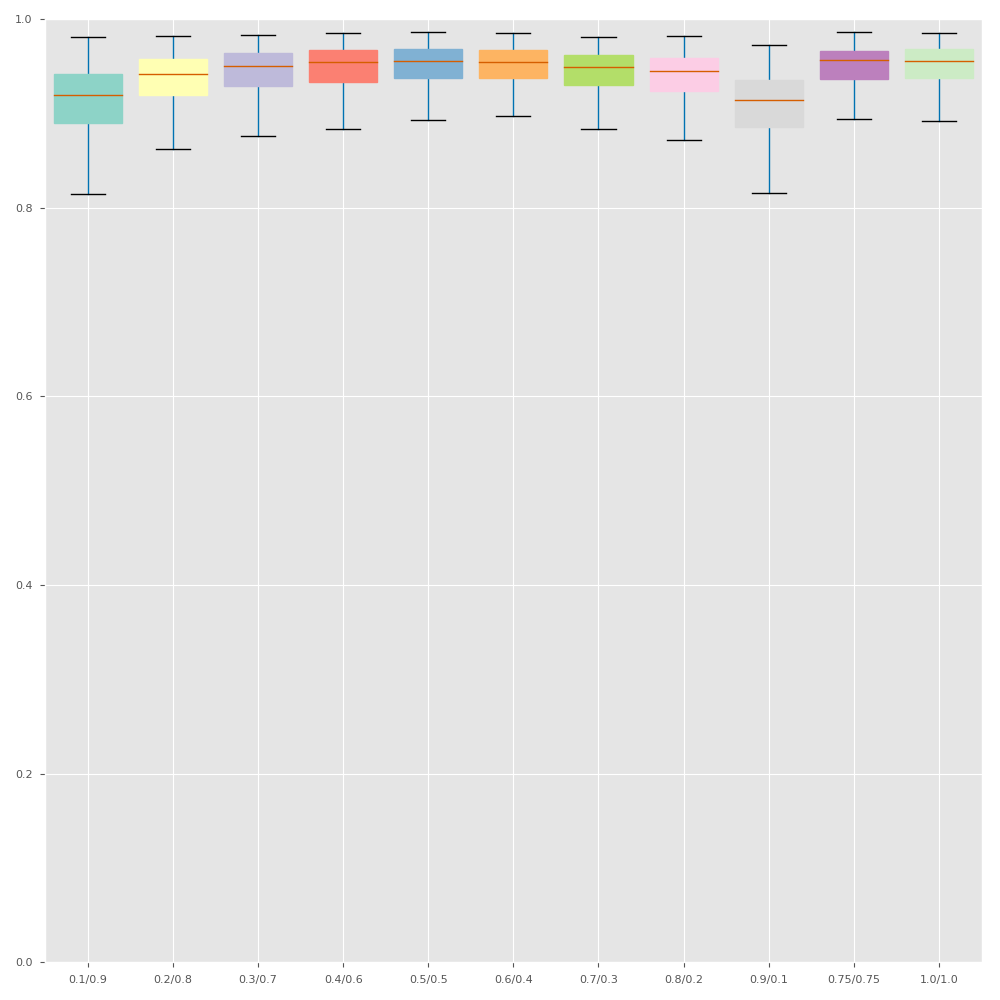} &  
        \includegraphics[width=\linewidth]{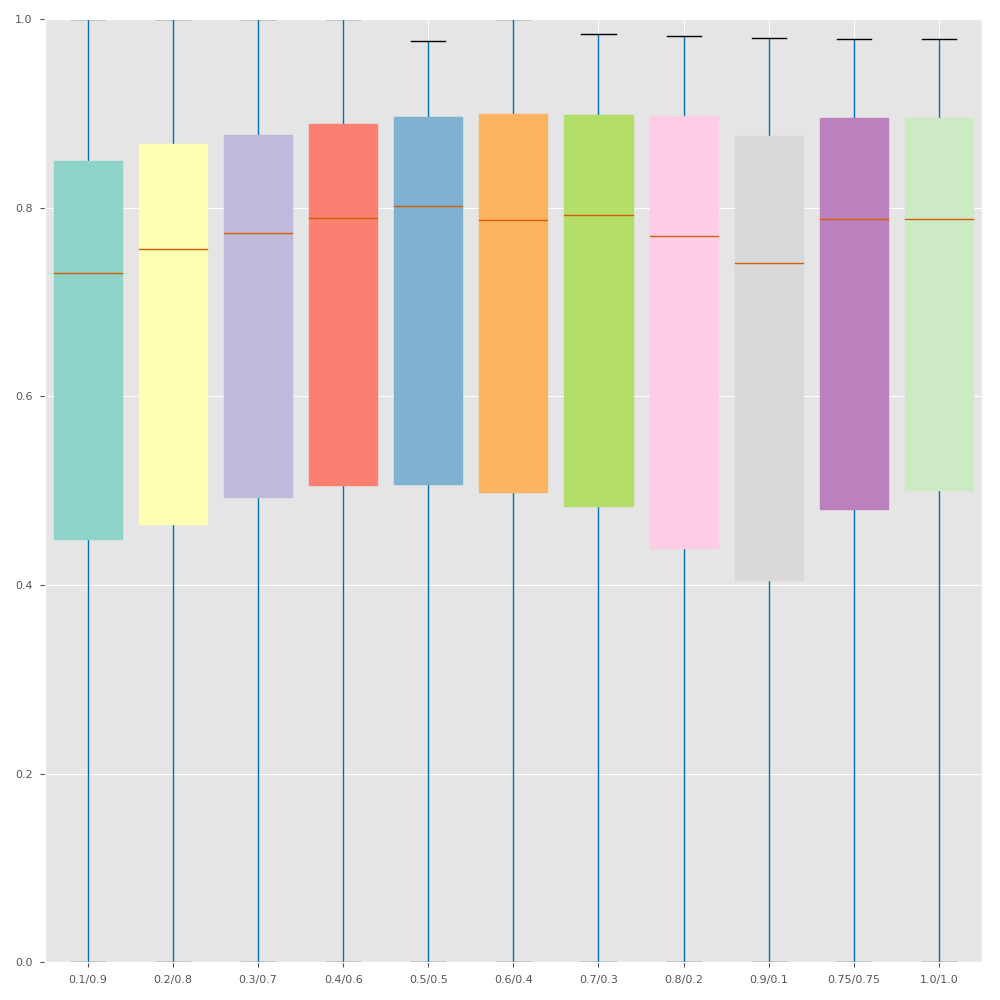} &  
        \includegraphics[width=\linewidth]{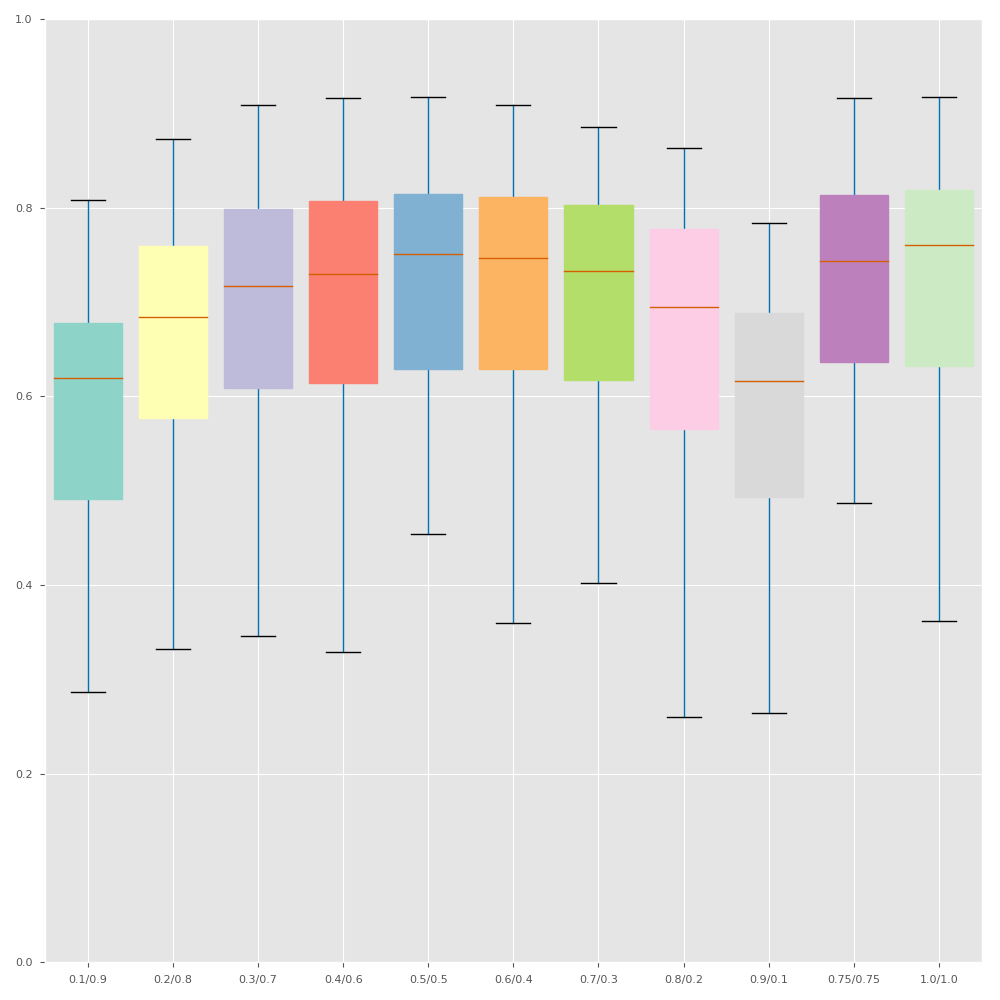} &  
        \includegraphics[width=\linewidth]{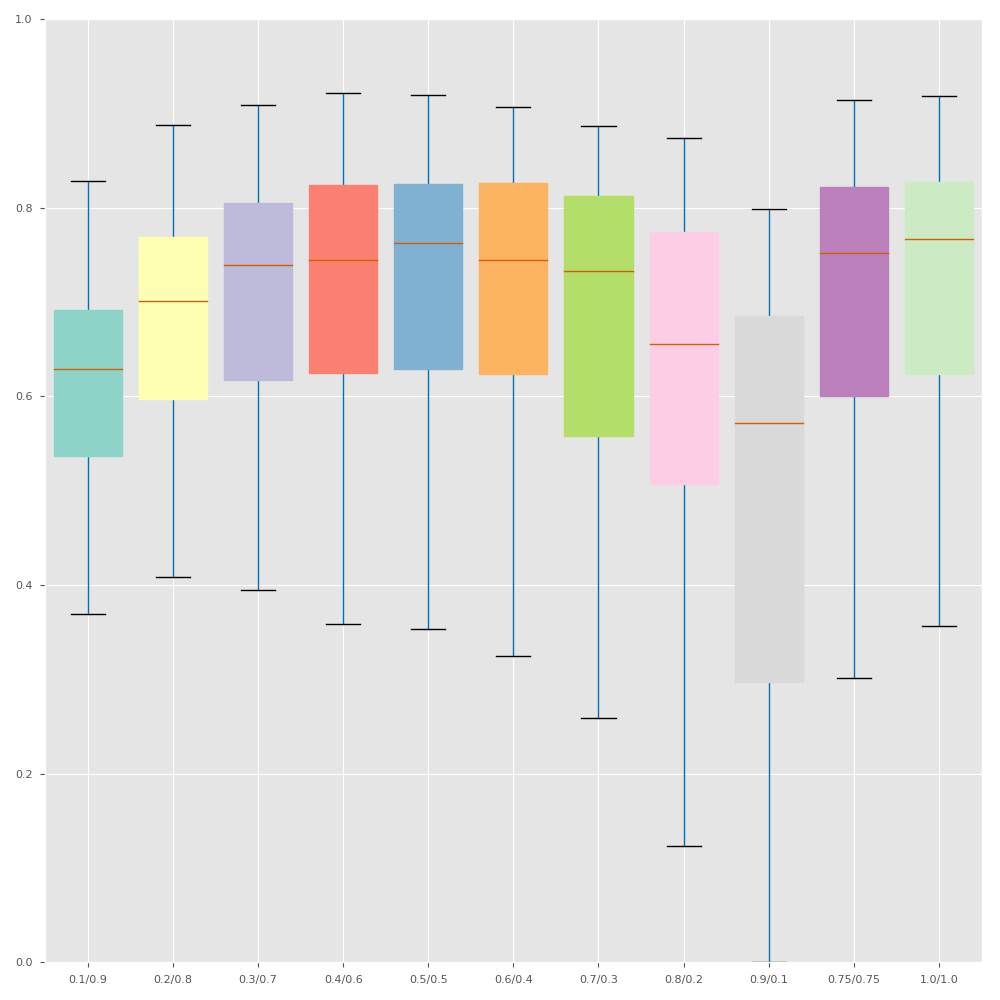}  \\ 
        
        \rotatebox{90}{\hspace{20pt} JACC} &
        \includegraphics[width=\linewidth]{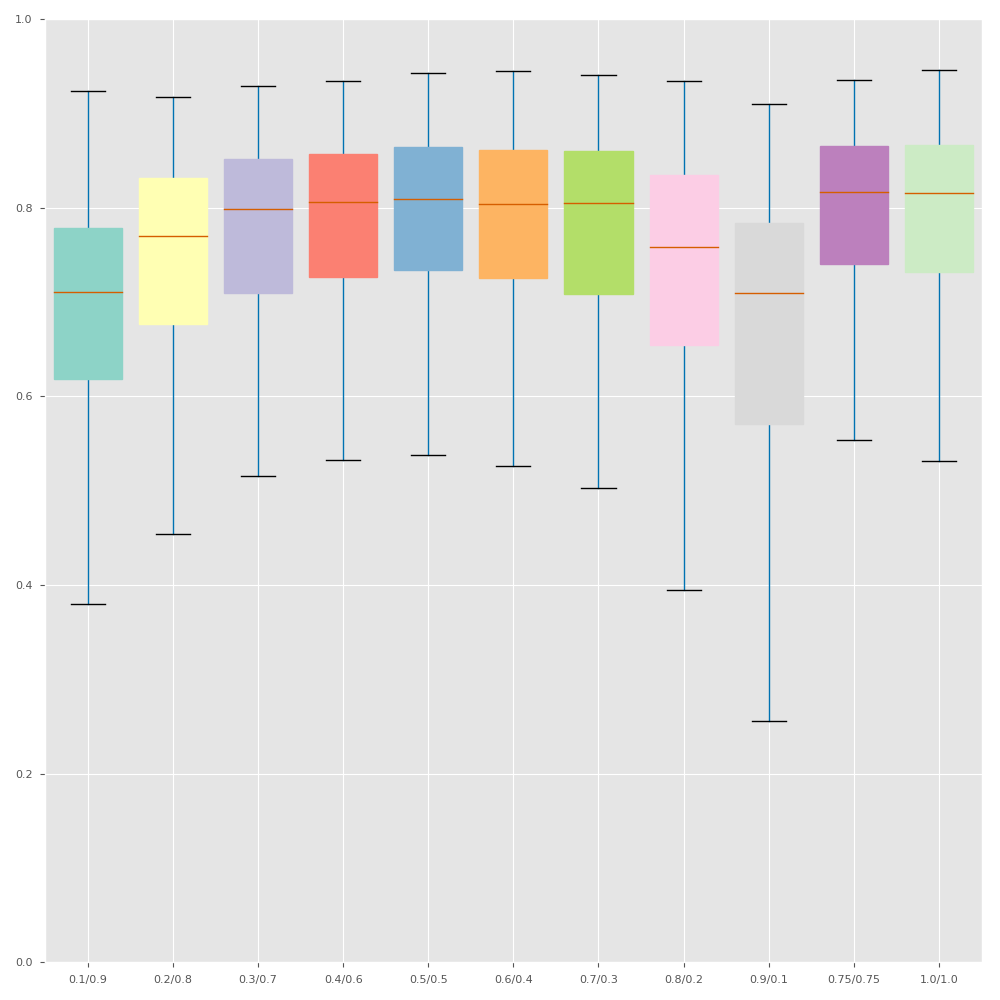} &  
        \includegraphics[width=\linewidth]{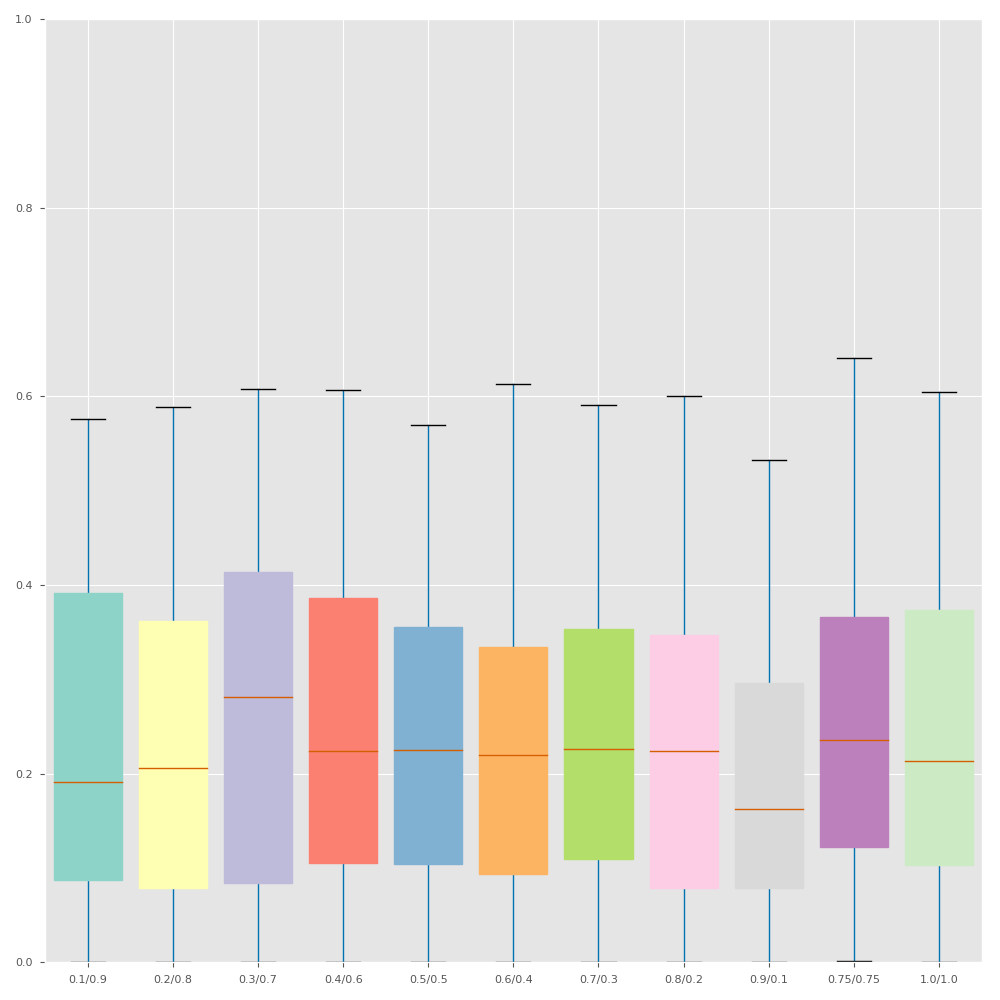} &  
        \includegraphics[width=\linewidth]{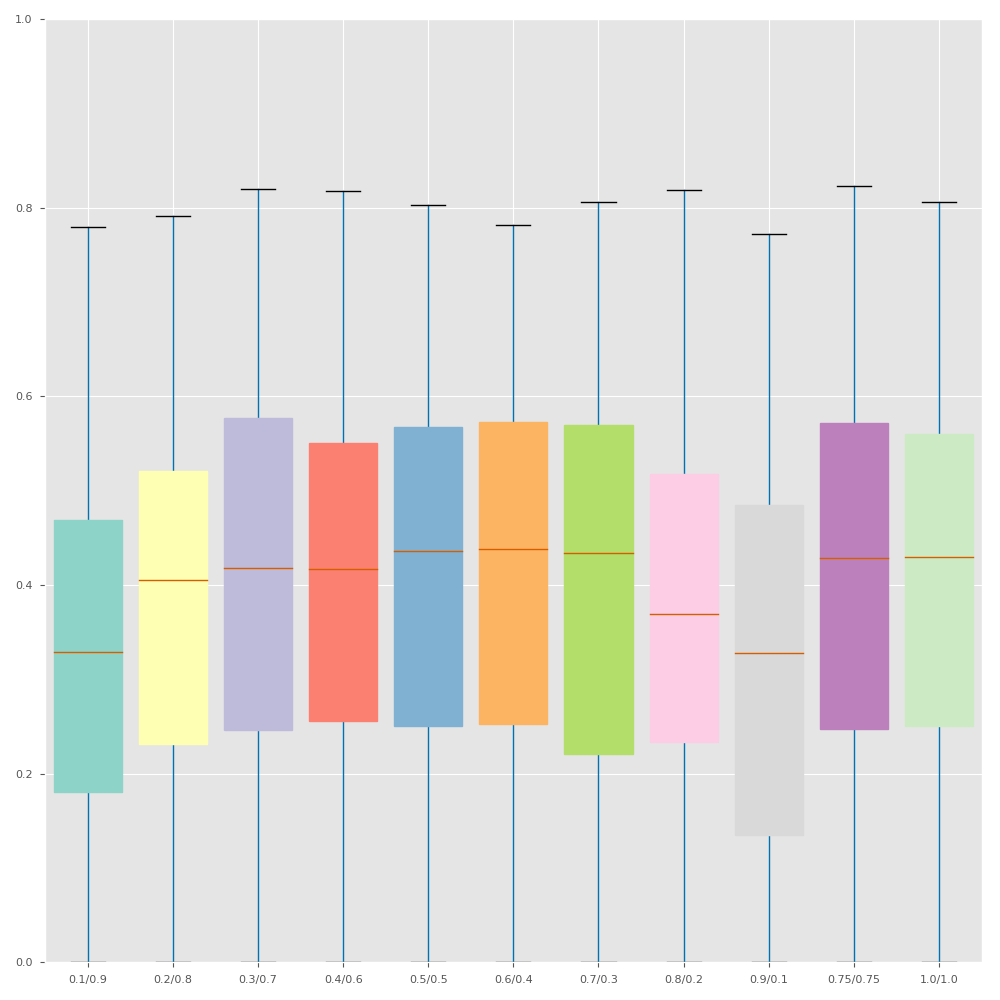} &  
        \includegraphics[width=\linewidth]{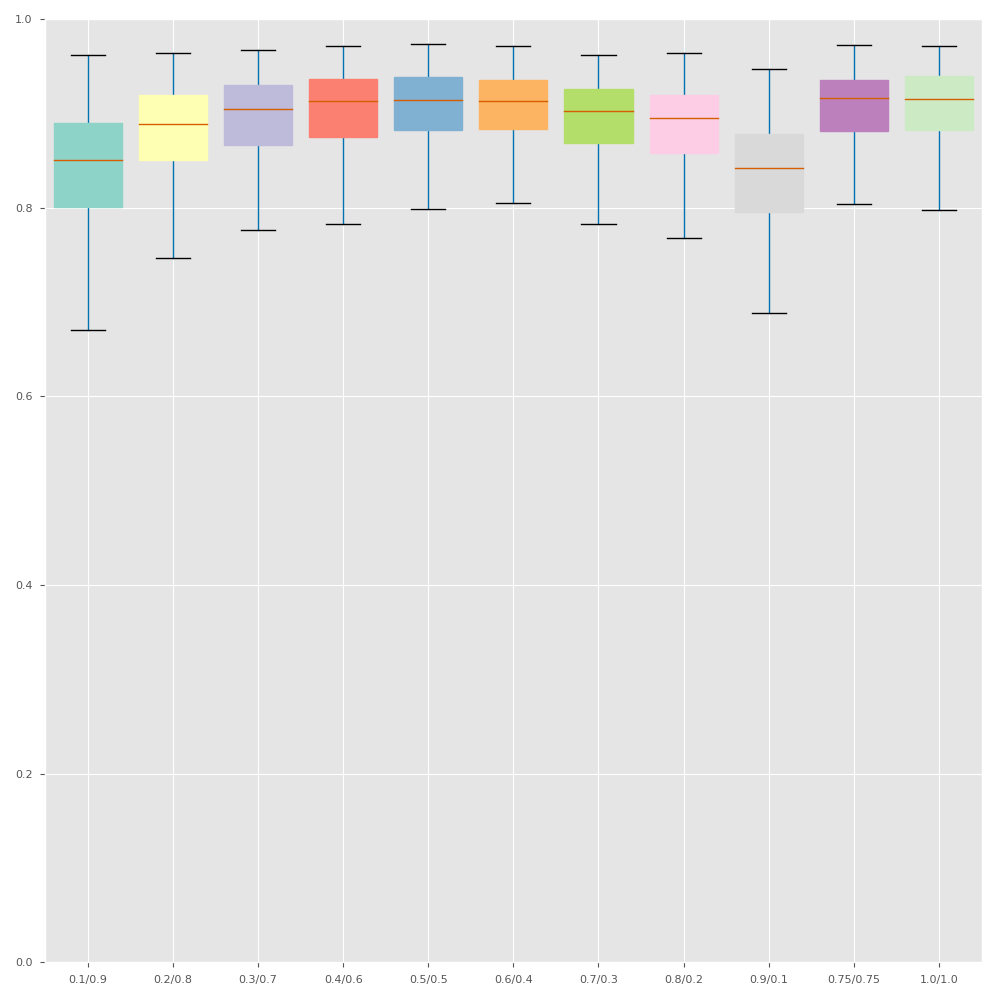} &  
        \includegraphics[width=\linewidth]{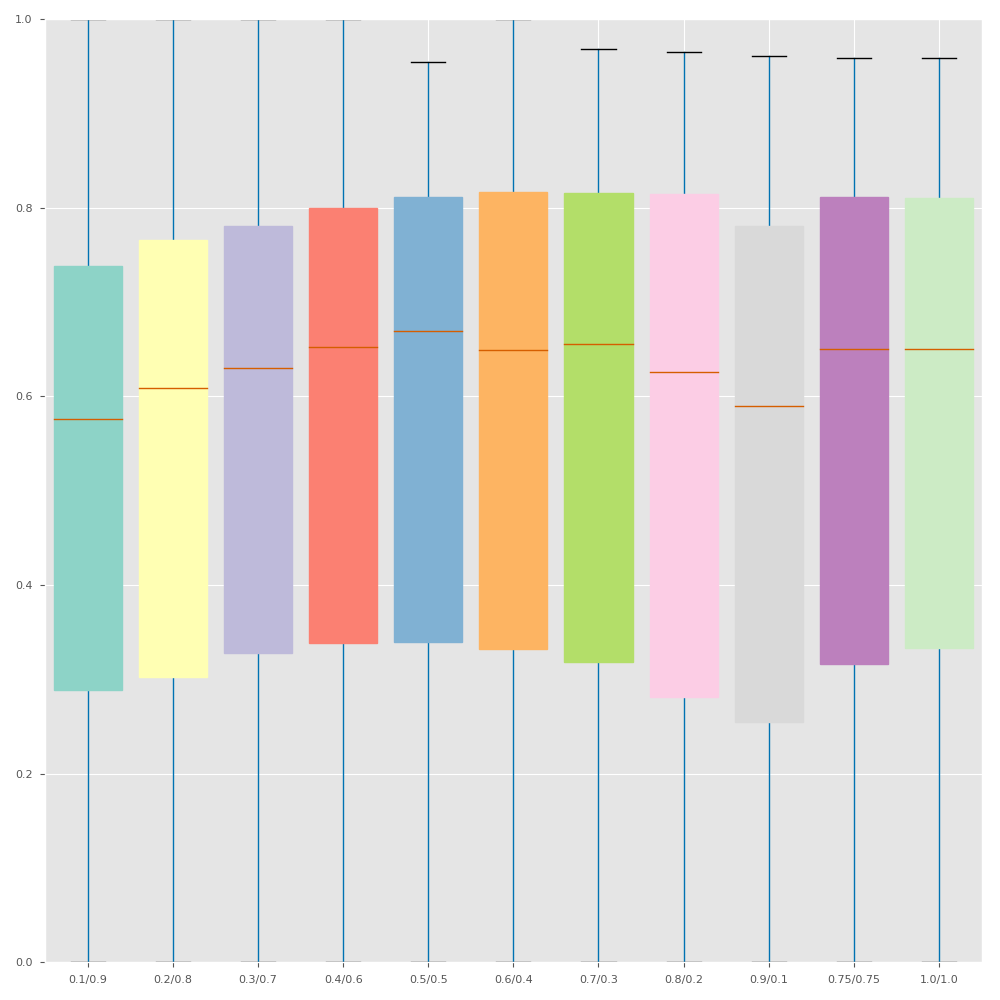} &  
        \includegraphics[width=\linewidth]{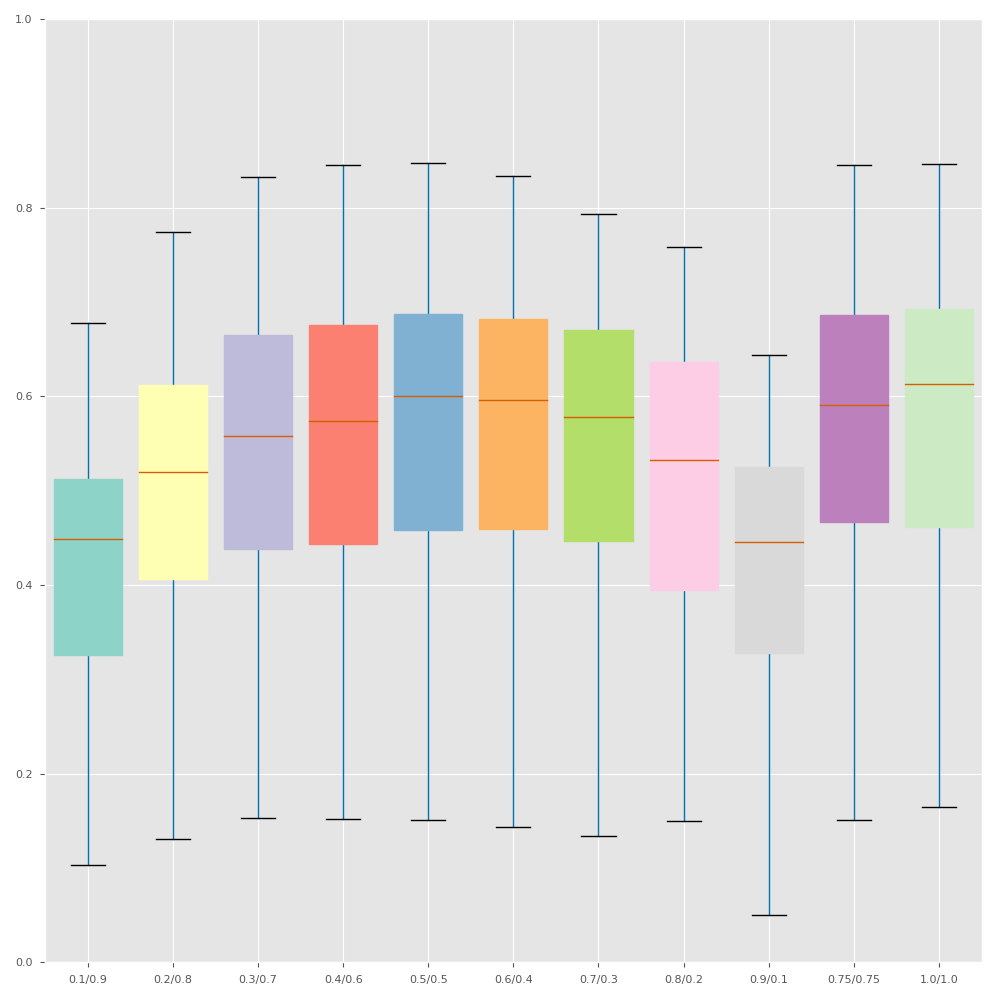} &  
        \includegraphics[width=\linewidth]{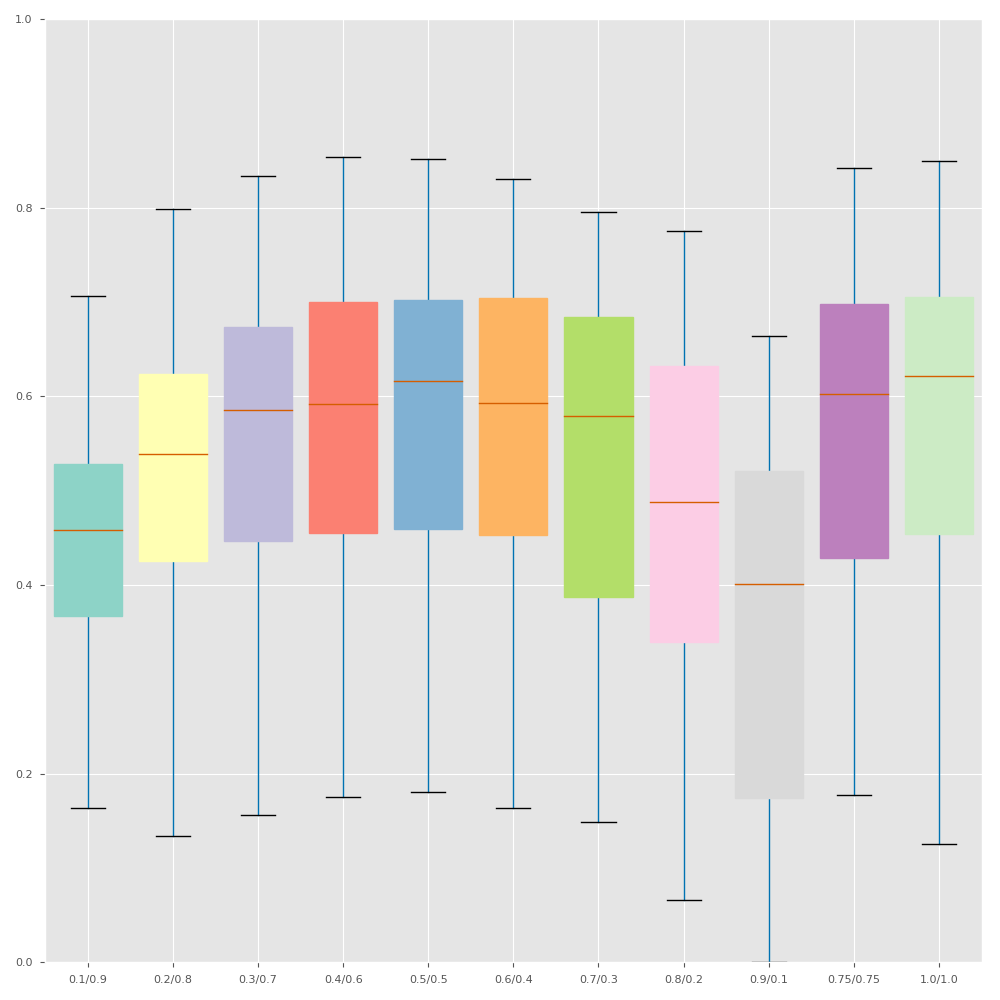}  \\ 

    \end{tabularx}
    }
        \caption{\reviewminorpar Boxplots for the Dice scores and Jaccard indexes obtained for all datasets using the range of sTversky losses with varying alpha/beta.}
    \label{fig:boxplots_main_tversky}
\end{figure*}

\begin{figure*}[!htbp]
    \centering
    \resizebox*{\linewidth}{!}{
    
    \def\arraystretch{0}
    \setlength{\tabcolsep}{0pt}
    \begin{tabularx}{\linewidth}{l @{\hspace{2pt}} YYYYYYY}

        & BR18 &   IS17 &  IS18 & MO18 & PO18 & WM17 & WM17\textsuperscript{DM}  \\
        
        \rotatebox{90}{\hspace{20pt} F0.5} &
        \includegraphics[width=\linewidth]{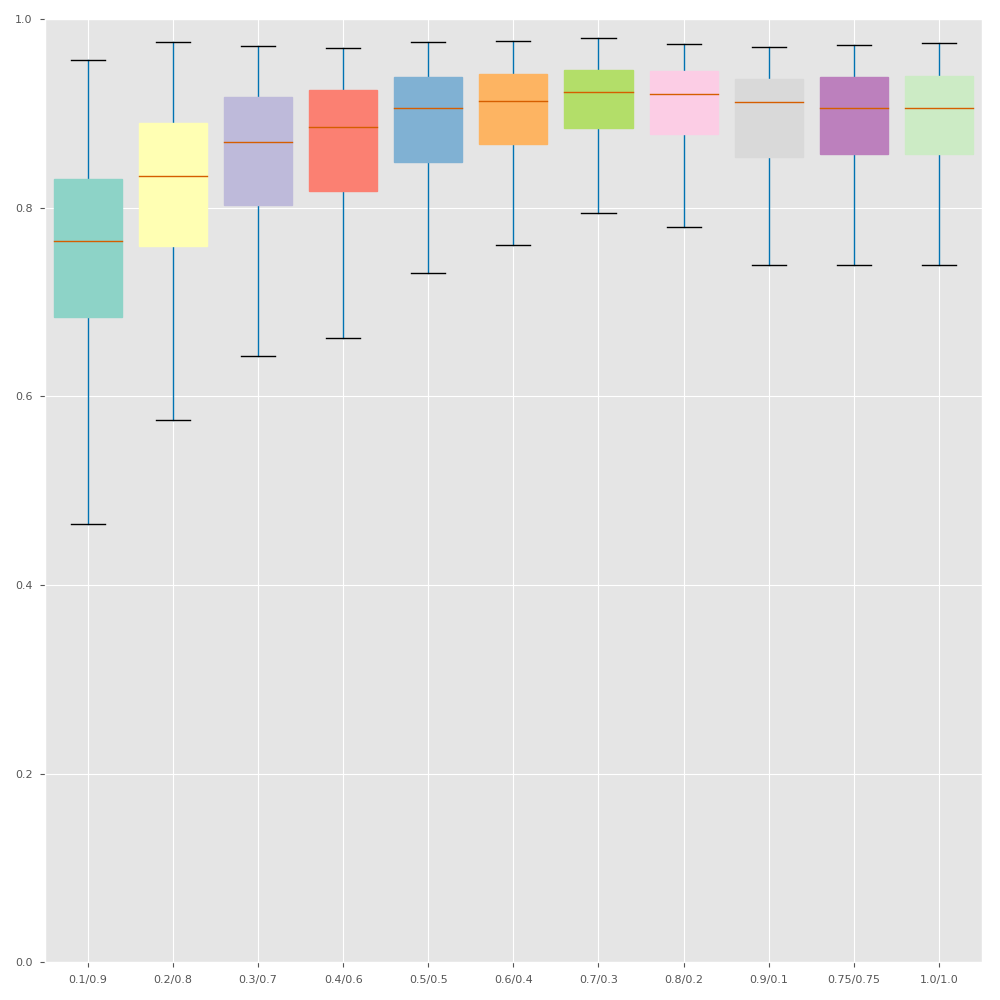} &  
        \includegraphics[width=\linewidth]{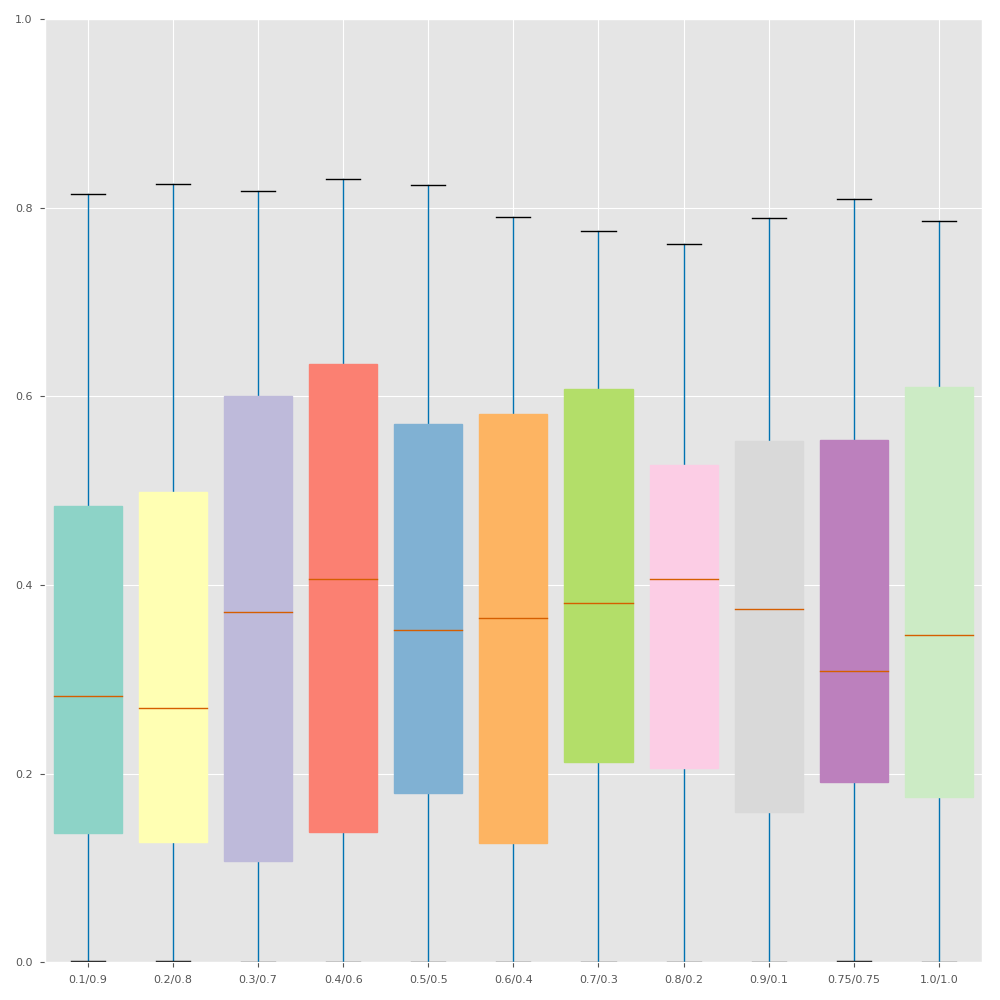} &  
        \includegraphics[width=\linewidth]{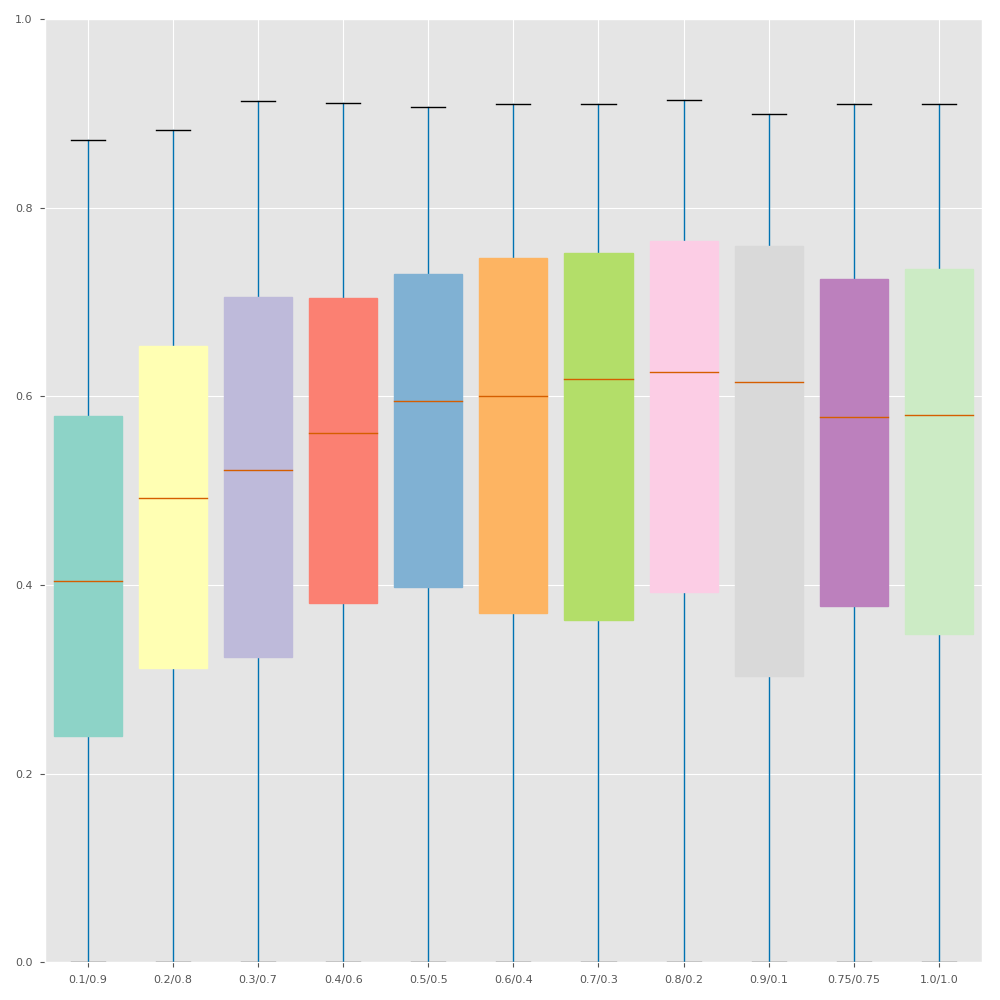} &  
        \includegraphics[width=\linewidth]{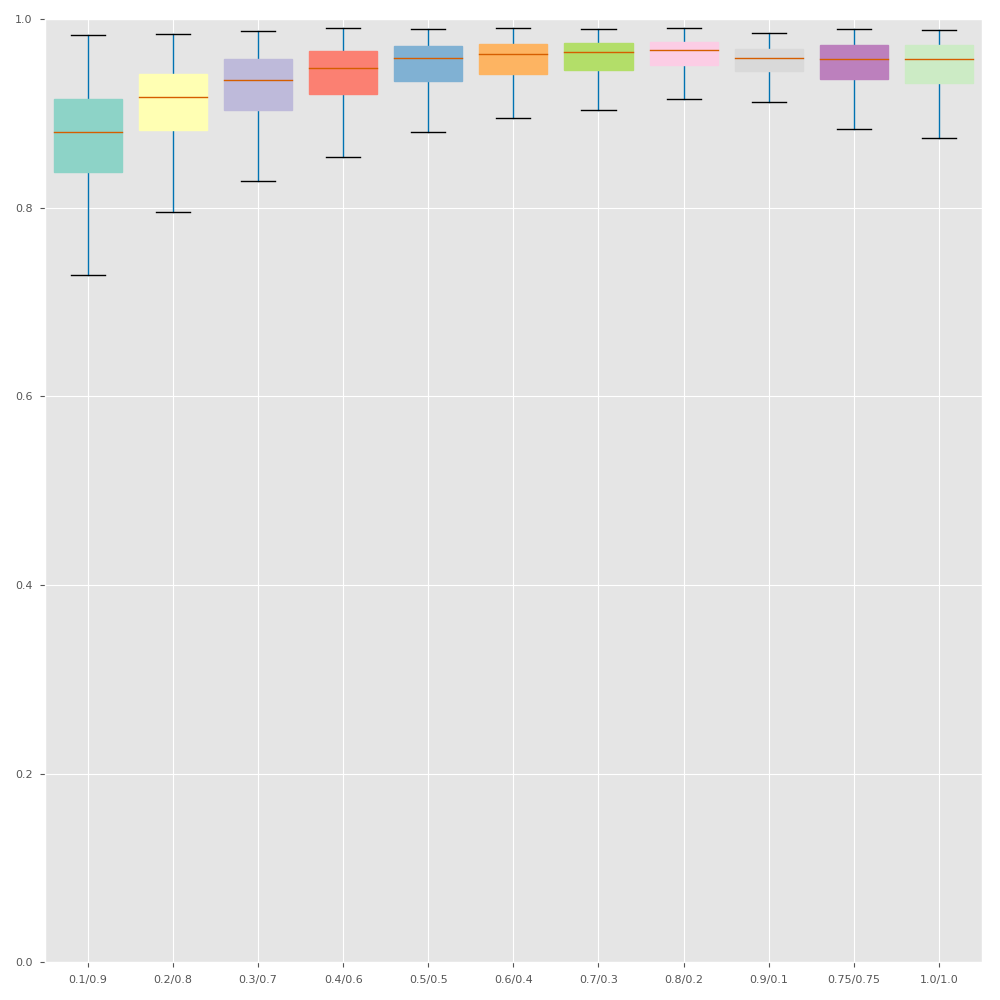} &  
        \includegraphics[width=\linewidth]{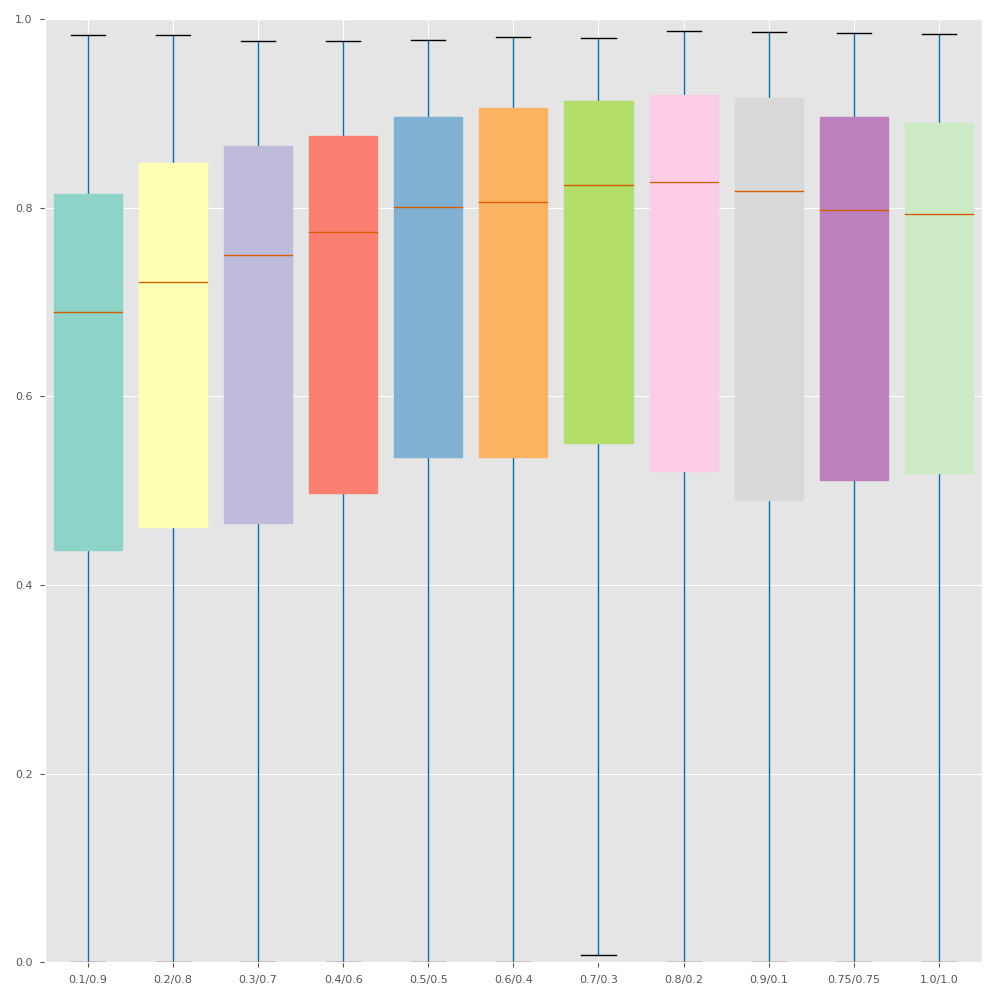} &  
        \includegraphics[width=\linewidth]{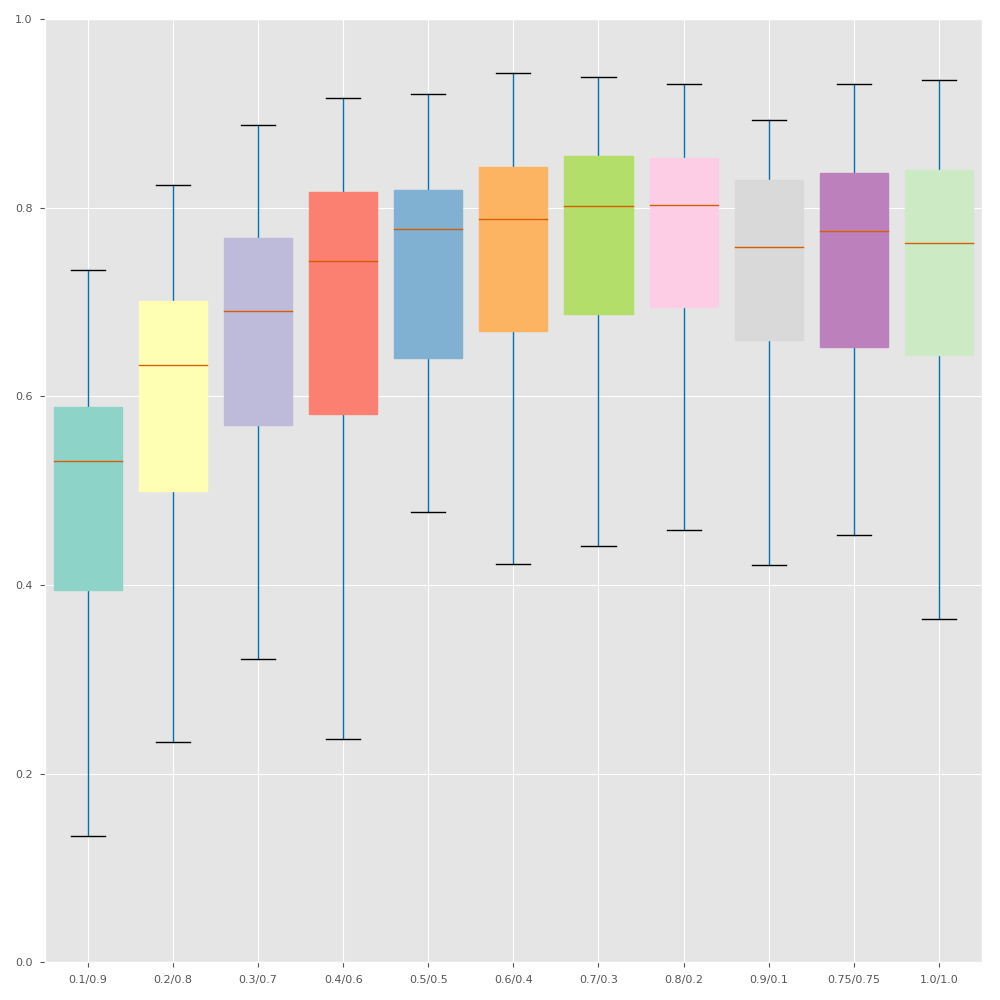} &  
        \includegraphics[width=\linewidth]{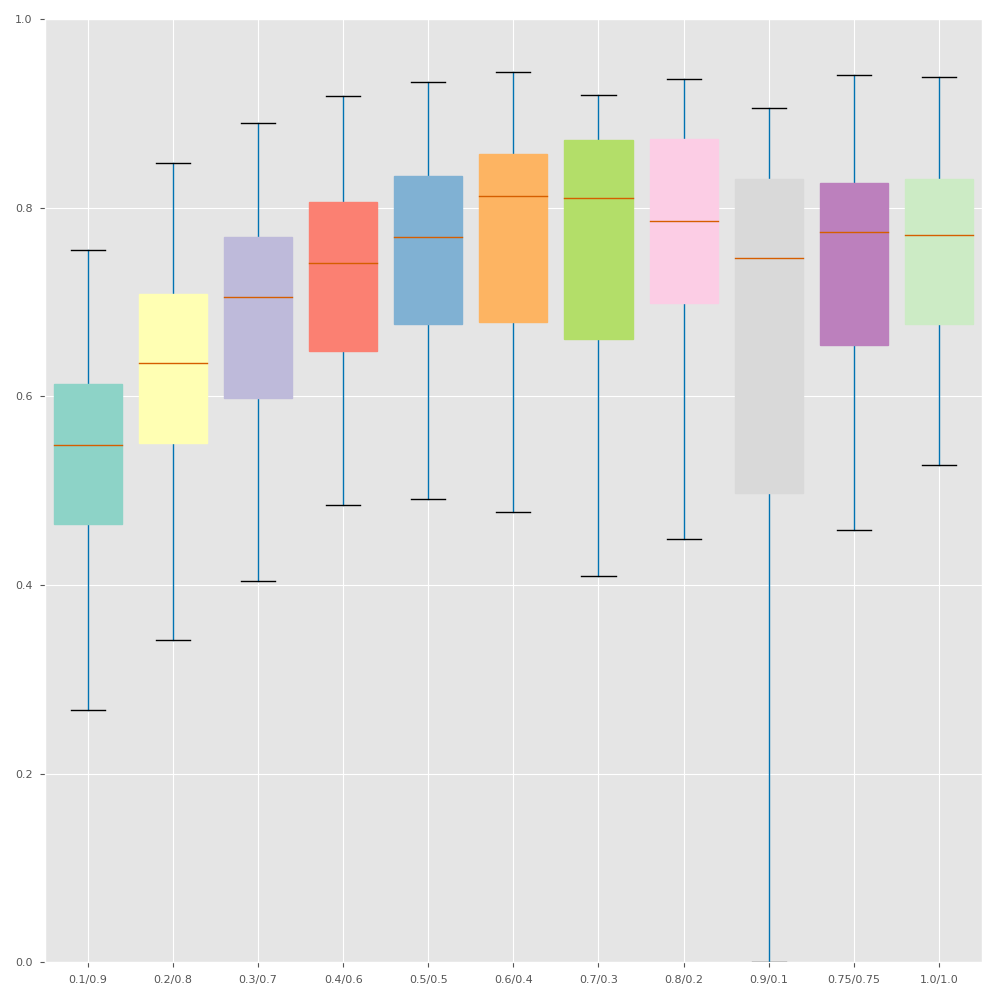}  \\ 
        
        \rotatebox{90}{\hspace{20pt} F1.0} &
        \includegraphics[width=\linewidth]{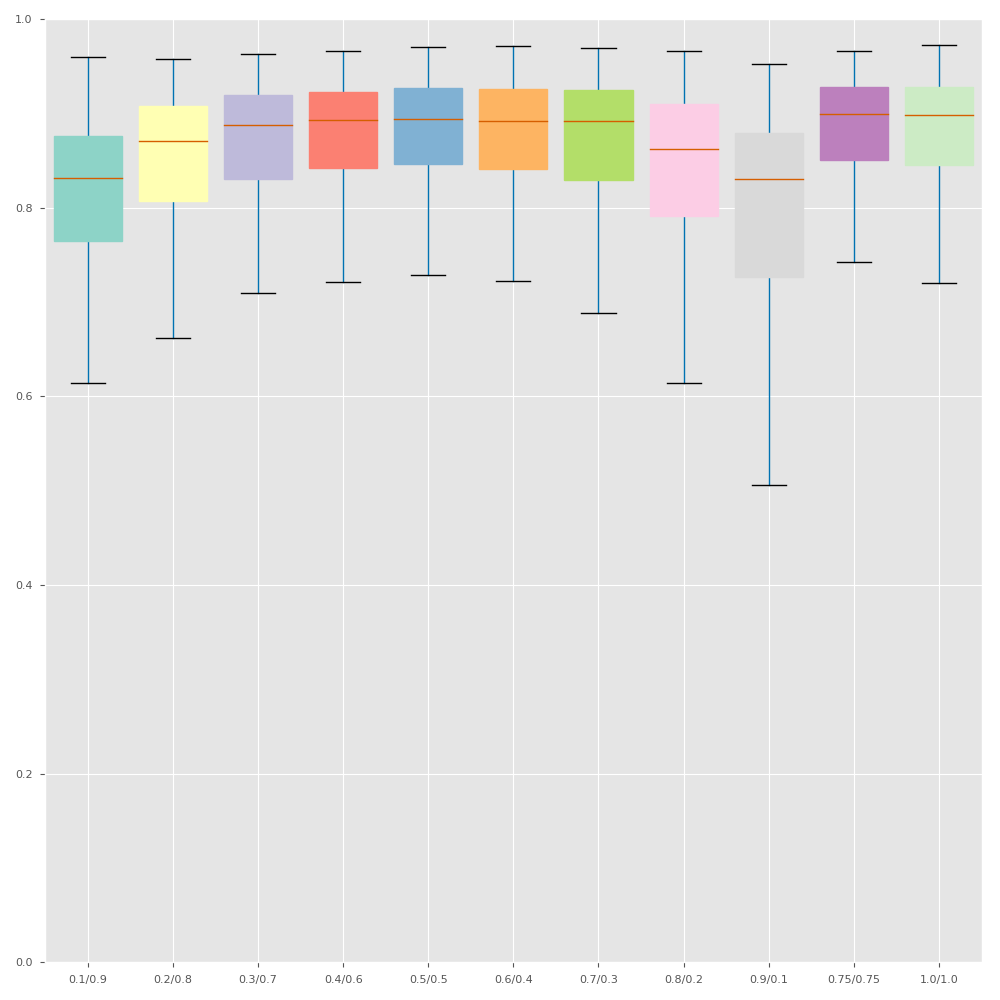} &  
        \includegraphics[width=\linewidth]{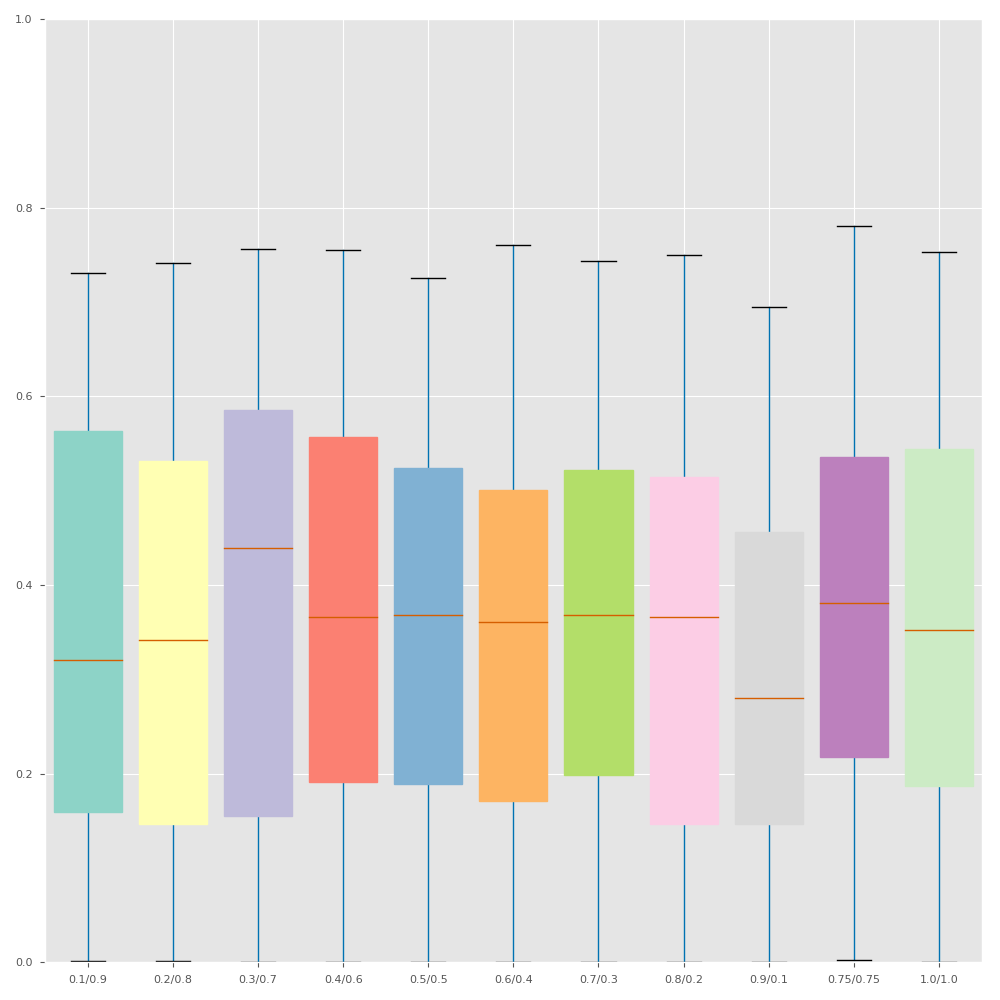} &  
        \includegraphics[width=\linewidth]{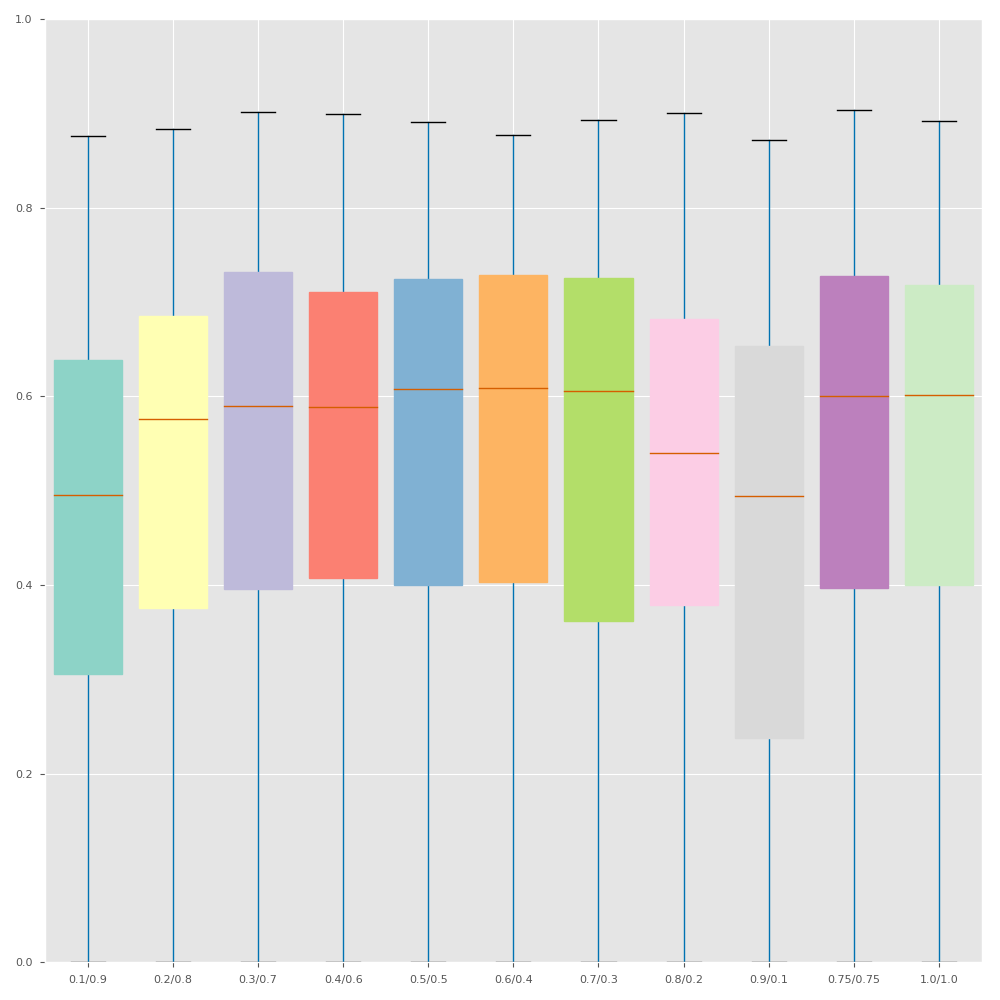} &  
        \includegraphics[width=\linewidth]{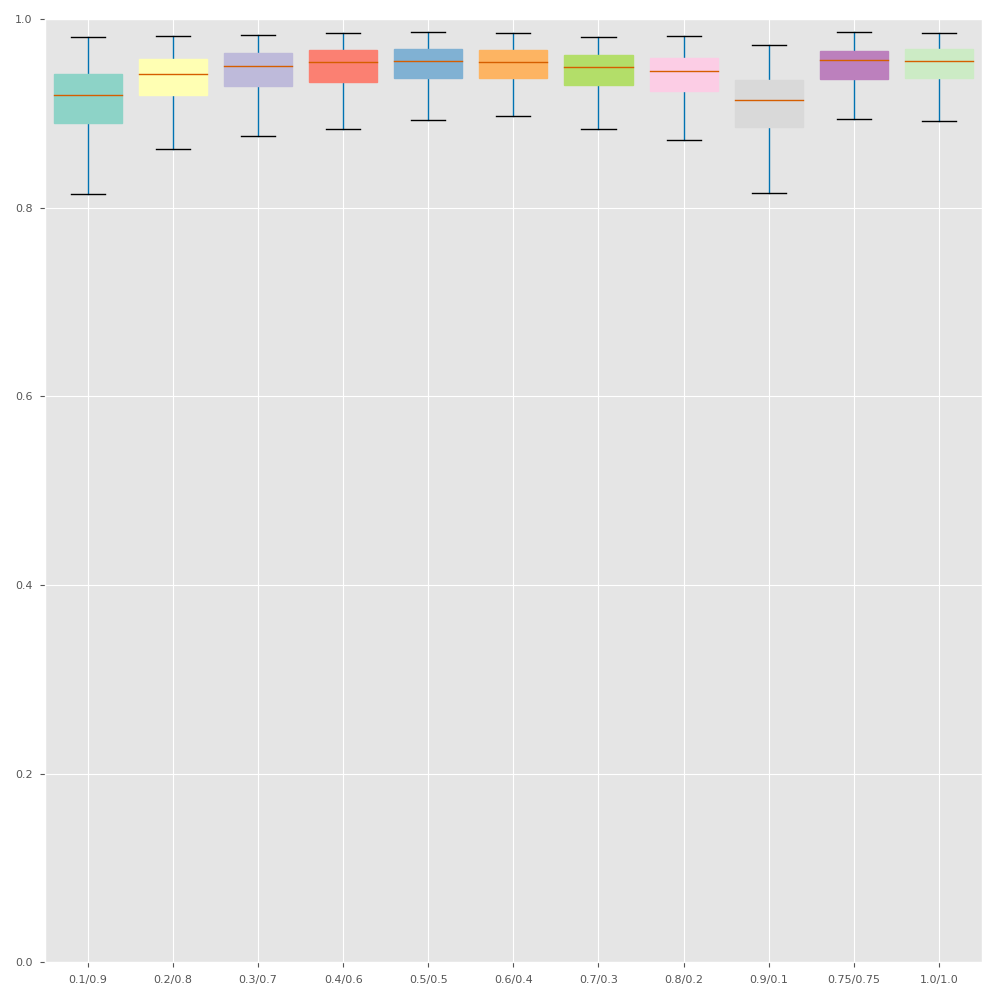} &  
        \includegraphics[width=\linewidth]{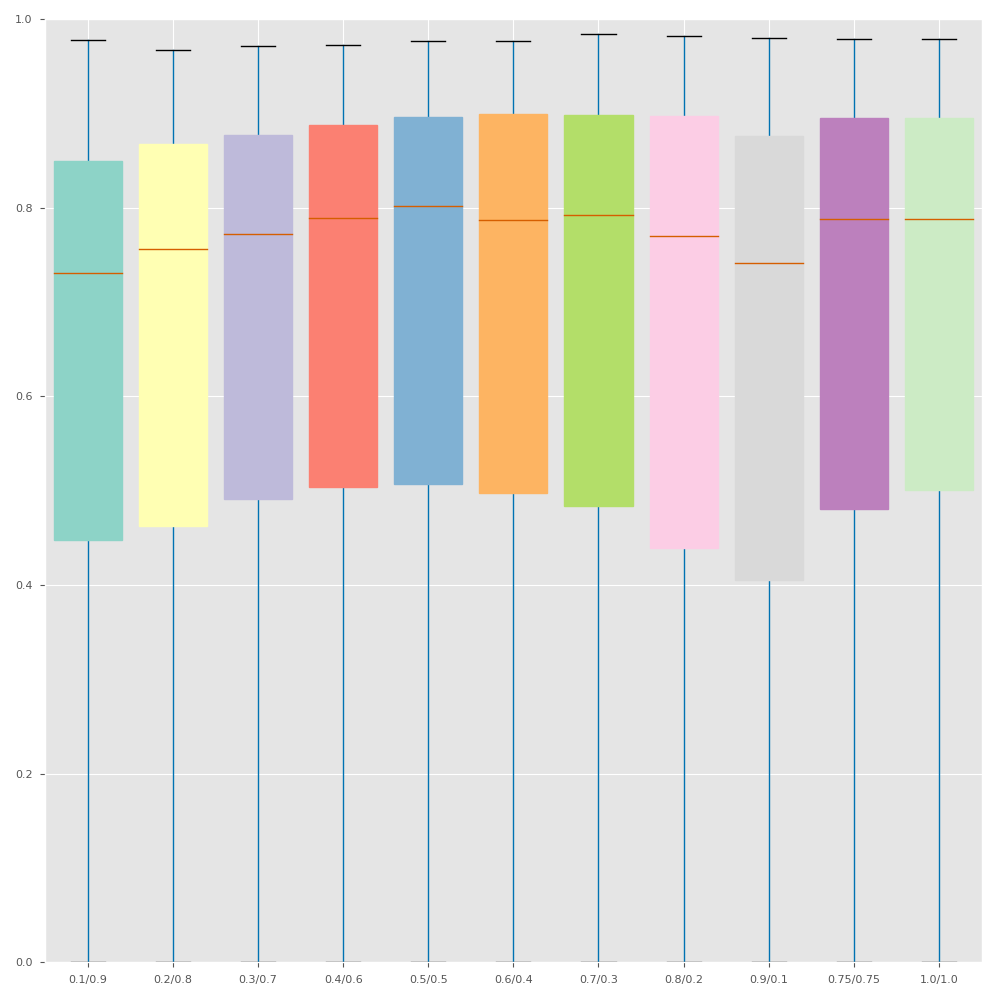} &  
        \includegraphics[width=\linewidth]{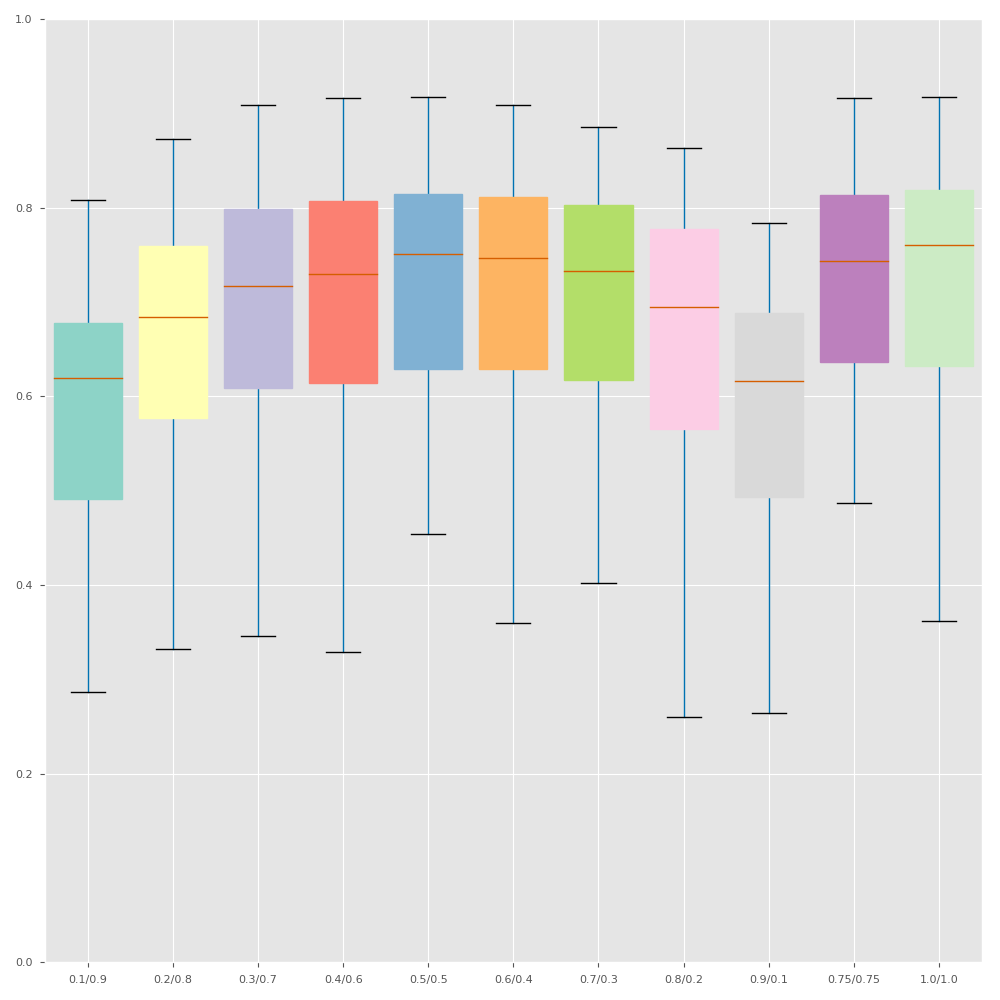} &  
        \includegraphics[width=\linewidth]{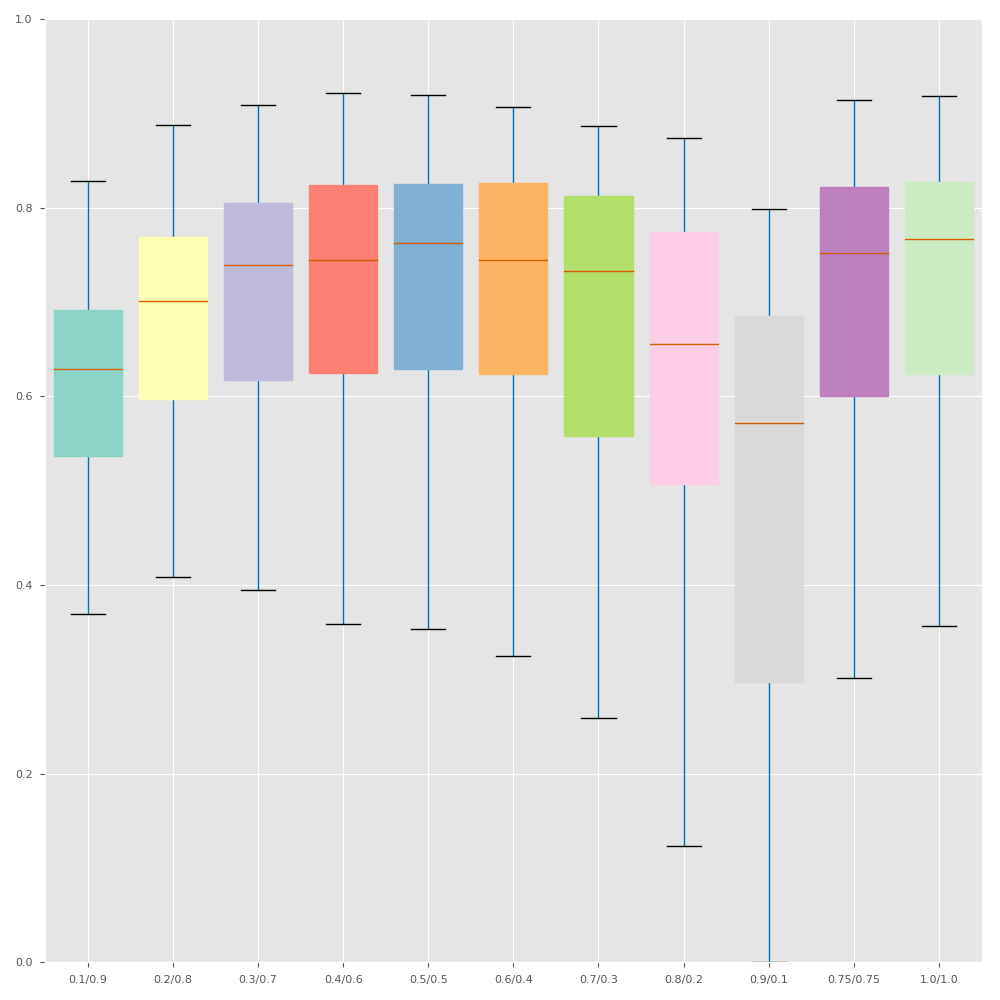}  \\ 
    
        \rotatebox{90}{\hspace{20pt} F1.5} &
        \includegraphics[width=\linewidth]{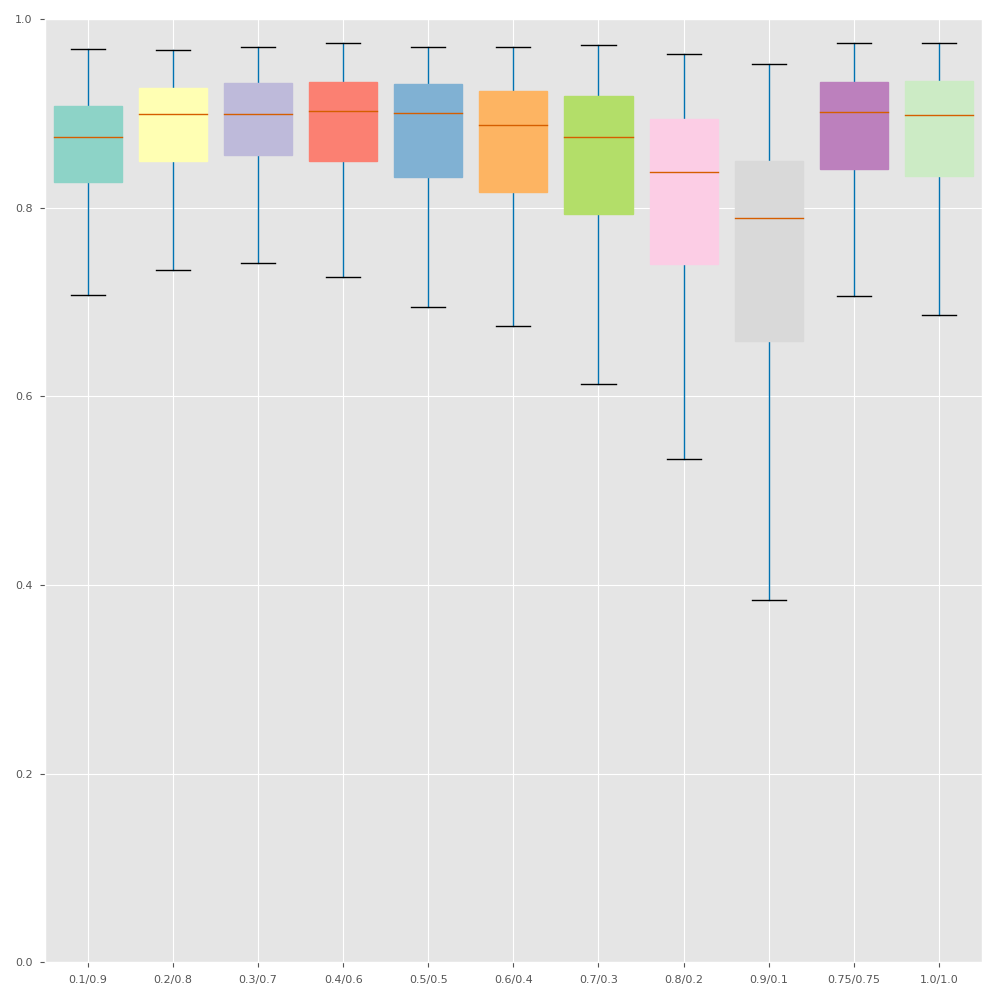} &  
        \includegraphics[width=\linewidth]{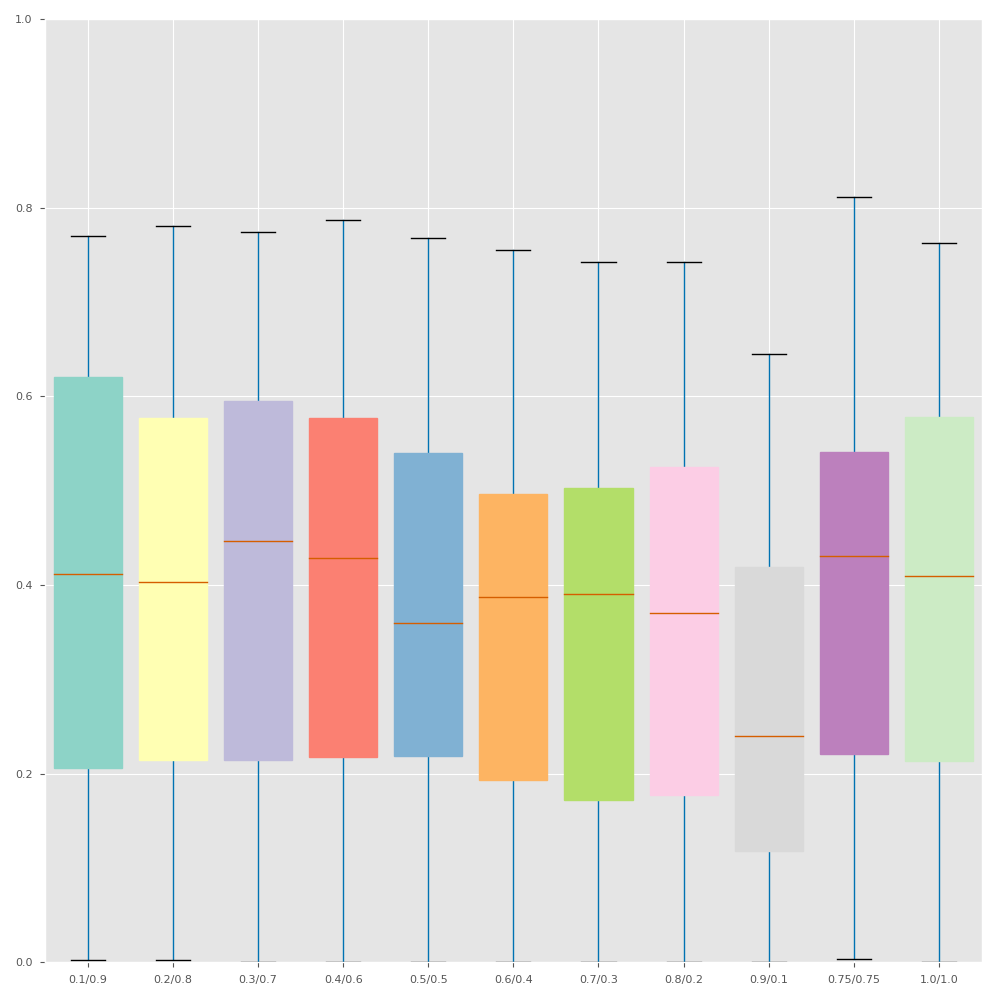} &  
        \includegraphics[width=\linewidth]{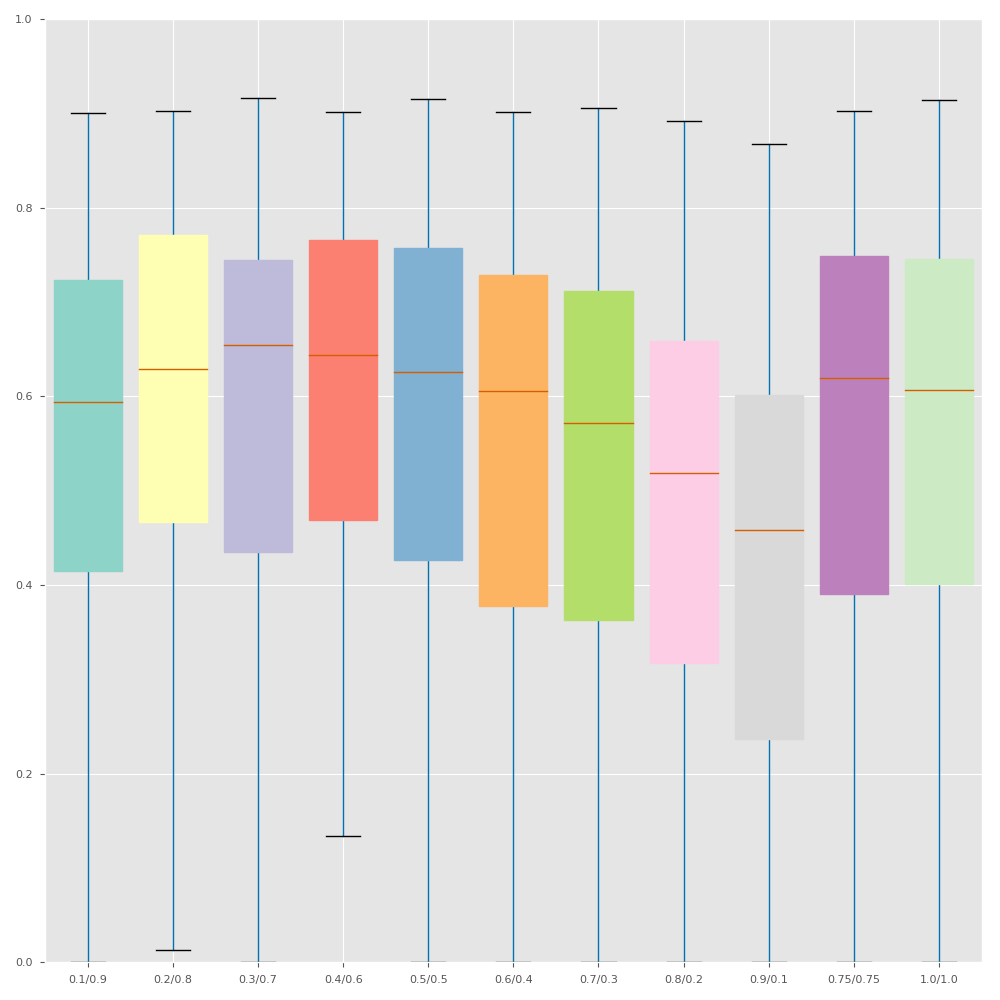} &  
        \includegraphics[width=\linewidth]{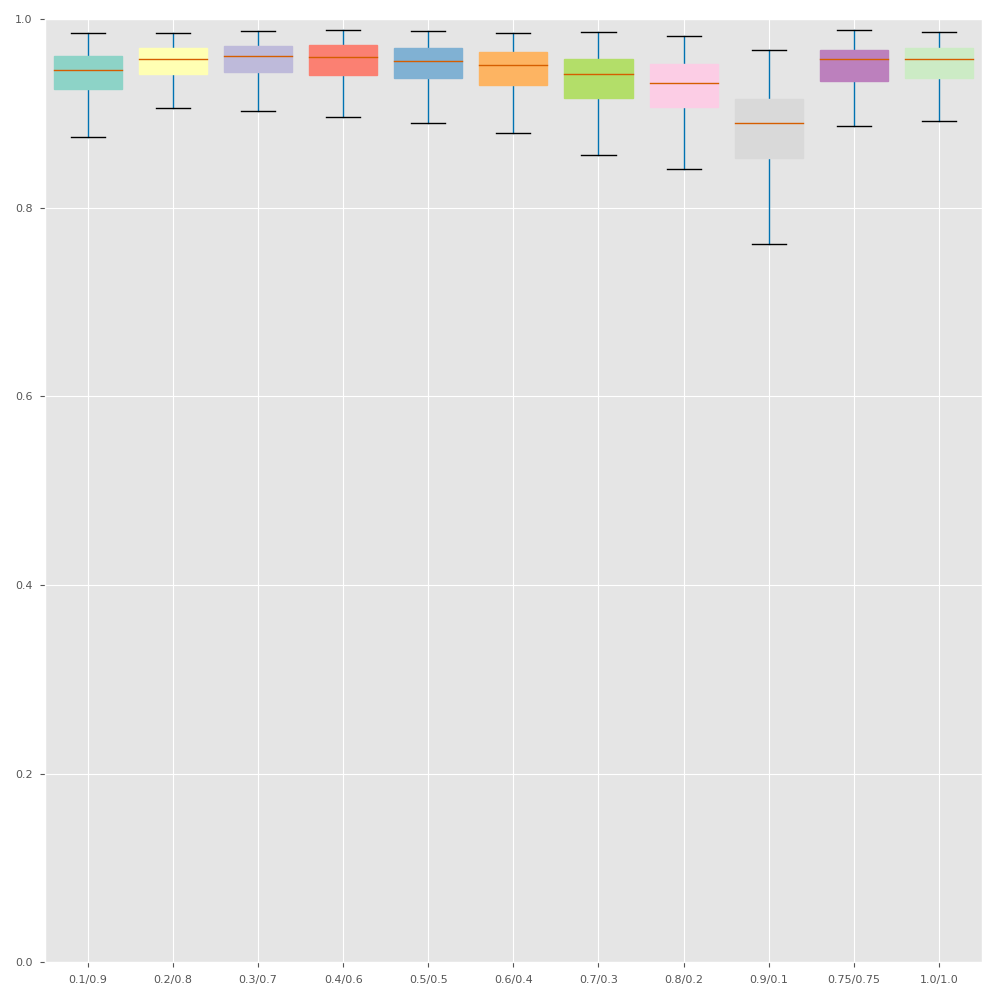} &  
        \includegraphics[width=\linewidth]{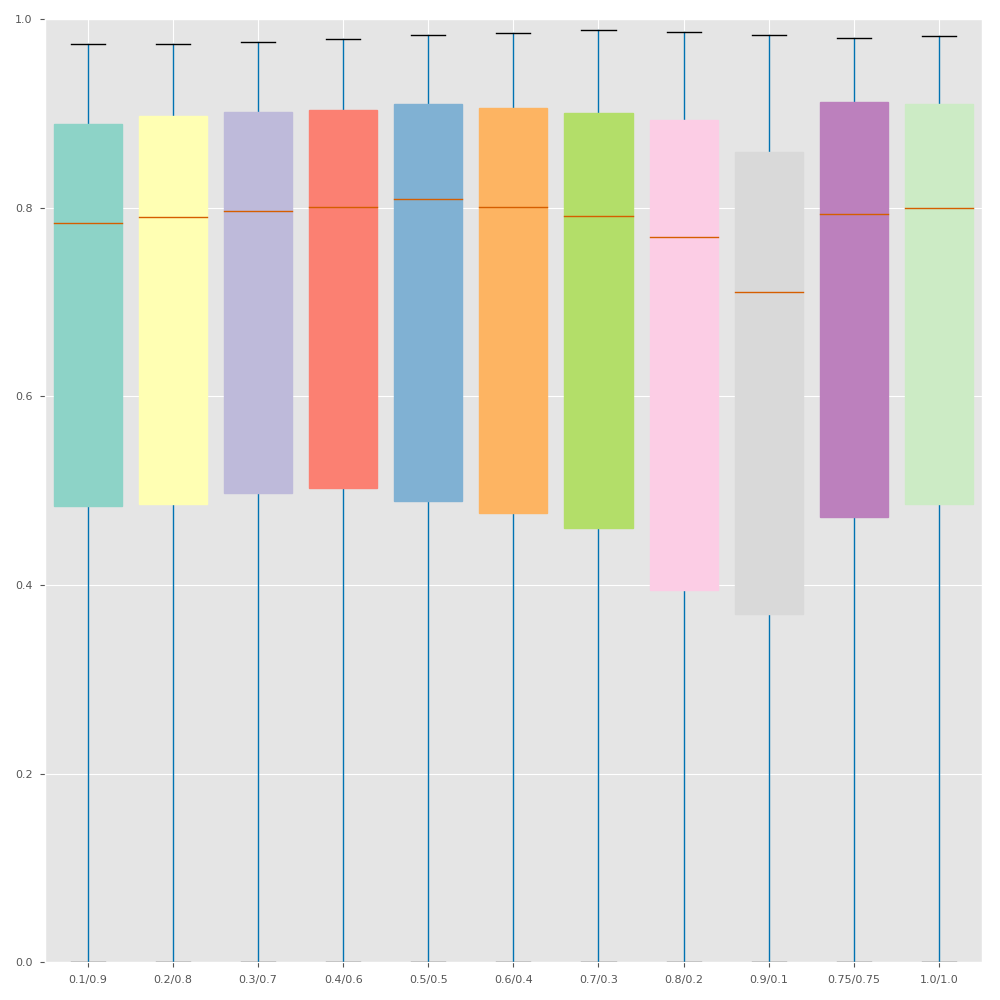} &  
        \includegraphics[width=\linewidth]{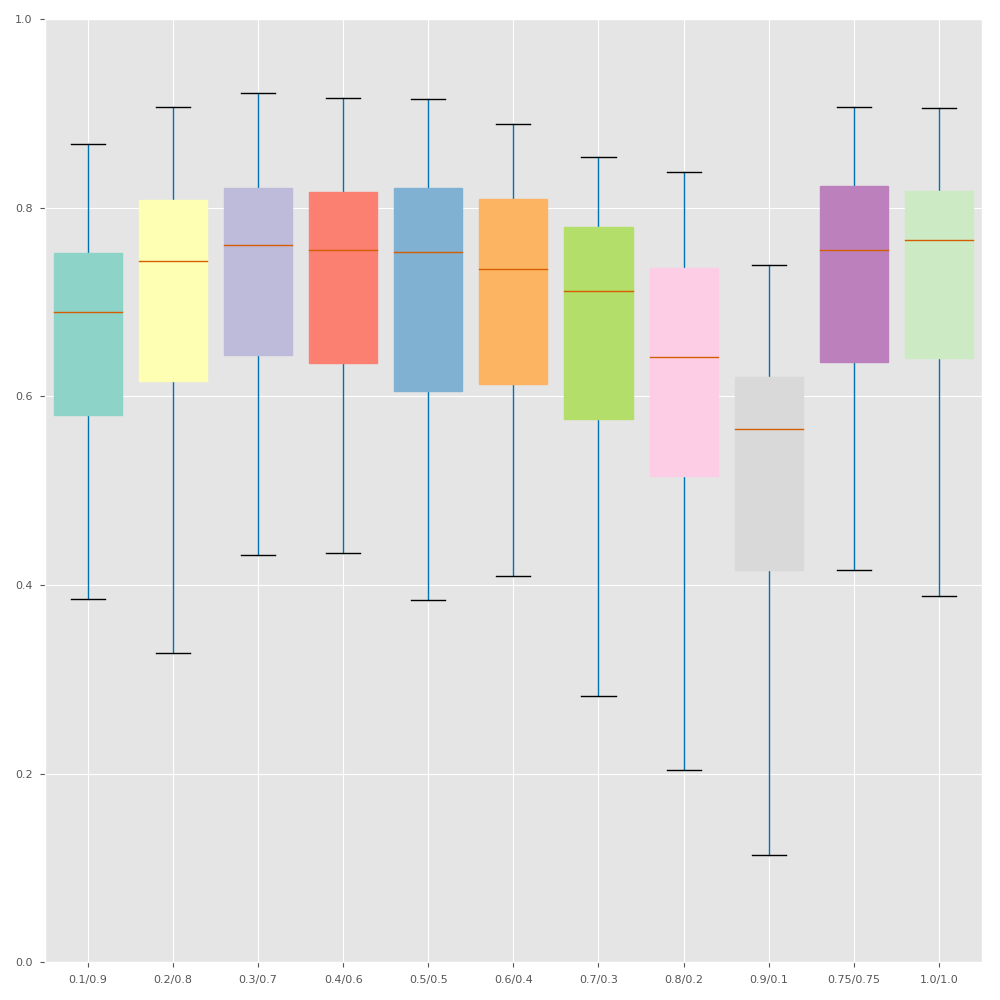} &  
        \includegraphics[width=\linewidth]{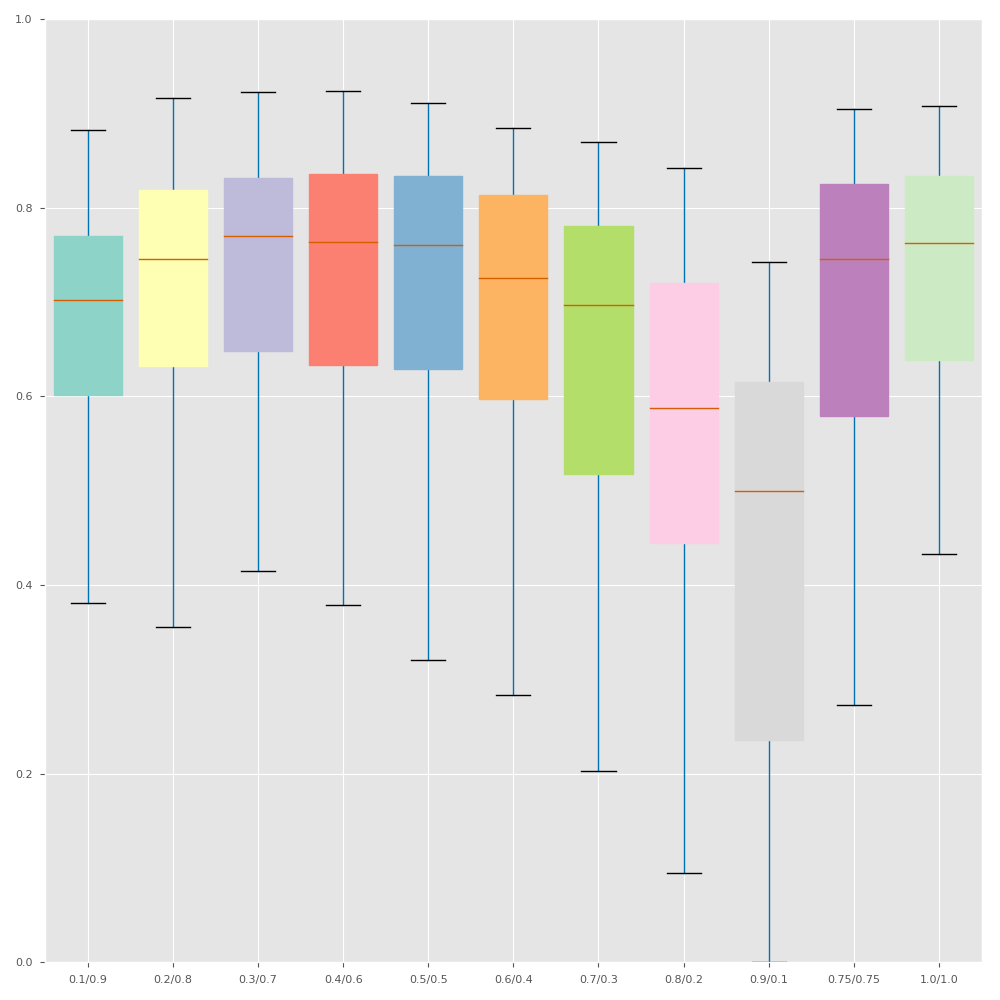}  \\ 

        \rotatebox{90}{\hspace{20pt} F2.0} &
        \includegraphics[width=\linewidth]{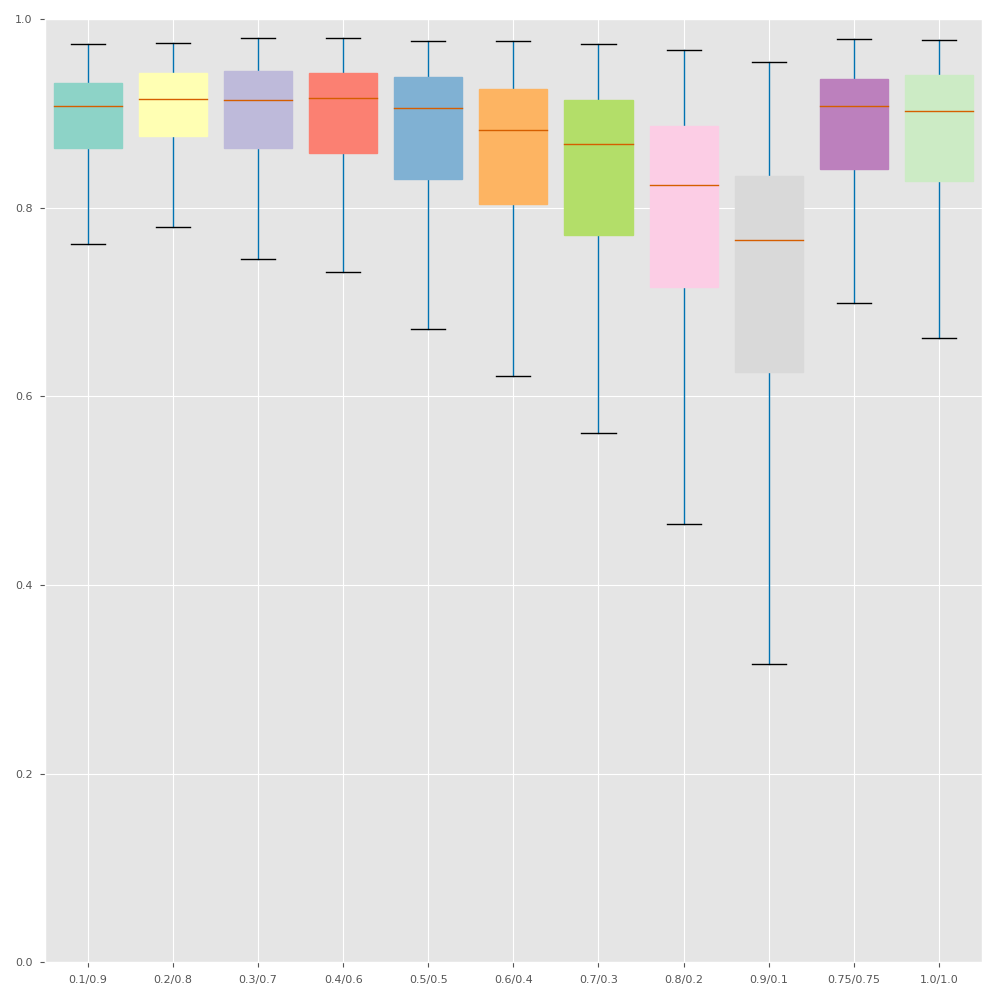} &  
        \includegraphics[width=\linewidth]{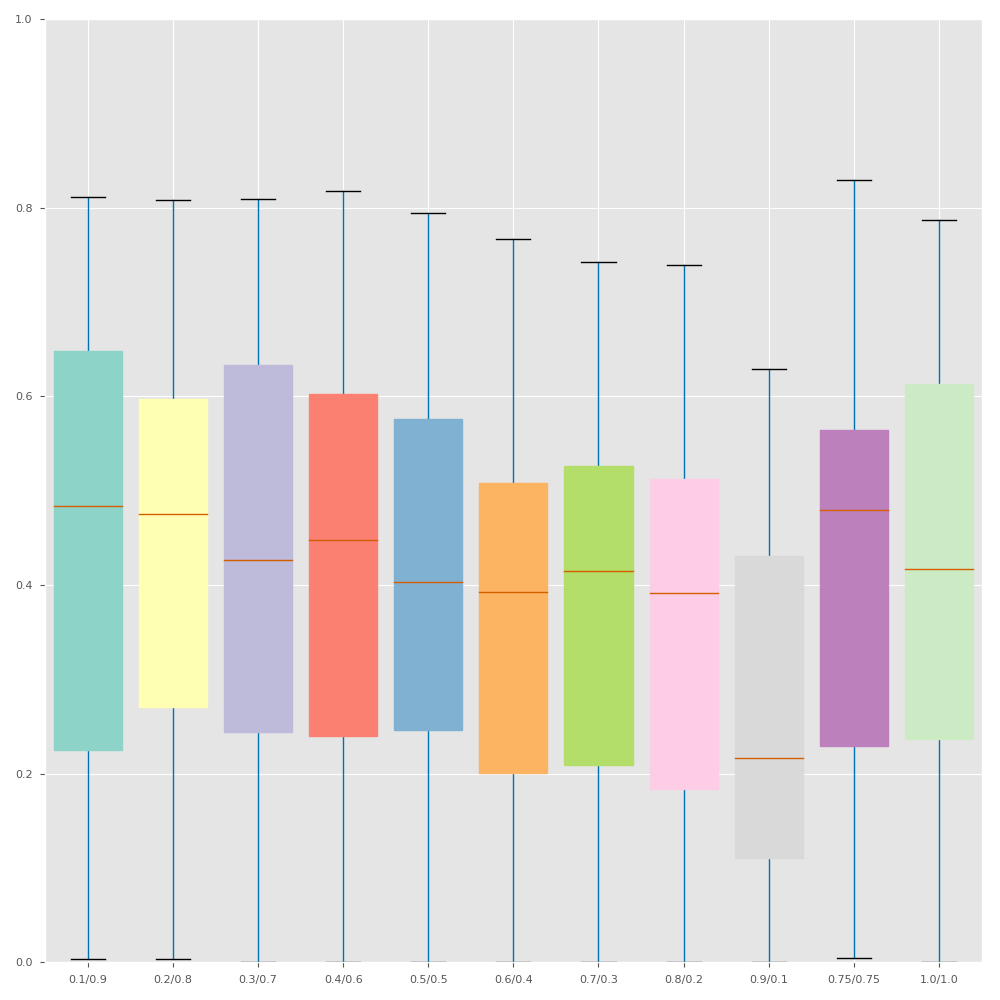} &  
        \includegraphics[width=\linewidth]{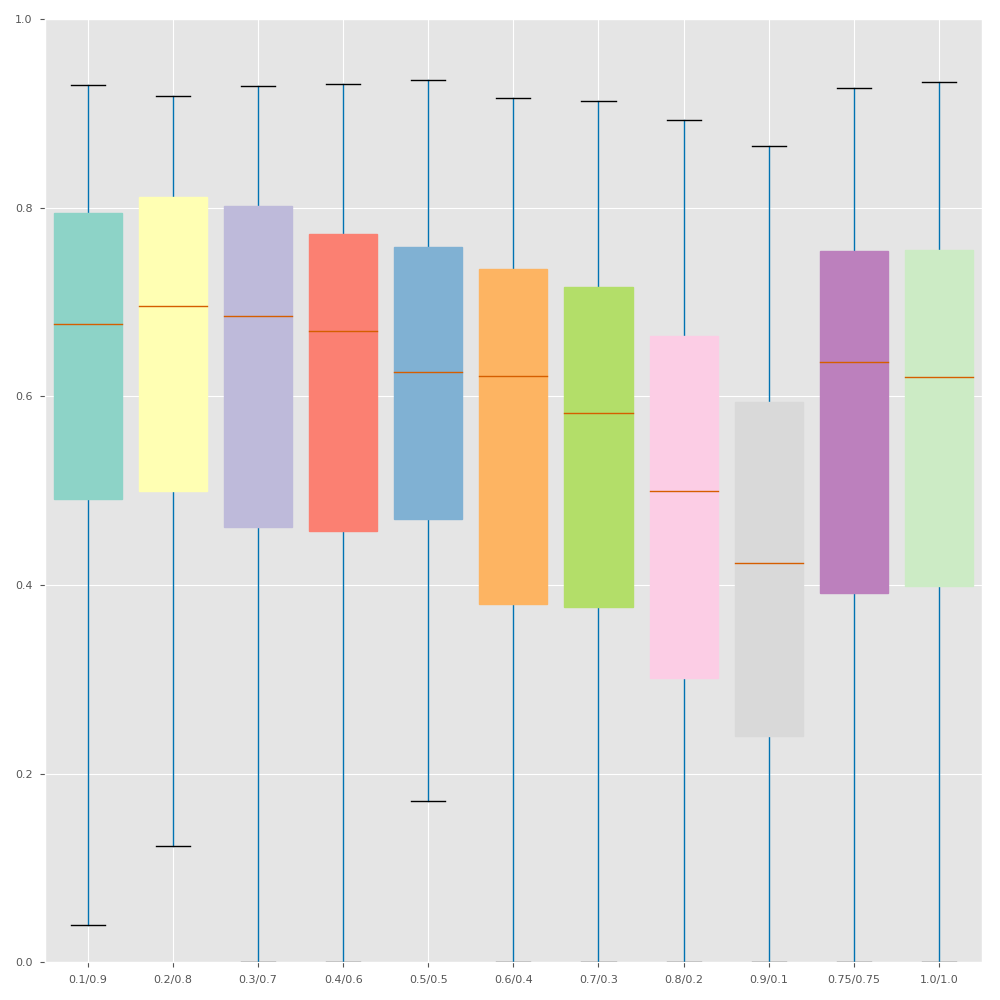} &  
        \includegraphics[width=\linewidth]{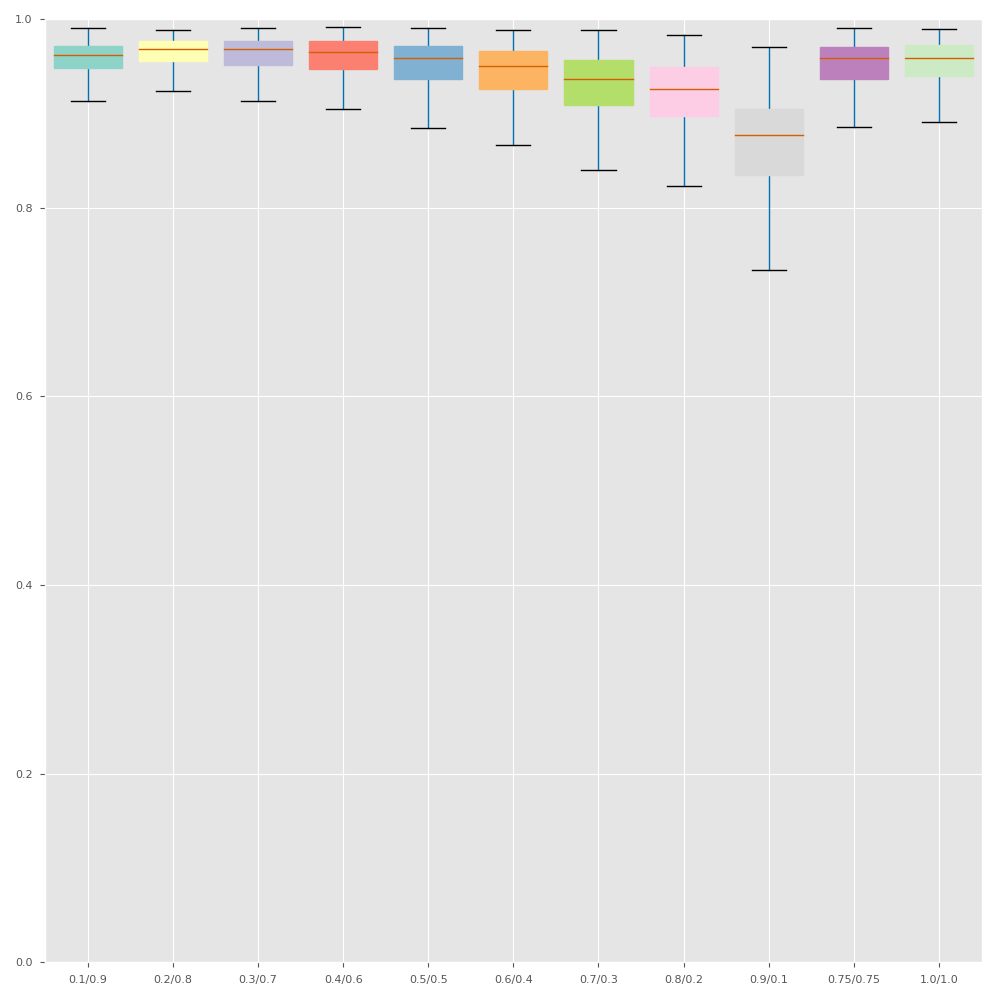} &  
        \includegraphics[width=\linewidth]{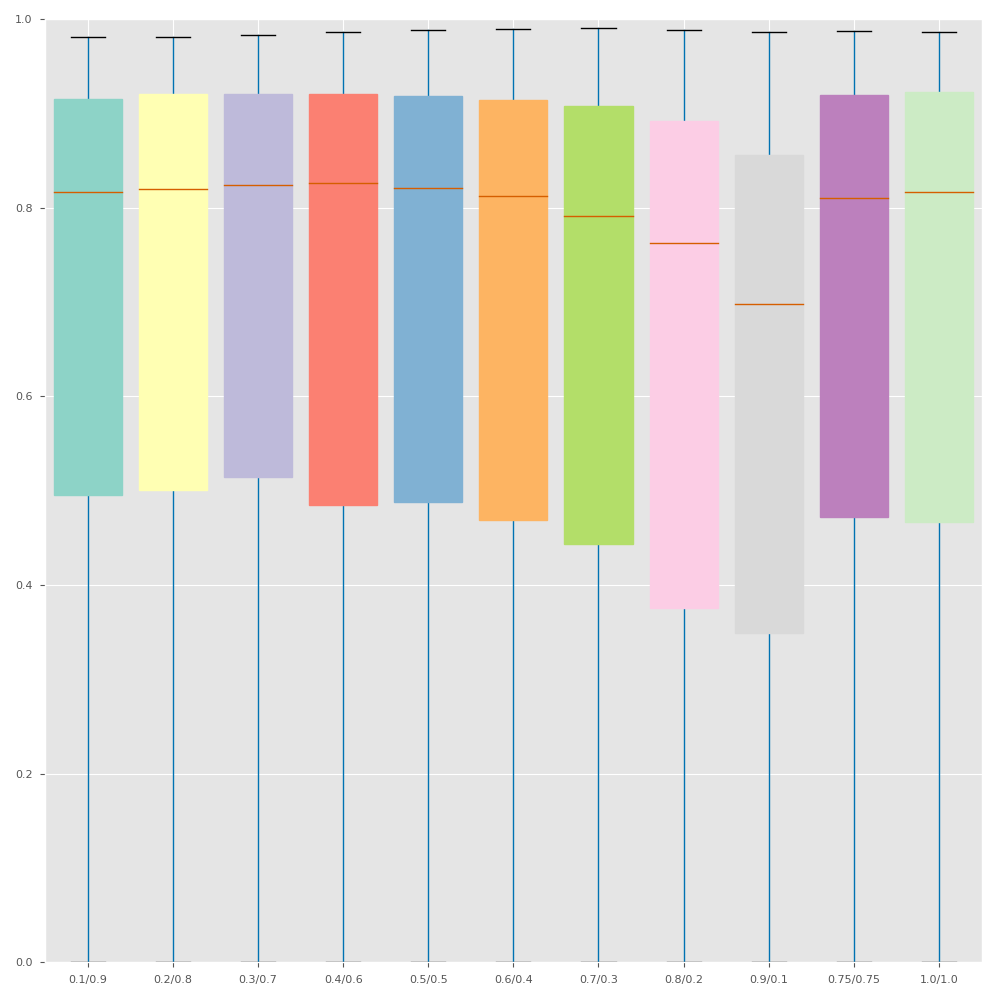} &  
        \includegraphics[width=\linewidth]{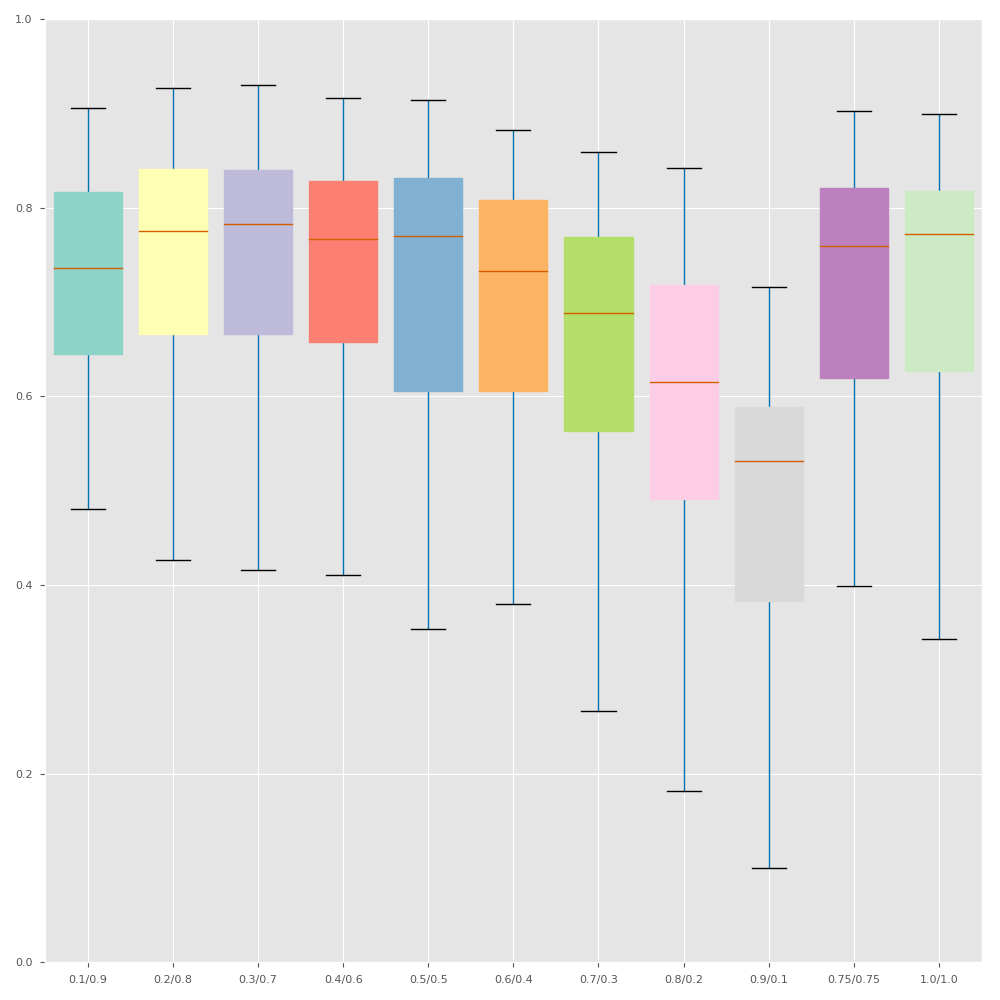} &  
        \includegraphics[width=\linewidth]{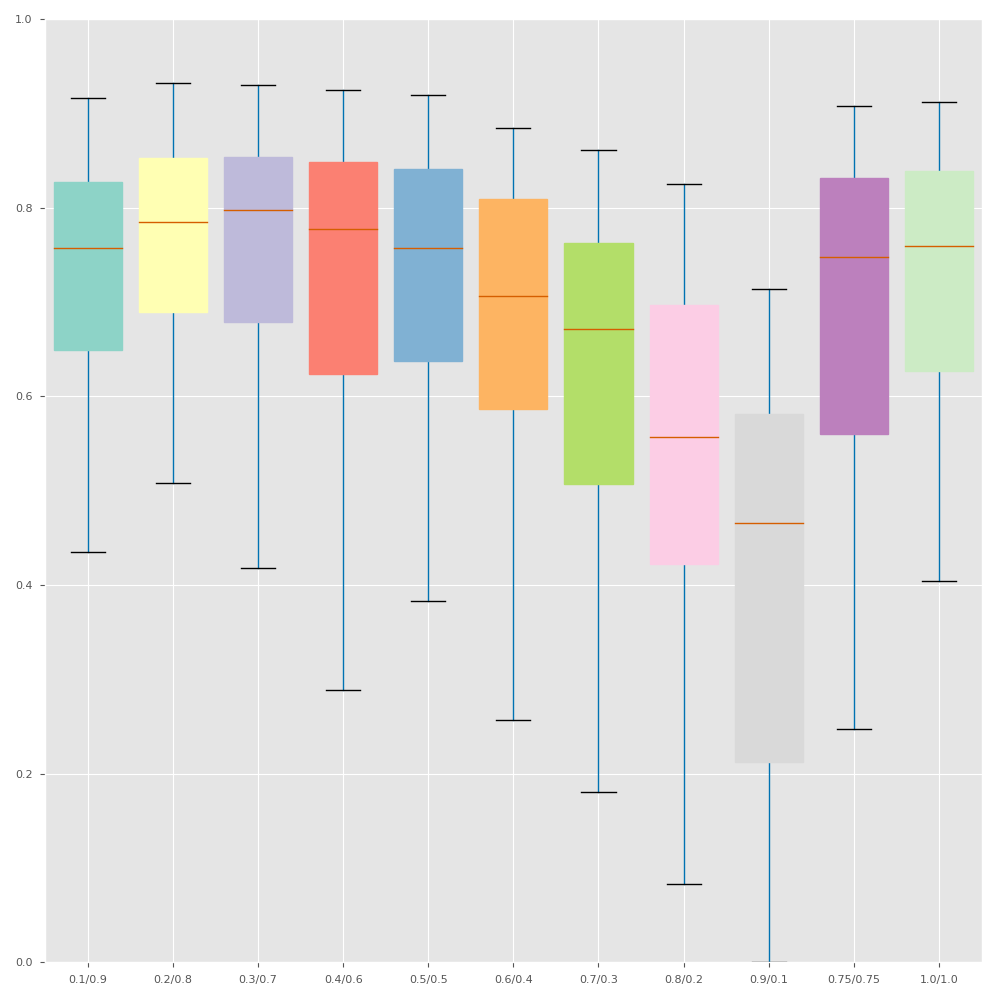}  \\ 

    \end{tabularx}
    }
        \caption{\reviewminorpar Boxplots of the alternative F measures obtained for each dataset using the range of sTversky losses with varying alpha/beta.}
    \label{fig:boxplots_alternative_tversky}
\end{figure*}

\end{document}